\documentclass[letterpaper,12pt,twoside,openright]{book}
\pdfoutput=1

\usepackage{ifthen}

\newboolean{BOOKFORM}
\setboolean{BOOKFORM}{false}

\ifthenelse{\boolean{BOOKFORM}}{

\usepackage[papersize={5.5in,8.5in},verbose,bmargin=1.7cm,rmargin=1.0cm,lmargin=2.1cm,tmargin=1.4cm,headsep=0.3cm,footskip=0.7cm]{geometry}
}
{
\usepackage[papersize={8.5in,11in},left=37mm, right=24mm, top=24mm, bottom=24mm, includefoot, headheight=13.6pt]{geometry}
}

\usepackage[title,titletoc]{appendix}

\usepackage{setspace}

\usepackage[small, bf, tableposition=top]{caption}
\usepackage{fancyhdr}
\makeatletter
\renewcommand{\chaptermark}[1]{\markboth{\textsc{\@chapapp}\  \thechapter:\ #1}{}}
\makeatother
\fancypagestyle{plain}{
\fancyhf{}

\fancyfoot[LE,RO]{\thepage}
}
\makeatletter
\def\cleardoublepage{\clearpage\if@twoside \ifodd\c@page\else
    \hbox{}
    \thispagestyle{plain}
    \newpage
    \if@twocolumn\hbox{}\newpage\fi\fi\fi}
\makeatother \clearpage{\pagestyle{plain}\cleardoublepage}
\renewcommand{\chaptermark}[1]{\markboth{ \emph{#1}}{}}
\usepackage[usenames]{color}

\usepackage{amsmath}
\usepackage{amsthm}%
\usepackage{amsfonts}

\usepackage{amssymb}
\usepackage{amstext}
\usepackage{graphicx}

\usepackage{shadethm}

\usepackage{pifont}  

\usepackage{bbold}

\usepackage{wrapfig}

\usepackage{subfigure}
\usepackage{stmaryrd}   

\usepackage{parskip}

\usepackage{tikz}
\usetikzlibrary{arrows,shapes,decorations,automata,backgrounds,petri}
\usetikzlibrary{shapes.gates.logic.US}
\usepackage[latin1]{inputenc}

\usetikzlibrary{calc,patterns,decorations.pathmorphing,decorations.markings}

\usepackage[bookmarks=true,pdfborder={0 0 0}]{hyperref}

\newshadetheorem{theorem}{Theorem}[chapter]
\newshadetheorem{corollary}{Corollary}[chapter]
\newshadetheorem{lemma}{Lemma}[chapter]
\newtheorem{conjecture}{Conjecture}[chapter]
\theoremstyle{definition}
\newtheorem{example}{Example}[chapter]

\newtheorem{claim}[theorem]{Claim}

\newtheorem{definition}{Definition}[chapter]

\newtheorem{proposition}[theorem]{Proposition}

\newcommand{\Tr}{{\rm Tr}}

\newcommand{\nc}{\newcommand}
\nc{\rnc}{\renewcommand}
\nc{\beq}{\begin{equation}}
\nc{\eeq}{{\end{equation}}}
\nc{\beqa}{\begin{eqnarray}}
\nc{\eeqa}{\end{eqnarray}}
\nc{\lbar}[1]{\overline{#1}}
\nc{\bra}[1]{\langle#1|}
\nc{\ket}[1]{|#1\rangle}
\nc{\ketbra}[2]{|#1\rangle\!\langle#2|}
\nc{\braket}[2]{\langle#1|#2\rangle}
\nc{\proj}[1]{| #1\rangle\!\langle #1 |}
\nc{\avg}[1]{\langle#1\rangle}
\nc{\smfrac}[2]{\mbox{$\frac{#1}{#2}$}}
\nc{\tr}{\operatorname{tr}}
\nc{\tracedist}[1]{\Delta_{}\!\left( #1 \right)}
\nc{\fid}[1]{F\!\left( #1 \right)}

\nc{\ox}{\otimes}
\nc{\dg}{\dagger}
\nc{\dn}{\downarrow}
\nc{\cA}{{\cal A}}
\nc{\cB}{{\cal B}}
\nc{\cC}{{\cal C}}
\nc{\cD}{{\cal D}}
\nc{\cE}{{\mathcal E}}
\nc{\cF}{{\cal F}}
\nc{\cG}{{\cal G}}
\nc{\cH}{{\cal H}}
\nc{\cI}{{\cal I}}
\nc{\cJ}{{\cal J}}
\nc{\cK}{{\cal K}}
\nc{\cL}{{\cal L}}
\nc{\cM}{{\cal M}}
\nc{\cN}{{\cal N}}
\nc{\cO}{{\cal O}}
\nc{\cP}{{\cal P}}
\nc{\cR}{{\cal R}}
\nc{\cS}{{\cal S}}
\nc{\cT}{{\cal T}}
\nc{\cU}{{\cal U}}
\nc{\cX}{{\cal X}}
\nc{\cZ}{{\cal Z}}

\nc{\entI}{{\bf I}}
\nc{\entIarrow}{{\bf I}^{\leftarrow}}
\nc{\entH}{{\bf H}}
\nc{\entS}{{\bf S}}
\nc{\entHmin}{H_{\min}}

\nc{\entF}{{\bf E}_f}

\nc{\isom}{\simeq}

\nc{\rank}{\operatorname{rank}}
\nc{\rar}{\rightarrow}
\nc{\lrar}{\longrightarrow}
\nc{\polylog}{\operatorname{polylog}}
\nc{\poly}{\operatorname{poly}}
\nc{\1}{{\openone}}

\nc{\weight}{\textbf{w}}
\nc{\hamdist}{d_{H}}

\nc{\Sp}{{{\mathbb S}}}
\nc{\RR}{{{\mathbb R}}}
\nc{\CC}{{{\mathbb C}}}
\nc{\FF}{{{\mathbb F}}}
\nc{\NN}{{{\mathbb N}}}
\nc{\ZZ}{{{\mathbb Z}}}
\nc{\PP}{{{\mathbb P}}}
\nc{\QQ}{{{\mathbb Q}}}
\nc{\UU}{{{\mathbb U}}}
\nc{\OO}{{{\mathbb O}}}
\nc{\EE}{{{\mathbb E}}}
\nc{\id}{{\operatorname{id}}}

\nc{\qubitchannel}{\id_2}
\nc{\bitchannel}{\overline{\id}_2}

\nc{\be}{\begin{equation}}
\nc{\ee}{\end{equation}}
\nc{\bea}{\begin{eqnarray}}
\nc{\eea}{\end{eqnarray}}
\nc{\<}{\langle}
\rnc{\>}{\rangle}
\nc{\Hom}[2]{\mbox{Hom}(\CC^{#1},\CC^{#2})}
\nc{\rU}{\mbox{U}}

\nc{\ob}[1]{#1}

\renewcommand{\exp}[1]	{\operatorname{exp}\left( #1 \right)}

\newcommand{\colvec}[1]{\left[\!\!\!\begin{array}{c}#1\end{array}\!\!\!\right]}

\def\mcal{\mathcal}

\newcommand{\ICcap}{\mathcal{C}_{\mathrm{IC}}\!\left(\mathcal{N}\right) }
\def\HKR{\mathcal{R}_{\mathrm{HK}}}
\def\RSato{\mathcal{R}_{\mathrm{Sato}}}

\def\MAConeR{\mathcal{R}^1_{\mathrm{CMG}}}
\def\delMAConeR{\partial\mathcal{R}^1_{\mathrm{CMG}}}
\def\MACtwoR{\mathcal{R}^2_{\mathrm{CMG}}}
\def\delMACtwoR{\partial\mathcal{R}^2_{\mathrm{CMG}}}
\def\bbR{\mathbb{R}}

\def\MAConeRHK{\mathcal{R}^{(o,1)}_{\textrm{HK}}}
\def\MACtwoRHK{\mathcal{R}^{(o,2)}_{\textrm{HK}}}

\def\MACone{QMAC$_1$ }
\def\MACtwo{QMAC$_2$ }
\def\Sasoglou{\c{S}a\c{s}o\u{g}lu}

\usepackage{ifthen}
\newboolean{WITHHMINS}
\setboolean{WITHHMINS}{false}

\newboolean{WITHAPDX}
\setboolean{WITHAPDX}{true}
\def\ifthen#1#2{\ifthenelse{#1}{#2}{} }

\def\PIMACavg{\Pi^n_{\bar{\rho},\delta}}
\def\PIMACone{\Pi^n_{m_1}}
\def\PIMAConepr{\Pi^n_{m_1^{\prime}}}
\def\PIMACtwo{\Pi^n_{m_2}}

\def\PIMAConetwo{\Pi^n_{m_1,m_2}}
\def\PIMAConeprtwo{\Pi^n_{m_1^{\prime},m_2}}
\def\PIMAConetwopr{\Pi^n_{m_1,m_2^{\prime}}}
\def\PIMAConeprtwopr{\Pi^n_{m_1^{\prime},m_2^{\prime}}}

\def\PIone{\Pi_{u^n_1(\ell_1)}}
\def\PIonepr{\Pi_{u^n_1(\ell_1^{\prime})}}

\def\PIsigUone{\Pi_{\sigma_{u^n_1(\ell_1)}} }
\def\PIsigUonePr{\Pi_{\sigma_{u^n_1(\ell_1^\prime)}} }

\def\rhoNulul{\rho_{u^n_1(\ell_1),u^n_2(\ell_2) } }
\def\rhouu{\rho^{B_1B_2}_{u_1,u_2}}
\def\omgu{\omega^{B_1}_{u_1}}
\def\omgNul{\omega^{B_1}_{u_1(\ell_1)}}
\def\omgNulpr{\omega^{B_1}_{u_1(\ell'_1)} }

\def\ExpXtwo{ \mathop{\mathbb{E}}_{X_2^n} }

\def\ExpXone{ \mathop{\mathbb{E}}_{X_1^n} }
\def\ExpX{ \mathop{\mathbb{E}}_{X^n} }

\def\ExpXonetwo{ \mathop{\mathbb{E}}_{\!\!X_1^n,X_2^n} }
\def\PIavg{\Pi^{B^n}_{\bar{\rho}}}

\def\ExpUn{ \mathop{\mathbb{E}}_{U^n} }

\def\indicator{ \ \mathbf{1} }
\def\NExpXtilde{ \mathop{\mathbb{E}}_{\tilde{X}^n} }
\def\NExpX{ \mathop{\mathbb{E}}_{X^n} }

\def\NPrYtildegXtilde{\mathop{\mathrm{Pr}}_{\tilde{Y}^n|\tilde{X}^n }}

\def\NExpYtildegXtildeU{ \mathop{\mathbb{E}}_{\tilde{Y}^n|\tilde{X}^nU^n} }
\def\NPrYtildegXtildeU{\mathop{\mathrm{Pr}}_{\tilde{Y}^n|\tilde{X}^nU^n }}
\def\NExpXtildegU{ \mathop{\mathbb{E}}_{\tilde{X}^n|U^n} }
\def\NExpXgU{ \mathop{\mathbb{E}}_{X^n|U^n} }

\def\prhom{\rho_{x^n(m)}^{B^n} } 
 
\def\pPm{P_{m}^{B^n} } 
\def\pPmpr{P_{m^\prime }^{B^n} }
\def\pPIm{\Pi_{x^n(m)} }

\def\pPIavg{\Pi_{\bar{\rho} }  }
\def\pLAMm{\Lambda_{m}^{B^n} }

\def\prhoMnew{\rho_{U^n,\tilde{X}^n}^{B^n} }

\usepackage{pifont}

\def\dingone{{\small \ding{172}} }
\def\dingtwo{{\small \ding{173} } }
\def\dingthree{{\small \ding{174} } }
\def\dingfour{{\small \ding{175} } }
\def\dingfive{{\small \ding{176} } }
\def\dingsix{{\small \ding{177} } }

\def\mdingone{\textrm{\ding{172}}}
\def\mdingtwo{\textrm{\ding{173}} }
\def\mdingthree{\textrm{\ding{174}} }
\def\mdingfour{\textrm{\ding{175}} }
\def\mdingfive{\textrm{\ding{176}} }
\def\mdingsix{\textrm{\ding{177}} }

\def\CMGR{\mathcal{R}_{\textrm{CMG}}}
\def\HKR{\mathcal{R}_{\textrm{HK}}}
\def\RSato{\mathcal{R}_{\textrm{Sato}}}

\def\delCMGR{\partial\mathcal{R}_{\textrm{CMG}}}
\def\delCMGRpr{\partial^{\prime}\mathcal{R}_{\textrm{CMG}}}
\def\cNpCMG{(\mcal{N},p_{\textrm{CMG}})}
\def\pCMG{p_{\textrm{CMG}}}

\def\shPulse{\,\rule[1mm]{0.7mm}{1pt}}
\def\medPulse{\,\rule[1mm]{2.1mm}{1pt}}
\def\longPulse{\,\rule[1mm]{4.2mm}{1pt}}

\newcommand{\prfsec}[1]{ \noindent{\bf #1:}\ }

\title{ \Large Network information theory for classical-quantum channels}

\author{Ivan Savov}

\begin{document}


\thispagestyle{empty}

\begin{center}
\vspace*{1in}
\LARGE{\textbf{Network information theory \\ for classical-quantum channels}} \ifthenelse{\boolean{BOOKFORM}}{\\[2cm]}{\\[5cm]}
Ivan Savov \\[2cm]
\normalsize{School of Computer Science}\\
\normalsize{McGill University, Montr\'eal}\\
July 2012\\[4.5cm]
A thesis submitted to McGill University in partial fulfillment of the requirements of the degree of Ph.D. \\[1cm]
\normalsize{\copyright Ivan Savov, 2012}\author{}\date{} \\

\end{center}

\pagenumbering{roman}                  


\onehalfspacing


\cleardoublepage

\frontmatter
\pagenumbering{roman}                  


\parindent=0.3in 


\singlespacing

\ \\ 
\noindent
\begin{center}
{\bf \Huge Acknowledgements}
\end{center}
\ \\

	This work would not have been possible without the continued support of my supervisor 	Patrick Hayden. 
	He introduced me to many interesting mathematical research questions 
	at the intersection of quantum physics and computer science.	
	His outstanding abilities as a researcher, teacher and explainer have been 
	an inspiration for me throughout the many years that I have known him. 
	%
	I am also very grateful to Mark M. Wilde for lending me his 
	expertise on all aspects of quantum Shannon theory.
	I would like to thank Omar Fawzi, Pranab Sen, Mai Vu and Saikat Guha for the numerous discussions 
	and their ability to distill complicated mathematical arguments into intuitive explanations.
	%
	I want to thank Olivier Landon-Cardinal, Adriano Ferrari, Grant Salton and Benno Salwey for 
	their help with the preparation of this manuscript. 
	There are many other people who deserve an honourable mention and my gratitude
	for either directly or indirectly
	influencing me: Jan Florjanczyk, Andie Sigler, Eren \Sasoglou,  Gilles Brassard and Andreas Winter.
	I also want to thank my family for supporting me in my scientific endeavours.
	%

\cleardoublepage

\onehalfspacing

\pagebreak
\hspace{2mm}
\begin{center}
{\bf \Huge Abstract}
\end{center}
\hspace{2mm}
\ifthenelse{\boolean{BOOKFORM}}{\singlespacing}{\onehalfspacing}

	Network information theory is the study of communication problems
	involving multiple senders, multiple receivers and intermediate relay stations.
	The purpose of this thesis is to extend the main ideas of 
	classical network information theory to the study of classical-quantum channels.
	We prove coding theorems for the following communication problems:
	quantum multiple access channels, quantum interference channels, 
	quantum broadcast channels and quantum relay channels.
	
	%
	A quantum model for a communication channel describes more accurately the channel's 
	ability to transmit information.
	By using physically faithful models for the channel outputs
	and the detection procedure, we obtain better communication
	rates than would be possible using a classical strategy. 
	In this thesis, we are interested in the transmission of classical information,
	so we restrict our attention to the study of \emph{classical-quantum} channels.
	These are channels with classical inputs and quantum outputs,
	%
	and so the coding theorems we present will use 
	classical encoding and quantum decoding.
	%

	We study the asymptotic regime where many copies of the channel
	are used in parallel, and the uses are assumed to be independent.
	In this context, 
	%
	%
	we can exploit information-theoretic techniques to calculate 
	the maximum rates for error-free communication 
	for any channel, given the statistics of the noise on that channel.	
	%
	These theoretical bounds can be used as a benchmark
	to evaluate 
	the rates achieved by practical communication protocols.

	Most of the results in this thesis consider 
	classical-quantum channels with finite dimensional output systems,
	which are analogous to classical discrete memoryless channels.
	%
	In the last chapter, we will show some applications of our
	results to a practical optical communication scenario,
	in which the information is encoded in continuous quantum 
	degrees of freedom, which are analogous to 
	classical channels with Gaussian noise.
	%


	%
	%
	%
	%
	%
	%
	%
	%
	%
	%
	%
	%

\cleardoublepage

\hspace{2mm}
\noindent
\begin{center}
{\bf \Huge R\'esum\'e}
\end{center}
\hspace{2mm}
\ifthenelse{\boolean{BOOKFORM}}{\singlespacing}{\onehalfspacing}


La th\'{e}orie de l'information multipartite \'{e}tudie les probl\`{e}mes de communication
avec plusieurs \'{e}metteurs, plusieurs r\'{e}cepteurs et des stations relais.  
L'objectif de cette th\`{e}se est d'\'{e}tendre les id\'{e}es centrales de la 
th\'{e}orie de l'information classique \`{a} l'\'{e}tude des canaux quantiques.
Nous allons nous int\'{e}resser aux sc\'{e}narios de communication suivants: les
canaux quantiques \`{a} acc\`{e}s multiples, les canaux quantiques \`{a} interf\'{e}rence, les
canaux quantiques de diffusion et les canaux quantiques \`{a} relais.
Dans chacun des ces sc\'{e}narios de communication, nous caract\'{e}risons les
taux de communication r\'{e}alisables pour l'envoi d'information classique sur ces
canaux quantiques.  

La mod\'{e}lisation quantique des canaux de communication est
importante car elle fournit une repr\'{e}sentation plus pr\'{e}cise de la capacit\'{e} du
canal \`{a} transmettre l'information. En utilisant des mod\`{e}les physiquement
r\'{e}alistes pour les sorties du canal et la proc\'{e}dure de d\'{e}tection, nous obtenons
de meilleurs taux de communication que ceux obtenus dans un mod\`{e}le classique.
En effet, l'utilisation de mesures quantiques collectives sur l'ensemble des
syst\`{e}mes physiques en sortie du canal permet une meilleure extraction
d'information que des mesures ind\'{e}pendantes sur chaque sous-syst\`{e}me.
Nous avons choisi d'\'{e}tudier les canaux \`{a} entr\'{e}e classique et sortie quantique
qui constituent une abstraction utile pour l'\'{e}tude de canaux quantiques
g\'{e}n\'{e}raux o\`{u} l'encodage est restreint au domaine classique. 

%

%
Nous \'{e}tudions le r\'{e}gime asymptotique o\`{u} de nombreuses copies de du canal sont utilis\'{e}es en parall\`{e}le, 
et les utilisations sont ind\'{e}pendantes. 
Dans ce contexte, il est possible de caract\'{e}riser les limites absolues sur la transmission d'information
d'un canal, si on connait les statistiques du bruit sur ce canal. 
Ces r\'{e}sultats th\'{e}oriques peuvent \^{e}tre utilis\'{e}es comme un point de rep\`{e}re 
pour \'{e}valuer la performance des protocoles de communication pratiques.

Nous consid\'{e}rons surtout les canaux o\`{u} les sorties 
sont des syst\`{e}mes quantiques de dimension finie, 
analogues aux canaux classiques discrets.
%
Le dernier chapitre pr\'{e}sente des applications pratiques de 
nos r\'{e}sultats \`{a} la communication optique,
o\`{u} syst\`{e}mes physiques auront des degr\'{e}s de libert\'{e} continus.
Ce contexte est analogue aux canaux classiques avec bruit gaussien.

\cleardoublepage

\singlespacing
\renewcommand\contentsname{Table of Contents}
\tableofcontents

\enlargethispage{\baselineskip}

\cleardoublepage

\parindent=0.3in 

\onehalfspacing

\cleardoublepage

%

%


\ \\ 
\noindent
\begin{center}
{\bf \Huge   \   Notation}
\end{center}
\ \\
\label{notation-page}

{ \small
\begin{center}
\begin{tabular}{rcl}
	{\sc Classical}			& &							{\sc Quantum} \\[4mm]
	$y_a \ \in \ \mcal{Y}$				&$\Longleftrightarrow$	&		$\ket{v}^B  \  \in \  \mcal{H}^B$ \\[-1mm]
	symbol from a finite set 			&					& 		vector in a Hilbert space \\[5mm]
	$p_Y \ \in \ \mcal{P}(\mcal{Y})$		&$\Longleftrightarrow$	&		$\rho^B \  \in \  \mcal{D}(\mcal{H}^B)$ \\[-1mm]
	probability distribution  			&					& 		density matrix $\equiv$ quantum state \\[-1mm]
	$p_Y(y) \geq 0, \ \forall y \in \mcal{Y}$  	& 				& 		$\bra{v}\rho^B\ket{v}\geq 0, \ \forall \ket{v} \in \mcal{H}^B$  \\[0mm]
	$\sum_y p_Y(y) = 1$	  			& 					& 		$\Tr[\rho^B] = 1$,  $\quad (\rho^B)^\dagger = \rho^B$  \\[5mm]
	$p_{Y|X}$						&$\Longleftrightarrow$	&		$\{ \rho_x^B \}, x \in \mcal{X}$ \\[-1mm]
	conditional probability distribution	&					& 		conditional states \\[-1mm]
	$\equiv$ classical-classical channel &					& 		$\equiv$ classical-quantum channel \\[5mm]
	$p_{XY}(x,y) \equiv p_X(x)p_{Y|X}(y|x)$	&$\Longleftrightarrow$	&	$\theta^{XB} \equiv \sum_x p_X(x)\;\ketbra{x}{x}^X \otimes \rho_x^B$ \\[-1mm]
	joint input-output distribution		&					& 		joint input-output state \\[5mm]
	$p_{\bar{Y}} \equiv \mathop{\mathbb{E}}_{X} p_{Y|X}$
								&$\Longleftrightarrow$	&
												$\bar{\rho}^B \equiv \mathop{\mathbb{E}}_{X} \rho_X^B$ \\[-1mm]
	average output distribution  		&					& 		average output state \\[5mm]
	$\indicator_{ \left\{  y^n \in \mathcal{T}^{(n)}_\delta\!(\bar{Y})  \right\} }$
								&$\Longleftrightarrow$	&
												$\Pi_{\bar{\rho}} \equiv \Pi_{\bar{\rho}^{\otimes n}, \delta}^{B^n}$ \\[0mm]
	indicator function for 			&					& 		projector onto the \\[-1mm]
	the output-typical set  			&					& 		output-typical subspace \\[5mm]
	$\indicator_{ \left\{  y^n \in \mathcal{T}^{(n)}_\delta\!(Y|x^n)  \right\} }$
								&$\Longleftrightarrow$	&
												$\Pi_{x^n} \equiv \Pi_{\rho_{x^n}, \delta}^{B^n}$ \\[0mm]
	indicator function for the  			&					& 		conditionally typical \\[-1mm]
	conditionally typical set 	 		&					& 		projector for the state $\rho^{B^n}_{x^n}$ \\[5mm]

	%
\end{tabular}

\end{center}
}

\clearpage

\onehalfspacing
\mainmatter

\fancyhead[LE]{{\small\emph{\thesection \  \rightmark}}}
\fancyhead[RO]{{\small\emph{Chapter \thechapter}: \leftmark}}
\pagestyle{fancy}
\pagenumbering{arabic}
\setcounter{page}{1}

\tikzset{
	cnode/.style={isosceles triangle,
				isosceles triangle apex angle=60,thick,
				draw=blue!75, fill=blue!20,minimum height=6mm,inner sep=0.3mm},
	qnode/.style={circle,thick,draw=blue!75,fill=blue!20,minimum size=6mm,inner sep=0.3mm},
	every label/.style=	{black, font=\footnotesize}
}



%
%


\chapter{Introduction}

    The central theme of this work is the transmission of information
    through noisy communication channels.
    The word \emph{information} means different things to different people,
    so it is worthwhile to begin the discussion with a clear definition of the term.
    Statements like 
    ``Canada has an information-based economy''
    suggest that information is some kind of commodity that can be 
    shipped on trains for export like oil or lumber.
    In the world of digital electronics, the word information is used 
    as a synonym for the word \emph{data}
    as in ``How much information can you store on your USB memory stick?''.
    In that context, most people would say that a 7MB \texttt{mp3} file
    contains just as much information as a 7MB file full of zeros.

    In this work we will use the term \emph{information} in the sense
    originally defined by Claude Shannon \cite{S48}.
    %
    Shannon realized that in order to study the problems of information
    storage and information transmission mathematically, 
    we must step away from the \emph{semantics} of the 
    messages and focus on their \emph{statistics}.
    Using the notions of entropy, conditional entropy and 
    mutual information,
    we can \emph{quantify} the information
    content of data sources and the information transmitting abilities of 
    noisy communication channels.

	We can arrive at an \emph{operational} interpretation of the 
	information content of a data source in terms of our ability to compress it.
	The more unpredictable the content of the data is,
	the more information it contains.
	Indeed, if we use \texttt{WinZip} to compress the \texttt{mp3} file
	and the file full of zeros, we will see that the latter 
	will result in a much smaller \texttt{zip} file,
	which is expected since a file full of zeros has less
	uncertainty and, by extension, contains less information.

	We can similarly give an operational interpretation of
	the information carrying capacity of a noisy communication channel
	in terms of our ability to convert it into a noiseless channel.
	Channels with more noise have a smaller capacity for carrying information.
	Consider a channel which allows us to send data at the rate
	of 1 MB/sec on which half of the packets sent get lost 
	due to the effects of noise on the channel.
	It is not true that the capacity of such a channel is 1 MB/sec,
	because we also have to account for the need to 
	retransmit lost packets.
	In order to correctly characterize the information carrying 
	capacity of a channel, we must consider the rate of
	the end-to-end \emph{protocol} which 
	converts many uses 	of the noisy channel into an effectively noiseless 
	communication channel.
	%

	%

%
	%

\section{Information theory}

		\begin{wrapfigure}{r}{0pt}%
		\begin{tikzpicture}[node distance=2.0cm,>=stealth',bend angle=45,auto,scale=3]
		  \begin{scope}
			\node [cnode] (Tx) [ label=left:Tx    ]                            {\small $x$} ; 
			\node [cnode] (Rx) [ label=right:Rx, right of=Tx, xshift=7mm] {\small $y$} 
				edge  [pre]             node[swap]  {\!\!$\lightning \equiv p(y|x)$}    (Tx) ;
		  \end{scope}
		  \begin{pgfonlayer}{background}
		    \filldraw [line width=4mm,join=round,black!10]
		      ([xshift=-1mm,yshift=+1mm]Tx.north -| Tx.east) rectangle (Tx.south -| Rx.west);
		  \end{pgfonlayer}
		\end{tikzpicture}
		\caption{\small A point-to-point channel 
				$\equiv p_{Y|X}(y|x)$. }
		\label{fig:cp-to-p}
		\end{wrapfigure}
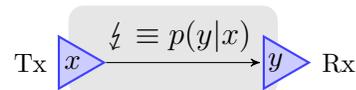
	

		Information theory  studies models of communication which are amenable
		to mathematical analysis.
		In order to model the effects of noise ($\lightning$) in a point-to-point communication 
		scenario, we represent the inputs and outputs of the channel probabilistically.
		We describe the channel as a triple $(\mcal{X}, p_{Y|X}(y|x), \mcal{Y})$,
		where $\mcal{X}$ is the set of possible symbols that the Transmitter (Tx) can send,
		$\mcal{Y}$ is the set of possible outputs that the Receiver (Rx) can obtain
		and $p_{Y|X}(y|x)$ is a conditional probability distribution describing the
		channel's transition probabilities.
		This model is illustrated in Figure~\ref{fig:cp-to-p}, where 
		random variables are
		pictorially represented as small triangles (\! %
		\begin{tikzpicture}[node distance=1cm,>=stealth',bend angle=45,auto, scale=0.5]
		  \tikzstyle{cnode}=[isosceles triangle,isosceles triangle apex angle=60,thick,draw=blue!75,
		  fill=blue!20,minimum height=4mm] 
		  \tikzstyle{qnode}=[circle,thick,draw=blue!75,fill=blue!20,minimum size=5mm] 
		  \tikzstyle{every label}=[black, font=\footnotesize]
		  \begin{scope}
			\node [cnode,scale=0.5,yshift=+4mm] (Tx) [   ]                            {} ;
		  \end{scope}
		  \end{tikzpicture}).
		%
		%
		For example, the noiseless binary channel 
		is represented as the triple $(\{0,1\}, p_{Y|X}(y|x)=\delta(x,y), \{0,1\} )$.
		%
		Using this model of the channel, it is possible to calculate  
		the optimal communication rates from Transmitter to Receiver 
		in the limit of many independent uses of the channel \cite{S48}. 
		These theoretical results have wide-reaching applications in
		many areas of communication engineering but also in other
		fields like cryptography, computer science, neuroscience
		and even economics.
		So long as a probabilistic model for the channel at hand 
		is available, we can use this model and the techniques of 
		information theory to arrive at precise mathematical statements
		about its suitability for a given communication task
		in the limit of many uses of the channel.

\section{Network information theory}

	Network information theory is the extension of Shannon's 
	model of noisy channels to communication scenarios with 
	multiple senders and multiple receivers \cite{el1980multiple,CT91,el2010lecture}.
	To model these channels probabilistically, we use multivariate conditional probability 
	distributions. 
	Some of the most important problems in network information theory
	are shown in Figure~\ref{fig:cnit}, and the relevant class of probability distributions
	is also indicated.
	%
	%
%

\begin{figure}[htb]
\begin{center}
\subfigure[MAC $\equiv p(y|x_1,x_2)$]{
\hspace{-10mm}
\begin{tikzpicture}[node distance=1.8cm,>=stealth',bend angle=45,auto,xscale=0.8,yscale=0.8]
  \begin{scope}
	\node [cnode] (Tx1) [ label=left:Tx1     ]                            {\footnotesize $x_1$} ; 
	\node [cnode] (Tx2) [ label=left:Tx2, below of=Tx1]                {\footnotesize $x_2$}; 
	\node [cnode] (Rx) [ label=right:Rx, right of=Tx1,yshift=-10mm] {\footnotesize $y$} 
		edge  [pre]             node[swap]  {$\lightning$}    (Tx1)
		edge  [pre]             node		   {$\lightning$}         (Tx2) ;
  \end{scope}
  \begin{pgfonlayer}{background}
    \filldraw [line width=4mm,join=round,black!10]
      ([xshift=-3mm]Tx1.north -| Tx1.east) rectangle ([xshift=+1mm]Tx2.south -| Rx.west);
  \end{pgfonlayer}
\end{tikzpicture}
}%
\subfigure[BC $\equiv p(y_1,y_2|x)$]{
\!\!\!\!\!\!\!
\begin{tikzpicture}[node distance=2.0cm,>=stealth',bend angle=45,auto,xscale=0.8,yscale=0.8]

  \begin{scope}
	\node [cnode] (Tx) [ label=left:Tx,yshift=-10mm   ]                            {$x$};
	\node [cnode] (Rx1) [ label=right:Rx1, right of=Tx,yshift=+10mm]	{\footnotesize $y_1$}
		edge  [pre]             node[swap]  		{$\lightning$}    (Tx);
	\node [cnode] (Rx2) [ label=right:Rx2, right of=Tx,yshift=-10mm] { $y_2$}
		edge  [pre]             node		   {$\lightning$}         (Tx) ;
  \end{scope}
  \begin{pgfonlayer}{background}
    \filldraw [line width=4mm,join=round,black!10]
      (Rx1.north -| Tx.east) rectangle (Rx2.south -| Rx2.west);
  \end{pgfonlayer}

\end{tikzpicture}
}\ifthenelse{\boolean{BOOKFORM}}{ }{}%
\subfigure[IC $\equiv p(y_1,y_2|x_1,x_2)$]{
\!\!\!\!\!\!
\begin{tikzpicture}[node distance=2.0cm,>=stealth',bend angle=45,auto,xscale=0.8,yscale=0.8]

  \begin{scope}
	\node [cnode] (ICTx1) [ label=left:Tx1   ]                            {\footnotesize $x_1$};
	\node [cnode] (ICTx2) [ label=left:Tx2, below of=ICTx1]		{\footnotesize $x_2$};
	\node [cnode] (ICRx1) [ label=right:Rx1, right of=ICTx1]	{\footnotesize $y_1$}
		edge  [pre]             		node[swap]  	{$\lightning$}	(ICTx1)
		edge  [pre,draw=red]	node[yshift=-1.5mm]	  	{$\lightning$}	(ICTx2);
	\node [cnode] (ICRx2) [ label=right:Rx2, right of=ICTx2] {\footnotesize $y_2$}
		edge  [pre,draw=red]	node[swap,yshift=+1mm]	  	{$\lightning$}	(ICTx1)
		edge  [pre]             node		   {$\lightning$}		(ICTx2) ;
  \end{scope}
  \begin{pgfonlayer}{background}
    \filldraw [line width=4mm,join=round,black!10]
      (ICTx1.north -| ICTx1.east) rectangle (ICRx2.south -| ICRx2.west);
  \end{pgfonlayer}

\end{tikzpicture}
}%
\subfigure[RC $\equiv p(y_1,y|x,x_1)$]{
\!\!\!\!\!\!\!
\begin{tikzpicture}[node distance=1.9cm,>=stealth',bend angle=45,auto,scale=0.5]

  \tikzstyle{cnode}=[isosceles triangle,isosceles triangle apex angle=60,thick,draw=blue!75,fill=blue!20,minimum height=6.5mm,inner sep=0mm] 
  \tikzstyle{qnode}=[circle,thick,draw=blue!75,fill=blue!20,minimum size=6mm] 

  \tikzstyle{every label}=[black, font=\footnotesize]

  \begin{scope}
	\node [cnode] (RCTx) [ label=left:Tx,yshift=-20mm   ]                            {$x$};
	\node [cnode] (RCRERx) [label=above:$\ \ \ \ \textrm{Re}$, above of=RCTx,xshift=+8mm ]		{$y_1$}
		edge  [pre]             		node[swap]	  	{$\lightning$}	(RCTx);
	\node [cnode] (RCRETx) [above of=RCTx,xshift=+13mm ]		{$x_1$};
	\node [cnode] (RCRx)  [ label=right:Rx, below of= RCRETx,xshift=+15mm]	{$y$}
		edge  [pre]             		node	  	{$\lightning$}	(RCTx)
		edge  [pre]				node[swap]  	{$\lightning$}	(RCRETx);
  \end{scope}
  \begin{pgfonlayer}{background}
    \filldraw [line width=4mm,join=round,black!10]
      ([xshift=2mm]RCRERx.south -| RCTx.north) rectangle (RCRx.south -| RCRx.west);
  \end{pgfonlayer}
  
\end{tikzpicture}%
}%
\end{center}%

\caption{\small Classical network information theory studies communication channels
with multiple senders and multiple receivers.
These include, among others, 
(a) multiple access channels (MACs), 
(b) broadcast channels (BCs),
(c)~interference channels (ICs),
and (d) relay channels (RCs).
}
\label{fig:cnit}
\end{figure}
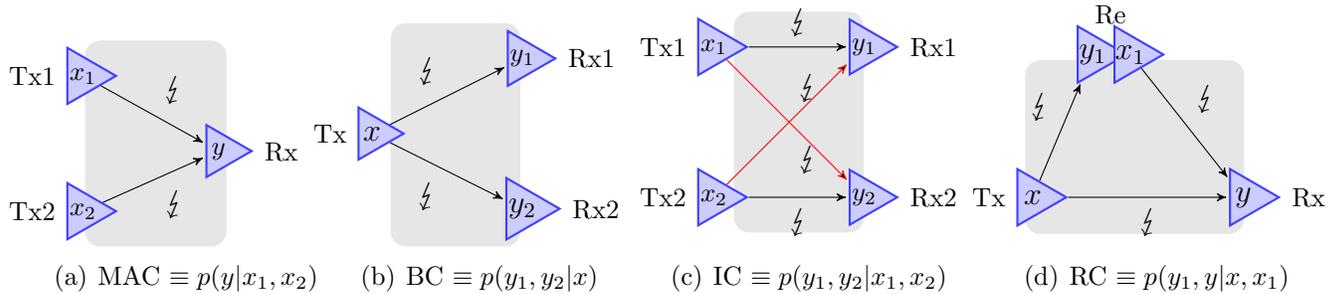

Each of the above channels is a model for some practical 
communication scenario.
%
In the multiple access channel, there are multiple transmitters 
trying to talk to a single base station,
and we can describe the tradeoff between the communication rates
that are achievable for the inbound communication links.
%
The broadcast channel is the dual problem 
in which a single transmit antenna emits 
multiple information streams intended for different receivers.
We can additionally have a \emph{common} information stream intended for both receivers.
%
Coding strategies for broadcast channels involve 
encodings that  can ``mix'' the information streams to
produce the transmit signal. 
%
Interference channels model situations where multiple independent
transmissions are intended, but \emph{crosstalk} occurs
because the  communication takes place in a shared medium.
%
The relay channel is a \emph{multi-hop} information network.
The Relay is assumed to decode the message during one block
of uses of the channel and re-transmit the information it has decoded
during the next block.
This allows the Receiver to collectively decode the information 
from both the Transmitter and the Relay
and achieve better communication rates than
what would be possible with point-to-point codes.

\section{Quantum channels}

	Classical models are not adequate for the characterization of the information 
	carrying capacity of 
	communication systems in which the information carriers are quantum systems.
	Such systems need not be exotic: in optical communication links, 
	the carriers are photons,
	which are properly described by quantum electrodynamics and only approximately described 
	by Maxwell's equations.
	A more general model for communication channels is one which
	takes into account the underlying laws of physics
	concerning the encoding, transmission and decoding of information using quantum systems.
	%
	Quantum decoding based on \emph{collective} measurements of all the channel outputs in parallel
	can be shown to achieve higher communication rates compared to 
	classical decoding strategies in which the channel outputs are measured individually.
	%

	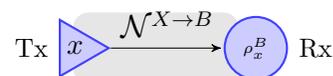
\begin{wrapfigure}{r}{0pt}%
	\begin{tikzpicture}[node distance=2.0cm,>=stealth',bend angle=45,auto,scale=3]
  \tikzstyle{cnode}=[isosceles triangle,isosceles triangle apex angle=60,thick,draw=blue!75,fill=blue!20,minimum height=6.5mm,inner sep=0mm] 
%
	  \tikzstyle{qnode}=[circle,thick,draw=blue!75,fill=blue!20,minimum size=6mm] 
	  \tikzstyle{every label}=[black, font=\footnotesize]

	  \begin{scope}
		\node [cnode] (Tx) [ label=left:Tx    ]                            {\small $x$} ; 
		\node [qnode] (Rx) [ label=right:Rx, right of=Tx,xshift=4mm] {\tiny $\rho_x^B$} 
			edge  [pre]             node[swap]  {$\mcal{N}^{X\to B}$}    (Tx) ;
	  \end{scope}
	  \begin{pgfonlayer}{background}
	    \filldraw [line width=4mm,join=round,black!10]
	      ([xshift=-1mm]Tx.north -| Tx.east) rectangle (Tx.south -| Rx.west);
	  \end{pgfonlayer}
	\end{tikzpicture}
	\caption{\small A point-to-point classical-quantum 
				channel
			$\{ \rho_x \}$.}
	\label{fig:cq-channel-diagram}
	\end{wrapfigure}
		
	Of particular interest are \emph{classical-quantum} channel models,
	which model the sender's inputs as classical variables
	and the receiver's outputs as quantum systems.
	A classical-quantum channel 
	$(\mcal{X}, \mcal{N}^{X\to B}(x)\!\equiv\! \rho^B_x, \ \mcal{H}^B)$
	is fully specified by the finite set of output states $\{ \rho^B_x \}$
	it produces for each of the possible inputs $x \in \mcal{X}$.
	%
	%
	%
	%
	%
	%
	%
	Figure~\ref{fig:cq-channel-diagram} depicts a
	classical-quantum channel, in which the quantum output system 
	is represented by a circle: 
		 	 \begin{tikzpicture}[node distance=2.0cm,>=stealth',bend angle=45,auto, scale=0.5]
			  \tikzstyle{cnode}=[isosceles triangle,isosceles triangle apex angle=60,thick,draw=blue!75,
			  fill=blue!20,minimum height=3mm] 
			  \tikzstyle{qnode}=[circle,thick,draw=blue!75,fill=blue!20,minimum size=5mm] 
			  \tikzstyle{every label}=[black, font=\footnotesize]
			  \begin{scope}
				\node [qnode,scale=0.6] (Tx) [   ]                            {} ;
			  \end{scope}
			  \end{tikzpicture}.
	Such channels 
	form a useful abstraction for studying the 
	transmission of classical data over quantum channels.
	The  Holevo-Schumacher-Westmoreland (HSW) Theorem (see 
	page \pageref{thm:HSWtheorem})
	establishes 
	the maximum achievable communication rates for classical-quantum channels 
	\cite{H98,SW97}.

	Note that a classical-quantum (\emph{c-q}) channel corresponds to the 
	use of a quantum-quantum (\emph{q-q}) channel in which the sender is 
	restricted to selecting from a finite set of signalling states.
	Any code construction for a \emph{c-q} channel can be augmented 
	with an optimization over the choice of signal states 
	to obtain a code for a \emph{q-q} channel.
	For this reason, we restrict our study here to that of \emph{c-q} channels.
	%
	%

		%
		%
		%
		The study of quantum channels finds practical applications in optical communications.
		Bosonic channels model the quantum aspects of optical communication links.
		%
		It is known that optical receivers based on collective quantum measurements 
		of the channel outputs
		outperform classical strategies,
		particularly in the low-photon-number regime \cite{GGLMSY04,guha2011structured,guha2012explicit}.
		In other words, quantum measurements are \emph{necessary} to
		achieve their ultimate information carrying capacity. 
		%
		In \cite{GGLMSY04} it is also demonstrated that
		classical encoding is \emph{sufficient} 
		to achieve the Holevo capacity of the lossy bosonic channel,
		giving further motivation for the theoretical study of classical-quantum models.

\section{Research contributions}

	This thesis presents a collection of results for  
	problems in network information theory for classical-quantum channels.
	As we stated before, the results here easily extend to quantum-quantum channels.
	The problems considered are illustrated in Figure~\ref{fig:qnit}.
	%

\begin{figure}[htbp]
\begin{center}

\subfigure[QMAC $\equiv \left\{ \rho^B_{x_1,x_2} \right\}$ ]{
\hspace{-9mm}
\begin{tikzpicture}[node distance=1.8cm,>=stealth',bend angle=45,auto]
  \begin{scope}
	\node [cnode] (Tx1) [ label=left:Tx1     ]                            {\footnotesize $x_1$} ; 
	\node [cnode] (Tx2) [ label=left:Tx2, below of=Tx1]                {\footnotesize $x_2$}; 
	\node [qnode] (Rx) [ label=right:Rx, right of=Tx1,yshift=-10mm] {\scriptsize $\rho^B_{x_1,x_2}$} 
		edge  [pre]             node[swap]  {$\lightning$}    (Tx1)
		edge  [pre]             node		   {$\lightning$}         (Tx2) ;
  \end{scope}
  \begin{pgfonlayer}{background}
    \filldraw [line width=4mm,join=round,black!10]
      ([xshift=-3mm]Tx1.north -| Tx1.east) rectangle ([xshift=+3mm]Tx2.south -| Rx.west);
  \end{pgfonlayer}
\end{tikzpicture}
}%
\subfigure[QIC $\equiv \left\{ \rho^{B_1B_2}_{x_1,x_2} \right\}$ ]{
\!\!\!\!
\begin{tikzpicture}[node distance=1.8cm,>=stealth',bend angle=45,auto]

  \begin{scope}
	\node [cnode] (ICTx1) [ label=left:Tx1   ]                            {\footnotesize $x_1$};
	\node [cnode] (ICTx2) [ label=left:Tx2, below of=ICTx1]		{\footnotesize $x_2$};
	\node [qnode] (ICRx1) [ label=right:Rx1, right of=ICTx1]	{\tiny $\rho^{B_1}_{x_1,x_2}$}
		edge  [pre]             		node[swap]  	{$\lightning$}	(ICTx1)
		edge  [pre,draw=red]								(ICTx2);
	\node [qnode] (ICRx2) [ label=right:Rx2, right of=ICTx2] {\tiny $\rho^{B_2}_{x_1,x_2}$}
		edge  [pre,draw=red]								(ICTx1)
		edge  [pre]             node		   {$\lightning$}		(ICTx2) ;
  \end{scope}
  \begin{pgfonlayer}{background}
    \filldraw [line width=4mm,join=round,black!10]
      ([xshift=-3mm,yshift=+2mm]ICTx1.north -| ICTx1.east) rectangle ([xshift=+3mm]ICRx2.south -| ICRx2.west);
  \end{pgfonlayer}

\end{tikzpicture}
}\ifthenelse{\boolean{BOOKFORM}}{ }{}%
\subfigure[QBC $\equiv \left\{ \rho^{B_1B_2}_{x} \right\}$ ]{
\!\!\!\!\!
\begin{tikzpicture}[node distance=1.8cm,>=stealth',bend angle=45,auto]

  \begin{scope}
	\node [cnode] (Tx) [ label=left:Tx,yshift=-10mm   ]                            {\footnotesize $x$};
	\node [qnode] (Rx1) [ label=right:Rx1, right of=Tx,yshift=+10mm]	{\scriptsize $\rho^{B_1}_{x}$}
		edge  [pre]             node[swap]  		{$\lightning$}    (Tx);
	\node [qnode] (Rx2) [ label=right:Rx2, right of=Tx,yshift=-10mm] {\scriptsize $\rho^{B_2}_{x}$}
		edge  [pre]             node		   {$\lightning$}         (Tx) ;
  \end{scope}
  \begin{pgfonlayer}{background}
    \filldraw [line width=4mm,join=round,black!10]
      ([xshift=-2mm]Rx1.north -| Tx.east) rectangle ([xshift=+2mm]Rx2.south -| Rx2.west);
  \end{pgfonlayer}

\end{tikzpicture}
}%
\subfigure[QRC $\equiv \left\{ \rho^{B_1B}_{x,x_1} \right\}$ ]{
\!\!\!\!\!
\begin{tikzpicture}[node distance=1.8cm,>=stealth',bend angle=45,auto]

  \begin{scope}
	\node [cnode] (RCTx) [ label=left:Tx,yshift=-20mm   ]                            {\footnotesize $x$};
	\node [qnode] (RCRERx) [ label=above:$\ \ \ \ \textrm{Re}$, above of=RCTx,xshift=+8mm ]		{\tiny $\rho^{B_1}_{x,x_1}$}
		edge  [pre]             		node[swap]	  	{$\lightning$}	(RCTx);
	\node [cnode] (RCRETx) [above of=RCTx,xshift=+14mm ]		{\footnotesize $x_1$};
	\node [qnode] (RCRx)  [ label=right:Rx, below of=RCRERx,xshift=+15mm]	{\tiny $\rho^{B}_{x,x_1}$}
		edge  [pre]             		node	  	{$\lightning$}	(RCTx)
		edge  [pre]				node[swap]  	{$\lightning$}	(RCRETx);
  \end{scope}
  \begin{pgfonlayer}{background}
    \filldraw [line width=4mm,join=round,black!10]
      ([xshift=+1mm,yshift=+2mm]RCRERx.south -| RCTx.north) rectangle ([xshift=+2.5mm]RCRx.south -| RCRx.west);
  \end{pgfonlayer}
  
\end{tikzpicture}
}

\end{center}
\caption{\small Network information theory can be extended to
channels with quantum outputs.
We call these ``classical-quantum channels,''
and consider the following communication scenarios:
(a) multiple access channels (QMACs), 
(b) interference channels (QICs),
(c) broadcast channels (QBCs),
and (d) relay channels (QRCs).
}
\label{fig:qnit}
\end{figure}
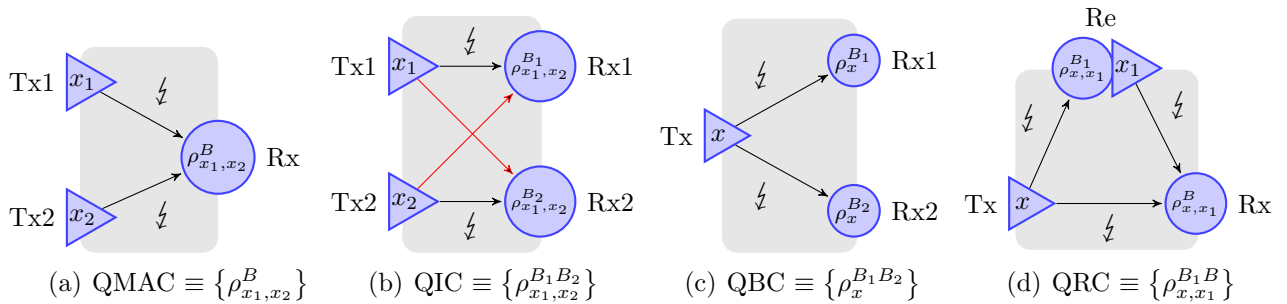

	Most of the results presented in this thesis have appeared in publication.
	%
	The new results on the quantum multiple access channel and 
	the quantum interference channel appeared in
	\cite{FHSSW11}, which is a collaboration between 
	Omar Fawzi, Patrick Hayden, Pranab Sen, Mark M. Wilde
	and the present author.
	That paper has been accepted for publication in the \emph{IEEE Transactions on Information Theory}.
	A more compact version of the same results was presented by the author
	at the 2011 \emph{Allerton conference} \cite{FHSSWallerton}.
	A follow-up paper on the bosonic quantum interference channel 
	was presented by the author at 
	the 2011 \emph{International Symposium on Information Theory},
	thanks to a collaboration with Saikat Guha and Mark M. Wilde \cite{GSW11bosonic}.
	A further collaboration with Mark M. Wilde led to the
	publication of \cite{SW11}, which describes two coding strategies
	for the quantum broadcast channel.
	Finally, a collaboration with Mark M. Wilde and Mai Vu
	led to the development of the coding strategy for 
	the quantum relay channel presented in \cite{savov2011partial}.
	The last two papers have been accepted for presentation at the
	2012 \emph{International Symposium on Information Theory}.


	Our aim has been to present a comprehensive collection
	of the state-of-the-art of current  knowledge in quantum network
	information theory analogous to the review paper 
	by Cover and El Gamal \cite{el1980multiple}.
	Indeed, the current work contains the classical-quantum 
	extension of many of the results presented in that paper.
	Towards this aim, we have chosen to include in the text
	the statement of several important results by others.
	These include a proof of the capacity theorems of
	the point-to-point c-q channels different from the original ones due to 
	Holevo, Schumacher and Westmoreland \cite{H98,SW97}
	and the capacity result for quantum multiple access channel,
	originally due to Winter \cite{winter2001capacity}.
	We will also present an alternate achievability proof of the 
	quantum Chong-Motani-Garg rate region for the QIC,
	which was originally proved by Sen \cite{S11a}.

\section{Thesis overview}

	Each of the communication problems covered in this thesis
	is presented in a separate chapter, and each chapter is organized
	in the same manner.
	The exposition in each chapter is roughly self-contained, 
	but the ideas developed in Chapter~\ref{chapter:MAC} 
	are of key importance to all other results in the thesis.	
	Chapters \ref{chapter:p-to-p} through \ref{chapter:RC} present
	results on classical-quantum (\emph{c-q}) channels where the output systems
	are arbitrary quantum states in finite dimensional Hilbert spaces.
	This class of channels generalizes the class of classical discrete memoryless channels.
	The last chapter, Chapter~\ref{chapter:bosonic}, introduces
	the basic notions of quantum optics and studies 
	bosonic quantum 
	channels, for which the output system is a quantum system
	with continuous degrees of freedom.

	Necessary background material on the notion of a
	classical typical set and its quantum analogue,
	the quantum typical subspace, is presented in Chapter~\ref{chapter:background}.
	A more detailed discussion about typicality is presented 
	in the appendix. Appendix~\ref{apdx:classical-typicality} concerns 
	classical typical sets whereas Appendix~\ref{apdx:quantum-typicality}
	reviews the properties of quantum typical subspaces,
	and quantum typical projectors. 
	Of particular importance are  
	conditionally typical projectors, which are used throughout 
	the proofs in this work.

	Our exploration of the classical-quantum world of 
	communication channels begins in Chapter~\ref{chapter:p-to-p},
	where we discuss classical and classical-quantum models of 
	point-to-point  communication.
	We will state and prove the capacity result for each class of channels:
	Shannon's classical channel coding theorem (Theorem~\ref{thm:shannon-ch-cap})
	and the Holevo-Schumacher-Westmoreland theorem (Theorem~\ref{thm:HSWtheorem})
	concerning the capacity of the classical-quantum channel.
	
	Chapter~\ref{chapter:MAC} presents results 
	on the quantum multiple access channel (QMAC) and 
	discusses the different coding strategies that can be employed.
	The capacity of the QMAC was established by Winter in \cite{winter2001capacity}
	(Theorem~\ref{thm:cqmac-capacity}) using a \emph{successive decoding} strategy.
	Our contribution to the quantum multiple access channel problem
	is Theorem~\ref{thm:sim-dec-two-sender}, which shows that the two-sender 
	\emph{simultaneous} decoding is possible.
	This result and the proof techniques used therein will 
	form key building blocks for the results in subsequent chapters.
	%
	%
	The proof of Theorem~\ref{thm:sim-dec-two-sender} 
	is the result of longstanding collaboration within our research group.
	%

	Chapter~\ref{chapter:IC} will present  results on quantum 
	interference channels.
	These include the calculation of
	the capacity region for the quantum interference channel in two special cases 
	and a description of the quantum Han-Kobayashi rate region
	\cite{FHSSWallerton,FHSSW11}.
	In that chapter, we also provide an alternate proof of the achievability of the
	quantum Chong-Motani-Garg rate region,
	which was first established by Sen in \cite{S11a}.
	This new proof is original to this thesis.	

	Chapter~\ref{chapter:BC} is dedicated to the quantum broadcast channel problem.
	We prove two theorems: the superposition coding inner bound 
	(Theorem~\ref{thm:sup-coding-inner-bound}), which was first proved in  \cite{YHD2006}
	using a different approach, and the Marton inner bound with no common message
	 (Theorem~\ref{thm:marton-no-common}).

	In Chapter~\ref{chapter:RC}, we will present
	Theorem~\ref{thm:PDF-for-QRC} which is a proof of the
	\emph{partial decode-and-forward} inner bound for the quantum relay channel.
	The \emph{decode-and-forward} and \emph{direct coding} strategies for the quantum
	relay channel are also established, since they are special cases of the more
	general Theorem~\ref{thm:PDF-for-QRC}.
	
	Chapter~\ref{chapter:bosonic} discusses the free-space optical communication 
	interference channel in the presence of background thermal noise.
	This is a model for the crosstalk between 	two optical communication links.
	This chapter demonstrates the practical aspect of the ideas developed in this thesis.

	We conclude with Chapter~\ref{chapter:conclusion} wherein we state
	open problems and describe avenues for future research.


\chapter{Background}
						\label{chapter:background}

	In this chapter we present all the necessary background 
	material which is essential to the results presented 
	in subsequent chapters.

	\section{Notation}
		
		We will denote the set $\{1,2,\ldots,n\}$ as $[1\!:\!n]$ or with the shorthand $[n]$.
		A random variable $X$, defined over a finite set $\mcal{X}$,
		is associated with a probability distribution $p_X(x)\equiv\Pr\{ X=x\}$,
		where the lowercase $x$ is used to denote particular values of the random variable.
		Furthermore, let $\mathcal{P}(\mathcal{X})$ denote
		the set of all probability mass functions on the finite set $\mathcal{X}$.
		Conditional probability distributions will be denoted as $p_{Y|X}(y|x)$ or simply $p_{Y|X}$.
		
		In order to help distinguish between the 
		classical systems (random variables)
		and the quantum systems in the equations, we use the
		following naming conventions.
		Classical random variables will be denoted by letters near to the
		end of the alphabet ($U$, $W$, $X_1$, $X_2$)
		and denoted as small triangles,
			\begin{tikzpicture}[node distance=2.0cm,>=stealth',bend angle=45,auto, scale=0.5]
			  \tikzstyle{cnode}=[isosceles triangle,isosceles triangle apex angle=60,thick,draw=blue!75,
			  fill=blue!20,minimum height=3mm] 
			  \tikzstyle{qnode}=[circle,thick,draw=blue!75,fill=blue!20,minimum size=5mm] 
			  \tikzstyle{every label}=[black, font=\footnotesize]
			  \begin{scope}
				\node [cnode,scale=0.6] (Tx) [   ]                            {} ;
			  \end{scope}
			  \end{tikzpicture},
		in the diagrams of this thesis.
		The triangular shape was chosen in analogy
		to the 2-simplex  $ \equiv \mcal{P}(\{1,2,3\})$.
		Quantum systems will be named with letters near 
		the beginning of the alphabet ($A$, $B_1$, $B_2$)
		and represented by circles,
		 	 \begin{tikzpicture}[node distance=2.0cm,>=stealth',bend angle=45,auto, scale=0.5]
			  \tikzstyle{cnode}=[isosceles triangle,isosceles triangle apex angle=60,thick,draw=blue!75,
			  fill=blue!20,minimum height=3mm] 
			  \tikzstyle{qnode}=[circle,thick,draw=blue!75,fill=blue!20,minimum size=5mm] 
			  \tikzstyle{every label}=[black, font=\footnotesize]
			  \begin{scope}
				\node [qnode,scale=0.6] (Tx) [   ]                            {} ;
			  \end{scope}
			  \end{tikzpicture},
			  in diagrams.
			  The circular shape is chosen in analogy with the Bloch sphere \cite{leifer2011formulating}.

		Consider a communication scenario
		with one or more senders (female) and one or more receivers (male).
		In diagrams, a sender is denoted \texttt{Tx} (short for Transmitter)
		and is associated with a random variable $X$.
		If there are multiple senders, then each of them 
		will be referred to as Sender~k and associated
		with a random variable $X_k$.
		Receivers will be denoted as $\texttt{Rx}~1$, $\texttt{Rx}~2$
		and each is associated with a different output of the channel.
		The outputs of a \emph{classical} channel will be denoted as $Y_1$, $Y_2$,
		and the outputs of a \emph{quantum} channel will be denoted as $\rho^{B_1}, \rho^{B_2}$.

		The purpose of a communication protocol is to transfer 
		bits of information from sender to receiver noiselessly.
		In this respect, the noiseless binary channel from sender to receiver 
		is the standard unit resource for this task:
		\be
			(\mcal{X}=\{0,1\}, \ p_{Y|X}(y|x)=\delta(x,y),\  \mcal{Y}=\{0,1\} )
			\ \   	\equiv \ \ 
			 [c \to c],
		\ee
		where we have also defined the more compact notation $[c \to c]$.
		We will use $[c \to c]$ to denote the \emph{communication resource}
		of being able to send one  bit of classical information from the sender to the 
		receiver \cite{DHW05b}.
		The 
		square brackets indicate that the resource is noiseless.
		%
		In order to describe multiuser communication scenarios, we extend
		this notation with superscripts indicating the sender and the receiver.
		Thus, in order to denote the noiseless classical communication 
		of one bit from Sender~$k$ to Receiver~$z$ we will use the notation
		$[c^k \to c^z]$.
		The communication resource which corresponds to the
		sender being able to \emph{broadcast} a 	message to Receiver~1 and Receiver~2  
		is denoted as $[c \to c^1c^2]$.
		All the coding theorems presented in this work are protocols
		for converting many copies of some noisy channel resource
		into noiseless classical communication between a particular
		sender and a particular receiver as described above.

		Codebooks $\{ x^n(m)\}_{m\in\mcal{M}}$ are lookup 
		tables for codewords representing a discrete set of messages
		$\mcal{M} = \{1, 2, 3, \ldots, |\mcal{M}| \}$ that could be transmitted.
		A communication rate $R$ is a real number which describes our
		asymptotic ability to construct codes for a certain communication task.
		We will use the notation $|\mcal{M}| = 2^{nR}$,
		and $\mcal{M} = \{1, 2, 3, \ldots, |\mcal{M}| \} \equiv [1 : 2^{nR} ]$, in which $2^{nR}$ should be interpreted 
		to indicate $\lfloor 2^{nR} \rfloor$.

		Let $\mathbb{R}_+^n \equiv \{  \vec{v} \in \mathbb{R}^n \ | \ v_i \geq 0, \forall i \in [1:n] \}$
		be the non-negative subset of $\mathbb{R}^n$.
		We will denote a \emph{rate region} as $\mcal{R} \subseteq \mathbb{R}_+^n$ and
		the boundaries of regions as $\partial\mcal{R}$.
		We denote points as $P \in \mathbb{R}^n$ and 
		denote the \emph{convex hull}  of a set of points $\{P_i\}$
		as $\mathrm{conv}(\{ P_i \})$.

	\section{Classical typicality} \label{sec:typ-review}
	We present here a number of properties of typical sequences \cite{CT91}.
	
	\subsection{Typical sequences}

		Consider the random variable $X$ with probability distribution $p_{X}(x)$ defined on a finite set $\cX$.
		Denote by $|\cX|$ the cardinality of $\cX$. 
		Let $H(X) \equiv H(p_X) \equiv - \sum_x p_X(x) \log_2 p_X(x)$ be the Shannon entropy of $p_X$,
		and it is measured in units of \emph{bits}.
		The binary entropy function is denoted  
		$H_2(p_0)\equiv -p_0\log_2(p_0) -(1-p_0)\log_2(1-p_0)\equiv H_2(p_1)$, where
		$p_0 \equiv p_X(0)$ and $p_1 \equiv 1-p_0$.

		Denote by $x^n$ a sequence $x_1x_2\dots x_n$, where each
		$x_i, i\in[n]$ belongs to the finite \emph{alphabet} $\cX$. 
		%
		%
		To avoid confusion, we use $i\in[1:n]$ to denote the index of a symbol $x$
		in the sequence $x^n$ and $a \in [1,2,\ldots,|\mcal{X}|]$ 
		to denote the different symbols in the alphabet~$\mcal{X}$.
		%


	Define the probability distribution $p_{X^n}(x^n)$ on $\cX^n$ to be the
	$n$-fold product of $p_{X}$:
	$p_{X^n}(x^n) \equiv \prod_{i=1}^n p_X(x_i)$. 
	The sequence $x^n$ is drawn from $p_{X^n}$ if
	and only if each letter $x_i$ is drawn independently from $p_X$.
	For any $\delta >0$,	define the set of entropy $\delta$-typical sequences of length 
	$n$ as: 
	\be
		\mcal{T}^{n}_\delta(X) 
		\! \equiv \!
		\left\{ \!
			x^n \in \mcal{X}^n  \colon \! \left|  -\frac{\log p_{X^n}(x^n)}{n}  - H(X) \right| \! \leq \delta 
		\right\}.%
		\label{eqn:def-typ-X}
	\ee

	Typical sequences enjoy many useful properties \cite{CT91}.
	For any $\epsilon,\delta>0$, and sufficiently large $n$, we have
	\begin{eqnarray}
		& 
		 \!\!\!\!\!\!\!\!\!\!\!\!\!\!\!\!\!
		\displaystyle\sum_{x^n\in \mcal{T}^{(n)}_\delta(X)} \!\!\!\! p_{X^n}\!\!\left( x^n\right)  
		& \geq  1-\epsilon,   \label{cc1BG}  \\[1.6mm]
		2^{-n[ H(X)+\delta ]}  \leq \!\!\!\!\!&
		\!\!\! p_{X^n}(x^n)  \!\!\!& \!\!\!\!\!
		\leq   2^{-n[H(X)-\delta]}  \  \   \forall x^n \in \mcal{T}^{(n)}_\delta(X),  \label{cc2BG} \\[0mm]
		\!\!\![1 - \epsilon] 2^{n[ H(X)-\delta]} \leq \!\!\!\!\!& 
		\!|\mcal{T}^{(n)}_\delta(X)| \!\!\!
		& \!\!\! \leq 2^{n[ H(X)+\delta]}. \label{cc3BG} 
	\end{eqnarray}

\def\xnlatex{x^n}

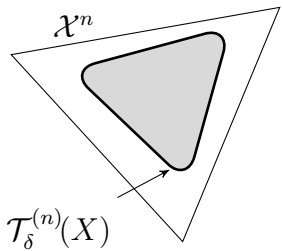
\begin{figure}[bt]
\begin{center}

	\begin{tikzpicture}[node distance=2.0cm,>=stealth',bend angle=45,auto]

	  \begin{scope}

		
		


		\draw  (3.3, 2.3) coordinate (XnCenter);
		\path (XnCenter) ++(35:20mm) coordinate (XnCorner1);
		\path (XnCenter) ++(165:20mm) coordinate (XnCorner2);		
		\path (XnCenter) ++(280:20mm) coordinate (XnCorner3);

		\draw (XnCorner1) -- (XnCorner2) -- (XnCorner3) -- cycle;						
		\path (XnCenter) ++(0,0.2) ++(35:12mm) coordinate (XnCorner1in);
		\path (XnCenter) ++(165:12mm) coordinate (XnCorner2in);		
		\path (XnCenter) ++(0.2,0)++(280:12mm) coordinate (XnCorner3in);

		\draw [fill=black!15] (XnCorner1in)  [rounded corners=10pt,line width=1pt] -- (XnCorner2in) -- (XnCorner3in) -- cycle;		

		\node at (2.2,3.2) {$\mathcal{X}^n$};
		
%
			\node (OnXnTyp) at (3.5,1.3) {};
			\node at (2,0.5) {$\mcal{T}_\delta^{(n)}\!(X)$} 
				edge [post] node[]  	{}  (OnXnTyp.center);

	  \end{scope}
	\end{tikzpicture}
\end{center}
\label{fig:X-typ-set}
\caption{The typical set. Property~\eqref{cc1BG} implies that draws of a 
random sequence $X^n \sim p_{X^n} \equiv \prod^n p_X$ 
are likely to fall inside the typical set $\mcal{T}_\delta^{(n)}\!(X) \subset \mcal{X}^n$ with high probability.
If draws from $X^n \sim \prod^n p_X$ are represented as points, then
after many draws the typical set will become darker as the shaded region in the diagram.
The probability mass density on $\mcal{T}_\delta^{(n)}\!(X)$ is
approximately uniform: 
it varies between $2^{-n[H(X)+\delta]}$ and $2^{-n[H(X)-\delta]}$ (Property~\eqref{cc2BG}),
and the size of the shaded area
will be at most $2^{n[H(X)+\delta]}$ (Property~\eqref{cc3BG}).
The non-typical set,  $\mcal{X}^n \setminus \mcal{T}_\delta^{(n)}\!(X)$,
will have at most $\epsilon$ weight in it (Property~\eqref{cc1BG}).
 }
\end{figure}

Property \eqref{cc1BG} indicates that a sequence $X^n$ of random variables
distributed according to $p_{X^n} = \prod^n p_X$ (identical and independently distributed),
is very likely to be typical, since all but $\epsilon$ of the weight of the
probability mass function is concentrated on the typical set, which follows
from the law of large numbers.
Property~\eqref{cc2BG} follows from the definition of the typical set \eqref{eqn:def-typ-X}.
The lower bound on the probability of the typical sequences from \eqref{cc2BG} can
be used to obtain an upper bound on the size of the typical set in \eqref{cc3BG}.
Similarly the upper bound from \eqref{cc2BG} and equation \eqref{cc1BG} can be
combined to give the lower bound on the typical set in \eqref{cc3BG}.

\subsection{Conditional typicality}

	Consider now the conditional probability distribution $p_{Y|X}(y|x)$
	associated with a communication channel.
	The induced joint input-output distribution is $(X,Y) \sim p_X(x)p_{Y|X}(y|x)$,
	when $p_X(x)$ is used as the input distribution.

	The conditional entropy $H(Y|X)$ for this distribution is
	\be
		H (Y|X) = H(X,Y) - H(X) =
		 \sum_{x_a \in \mcal{X}} 
			p_{X}(x_a) 
 			H(Y|x_a).
	\ee	
	where $H(Y|x_a) = - \sum_y p_{Y|X}(y|x_a) \log p_{Y|X}(y|x_a)$.

	We define the $x^n$-conditionally typical set 
	$\mcal{T}^{(n)}_\delta(Y|x^n) \subseteq \mcal{Y}^n$
	to consist of all sequences $y^n$ which are 
	typically output when the input to the channel is $x^n$:
	\be
		\!\! \mcal{T}^{(n)}_\delta(Y|x^n) 
		\! \equiv \!
		\left\{ \!
			y^n \in \mcal{Y}^n  
			\colon \! \left|  -\frac{\log p_{Y^n|X^n}(y^n|x^n)}{n}  - H(Y|X) \right| \! \leq \delta 
		\right\}\!\!,%
		\label{eqn:def-typ-YgX}
	\ee	
	with $p_{Y^n|X^n}(y^n|x^n) = \prod_{i=1}^n p_{Y|X}(y_i|x_i)$.
	The definition in \eqref{eqn:def-typ-YgX} can be rewritten as:
	\begin{align}
		2^{-n[ H(Y|X)+\delta ]}  \leq & \ \ 
		p_{Y^n|X^n}(y^n|x^n) \ \ 
		\leq   2^{-n[H(Y|X)-\delta]},  \ \    \forall y^n \in \mcal{T}^{(n)}_{\delta}\!(Y|x^n),  \label{cc5}
	\end{align}
	for any sequence $x^n$.
	
	Suppose that a random input sequence $X^n \sim p_{X^n}=\prod^n p_X$ 
	is passed through the channel $p_{Y^n|X^n}$.
	Then a conditionally typical sequence is likely to occur.
	More precisely, we have that for any $\epsilon,\delta>0$, and sufficiently large $n$
	the statement is true under the expectation over the input sequence
	$X^n$:
	\begin{align}
		 \NExpX
		 \!\!\!\!\!
		\displaystyle
		\sum_{ y^n \in \mcal{T}^{(n)}_{\delta}\!(Y|X^n) } \!\! p_{Y^n|X^n}\!\!\left( y^n|X^n\right)  
		& = \sum_{x^n} p_{X^n}(x^n) \!\!\!\!\!\!
			\sum_{ y^n \in \mcal{T}^{(n)}_{\delta}\!(Y|x^n) } 
			\!\!\!\!\! 
			p_{Y^n|x^n}\!\!\left( y^n|x^n\right)   \nonumber \\
		& \geq  1-\epsilon.   \label{cc4BG}
	\end{align}
	We also have the 
	following bounds on the expected size 
	of the conditionally typical set:	
	\be
		[1 - \epsilon] 2^{n[ H(Y|X)-\delta]} \leq \ \   
		\NExpX
		\left|  \mcal{T}^{(n)}_{\delta}\!(Y|X^n)  \right| 
		\ \  \leq 2^{n[ H(Y|X)+\delta]}. \label{cc6BG} 
	\ee

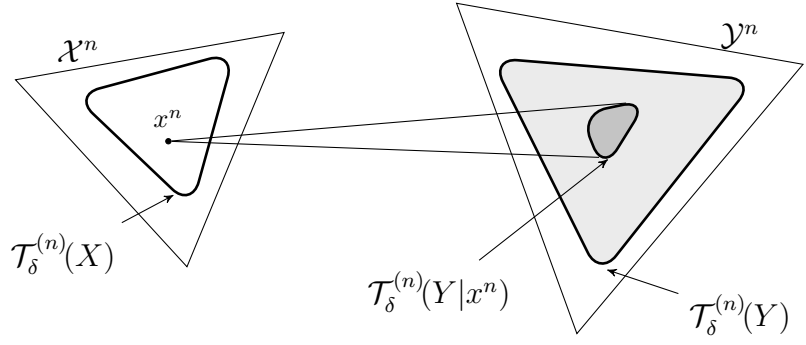
\begin{figure}
\begin{center}

	\begin{tikzpicture}[node distance=2.0cm,>=stealth',bend angle=45,auto]

	  \begin{scope}

		\draw  (3.3, 3.3) coordinate (XnCenter);
		\path (XnCenter) ++(35:20mm) coordinate (XnCorner1);
		\path (XnCenter) ++(165:20mm) coordinate (XnCorner2);		
		\path (XnCenter) ++(280:20mm) coordinate (XnCorner3);

		\draw (XnCorner1) -- (XnCorner2) -- (XnCorner3) -- cycle;						
		\path (XnCenter) ++(0,0.2) ++(35:12mm) coordinate (XnCorner1in);
		\path (XnCenter) ++(165:12mm) coordinate (XnCorner2in);		
		\path (XnCenter) ++(0.2,0)++(280:12mm) coordinate (XnCorner3in);

			\draw (XnCorner1in)  [rounded corners=10pt,line width=1pt] -- (XnCorner2in) -- (XnCorner3in) -- cycle;		

		%
		\node at (2.2,4.2) {$\mathcal{X}^n$};
		
		\node (OnXnTyp) at (3.5,2.3) {};
		\node at (2,1.5) {$\mcal{T}_\delta^{(n)}\!(X)$} 
			edge [post] node[]  	{}  (OnXnTyp.center);

		\node at (11,4.4) {$\mathcal{Y}^n$};
		\draw  (9.3, 3.1) coordinate (YnCenter);
		\path (YnCenter) ++(20:27mm) coordinate (YnCorner1);
		\path (YnCenter) ++(140:27mm) coordinate (YnCorner2);		
		\path (YnCenter) ++(260:27mm) coordinate (YnCorner3);
		\draw (YnCorner1) -- (YnCorner2) -- (YnCorner3) -- cycle;						

		\path (YnCenter) ++(0,0.1) ++(18:20mm) coordinate (YnCorner1in);
		\path (YnCenter) ++(148:19mm) coordinate (YnCorner2in);		
		\path (YnCenter) ++(-0.1,0)++(268:19mm) coordinate (YnCorner3in);
		\draw [fill=black!8] (YnCorner1in)  [rounded corners=10pt,line width=1pt] -- (YnCorner2in) -- (YnCorner3in) -- cycle;		
		\node (OnYnTyp) at (9.2,1.3) {};
		\node at (11, 0.7) {$\mcal{T}_\delta^{(n)}\!(Y)$} 
			edge [post] node[]  	{}  (OnYnTyp.center);

		\path (YnCenter) ++(0,0.3) ++(18:5mm) coordinate (YngXCorner1in);
		\path (YnCenter) ++(148:5mm) coordinate (YngXCorner2in);		
		\path (YnCenter) ++(-0.1,0)++(268:5mm) coordinate (YngXCorner3in);
		\draw [fill=black!23] (YngXCorner1in)  [rounded corners=10pt,line width=1pt] -- (YngXCorner2in) -- (YngXCorner3in) -- cycle;		
		\node (OnYngXTypTop) at (9.63,3.52) {};
		\node (OnYngXTypBot) at (9.27,2.78) {};
		\node at (7, 1) {$\mcal{T}_\delta^{(n)}\!(Y|x^n)$} 
			edge [post] node[]  	{}  (OnYngXTypBot.center);

		\node [circle,fill,draw,inner sep=0pt,minimum size=2pt,label=above:$\xnlatex$,  xshift=1mm,yshift=-3mm] (xn) at (XnCenter)  {};
		\draw (OnYngXTypTop) -- (xn);
		\draw (OnYngXTypBot) -- (xn);

	  \end{scope}
	\end{tikzpicture}
\end{center}
\caption{Illustration of the conditionally typical set $\mcal{T}_\delta^{(n)}\!(Y|x^n)$ 
	and the output-typical set $\mcal{T}_\delta^{(n)}\!(Y)$.
	The ``density'' of $\mcal{T}_\delta^{(n)}\!(Y)$,
	the lightly shaded area, is at least $2^{-n[H(Y) +\delta]}$,
	and the size of $\mcal{T}_\delta^{(n)}\!(Y)$ is at most $2^{n[H(Y)+\delta]}$.
	The size of $\mcal{T}_\delta^{(n)}\!(Y|x^n)$,
	the darker shaded region, is no greater than 
	$2^{n[H(Y|X)+\delta]}$ for an $x^n$ picked on average. 
	}
\label{fig:Y-X-cond-typ-set}
\end{figure}

	\subsection{Output-typical set}
	
	Consider the distribution over symbols $y \in \mcal{Y}$ 
	induced by the channel $\mcal{N}\equiv p_{Y|X}(y|x)$ 
	whenever the input distribution is $p_X(x)$:
	\be
		p_Y(y) \equiv \sum_{x} p_{Y|X}(y|x) p(x) = \mathbb{E}_{X} \mcal{N}.
	\ee

	We define the output typical set as 
	\be
		\mcal{T}^{(n)}_{\delta}\!(Y)
		\! \equiv \!
		\left\{ \!
			y^n \in \mcal{Y}^n  \colon \! \left|  -\frac{\log p_{Y^n}(y^n)}{n}  - H(Y) \right| \! \leq \delta 
		\right\},%
		\label{eqn:def-typ-Y}
	\ee
	where $p_{Y^n}=\prod^n p_{Y}$. 
	Note that the output-typical set is just a special case of the 
	general typical set shown in \eqref{eqn:def-typ-X}.
	The terminology \emph{output-typical} is introduced to
	help with the exposition.

	When the input sequences are chosen according to 
	$X^n \sim p_{X^n}=\prod^n p_X$,
	then output sequences are likely to be output-typical:	
	\be
		\ExpX
		\sum_{ y^n \in \mcal{T}^{(n)}_{\delta}\!(Y) } \!\! p_{Y^n|X^n}\!\!\left( y^n|X^n\right)  
		\ \ \geq \ \  1-\epsilon.   
		\label{cc7BG}
	\ee

	An illustration and an intuitive interpretation of \eqref{cc4BG}, \eqref{cc6BG} and \eqref{cc7BG} is presented in 
	Figure~\ref{fig:Y-X-cond-typ-set}.
	The expression in \eqref{cc4BG} for the property of the conditionally
	typical set $\mcal{T}^{(n)}_{\delta}\!(Y|x^n)$ is the analogue 
	of the typical property \eqref{cc1BG} for  $\mcal{T}^{(n)}_{\delta}\!(X)$.
	The interpretation is that the codewords of a random codebook
	are likely to produce output sequences that fall within
	their conditionally typical sets.
	This property will be used throughout this thesis to guarantee
	that the decoding strategies based on conditionally typical sets
	correctly recognize the channel outputs.
	On the other hand, \eqref{cc6BG} gives us both an upper bound and a lower
	bound on the size of the conditionally typical set for a random 
	codebook.
	Finally, Property~\eqref{cc7BG} tells us that the outputs of the channel
	which are not output-typical are not likely.
	
%
%

%

\subsection{Joint typicality}

	Consider now the joint probability distribution $p_{XY}(x,y) \in \mcal{P}(\mcal{X},\mcal{Y})$.
	Let $(X^n,Y^n)$ be a pair of random variables distributed
	according to the product distribution $\prod^n p_{XY}$.
	
	We define the jointly typical set $\mcal{J}^{(n)}_\delta(X,Y) \subseteq \mcal{X}^n \times \mcal{Y}^n$
	to be the set of sequences that are typical with respect to the joint
	probability distribution $p_{XY}$ and with respect to the marginals $p_X$ and $p_{Y}$.
	\be
		\!\!\! \! \mcal{J}^{(n)}_\delta(X,Y)
		\equiv
		\left\{ \!
			(x^n,y^n)  \in \mcal{X}^n \times \mcal{Y}^n  
			\left| 
			\begin{array}{l}
				x^n \in \mcal{T}^{(n)}_\delta(X)   \\
				y^n \in \mcal{T}^{(n)}_\delta(Y)   \\
				(x^n,y^n) \in \mcal{T}^{(n)}_\delta(X,Y)  \!\!
			\end{array}
			\right.
		\right\}\!.%
		\label{eqn:def-Jtyp}
	\ee
	A multi-variable sequence, therefore, is jointly typical if and only
	if all the sequences in the subsets of the variables are jointly typical.
	%

	The probability that two random sequences drawn from the marginals
	$X^n \sim \prod^n p_X$ and $Y^n \sim \prod^n p_Y$
	are jointly typical can be bounded from above by $2^{-n[ I(X;Y) - \delta ] }$.
	This is straightforward to see from the definition in \eqref{eqn:def-Jtyp}
	and the properties of typical sets. 
	If $(x^n,y^n)$ is such that $x^n \in \mcal{T}^{(n)}_{\delta^{\prime}}(X)$
	and $y^n \in \mcal{T}^{(n)}_{\delta^{\prime}}(Y)$ then
	$p_{X^n}(x^n) \leq 2^{-n[H(X) - \delta^\prime]}$
	and
	$p_{Y^n}(y^n) \leq 2^{-n[H(Y) - \delta^\prime]}$.
	On the other hand, we know that the number
	of sequences that are typical according to the joint distribution 
	is no larger than $2^{n[H(XY)+\delta^{\prime\prime}]}$.
	Combining these two observations we get:
	\begin{align}
		\sum_{ (x^n,y^n) \in \mcal{T}^{(n)}_{\delta^{\prime\prime}}(X,Y)   } 
		\!\!\!\!\!\!\!
		p_{X^n}(x^n) 
		p_{Y^n}(y^n)
		& \leq 
			\left| \mcal{T}^{(n)}_{\delta^{\prime\prime}}(X,Y)  \right|
			2^{-n[H(X) - \delta^\prime]}
			2^{-n[H(Y) - \delta^\prime]} \nonumber \\
		& \leq 
			2^{n[H(XY) + \delta^{\prime\prime}]}
			2^{-n[H(X) - \delta^\prime]}
			2^{-n[H(Y) - \delta^\prime]} \nonumber \\		
		& =
			2^{-n[ I(X;Y) - \delta ] }.
	\end{align}
	Note that the parameter 
	$\delta = 2\delta^\prime + \delta^{\prime\prime}$
	is a function of our choice of typicality parameters
	for the typical sets.
	%
	%

\section{Introduction to quantum information}
	
		The use of quantum systems for information processing tasks
		is no more mysterious than the use of digital technology for information processing.
		The use of an \emph{analog to digital converter} (ADC) to transform an analog
		signal to a digital representation
		and the use of a \emph{digital to analog converter} (DAC) to transform
		from the digital world back into the analog world are similar 
		to the \emph{state preparation} and the \emph{measurement} 
		steps used in quantum information science.
		The \emph{digital world} is sought after because of the computational,
		storage and communication benefits associated with 
		manipulation of discrete systems instead of continuous signals.
		Similarly, there are benefits associated with using the
		\emph{quantum world} (Hilbert space) in certain computation problems
		\cite{shor1994algorithms,shor1995polynomial}.
		The use of digital and quantum technology can therefore both be
		seen operationally as a black box process with information encoding,
		processing and readout steps.

		The focus of this thesis is the study of \emph{quantum} aspects 
		of communication which are relevant for \emph{classical communication}  tasks.
		In order to make the presentation more self-contained,
		we will present below a brief introduction to the subject
		which describes how quantum systems are represented,
		how information can be encoded and how information can be read out.

		\subsection{Quantum states}
	
			In order to describe the \emph{state} of a quantum system $B$
			we use a density matrix $\rho^{B}$ acting on a $d$-dimensional 
			complex vector space $\mathcal{H}^{B}$ 		(Hilbert space).
			To be a density matrix, the operator 
			%
			$\rho^B$ has to be Hermitian, positive semidefinite and have unit trace.
			We denote the set of density matrices on a Hilbert space $\mcal{H}^B$
			as $\mcal{D}(\mcal{H}^B)$.

			A common choice of basis for $\mcal{H}^B$ is the standard basis 
			$\{ \ket{0}, \ket{1},$ $\ldots, \ket{d-1} \}$:
			\be
				\ket{0}  \equiv   \colvec{ 1 \\ 0 \\ \vdots \\ 0}\!\!,  \ \ \
				\ket{1}  \equiv   \colvec{ 0 \\ 1 \\ \vdots \\ 0}\!\!,  \ \ \
				\ \ldots,  \ \ \
				\ket{d-1} \equiv \colvec{ 0 \\  \vdots \\ 0  \\ 1}\!\!,
			 \ee
			 which is also known as the \emph{computational} basis.
			 
			 In two dimensions, another common basis is the
			 \emph{Hadamard} basis:
			\begin{align}
				\ket{+} & \equiv   \frac{1}{\sqrt{2}}\ket{0} + \frac{1}{\sqrt{2}}\ket{1},   \\
				 \ket{-} & \equiv    \frac{1}{\sqrt{2}}\ket{0} - \frac{1}{\sqrt{2}}\ket{1}.
			\end{align}
			%


			%

			The eigen-decomposition of the density matrix $\rho^B$ gives us another choice of
			basis  in which to represent the state.
			Any density matrix can be written in the form:
			\be
			\!\!
			\rho^B \!
			\equiv \!
				\left[ \!\!\begin{array}{cccc} 
					\vspace{4mm} 	&  			&  			&  \\
					\ket{e_{\rho;1}} \!& \ket{e_{\rho;2}} 	\!		& \!\!\cdots \! \!			& \ket{e_{\rho;d}}\\  
					\vspace{4mm} 	&   			&  			& 
				\end{array} \!\!\right]		
				\!
				\left[ \begin{array}{cccc} 
					\lambda_{\rho;1} 	& 0 			& \cdots 			& 0\\  
					0 		& \lambda_{\rho;2} 		&  				& 0 \\
					\vdots 	&  			& \ddots 			& \vdots \\
					0 		& 0 			& \cdots			& \lambda_{\rho;d}
				\end{array} \right]		
				\!
				\left[ \begin{array}{c}
					\ifthenelse{\boolean{BOOKFORM}}{
					\hspace{2mm}  \ \  \bra{e_{\rho;1}}  \ \hspace{2mm}     \\ 
					}
					{
					\hspace{6mm}  \ \  \bra{e_{\rho;1}}  \ \hspace{6mm}     \\ 
					}
					\bra{e_{\rho;2}} \\
					\vdots \\
					\bra{e_{\rho;d}} 								
				\end{array} \right]\!\!,
			\ee
			where the eigenvalues $\lambda_{\rho;i}$
			are all real and nonnegative.
			In our notation, column vectors are denoted as 
			\emph{kets} $\ket{e_{\rho;i}}$ 
			and the dual (Hermitian conjugate) of a ket is the 
			\emph{bra}:  $\bra{e_{\rho;i}} \equiv \ket{e_{\rho;i}} ^\dagger$
			(a row vector).
			%
			%
			%
			We say that $\rho^B$ is a \emph{pure state} if it has only a single non-zero
			eigenvalue: $\lambda_{\rho;1}=1$, $\lambda_{\rho;i}=0$, $\forall i > 1$.

			Because the density matrix is positive semidefinite and has unit trace ($\sum_i \lambda_{\rho;i}=1$),
			we can identify the eigenvalues of $\rho^B$ with a probability distribution:
			$p_Y(y) \equiv \lambda_{\rho;y}$.
			A density matrix, therefore, corresponds to the probability distribution
			$p_Y(y)$ over the subspaces: $\ketbra{e_{\rho;y}}{e_{\rho;y}}$.
			This property will be important when we want to define 
			the typical subspace for the tensor product state:
			$(\rho^{B})^{\otimes n} \equiv \rho^{B_1} \otimes \rho^{B_2} \otimes \cdots \otimes \rho^{B_n}$.

		%

		%

		    Suppose that we have a two-party quantum state $\rho^{AB}$ such that
		    Alice has the subsystem $A$ and Bob has the subsystem $B$.
		    %
		    The state in Alice's lab is described by $\rho^A = \Tr_B[\rho^{AB}]$, where
		    $\Tr_B$ denotes a partial trace over Bob's degrees of freedom.

		%
		%
		In order to describe the ``distance'' between two quantum states, we use the notion
		of \emph{trace distance}.
		The trace distance between states $\sigma$ and $\rho$ is
		$\|\sigma-\rho\|_1 = \mathrm{Tr}|\sigma - \rho|$, where  $|X| = \sqrt{X^{\dagger}X}$.
		Two states that are similar have trace distance close to zero,
		whereas states that are perfectly distinguishable have trace distance equal to two.

		Two quantum states can ``substitute'' for one
		another up to a penalty proportional to the trace distance between them:
		%
		\begin{lemma} \label{lem:trace-inequality}
		Let 	
		$0\leq \rho, \sigma, \Lambda \leq I$. Then
		\be
		\mathrm{Tr}\left[  \Lambda\rho\right]
			\leq
		\mathrm{Tr}\left[ \Lambda \sigma\right]  
			+ \left\Vert \rho-\sigma\right\Vert _{1}.
		\label{eqn:tr-trick}
		\ee
		\end{lemma}
		\begin{proof}
			This follows from a variational characterization of 
			trace distance as the distinguishability of
			the states under an optimal measurement operator $M$:
			\begin{align*}
				\left\Vert \rho-\sigma\right\Vert _{1} 
					&\equiv  \ 	2 \max_{0 \leq M \leq I} \mathrm{Tr}\left[ M(\rho-\sigma) \right] \\
					& \geq   \ 	\max_{0 \leq M \leq I} \mathrm{Tr}\left[ M(\rho-\sigma) \right] \\
					& \overset{\mdingone}{\geq} 	\  \mathrm{Tr}\left[ \Lambda(\rho-\sigma) \right] \\
					& \geq 	 \ \mathrm{Tr}\left[ \Lambda  \rho \right] 
									- 	\mathrm{Tr}\left[ \Lambda \sigma \right].
			\end{align*}%
		Equation \dingone follows since the operator $\Lambda$,
		$0\leq \Lambda\leq 1$, 
		is a particular choice of  the measurement operator $M$.%
		\end{proof}
		
		Most of the quantum systems considered in this document are  
		finite dimensional,
		but it is worth noting that there are also quantum systems 
		with continuous degrees of freedom which are represented
		in infinite dimensional Hilbert spaces.
		We will discuss the infinite dimensional case in 
		Chapter~\ref{chapter:bosonic}, where we consider the 
		quantum aspects of optical communication.


	\subsection{Quantum channels}

		By convention we will denote the input state 
		as $\sigma$ (for \emph{s}ender) and the outputs of the channel as $\rho$ (for \emph{r}eceiver).
		A noiseless quantum channel is represented by a unitary operator $U$ which acts on the input state 
		$\sigma$ by conjugation to produce the output state $\rho = U\sigma U^\dag$.
		General quantum channels are represented by completely-positive trace-preserving (CPTP) 
		maps $\mathcal{N}^{A\rightarrow B}$, which accept input states in
		$A$ and produce output states in $B$: $\rho^B = \mathcal{N}^{A\rightarrow B}(\sigma^A)$.

		If the sender wishes to transmit some classical message $m$ to the receiver
		using a quantum channel, her encoding procedure will consist of 
		a classical-to-quantum encoder $\mcal{E} \colon m  \to \sigma^A$, 
		to prepare a message state $\sigma^A \in \mcal{D}(\mcal{H}^A)$ suitable as input for the channel.
		We call this the \emph{state preparation} step.

		If the sender's encoding is restricted to transmitting a 
		finite set of orthogonal states $\{ \sigma^A_x \}_{x \in \mcal{X}}$,
		then we can consider the choice of the signal states $\{ \sigma^A_x \}$ to be
	    part of the channel.
	    Thus we obtain a channel with classical inputs $x \in \mcal{X}$
	    and quantum outputs: $\rho_x^B = \mcal{N}^{X \to B}(x) \equiv \mcal{N}^{A\to B}(\sigma_x^A)$.
		A classical-quantum channel, $\mcal{N}^{X\to B}$, 
		is represented by the set of $|\mcal{X}|$ possible output states 
		$\{ \rho^B_x \equiv \mcal{N}^{X\to B}\!(x) \}$, meaning that each
		classical input of $x\in \mcal{X}$ leads to a different quantum output $\rho^B_x\in \mcal{D}(\mcal{H}^B)$.

	\subsection{Quantum measurement}

		The decoding operations performed by the 
		receivers correspond to quantum measurements on the outputs of the channel.
		A quantum measurement 
		is a positive operator-valued measure (POVM) 
		$\left\{ \Lambda_{m}\right\}_{m\in\left\{  1,\ldots,|\mathcal{M}|\right\}  }^{B\to M'}$ on
		the system $B$, the output of which we denote $M^{\prime}$. 
		The probability of outcome $M^{\prime}=m$ when the state $\rho^B$ is measured
		is given by the Born rule:
		\be
			\Pr\{M^{\prime}=m\} \equiv \Tr[ \Lambda^B_m \rho^B ].
		\ee
		To be a valid POVM, the set of $|\mcal{M}|$ operators $\Lambda_{m}$ must all be positive semidefinite
		and sum to the identity:  $\Lambda_{m} \geq 0, \,\,\, \sum_{m}\Lambda_{m}=I$.


		A quantum instrument $\{  \Upsilon_k \}^{A\to B}$ is a more general operation which 
		consists of a collection of completely positive (CP) maps such that $\sum_k  \Upsilon_k$ is 
		trace preserving \cite{DL70}.
		When applied to a quantum state $\sigma^A$, the different elements are applied
		with probability $p_k = \Tr\!\left[ \Upsilon_k(\sigma^A) \right]$ 
		resulting in different normalized outcomes 
		$\rho_k^B=\frac{1}{p_k}\Upsilon_k(\sigma^A)$.

	\subsection{Quantum information theory}
	
		Many of the fundamental ideas of quantum information theory are analogous to 
		those of classical information theory. 
		For example,
		we quantify the information content of quantum systems 
		using the notion of entropy.
		
		\begin{definition}[von Neumann Entropy] 
			Given the density matrix $\rho^A \in \mcal{D}(\cH^A)$, the expression
			\be
				H(A)_\rho=-\Tr\left(\rho^A\log\rho^A\right)
			\ee
			is known as the \emph{von Neumann entropy} of the state $\rho^A$. 
		\end{definition}
		
	            Note that the symbol $H$ is used for both classical and quantum entropy.
	            The von~Neumann entropy of quantum state $\rho^A$ 
	            with spectral decomposition $\rho^A = \sum_i \lambda_i \proj{e_i}$, is
	            equal to the Shannon entropy of its eigenvalues.
	            \be
		            H(A)_\rho=-\Tr\left(\rho^A\log\rho^A\right) = - \sum_i \lambda_i \log \lambda_i = H(\{\lambda_i\}).
	            \ee

			For bipartite states $\rho^{AB}$ we can also define the quantum conditional entropy
			\be
				H(A|B)_\rho 	\equiv 		H(AB)_\rho - H(B)_\rho,					\label{cond-entrpy} 
			\ee
			where $H(B)_\rho = -\Tr\left(\rho^B\log\rho^B\right)$ is the entropy of the reduced density matrix
			$\rho^B = \Tr_A\!\left( \rho^{AB}\right)$. In the same fashion we can also define the 
			quantum mutual  information
			\be
				I(A;B)_\rho 	\equiv		H(A)_\rho + H(B)_\rho - H(AB)_\rho,
			\ee
			and in the case of a tripartite system $\rho^{ABC}$ we define the conditional mutual information 
			as 
			\bea
				I(A;B|C)_\rho 	&\equiv&	H(A|C)_\rho + H(B|C)_\rho - H(AB|C)_\rho \label{cond-mut-info} \\
								&=&		H(AC)_\rho + H(BC)_\rho - H(ABC)_\rho - H(C)_\rho.
			\eea
		    
			\noindent It can be shown that $I(A;B|C)$ is non negative for any tripartite state $\rho^{ABC}$.
			The formula $I(A;B|C)\geq 0$ can also be written in the form
			\be	\label{strong-subadditivity}
				H(AC) + H(BC) 	\geq	H(C) + H(ABC).
			\ee
			This inequality, originally proved in \cite{LR73}, is called the \emph{strong subadditivity} of von Neumann 
			entropy and forms an important building block of quantum information theory.

		Consider the classical-quantum state $\rho^{XB}$ given by:
		\be
			\rho^{XB}  = \sum_{x \in \mcal{X}} p_X(x) \ketbra{x}{x}^X \otimes \rho_x^{B}.
		\ee
		The conditional entropy $H(B|X)$ of this state is equal to:
		\be
			H(B|X) 
			= \sum_{x \in \mcal{X}}  p_X(x) H(\rho_x^B)
			= \sum_{x \in \mcal{X}}  p_X(x) H(B)_{\rho_x}.
		\ee


	



\section{Quantum typicality}

	The notions of typical sequences and typical sets generalize to the quantum setting by virtue of the spectral
	theorem. 
	Let $\cH^B$ be a $d_B$ dimensional Hilbert space and let 
	$\rho^B  \in \mcal{D}(\cH^B)$ be the density matrix associated with a quantum state.
	%
	We identify the eigenvalues of $\rho^B$ with the probability 
	distribution $p_Y(y)=\lambda_{\rho;y}$ 
	and write the 
	spectral decomposition as:
	\be
		\rho^B = \sum_{y=1}^{d_B} p_Y(y) \ket{e_{\rho;y}}\bra{e_{\rho;y}}^B
		\label{eq:spectral-decomp-rhoBG}
	\ee
	where $\ket{e_{\rho;y}}$ is the eigenvector of $\rho^B$ corresponding to eigenvalue $p_Y(y)$.
	%

	Define the set of  $\delta$-typical eigenvalue labels according to the eigenvalue distribution $p_Y$ as
	\be
		\mcal{T}^{(n)}_\delta(Y) 
		\! \equiv \!
		\left\{ \!
			y^n \in \mcal{Y}^n  \colon \! \left|  -\frac{\log p_{Y^n}(y^n)}{n}  - H(Y) \right| \! \leq \delta 
		\right\}.%
	\ee
	For a given string $y^n = y_1y_2\ldots y_i\ldots y_n$ we define the 
	corresponding eigenvector as 
	\be
		\ket{e_{\rho;y^n}} = \ket{e_{\rho;y_1}} \otimes \ket{e_{\rho;y_2}} \otimes \cdots \otimes \ket{e_{\rho;y_n}},
	\ee
	where for each symbol $y_i=b \in \{1,2,\ldots,d_B\}$ we select the b$^\textrm{th}$ eigenvector 
	$\ket{e_{\rho;b}}$.
	
	The typical subspace associated with the density matrix $\rho^B$
	is defined as
	\be
	{A}^n_{\rho, \delta}
		= \textrm{span} \left\{  \ket{e_{\rho;y^n}} \colon y^n \in \mcal{T}^{(n)}_\delta(Y)  \right\}.
	\ee
	The typical projector is defined as 
	\be
		\Pi^{n}_{\rho^B, \delta}
		= \sum_{y^n\in   
			\mcal{T}^{(n)}_\delta(Y)}
			\ket{e_{\rho;y^n}}\! \bra{e_{\rho;y^n}}.
	\ee
	Note that the typical projector is linked twofold to the spectral decomposition
	of \eqref{eq:spectral-decomp-rhoBG}: the sequences $y^n$ are selected 
	according to $p_Y$ and the set of typical vectors are built from tensor
	products of orthogonal eigenvectors $\ket{e_{\rho;y}}$.

	Properties analogous to (\ref{cc1BG}) - (\ref{cc3BG}) hold.
	For any $\epsilon,\delta>0$, and all sufficiently large $n$ we have
	\begin{eqnarray}
		\label{eqn:TypP-prop-one}
		&\!\!\!\!\!\!\!\Tr\{\rho^{\otimes n} \Pi^{n}_{\rho, \delta}\}  &\geq  1-\epsilon  \\
		\label{eqn:TypP-prop-two}
		2^{-n[ H(B)_\rho+\delta]}\Pi^n_{\rho, \delta} 
		\leq  
		&\!\!\Pi^n_{\rho, \delta} \rho^{\otimes n} \Pi^{n}_{\rho, \delta}
		\!\!&
		\leq 
		2^{-n[ H(B)_\rho-\delta ]}\Pi^{n}_{\rho, \delta}, \\
	\label{eqn:TypP-prop-three}
		[1 - \epsilon] 2^{n[ H(B)_\rho-\delta]}
		\leq  
		&\Tr\{\Pi^{n}_{\rho, \delta}\} &
		\leq 
		2^{n[ H(B)_\rho+\delta]}.
	\end{eqnarray}%

	Equation \eqref{eqn:TypP-prop-one} tells us that most of the support of the 
	state $\rho^{\otimes n}$ is within the typical subspace.
	%
	The interpretation of \eqref{eqn:TypP-prop-two} is that the eigenvalues 
	of the  state $\rho^{\otimes n}$ are bounded between
	$2^{-n[ H(B)_\rho+\delta ]}$ and $2^{-n[ H(B)_\rho-\delta ]}$
	on the typical subspace ${A}^n_{\rho, \delta}$.
	
	\bigskip 
	
	\noindent
	{\bf Signal states\ \ }
	Consider now a set of quantum states $\{\rho^B_{x_a}\}$,
	$x_a \in \mcal{X}$.
	We perform a spectral decomposition of each $\rho^B_{x_a}$
	to obtain 
	\be
		\rho^B_{x_a}= \sum_{y=1}^{d_B} p_{Y|X}(y|x_a) \ket{e_{\rho_{x_a};y}}\bra{e_{\rho_{x_a};y}}^B,
		\label{eq:cond-spectral-decomp-rhoBG}
	\ee
	where $p_{Y|X}(y|x_a)$  is the $y^\textrm{th}$ eigenvalue of  $\rho^B_{x_a}$ and 
	$\ket{e_{\rho_{x_a};y}}$ is the corresponding eigenvector.
	
	We can think of $\{\rho^B_{x_a}\}$ as a classical-quantum (\emph{c-q}) channel
	where the input is some $x_a \in \mcal{X}$ and the output is the 
	corresponding quantum state $\rho^B_{x_a}$.
	If the channel is memoryless, then for  each input sequence 
	$x^n=x_1x_2\cdots x_n$ we have the corresponding tensor product output state: 
	\be
		\rho^{B^n}_{x^n} = \rho^{B_1}_{x_1} \otimes \rho^{B_2}_{x_2} \otimes \cdots \otimes \rho^{B_n}_{x_n}.
	\ee


\subsection{Quantum conditional typicality}

\begin{figure}
\begin{center}

	\begin{tikzpicture}[node distance=2.0cm,>=stealth',bend angle=45,auto]

	  \begin{scope}

		
		


		\draw  (3.3, 3.3) coordinate (XnCenter);
		\path (XnCenter) ++(35:20mm) coordinate (XnCorner1);
		\path (XnCenter) ++(165:20mm) coordinate (XnCorner2);		
		\path (XnCenter) ++(280:20mm) coordinate (XnCorner3);

		\draw (XnCorner1) -- (XnCorner2) -- (XnCorner3) -- cycle;						
		\path (XnCenter) ++(0,0.2) ++(35:12mm) coordinate (XnCorner1in);
		\path (XnCenter) ++(165:12mm) coordinate (XnCorner2in);		
		\path (XnCenter) ++(0.2,0)++(280:12mm) coordinate (XnCorner3in);

			\draw (XnCorner1in)  [rounded corners=10pt,line width=1pt] -- (XnCorner2in) -- (XnCorner3in) -- cycle;		

		%
		\node at (2.2,4.2) {$\mathcal{X}^n$};
		
		\node (OnXnTyp) at (3.5,2.3) {};
		\node at (2,1.5) {$\mcal{T}_\delta^{(n)}\!(X)$} 
			edge [post] node[]  	{}  (OnXnTyp.center);

		%





		\node [circle,fill,draw,inner sep=0pt,minimum size=2pt,label=above:$\xnlatex$,  xshift=1mm,yshift=-3mm] (xn) at (XnCenter)  {};

		\node at (8,4.8) {$\mathcal{H}^{B^n}$};
		\draw   (9, 2.5) ellipse (15mm and 25mm); 

		\draw [line width=1pt, fill=black!8]  (9,2.5) ellipse (13mm and 20mm); 
		\node (OnBnTyp) at (9.9,3.9) {};
		\node at (10.7,4.7) {$\Pi_{\bar{\rho}}$} 
			edge [post] node[]  	{}  (OnBnTyp.center);

		\node[ellipse,draw=black!80,fill=black!23,thick,inner sep=3pt,
			rotate=90,
			label=right:$\Pi_{\rho_{x^n}}$] 
			(XWCondOut) 
			at (8.5,2.7) { $\phantom{ {aa} }$}; 		
		\draw (XWCondOut.east) -- (xn);
		\draw (XWCondOut.west) -- (xn);
		
	  \end{scope}
	\end{tikzpicture}
\end{center}
\label{fig:Q-cond-typ-set}
\caption{Illustration of a conditionally typical subspace for some sequence $x^n$,
and the output-typical subspace.}
\end{figure}
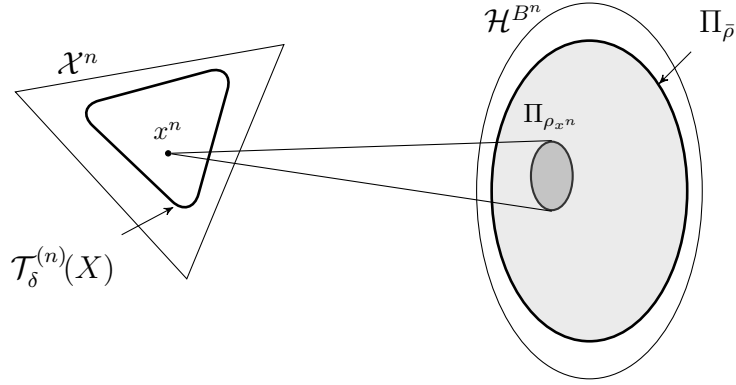

	\medskip
	\noindent
	{\bf Conditionally typical projector \ \ }
	Consider the ensemble $\left\{  p_{X}\!\left(  x_a\right)  ,\rho_{x_a}\right\}$.
	The choice of distributions induces the following classical-quantum state:
	\be
		\rho^{XB} =
	       	\sum_{x_a} 
			p_{X}\!\left(x_{a}\right)
			\ketbra{x_a}{x_a}^{X} \!\!
			\otimes \!
			\rho^{B}_{x_a}.
	\ee

	We define 
	$
		H (B|X)_\rho
		\equiv
		 \sum_{x_a \in \mcal{X}} 
			p_{X}(x_a) 
 			H (\rho_{x_a})
	$
	to be the conditional entropy of this state.
	Expressed in terms of the eigenvalues of the signal states,
	the conditional entropy becomes
	\be
		H (B|X)_\rho
		\equiv
		H(Y|X)
		\equiv
		\sum_{x_a}p_{X}(x_a)H(Y|x_a),
	\ee	
	where $H(Y|x_a) = - \sum_y p_{Y|X}(y|x_a) \log p_{Y|X}(y|x_a)$
	is the entropy of the eigenvalue distribution shown in \eqref{eq:cond-spectral-decomp-rhoBG}.

	We define the $x^n$-conditionally typical projector
	as follows:
	\be
		\Pi^{n}_{\rho^B_{x^n}, \delta}
		=  \sum_{y^n \in 
			\mcal{T}^{(n)}_\delta(Y|x^n)
			}
			\ket{e_{\rho_{x^n};y^n}}\!\bra{e_{\rho_{x^n};y^n}},
		\label{eqn:cond-typ-projector-def}
	\ee
	where the set of conditionally typical eigenvalues $\mcal{T}^{(n)}_\delta(Y|x^n)$
	consists of 
	all sequences $y^n$ which satisfy: 
	\be
		\mcal{T}^{(n)}_\delta(Y|x^n) 
		\! \equiv \!
		\left\{ \!
			y^n 
			\colon \! \left|  -\frac{\log p_{Y^n|X^n}(y^n|x^n)}{n}  - H(Y|X) \right| \! \leq \delta 
		\right\},%
	\ee	
	with $p_{Y^n|X^n}(y^n|x^n) = \prod_{i=1}^n p_{Y|X}(y_i|x_i)$.
	
	%
	The states $\ket{e_{\rho_{x^n};y^n}}$ are built from tensor products of eigenvectors
	for the individual signal states:
	\be
		\ket{e_{\rho_{x^n};y^n}} = \ket{e_{\rho_{x_1};y_1}} \otimes \ket{e_{\rho_{x_2};y_2}} \otimes \cdots \otimes \ket{e_{\rho_{x_n};y_n}},
	\ee
	where the string $y^n = y_1y_2\ldots y_i\ldots y_n$ varies over different choices of bases for $\mcal{H}^B$.
	For each symbol  $y_i=b \in \{1,2,\ldots,d_B\}$ we select $\ket{e_{\rho_{x_a};b}}$:
	the b$^\textrm{th}$ eigenvector from the eigenbasis of $\rho_{x_a}$ corresponding to the letter $x_i = x_a \in \mathcal{X}$.
	%

	The following bound on the rank of the conditionally typical projector
	applies:
	\be
		\Tr\{  \Pi^{n}_{\rho^B_{x^n}, \delta} \} \leq 2^{n[H(B|X)_\rho + \delta] }.
		\label{eqn:bound-on-size}
	\ee

\section{Closing remarks}

In the next chapter, we will show how the properties of the typical sequences 
and typical subspaces can be used to construct coding theorems for 
classical and classical-quantum channels.


\chapter{Point-to-point communication}

									\label{chapter:p-to-p}

	In this chapter we describe the point-to-point communication 
	scenario in which there is a single sender and a single receiver.
	In Section~\ref{sec:clas-chan-code}, we review Shannon's channel coding theorem and give the details
	of the achievability proof in order to introduce the idea of \emph{random coding}
	in its simplest form.
	Our presentation is somewhat unorthodox since we use only the properties
	of the conditionally typical sets and not the jointly typical sets.
    Though, following this approach allows us to directly generalize our proof techniques
    to the quantum case.

	In Section~\ref{sec:quan-chan-code} we will discuss the Holevo-Schumacher-Westmoreland (HSW) Theorem
	and show an achievability proof.
	We do so with the purpose of introducing important background material on the
	construction of quantum decoding operators.
	We show how to construct a decoding POVM defined in terms of the conditionally typical projectors.
	%
	Readers interested only in the essential parts should consult 
	Lemma~\ref{lem:HN-inequality} and Lemma~\ref{lem:gentle-operator}, 
	since they will be used throughout the remainder of the text.




%


\section{Classical channel coding}
															\label{sec:clas-chan-code}

	The fundamental problem associated with communication channels 
	is to calculate and formally prove their capacity for information transmission.
	We can think of the use of a channel $\mcal{N}$ as a \emph{communication resource},
	of which we have $n$ instances.
	Each use of the channel is assumed to be independent,
	and modelled by the conditional probability distribution
	$p_{Y|X}(y|x)$, where $x$ and $y$ are elements from 
	the  finite sets $\mcal{X}$, $\mcal{Y}$.
	This is called the discrete memoryless setting.

	Our goal is to study the 
	\emph{rate} $R$ at which the channel $\mcal{N}$ can be converted
	into copies of the noiseless binary channel $[c\to c] \equiv 
	\delta(x,y), \ x,y \in \{0,1\}$,
	which represents the canonical unit resource of communication.
	This conversion can be expressed as follows:
	\be
		n \cdot \mcal{N}  \ \ \overset{ (1-\epsilon)}{\longrightarrow} \ \ nR \cdot  [c \to c].
        \label{ptopyo}
	\ee
	This equation describes a protocol in which $n$ units of the noisy communication 
	resource $\mcal{N}$ are transformed into $nR$ bits of noiseless transmission,
	and the protocol succeeds with probability $(1-\epsilon)$.
	Note that we allow the communication protocol to fail
	with probability $\epsilon$, but $\epsilon$ is an arbitrarily small number for sufficiently large $n$.
	%
	To prove that the rate $R$ is \emph{achievable},
	one has to describe the coding 
	strategy and prove that 
    the probability of error for that strategy can be made arbitrarily small.
	Usually, the right hand side in equation \eqref{ptopyo} 
	is measured as the number of different messages
	$\mcal{M} \equiv \{1, 2, \ldots, 2^{nR} \}  \equiv [1:2^{nR}]$
	that can be transmitted using $n$ uses of the channel.
	One can think of the $nR$ individual bits of the message
	as being noiselessly transmitted to the receiver.
	The channel coding pipeline can then be described as follows:
	
	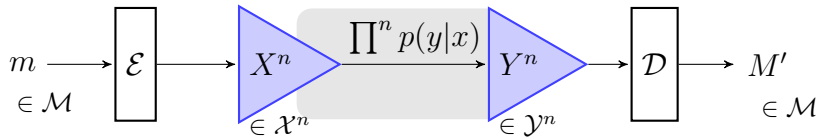
\begin{figure}[htbp]
	\begin{center}

	\begin{tikzpicture}[node distance=1.8cm,>=stealth',bend angle=45,auto]
	  \tikzstyle{cnode}=[isosceles triangle,isosceles triangle apex angle=60,thick,draw=blue!75,fill=blue!20,minimum height=0.3cm] 
	  \tikzstyle{processing}=[rectangle,thick,draw=black,fill=white,minimum height=1.5cm] 
	  \tikzstyle{measurment}=[rectangle,thick,draw=blue!75,fill=blue!20,minimum height=0.3cm] 
	  \tikzstyle{qnode}=[circle,thick,draw=blue!75,fill=blue!20,minimum size=6mm] 
	  \tikzstyle{every label}=[black, font=\footnotesize]

	  \begin{scope}
		\node (m) [ label=below:$\ \ \ \ \ \in \mcal{M}$     ]  {$m$} ;
		\node [processing,right of=m,xshift=-3mm] (Enc)  {$\mcal{E}$} 
				edge  [pre]				node[swap]  	{}	(m);
		\node [cnode] (Xn) [ label=below:$\ \ \in \mcal{X}^n$, right of=Enc]                {$X^n$}
				edge  [pre]				node[swap]  	{}	(Enc);
		\node [cnode] (Yn) [ label=below:$\ \ \in \mcal{Y}^n$, right of=Xn, xshift=1.5cm]                {$Y^n$}
				edge  [pre]				node[swap]  	{$ \prod^n p(y|x)$}	(Xn);	
		\node [processing,right of=Yn] (Dec)  {$\mcal{D}$} 
				edge  [pre]				node[swap]  	{}	(Yn);	
		\node (Mpr) [ label=below:$\ \ \ \ \ \in \mcal{M}$,right of=Dec,xshift=-3mm   ]  {$M^\prime$} 
				edge  [pre]				node[swap]  	{}	(Dec);	
	  \end{scope}
	  \begin{pgfonlayer}{background}
	    \filldraw [line width=4mm,join=round,black!10]
	      ([xshift=-4mm]Xn.north -| Xn.east) rectangle (Yn.south -| Yn.west);
	  \end{pgfonlayer}
	\end{tikzpicture}

	\end{center}
	\caption{\small Classical channel coding setup.
			The diagram shows the encoding, transmission and decoding
			steps of a communication protocol that uses $n$ copies of the 
			classical channel  $\mcal{N}=(\mcal{X}, p_{Y|X}(y|x), \mcal{Y})$. }
	\label{fig:classical-transmission}

	\end{figure}

The probability of error when sending message $m$ is defined as
$p_e(m) \equiv \textrm{Pr}\{ M^\prime \neq m \}$, 
where $M^\prime \equiv \mcal{D} \circ \mcal{N}^n \circ \mcal{E}(m)$
is the random variable associated with the output of the protocol.
The average probability of error over all messages is
\be
	\bar{p}_e \equiv  \frac{1}{|\mcal{M}|}	\sum_{m \in \mcal{M}} \textrm{Pr}\{ M^\prime \neq m \}.
\ee
This is the quantity we have to bound when we perform an \emph{error analysis}
of some coding protocol.

	    \begin{definition}%
		An $(n,R,\epsilon)$ coding protocol consists of
		a message set $\mcal{M}$, where $|\mcal{M}|=2^{nR}$,
		an encoding map
		$\mcal{E}:\mcal{M}\to \mcal{X}^n$
		described by a codebook 	$\{x^n(m)\}_{m \in \mathcal{M}}$,
		and a decoding map $\mcal{D}:\mcal{Y}^n \to \mcal{M}$
		such that the average probability of error 
		 is bounded from above as $\overline{p}_{e} \leq \epsilon$.

	    \end{definition}  
		
		A rate  $R$ is \textit{achievable} if there exists
		an $\left(  n,R-\delta,\epsilon\right)$ coding protocol 
		for all $\epsilon,\delta>0$  as $n\to \infty$.
		%
		
	\subsection{Channel capacity}
	
		The capacity $C$ of a channel is the maximum of the rates
		$R$ that are achievable, and is established in Shannon's 
		channel coding theorem.
		
	\begin{theorem}[Channel capacity   \cite{S48,feinstein1954new}] 
		\label{thm:shannon-ch-cap}
		The communication capacity of a
		discrete memoryless channel $(\mcal{X}, p_{Y|X}(y|x), \mcal{Y})$ is given by
		\be
			C = \max_{p_X} I(X;Y),
		\ee
		where the optimization is taken over all possible input distributions $p_X(x)$.
		The mutual information is calculated on the
		induced joint probability distribution
		\be
			(X,Y) \sim p_{XY}(x,y)=p_{X}(x)p_{Y|X}(y|x).
		\ee
	\end{theorem}

	The proof of a capacity theorem usually contains two parts:
	\begin{itemize}
	
	\item
        
		A direct coding part that shows that for all $\epsilon,\delta >0$,
		there exists a codebook $\mcal{E}(m)\equiv\{x^n(m)\}$ of rate 
		$R = C-\delta$ and a decoding map $\mcal{D}$ with
		average probability of error $\bar{p}_e \leq \epsilon$.
		
	\item	
		A converse part that shows that the rate $C$ is the maximum rate possible.
		A converse theorem establishes that the probability of error for a 
		coding protocol $(n,C+\delta,\epsilon)$ is bounded away from zero (weak converse),
		or that the probability of error goes exponentially to 1 (strong converse).
		
	\end{itemize}

	%






	\begin{proof}
	We give an overview of the achievability proof of Theorem~\ref{thm:shannon-ch-cap}
	in order to introduce key concepts, which will be used in the other proofs in this thesis.
	%
	
	We use a random codebook with $2^{nR}=|\mcal{M}|$ 
	codewords $x^n \in \mcal{X}^n$ generated independently
	from the product distribution $p_{X^n}(x^n)=\prod^n p_X(x_i)$.
	When the sender wants to send the message
	$m \in \mcal{M}$, she will input the $m^{\textrm{th}}$
	codeword, which we will denote as $x^n(m)$.
	Let $Y^n$ denote the resulting output of the channel.
	The distribution on the output symbols induced
	by the input distribution is $p_Y(y) \equiv \sum_{x} p_{Y|X}(y|x) p(x)$,
	and define the set of output-typical sequences
	$\mcal{T}_\delta^{(n)}(Y)$ according to the distribution $p_Y$.
	For any sequence $x^n$, denote the set of conditionally
	typical output sequences $\mcal{T}_\delta^{(n)}(Y| x^n )$.
	%
	
	%
	Given the output of the channel $y^n$,
	the receiver will use the following algorithm:
	\begin{enumerate}
	\item
		If $y^n \not \in \mcal{T}_\delta^{(n)}(Y)$, then an error is declared. 

	 \item
		Return $m$ if $y^n$ is an element of 
		the conditionally typical set $\mcal{T}_\delta^{(n)}(Y| x^n(m) )$.
		Report an error if no match or multiple matches are found.
	\end{enumerate}

	We now define the three types of errors that may occur in the protocol
	when the message $m$ is being sent.
	\begin{description}
	
	\item[$\mathbf{(E0)}$:]%
	The event that the channel output 
	$Y^n$ is not 
	output-typical: $\{ Y^n \not \in  \mcal{T}_\delta^{(n)}(Y) \}$.

	\item[$\mathbf{(E1)}$:]%
	The event that the channel output sequence $Y^n$ 
	is not in the conditionally typical set 
	$\{ Y^n \not \in \mcal{T}_\delta^{(n)}(Y| x^n(m) ) \}$,
	which corresponds to the message $m$.

	\item[$\mathbf{(E2)}$:]%
	The event that $Y^n$ is output-typical and it
	falls in the conditionally typical	set for another message:
	\be
	\!\!\!\!\!\!\!
	\{ Y^n \in  \mcal{T}_\delta^{(n)}(Y) \}
	\cap
	\left(
	\bigcup_{m^\prime \neq m}
	\{ Y^n \in \mcal{T}_\delta^{(n)}(Y| x^n(m^\prime) ), m^\prime \neq m  \}
	\right)\!\!.
	\ee

	\end{description}

	We can bound the probability of all three events
	when a random codebook is used, that is,
	we will take the expectation over the random choices of the symbols
	for each codeword.
	We define the expectation of an event as the
	expectation of the associated indicated random variable.
	
The bound  $\NExpX \mathbf{(E0)} \leq \epsilon$ follows from  \eqref{cc7BG}. 
The crucial observation for the proof is to use the symmetry of the code
construction: if the codewords for all the messages are constructed identically,
then it is sufficient to analyze the probability of error for any one fixed message.
%
We obtain a bound $\NExpX \mathbf{(E1)} \leq \epsilon$ from \eqref{cc4BG}.

In order to bound the probability of error event $\mathbf{(E2)}$, 
we will use the \emph{classical packing lemma}, Lemma~\ref{lem:classical-packing}
in Appendix~\ref{apdx:classical-coding-theorem}.
Using the packing lemma with $U=\emptyset$, we obtain a bound 
on the probability that the conditionally typical sets 
for different messages will overlap. 
We can thus bound the expectation of the probability of error event $\mathbf{(E2)}$
as follows:
 \begin{align*}
\NExpX 
 \textrm{Pr}\!\left\{ \mathbf{(E2)}  \right\}
	 &\leq
		 |\mathcal{M}| \;2^{-n[ I(X;Y) - \delta ]}.
\end{align*}

%
%

\bigskip

We can now use the union bound to bound the overall probability
of error for our code as follows:
 \begin{align*}
\NExpX \left\{   \bar{p}_e \right\}
	&=
	 	\NExpX \textrm{Pr}\!\left\{ \mathbf{(E0)} \cup \mathbf{(E1)} \cup \mathbf{(E2)} \right\} \\
 	&\leq
		\NExpX \textrm{Pr}\!\left\{ \mathbf{(E0)}  \right\}
		 +
		 \NExpX \textrm{Pr}\!\left\{  \mathbf{(E1)} \right\}
		  + 
		   \NExpX \textrm{Pr}\!\left\{ \mathbf{(E2)} \right\} \\
 &\leq
		\qquad \ 
		\epsilon 
		\qquad  \ 
		+ \ 
		\qquad  \
		\epsilon  
		\qquad \ 
		+  \ 
	 |\mathcal{M}|\;2^{-n[ I(X;Y) - \delta ]} \\
 & =
		\qquad \ 
		\epsilon 
		\qquad  \ 
		+ \ 
		\qquad  \
		\epsilon  
		\qquad \ 
		+  \ 
		 2^{-n[ I(X;Y) - R - \delta ]}.
\end{align*}

	Thus,
	in the limit of many uses of the channel, we have:
	\be
		\NExpX \left\{   \bar{p}_e \right\} \ \ \leq \ \ \epsilon^\prime,
	\ee
	provided the rate $R \leq I(X;Y)-2\delta$.

	The last step is called \emph{derandomization}. 
	If the expected probability of error of a random codebook can be bounded as above,
	then there must exist a particular codebook with $\bar{p}_e \leq  \epsilon^\prime$,
	which completes the proof.
	\end{proof}
	
	Note that it is possible to use an \emph{expurgation} step
	and throw out the worse half of the codewords 
	in order to convert the bound on the average probability of 
	error $\bar{p}_e$ into a bound on the maximum probability 
	of error $\bar{p}_e^{\max}=\max_m p_e(m)$ \cite{CT91}.





\section{Quantum communication channels}
	\label{sec:quantum-info-theory}


	\begin{wrapfigure}[6]{r}{0pt}%
	\begin{tikzpicture}[node distance=2.0cm,>=stealth',bend angle=45,auto,scale=3]

	  \begin{scope}
		\node [qnode] (Tx) [ label=left:Tx    ]                            {\footnotesize $\sigma^A$} ; 
		\node [qnode] (Rx) [ label=right:Rx, right of=Tx, xshift=4mm] {\footnotesize $\rho^B$} 
			edge  [pre]             node[swap]  {$\mcal{N}^{A\to B}$}    (Tx) ;
	  \end{scope}
	  \begin{pgfonlayer}{background}
	    \filldraw [line width=4mm,join=round,black!10]
	      ([xshift=0mm,yshift=0.5mm]Tx.north -| Tx.east) rectangle ([xshift=1mm]Tx.south -| Rx.west);
	  \end{pgfonlayer}
	\end{tikzpicture}
	\caption{\small A point-to-point quantum channel
			$\mcal{N}^{A\to B}$. }
		\label{fig:qqchannel}
	\end{wrapfigure}

	%
	A quantum channel $(\mcal{H}^A, 	\mcal{N}^{A\to B}, \mcal{H}^B )$ 
	is described as a completely positive trace-preserving 
	map 	$\mcal{N}^{A\to B}$ which takes a quantum system in state $\sigma^A \in  \mcal{D}(\mcal{H}^A)$ as input
	and outputs a quantum system $\rho^B \in \mcal{D}(\mcal{H}^B)$.
	Figure~\ref{fig:qqchannel} shows an example of such a channel.
	%
	In recent years, the techniques of classical information theory have been
	extended to the study of quantum channels. For a review of  the subject
	see \cite{wilde2011book}.
	
	In addition to the standard problem of \emph{classical} transmission of information
	(denoted $[c \to c]$), for quantum channels we can study the 
	transmission of quantum information (denoted $[q \to q]$).
	If pre-shared entanglement between Transmitter
	and Receiver is available, it can be used in order to improve the communication rates
	using an \emph{entanglement-assisted} protocol.
	%
	%
        There are multiple communication tasks and different capacities 
        associated with each task for any given quantum channel $\mathcal{N}$ \cite{BSST99}.
        Some of the possible communication tasks, along with their associated
        capacities are:
        \begin{itemize}
            \item Classical data capacity: $\mathcal{C}\!\left(\mathcal{N}\right)$ \vspace{-0.1in}
            \item Quantum data capacity: $\mathcal{Q}\!\left(\mathcal{N}\right)$ \vspace{-0.1in}
            \item Entanglement-assisted classical data capacity: $\mathcal{C}_{\texttt{E-A}}\!\left(\mathcal{N}\right)$ \vspace{-0.1in}
            \item Entanglement-assisted quantum data capacity: $\mathcal{Q}_{\texttt{E-A}}\!\left(\mathcal{N}\right)$
        \end{itemize}
	The latter two are actually equivalent up to a factor of $2$, because we can 
	use the \emph{superdense coding} and 
	\emph{quantum teleportation} protocols to convert between them
	in the presence of free entanglement \cite{PhysRevLett.69.2881,PhysRevLett.70.1895}.
	
	In the context of quantum information theory, 
	pre-shared quantum entanglement between sender and receiver
	must be recognized as a communication resource.
	We denote this resource $[qq]$ and must take into account the rates 
	at which it is consumed or generated as part of a communication protocol
	\cite{DHW05b}.
	It is interesting to note that shared randomness (denoted $[cc]$),
	which is the classical equivalent of shared entanglement, 
	does not increase the capacity of point-to-point  classical channels.

	\subsubsection{Classical-quantum channels}

	%
	\begin{wrapfigure}{r}{0pt}%
	\begin{tikzpicture}[node distance=2.0cm,>=stealth',bend angle=45,auto,scale=3]
  \tikzstyle{cnode}=[isosceles triangle,isosceles triangle apex angle=60,thick,draw=blue!75,fill=blue!20,minimum height=6.5mm,inner sep=0mm] 
%
	  \tikzstyle{qnode}=[circle,thick,draw=blue!75,fill=blue!20,minimum size=6mm] 
	  \tikzstyle{every label}=[black, font=\footnotesize]

	  \begin{scope}
		\node [cnode] (Tx) [ label=left:Tx    ]                            {\small $x$} ; 
		\node [qnode] (Rx) [ label=right:Rx, right of=Tx,xshift=4mm] {\tiny $\rho_x^B$} 
			edge  [pre]             node[swap]  {$\mcal{N}^{X\to B}$}    (Tx) ;
	  \end{scope}
	  \begin{pgfonlayer}{background}
	    \filldraw [line width=4mm,join=round,black!10]
	      ([xshift=-1mm]Tx.north -| Tx.east) rectangle (Tx.south -| Rx.west);
	  \end{pgfonlayer}
	\end{tikzpicture}
	\caption{\small A point-to-point c-q channel
			$\{ \rho_x \}$. }
	\end{wrapfigure}
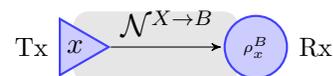

	In the previous section we introduced some of the main communication problems of 
	quantum information theory.
	The focus of this thesis will be the study of 
	\emph{classical communication}  ($[c\to c]$)
	over quantum channels, with no entanglement assistance.
	%
	%
	For this purpose, we will use the 
	\emph{classical-quantum} (c-q) channel model,
	which corresponds to the use of a quantum channel where the 
	Sender is restricted to sending a finite set of \emph{signal} states
	$\{ \sigma^A_x \}_{x \in \mcal{X}}$.
	If we consider the choice of the signal states $\{ \sigma^A_x \}$ to be 
	part of the channel, we obtain a channel with classical inputs $x \in \mcal{X}$
	and quantum outputs: $\mcal{N}^{X \to B}(x) \equiv \mcal{N}^{A\to B}(\sigma_x^A)$.
	%
	Note that
	a classical-quantum channel 
	$(\mcal{X}, \mcal{N}^{X\to B}(x)\!\equiv\! \rho^B_x, \ \mcal{H}^B)$
	is fully specified by the finite set of output states $\{ \rho^B_x \}$
	it produces for each of the possible inputs $x \in \mcal{X}$.
	%
	%
	%
	%
	%
	%
	%
	%
	This channel model is a useful abstraction for studying the 
	transmission of classical data over quantum channels.
	Any code construction for a c-q channel can be augmented 
	with an optimization over the choice of signal states
	$\{ \sigma^A_x \}_{x \in \mcal{X}}$ to obtain
	a code for a quantum channel.
	The  Holevo-Schumacher-Westmoreland Theorem
	establishes 
	the classical capacity of the classical-quantum channel \cite{H98,SW97}.
	The strong converse was later proved in \cite{ogawa1999strong}.

	\subsection{Classical-quantum channel coding}
												\label{sec:quan-chan-code}

	The quantum channel coding problem for 
	a point-to-point classical-quantum channel
	$(\mcal{X}, \mcal{N}^{X\to B}(x)\!\equiv\! \rho^B_x, \ \mcal{H}^B)$
	is studied in the following setting.

	\begin{figure}[htbp]
	\begin{center}

	\begin{tikzpicture}[node distance=2cm,>=stealth',bend angle=45,auto]
	  \tikzstyle{cnode}=[isosceles triangle,isosceles triangle apex angle=60,thick,draw=blue!75,fill=blue!20,minimum height=0.3cm] 
	  \tikzstyle{processing}=[rectangle,thick,draw=black,fill=white,minimum height=1.5cm] 
	  \tikzstyle{measurment}=[rectangle,thick,draw=blue!75,fill=blue!20,minimum height=0.3cm] 
	  \tikzstyle{qnode}=[circle,thick,draw=blue!75,fill=blue!20,minimum size=6mm] 
	  \tikzstyle{every label}=[black, font=\footnotesize]

	  \begin{scope}
		\node (m) [ label=below:$\ \ \ \ \ \in \mcal{M}$     ]  {$m$} ;
		\node [processing,right of=m,xshift=-3mm] (Enc)  {$\mcal{E}$} 
				edge  [pre]				node[swap]  	{}	(m);
		\node [cnode] (Xn) [ label=below:$\ \ \in \mcal{X}^n$, right of=Enc,xshift=-3mm]                {$X^n$}
				edge  [pre]				node[swap]  	{}	(Enc);
	\node [qnode] (Bn) [ label=below:$\ \ \in \mcal{H}^{B^n}$, right of=Xn, xshift=1cm]   {$\rho^{B^n}_{X^n}$}
				edge  [pre]				node[swap]  	{$ \mcal{N}^{\otimes n}$}	(Xn);	
		\node [processing,right of=Bn,xshift=-1mm] (Dec)  {$\left\{ \Lambda_{X^n(m)} \right\} $} 
				edge  [pre]				node[swap]  	{}	(Bn);	
		\node (Mpr) [ label=below:$\ \ \ \ \ \in \mcal{M}$,right of=Dec, xshift=-2mm]  {$M^\prime$} 
				edge  [pre]				node[swap]  	{}	(Dec);	
	  \end{scope}
	  \begin{pgfonlayer}{background}
	    \filldraw [line width=4mm,join=round,black!10]
	      ([xshift=-4mm]Xn.north -| Xn.east) rectangle (Bn.south -| Bn.west);
	  \end{pgfonlayer}
	\end{tikzpicture}

	\end{center}
	\caption{\small HSW coding setup.}
	\label{fig:quantum-transmission}

	\end{figure}
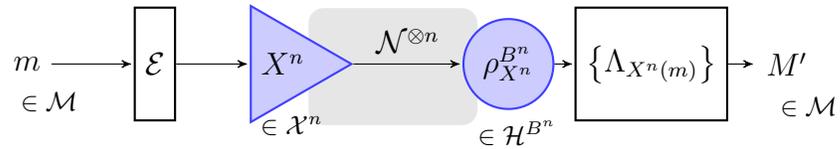


		Let $x^{n}\!\left(  m\right) \equiv x_1x_2\cdots x_n  \in \mathcal{X}^n$ be the codeword
		which is input to the channel when we want to send message $m$.
		%
		The output of the channel will be the $n$-fold tensor product state:
		\be
			\mathcal{N}^{\otimes n}\!\left( x^{n}(m) \right)
			\equiv			
			\rho_{x^{n}(m)}^{B^{n} }
			\equiv 
			\rho_{x_{1}(m)}^{B_1}
			\otimes
			\rho_{x_{2}(m)}^{B_2}
			\otimes
			\cdots
			\otimes
			\rho_{x_{n}(m)}^{B_n}.
		\label{eqn:output-state-tensor-prod}
		\ee

	To extract the classical information encoded into this
	state, we must perform a quantum measurement.
	The most general quantum measurement is described
	by a positive operator-valued measure (POVM) 
	$\left\{ \Lambda_{m}\right\}_{m\in \mcal{M}}$ 
	on the system $B^n$. 
	To be a valid POVM, the set $\{\Lambda_{m}\}$ of 
	$|\mcal{M}|$ operators should all be positive semidefinite
	and sum to the identity:  $\Lambda_{m} \geq 0, \,\,\, \sum_{m}\Lambda_{m}=I$.
	%
	%

	In the context of our coding strategy,	 the
	decoding measurement aims to distinguish the $|\mcal{M}|$ 
	possible states of the form \eqref{eqn:output-state-tensor-prod}.
	The advantage of the quantum coding paradigm is that it allows for 
	joint measurements on all the outputs of the channel,
	which is more powerful than measuring the systems individually.
	
	We define the average probability of error
	for the end-to-end protocol as
	\begin{align}
		\bar{p}_{e}
		\equiv 
			\frac{1}{|\mathcal{M}|}\sum_{m}					
			\text{Tr}\!
			\left\{  
				\left(  I-\Lambda^{B^{n}} _{x^n(m)} \right)
				\rho_{x^n(m)}^{B^{n}} 
			\right\},
	\end{align}
	where the operator  $\left(  I-\Lambda^{B^{n}} _{x^n(m)} \right)$ 
	corresponds to the complement of the correct decoding outcome.

	%

	    \begin{definition}%
		An $(n,R,\epsilon)$ classical-quantum coding protocol consists of
		a message set $\mcal{M}$, where $|\mcal{M}|=2^{nR}$,
		an encoding map
		$\mcal{E}:\mcal{M}\to \mcal{X}^n$
		described by a codebook 	$\{x^n(m)\}_{m \in \mathcal{M}}$,
		and a decoding measurement
		(POVM) $\{ \Lambda_{x^n(m)} \}_{m \in 
		\mcal{M} }$
		such that the average probability of error 
		 is bounded from above as $\overline{p}_{e} \leq \epsilon$.

	    \end{definition}

%
%

%

	\begin{theorem}[HSW Theorem  \cite{H98,SW97}] 
												\label{thm:HSWtheorem}
		The classical communication capacity of a
		classical-quantum channel $(\mcal{X}, \rho^B_x , \mcal{H}^B)$ is
		given by:
		\be
			\mcal{C}(\mcal{N})  
			= 		 \max_{p_X} 	I(X;B)_\theta 
		\ee
		where the optimization is taken over all possible input distributions $p_X$,
		and where entropic quantities are calculated with respect
		to the following state: 
		\be
			\theta^{XB} = \sum_x p_X(x) \ \ketbra{x}{x}^X \otimes \rho^B_x.
			\label{eqn:code-state}
		\ee
	\end{theorem}
	%


		%

		


\bigskip

The classical-quantum state $\theta^{XB}$ is the state with 
respect to which we will calculate  
mutual information quantities.
We call this state the \emph{code state} and it extends 
the classical joint probability distribution induced by a channel,
when the input distribution $p_X$ is used to construct the codebook: $p_X(x)p_{Y|X}(y|x)$.
In the case of the classical-quantum channel,
the outputs are quantum systems.
Information quantities taken with respect to classical-quantum
states are called ``Holevo'' quantities in honour of Alexander Holevo who was first
to recognize the importance of this expression
by proving that it is an upper bound to the accessible information of an ensemble
\cite{holevo1973bounds,holevo1979capacity}.
Holevo quantities are expressed as a difference of two entropic terms:
\be
 	I(X;B)_\theta 
	\equiv 	H(B)_\theta - H(B|X)_\theta
	\equiv 	H\!\bigg(\sum_x p_X(x) \rho_x^B \bigg) - \sum_x p_X(x)H(\rho_x^B).
\ee
Holevo quantities are in some sense partially classical, 
since the entropies are with respect to quantum systems, 
but the conditioning is classical.

	\subsubsection{Quantum decoding}

		When devising coding strategies for classical-quantum channels, 
		the main obstacle to overcome is the construction of a decoding POVM 
		that correctly identifies the messages.
		%
		Using the properties of quantum typical subspaces 
		we can construct a set of positive operators $\{ P_m \}_{m\in\mcal{M}}$
		which, analogously to the classical conditionally typical indicator functions,
		are good at detecting ($\text{Tr}\!\left[P_m\;\rho_m\right] \geq 1- \epsilon$)
		and distinguishing  ($\text{Tr}\!\left[P_{m}\;\rho_{m^\prime\neq m}\right] \leq \epsilon$)
		the output states produced by each message.
		We can construct a valid POVM by \emph{normalizing} these operators:
		\begin{align}
		\Lambda_{m} &  \equiv 
		\left(  
			\sum_{k}P_{k}
		\right)^{\!\!\!-1/2}
		\!\! P_{m}
		\left(
			\sum_{k} P_{k}
		\right)^{\!\!-1/2}, \label{eq:square-root-POVM-generic} 
		\end{align}
		so that we will have $\sum_{m}\Lambda_{m}=I$.
		This is known as the \emph{square root} measurement 
		or the \emph{pretty good} measurement \cite{H98,SW97}.
		

\bigskip

		%
		The achievability proof of Theorem~\ref{thm:HSWtheorem} 
		is based on the properties of typical subspaces 
		and the square root measurement.
		We construct a set of unnormalized positive operators
		\be
			\pPm  
			\equiv 
			\pPIavg  \ \pPIm \ \pPIavg,
			\label{eqn:pPm-def}
		\ee
		where $\pPIm \equiv  \Pi_{\rho_{x^n(m)},\delta}^{B^n}$
		is the conditionally typical projector that corresponds to the input
		sequence $x^n(m)$
		and $\pPIavg \equiv  \Pi_{\bar{\rho}^{\otimes n}\!,\delta }^{B^n}$
		is the output-typical projector for the average output
		state $\bar{\rho} = \sum_x p_X(x) \rho_x^B$.
		The operator ``sandwich'' in equation \eqref{eqn:pPm-def} corresponds
		directly to the decoding criteria used in the classical coding theorem.
		We require the state to be in the output-typical subspace 
		\emph{and} inside the conditionally typical subspace for the correct codeword $x^n(m)$.
		The decoding POVM is then constructed as in \eqref{eq:square-root-POVM-generic}.
		
		By using the properties of the typical projectors,
		we can show that the probability of error of this coding
		scheme vanishes provided $R \leq I(X;B) -\delta$.
		%
		%
		An effort has been made to present the proofs of the 
		classical and quantum coding theorems  in a similar fashion  
		in order to highlight similarities in the reasoning.
		%


\section{Proof of HSW Theorem}

	\label{sec:HSW-proof-apdx}

	In this section we give the details of the POVM construction
	and the error analysis for the decoder used by the receiver
	in the HSW Theorem.

	%
	Recall the classical-quantum state (\ref{eqn:code-state}), with 
	respect to which our code is constructed:
	\be
		\theta^{XB} =
		\sum_{x}  p_X\!\left(  x\right)  
		\left\vert x \right\rangle\!\!\left\langle x\right\vert ^{X}
		\otimes
		\rho_{x}^{B}.
	\ee

	For each input sequence $x^n$, 
	there is a corresponding $\delta$-conditionally typical projector:
	$\Pi_{x^n} 
	 \equiv
	  \Pi_{\rho_{x^{n}}, \delta }^{B^{n}}$. 

	Define also the average output state
	$	\bar{\rho} 	\equiv 
		\sum_{x}  p_X\!\left(  x\right)  \rho_{x}^{B}$,
	and the corresponding average-output-typical projector
	$
	\pPIavg 
	\equiv 
	 \Pi_{\bar{\rho}^{\otimes n},\delta }^{B^{n}}$.
	
	The Receiver constructs a decoding POVM $\{ \Lambda_{m} \}_{m \in \mcal{M}}$
	by starting from the \emph{projector sandwich}:
		\be
			\pPm  
			\equiv 
			\pPIavg  \ \pPIm \ \pPIavg,
			\label{eqn:pPm-def2}
		\ee
		and \emph{normalizing} the operators:
		\begin{align}
		\Lambda_{m} &  \equiv 
		\left(  
			\sum_{k}P_{k}
		\right)^{\!\!\!-1/2}
		\!\! P_{m}
		\left(
			\sum_{k} P_{k}
		\right)^{\!\!-1/2}\!\!. \label{eq:square-root-POVM-generic-bis} 
		\end{align}
		%
		
		The error analysis of a square root measurement is  
		greatly simplified by using the Hayashi-Nagaoka operator inequality.
		
		\begin{lemma}[Hayashi-Nagaoka \cite{hayashi2003general}] 
														\label{lem:HN-inequality}
		If $S$ and $T$ are operators such that  
		$0\leq T$ and $0\leq S\leq I$, then
		\be
			I-\left(  S+T\right)  ^{-\frac{1}{2}}S\left(  S+T\right)  ^{-\frac{1}{2}}%
			\ \: \leq \ \:
			2\left(  I-S\right) \  + \ 4T. 
		\ee
		\end{lemma}%

		If we let $S=P_m$ and $T=\sum_{m^\prime \neq m} P_{m^\prime}$
		in the above inequality we obtain
		\be
			I - \Lambda_m 	\ \ \
			\leq \ \ \
			2\left(  I-P_m\right)   \ +  \ 4 \textstyle \sum_{m^\prime \neq m} P_{m^\prime},
		\ee
		%
		which corresponds to 
		the decomposition of the error outcome $(I-\Lambda_m)$ into 
		two contributions: 
		\renewcommand{\labelenumi}{\Roman{enumi}.}
		\begin{enumerate}
		\item
			The probability that the correct detector does not 
			``click'': $\left(  I-P_m\right)$. 
			This corresponds to the error events $\mathbf{(E0)}$ and $\mathbf{(E1)}$ in
			the classical coding theorem.
		\item 
			The probability that a wrong detector ``clicks'': 
			$\sum_{m^\prime \neq m} P_{m^\prime}$.
			This corresponds to the error event $\mathbf{(E2)}$ in the classical case.

		\end{enumerate}
		\renewcommand{\labelenumi}{\arabic{enumi}.}

%
	%

	We will show that the average probability of error 
	\begin{align*}
		\bar{p}_{e}
		\equiv 
			\frac{1}{|\mathcal{M}|}\sum_{m}					
			\text{Tr}\!
			\left\{  
				\left(  I-\Lambda^{B^{n}}_{m} \right)
				\prhom
			\right\},
	\end{align*}
	will be small 	
	provided the rate $R  \leq I(X;B)-\delta=H(B)-H(B|X)-\delta$.
	The bound follows from the following properties of typical projectors:
	\begin{align}
		\Tr[ \pPIm ] &\leq 2^{n[H(B|X)+\delta ]},  \label{RE-typ-1-ptop}\\
		 \pPIavg  \bar{\rho}^{\otimes n}  \ \pPIavg & \leq 2^{-n[H(B)-\delta] } \pPIavg, \label{RE-typ-2-ptop}
	\end{align}
	and reasoning analogous to that used in the classical coding theorem.
	Note that by the symmetry of both the codebook construction and the decoder we can 
	study the error analysis for a fixed message $m$.
	
	Consider the probability of error when the message  $m$ is sent,
	and let us apply the Hayashi-Nagaoka operator inequality
	(Lemma~\ref{lem:HN-inequality}) to split the error into two terms:
	\begin{align}
	 \bar{p}_e
	 & \equiv
	\text{Tr}\!\left[ 
		\left(I   - \!\pLAMm \right)  
		\prhom
	  \right]  \nonumber \\
	& \leq  2
	\underbrace{
	\text{Tr}\left[  \left(  I 
	-\pPm
	\right)   \ 
	\prhom 
	\right]  
	}_{(\textrm{I})}
	\ \  + \;4 \! 
	\underbrace{
	\sum_{m^{\prime}   \neq m }
	\!\!
	\text{Tr}
	\left[  
	\pPmpr
	\
	\prhom
	\right]
	}_{(\textrm{II})}
	\!.
	\label{eqn:after-HN-in-HSW}
	\end{align}
	
	We bound the expectation 
	of the average probability of error
	by bounding the individual terms.
	
	We now state two useful results, which we need
	to bound the first error term.
	First, recall the inequality from Lemma~\ref{lem:trace-inequality} which states that:
	\be
	\mathrm{Tr}\left[  \Lambda\rho\right]
		\leq
	\mathrm{Tr}\left[ \Lambda \sigma\right]  
		+ \left\Vert \rho-\sigma\right\Vert _{1},
	\label{eqn:tr-trickHSW}
	\ee
	holds for all operators such that   $0\leq \rho, \sigma, \Lambda \leq I$. 
	
	The second ingredient is the gentle measurement lemma.
	\begin{lemma}[Gentle operator lemma for ensembles \cite{itit1999winter}] 
								\label{lem:gentle-operator}
	Let $\left\{  p\!\left(  x\right)  ,\rho_{x}\right\}$ be an ensemble
	and let $\bar{\rho} \equiv\sum_{x}p\!\left(  x\right)  \rho_{x}$. 
	If an operator $\Lambda$, where $0 \leq \Lambda \leq I$, 
	has high overlap with the average state, 
	$\mathrm{Tr}\left[  \: \Lambda \: \bar{\rho} \: \right]  \geq1-\epsilon$,
	then the subnormalized state $\sqrt{\Lambda}\rho_{x}\sqrt{\Lambda}$ is close in 
	trace distance to the original state $\rho_{x}$ on average: %
	$ 
	\mathbb{E}_{X}\left\{  \left\Vert \sqrt{\Lambda}\rho_{X}\sqrt{\Lambda}%
	-\rho_{X}\right\Vert _{1}\right\}  \leq2\sqrt{\epsilon}.
	$
	\end{lemma}

	We bound the expectation over the code randomness
	for the first term in \eqref{eqn:after-HN-in-HSW} as follows:
	{\allowdisplaybreaks
	\begin{align*}
	\ExpX   (\text{I})  
	&=	\ExpX   
		\text{Tr}\left[  
			\left(  I  -\pPm\right)   \ 
			\prhom 
		\right] \\
	&=	\ExpX   
		\text{Tr}\!\left[  
			\left(  I- \pPIavg \pPIm \pPIavg \right) 
			\prhom  
		\right]  \\
	&=	
		1 -
		\ExpX  \!  \bigg\{
		\text{Tr}\!\left[  
			 \pPIm \ \ \pPIavg 
			\prhom  
			\pPIavg
		\right]  
		\bigg\} \\
	&\overset{\mdingone}{\leq}
		1 -
		\ExpX  \!  \bigg\{
		\text{Tr}\!\left[  
		  \pPIm \ 
			\prhom  
		\right]  \ 	  +  
		\left\| \pPIavg \prhom \pPIavg   - \prhom \right\|_1
		\bigg\} \\
	& = 
		1 -
		\ExpX 
		\text{Tr}\!\left[  
		  \pPIm \ 
			\prhom  
		\right]  \ 	  +  
		\ExpX  \left\| \pPIavg \prhom \pPIavg   - \prhom \right\|_1 \\
	&\overset{\mdingtwo}{\leq}
		1 - 
		\ExpX 
		\text{Tr}\!\left[  
		  \pPIm \ 
			\prhom  
		\right]  \ 	  +  
		2 \sqrt{\epsilon}  \\
	&\overset{\mdingthree}{\leq}
		1 - (1-\epsilon) + 2\sqrt{\epsilon} \ \  = \ 	\epsilon + 2\sqrt{\epsilon}.
	\end{align*}}%
	The inequality~\dingone follows from equation \eqref{eqn:tr-trickHSW}.
	The inequality~\dingtwo follows from Lemma~\ref{lem:gentle-operator}
	and the property of the average output state $\text{Tr}\!\left[  \pPIavg \ \bar{\rho}^{\otimes n} \right] \geq 1- \epsilon$.
	The inequality~\dingthree follows from: $
		\ExpX \text{Tr}\!\left[ 	
			\Pi_{X^n(m)}  \rho_{X^n(m)} 
		 \right] 
		 \geq 1- \epsilon$.

	The crucial Holevo information-dependent bound on the expectation of the second term in \eqref{eqn:after-HN-in-HSW}
	can be obtained by using the quantum packing lemma. 
	The quantum packing lemma (Lemma~\ref{lem:q-packing}) given in Appendix~\ref{apdx:q-packing}, provides a bound on the 
	amount of overlap between the conditionally typical subspaces for the codewords in our code construction
	and is analogous to the classical packing lemma (Lemma~\ref{lem:classical-packing}),
	which we used to prove the classical channel coding theorem. 
	Note that Lemma~\ref{lem:q-packing} is less general than the quantum packing lemmas which appear
	in \cite{itit2008hsieh} and \cite{wilde2011book}.

	The overall probability of error is thus bounded as
	\be
		\ExpX \bar{p}_{e} \ \ 
		\leq \ 
		2( \epsilon + 2\sqrt{\epsilon} )
		+ 4\left( 2^{-n[I(X;B)-2\delta\ -R] } \right),
	\ee
	and if we choose
	$R \leq I(X;B)-3\delta$,
	the probability of error 
	is bounded from above by $\epsilon$ in the limit $n\to \infty$.


%
%
%
%
%
 
 \bigskip

\begin{example}[Point-to-point channel]
									\label{ex:PtoPcVScq}
	Consider the classical-quantum channel $\mcal{N}\equiv ( \{0,1\}, \rho^B_x, \mathbb{C}^2 )$,
	which takes a classical bit as input and outputs a qubit (a two-dimensional quantum system).
	Suppose the channel map is the following:
	\be
	    \!\!0   \rightarrow \rho_0 \equiv \ketbra{0}{0} = \left[\begin{array}{cc}1&0\\0&0\end{array}\right]\!\!,   \qquad 
	    1   \rightarrow \rho_1 \equiv \ketbra{+}{+} = \!\left[\!\begin{array}{cc}\tfrac{1}{2}&\tfrac{1}{2}\\\tfrac{1}{2}&\tfrac{1}{2}\end{array}\!\right]\!\!.
	\ee

	We calculate the channel capacity for three different measurement strategies:
	two classical strategies where the channel outputs are measured independently,
	and a quantum strategy that uses collective measurements on blocks
	of $n$ channel outputs.
	Because the input is binary, it is possible to plot the achievable rates for 
	all input distributions $p_X$. See Figure~\ref{fig:PtoP_example} for 
	a plot of the achievable rates for these three strategies.

	\noindent
	{\bf a) Basic classical decoding: }
	A classical strategy for this channel corresponds to the channel outputs 
	being individually measured in the computational basis:
	\be
		\Lambda_0 = \ketbra{0}{0}, \ \ \Lambda_1 = \ketbra{1}{1}, 
		\qquad 
		\Lambda_{y^n}^{B^n} \equiv \Lambda_{y_1}\otimes \Lambda_{y_2} \otimes \cdots \otimes \Lambda_{y_n}.
	\ee
	Such a communication model for the channel is \emph{classical} since we have
	$\Tr\left[ \Lambda_{y^n}^{B^n} \ \rho^{B^n}_{x^n} \right] \equiv p_{Y^n|X^n}(y^n|x^n)$.
	More specifically, $p_{Y^n|X^n}(y^n|x^n) = \prod^n p^{(a)}_{Y|X}(y_i|x_i)$,
	where $p^{(a)}_{Y|X}(y|x)$ is a classical $Z$-channel with transition probability $p_z \equiv p^{(a)}_{Y|X}(0|1)=\Tr[ \Lambda_0 \ketbra{+}{+} ]=0.5$.

	The capacity of the classical $Z$-channel is given by:
	\be
		C^{(a)}(\mcal{N}) = \max_{0 \leq p_0 \leq 1}   H\!\big( (1-p_0)(1-p_z) \big) - (1-p_0)H( p_z ),
	\ee
	where we parametrize in terms of $p_0 = p_X(0)$. 
	For this model, the capacity achieving input distribution has $p_0=0.6$
	and the capacity is $C^{(a)}=H_2(0.2)-0.4 \approx 0.3219$.

	\noindent	
	{\bf b) Aligned classical decoding: }
	A better classical model is to use a ``rotated'' quantum measurement such that
	the measurement operators are symmetrically aligned with the 
	channel outputs. The measurement directions $-\pi/8$ and $\pi/4 + \pi/8$
	are symmetric around the output states $\ket{0}$ and $\ket{+}$.
	Define the notation $c_{\pi_8}=\cos(\pi/8)$ and $s_{\pi_8}=\sin(\pi/8)$.
	The measurement along the $-\pi/8$ and $\pi/4 + \pi/8$ directions
	corresponds to the following POVM operators:
	\begin{align*}
		\Lambda_0 & = (c_{\pi_8} \ket{0}-s_{\pi_8} \ket{1})(c_{\pi_8}\bra{0}-s_{\pi_8}\bra{1})
					= \!\left[\!\begin{array}{cc}c_{\pi_8}^2 & -c_{\pi_8}s_{\pi_8}\\ -s_{\pi_8}c_{\pi_8}& s_{\pi_8}^2\end{array}\!\right]_{\!\{ \ket{0}, \ket{1} \} }  \\
		\Lambda_1 & = (c_{\pi_8} \ket{+}-s_{\pi_8} \ket{-})(c_{\pi_8}\bra{+}-s_{\pi_8}\bra{-})
					= \!\left[\!\begin{array}{cc}c_{\pi_8}^2 & -c_{\pi_8}s_{\pi_8}\\ -s_{\pi_8}c_{\pi_8}& s_{\pi_8}^2\end{array}\!\right]_{\!\{ \!\ket{+}, \ket{-} \} } 
	\end{align*}
	where the matrix representations are expressed in the basis indicated in subscript.

	Using this measurement on channel outputs $\rho_x^B$ 
	induces a classical channel $p^{(b)}_{Y|X}$ with 
	transition probabilities   
	\be
	    p^{(b)}_{Y|X}(0|0) = c_{\pi_8}^2, \ \  p_{Y|X}(1|0)=s_{\pi_8}^2,  
	    \quad 
	    p^{(b)}_{Y|X}(1|1) = c_{\pi_8}^2, \ \  p_{Y|X}(0|1)=s_{\pi_8}^2,  
	\ee
	which corresponds to a binary symmetric channel (BSC)
	with crossover probability $p_e = s_{\pi_8}^2 = \sin^2(\pi/8)$
	and success probability $p_s  = c_{\pi_8}^2$.
	The capacity of this BSC is given by:
	\be
		C^{(b)}(\mcal{N}) = 1 - H( p_s ) = 1 - H\big(\cos^2(\pi/8) \big) \approx 0.3991. 
	\ee

\begin{figure}[ptb]
\begin{center}
	\includegraphics[width=0.618\textwidth]{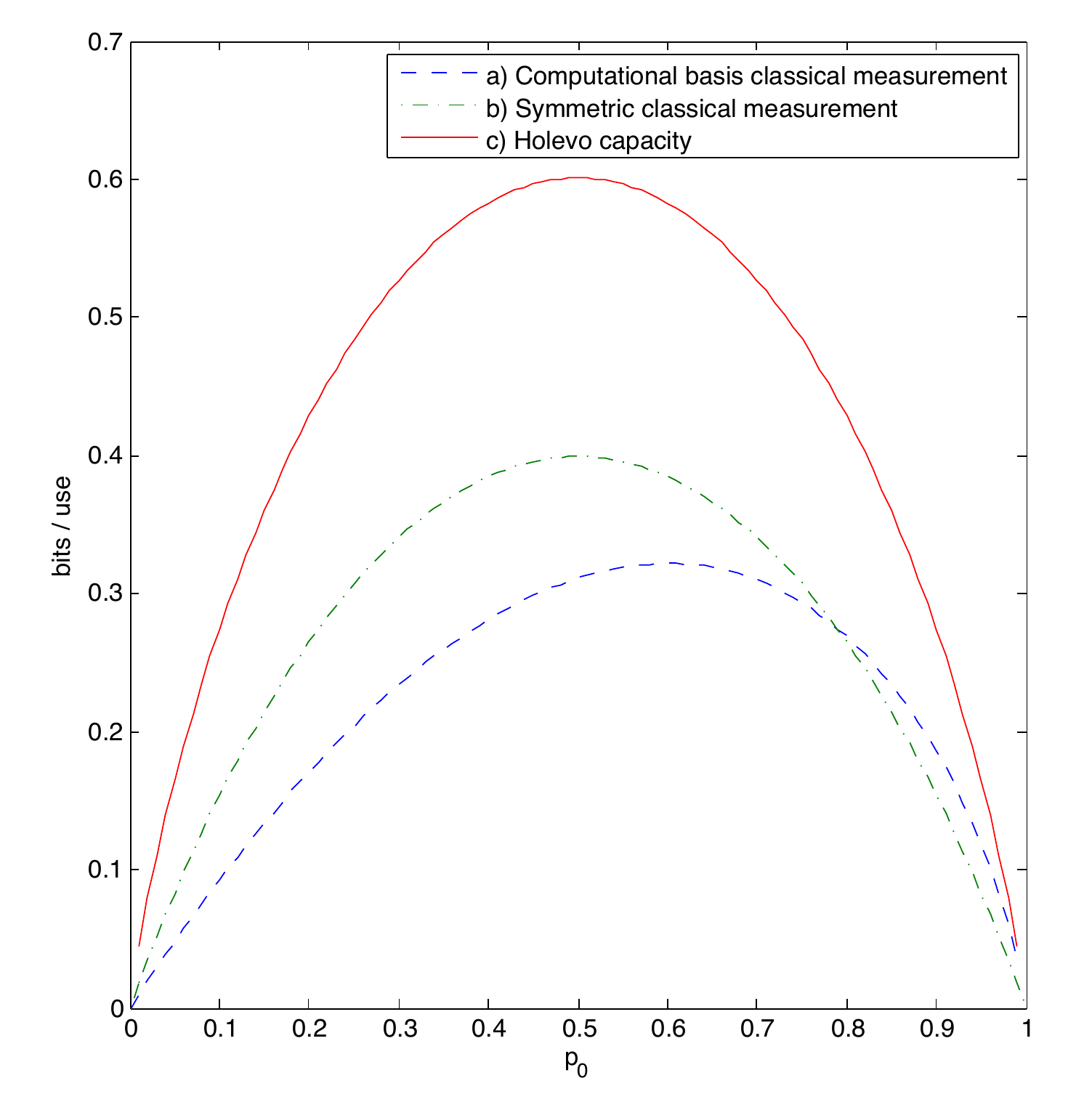}
\end{center}
\caption{Plot of the achievable rates for the point-to-point channel $\rho_x^B$
given by the map $0 \to \ketbra{0}{0}^B, 1\to \ketbra{+}{+}^B$ under three models.
The horizontal axis corresponds to the parameter $p_0 = p_X(0)$ of the input distribution.
The first model treats 
each output of the channel as a classical bit $Y^{(a)} \in \{0,1\}$ corresponding to the output of a measurement
in the computational basis: $\{ \Lambda^{(a)}_y \}_{y\in\{0,1\} } =  \{ \ketbra{0}{0}, \ketbra{1}{1} \}$.
The mutual information $I(X;Y^{(a)})$ for all input distributions $p_X$ is plotted as a dashed line.
Under this model, the channel $\mcal{N}$ corresponds to a classical $Z$-channel.
%
%
%
A better approach is to use a symmetric measurement with output denoted as $Y^{(b)}$,
which corresponds to a classical binary symmetric channel.
The mutual information $I(X;Y^{(b)})$ is plotted as a dot-dashed line.
The best coding strategy is to use block measurements.
The Holevo quantity $H\!\big(\sum_x p_X(x) \rho_x^B \big) - \sum_x p_X(x)H(\rho_x^B)$
for all input distributions is plotted as a solid line.
The capacity of the channel under each model is given by the maximum of each function
curve: $C^{(a)}(\mcal{N}) \approx 0.3219$, 
$C^{(b)}(\mcal{N})  \approx 0.3991$,
and 
$\mcal{C}^{(c)}(\mcal{N})  =   H_2(\cos^2(\pi/8)) \approx 0.6009$.
For this particular channel the quantum decoding strategy leads to a 50\% improvement in the 
achievable communication rates relative to the best classical strategy.
}%
\label{fig:PtoP_example}%
\end{figure}

	\noindent
	{\bf c) Holevo limit: }
	The HSW Theorem tells us the \emph{ultimate} capacity of this
	channel is given by
	\be
		\mcal{C}^{(c)}(\mcal{N}) 
		\equiv 	 \max_{p_X} 	H\!\bigg(\sum_x p_X(x) \rho_x^B \bigg) - \sum_x p_X(x)H(\rho_x^B).
	\ee
	In our case, the capacity is achieved using the uniform input distribution.
	The capacity for this channel using a quantum measurment is therefore:
	\be
		\mcal{C}^{(c)}(\mcal{N})  =   H_2(\cos^2(\pi/8))   \approx 0.6009.
	\ee
	
	In general, a collective measurement on blocks of $n$ outputs of the channel
	are required to achieve the capacity.
	This means that the POVM operators $\{ \Lambda_{x^n}^{B^n} \}$ cannot be written as a tensor 
	product of measurement operators on the individual output systems.
	The channel capacity can be achieved using the random coding approach
	and the square root measurement based on conditionally typical projectors
	as shown in the proof of Theorem~\ref{thm:HSWtheorem}.
	%

\end{example}

\section{Discussion}

	This chapter introduced the key concepts of the classical and quantum channel coding
	paradigms.
	The situation considered in Example~\ref{ex:PtoPcVScq} serves as an illustration
	of the potential benefits that exist for modelling communication channels using quantum
	mechanics.
	%

	The key take-away from this chapter is that collective measurements
	on blocks of channel outputs are necessary in order to achieve the
	\emph{ultimate} capacity of classical-quantum communication channels,
	and that classical strategies which measure the channel outputs individually are
	suboptimal.
	The increased capacity is perhaps the most notable difference that exists between the
	classical and classical-quantum paradigms for communication  \cite{ElGamalQuestion}.
	
	In the remainder of this thesis, we will study multiuser classical-quantum communication models
	and see various coding strategies, measurement constructions and
	error analysis techniques 	which are necessary in order to prove coding theorems.
	%


\chapter{Multiple access channels}

													\label{chapter:MAC}

	The multiple access channel is a communication model
	for situations in which
	multiple senders are trying to transmit information to a single receiver.
	To fully solve the multiple access channel problem is to characterize
	all possible transmission rates for the senders which are decodable by
	the receiver.
	%
	%
	We will see that there is a natural tradeoff between the rates of the senders;
	the louder that one of the senders ``speaks,''
	the more difficult it will be for the receiver to ``hear'' the other senders.
	%

\section{Introduction}
													\label{sec:MAC-intro}

	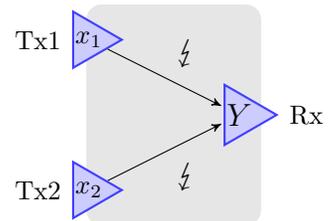
\begin{wrapfigure}[8]{r}{0pt}
	\begin{tikzpicture}[node distance=2.0cm,>=stealth',bend angle=45,auto]

	  \begin{scope}
		\node [cnode] (Tx1) [ label=left:Tx1     ]                            {\footnotesize $x_1$} ; 
		\node [cnode] (Tx2) [ label=left:Tx2, below of=Tx1]                {\footnotesize $x_2$}; 
		\node [cnode] (Rx) [ label=right:Rx, right of=Tx1,yshift=-10mm] { $Y$} 
			edge  [pre]             node[swap]  {$\lightning$}    (Tx1)
			edge  [pre]             node		   {$\lightning$}         (Tx2) ;
	  \end{scope}
	  \begin{pgfonlayer}{background}
	    \filldraw [line width=4mm,join=round,black!10]
	      ([xshift=-3mm]Tx1.north -| Tx1.east) rectangle ([xshift=+3mm]Tx2.south -| Rx.west);
	  \end{pgfonlayer}
	\end{tikzpicture}
		\caption{A classical multiple access channel.
			}
		\label{ex:CMACchannelTriangles}
		\end{wrapfigure}
	
		The classical multiple access channel $\mcal{N}^{X_1X_2\to Y}$ 
		is a triple 
		$(\mcal{X}_1 \times \mcal{X}_2, \mcal{N}(x_1,x_2)\equiv p_{Y|X_1X_2}(y|x_1,x_2), \mcal{Y} )$,
		where $\mcal{X}_1$ and $\mcal{X}_2$ are the input alphabets
		for the two senders, $\mcal{Y}$ is the output alphabet and 
		$p_{Y|X_1X_2}(y|x_1,x_2)$ is a conditional probability distribution
		which describes the channel behaviour.

		Our task is to characterize the communication rates
		$(R_1,R_2)$ that are achievable
		from Sender~1 to the receiver and 
		from Sender~2 to the receiver.

		\newcounter{exContinuedTMP}[example]
	
		\begin{example}
						\label{ex:lasersCMAC}
		Consider a situation in which two senders 
		use laser light pulses to communicate to a distant receiver
		equipped with an optical instrument and a photodetector.
		In each time instant, Sender 1 can choose 
		to send either a weak pulse of light or a strong pulse:
		$\mcal{X}_1 = \{ \shPulse, \medPulse \}$.
		Sender~2 similarly has two possible inputs $\mcal{X}_2 = \{ \shPulse, \medPulse \}$.
		The receiver measures the light intensity coming into
		the telescope,
		and we model his reading as the following output space
		$\mcal{Y}=\{ \shPulse, \medPulse, \longPulse \}$.
		%
		The output signal is the sum of the incoming signals:
		$Y=X_1+X_2$.
		We have $p_{Y|X_1X_2}(\shPulse|\shPulse,\shPulse)=1$,
		$p_{Y|X_1X_2}(\medPulse|\medPulse,\shPulse)
		 =p_{Y|X_1X_2}(\medPulse|\shPulse,\medPulse)=1$ 
		 and 
		 $p_{Y|X_1X_2}(\longPulse|\medPulse,\medPulse)=1$.

	\begin{figure}
	\begin{center}
		\begin{tikzpicture}[node distance=2.0cm,>=stealth',bend angle=45,auto,scale=0.83, every node/.style={scale=0.83}]

	  \begin{scope}
		\node [draw=black,anchor=south west,inner sep=0] (bglayer) at (0,0) {\includegraphics[width=13cm]{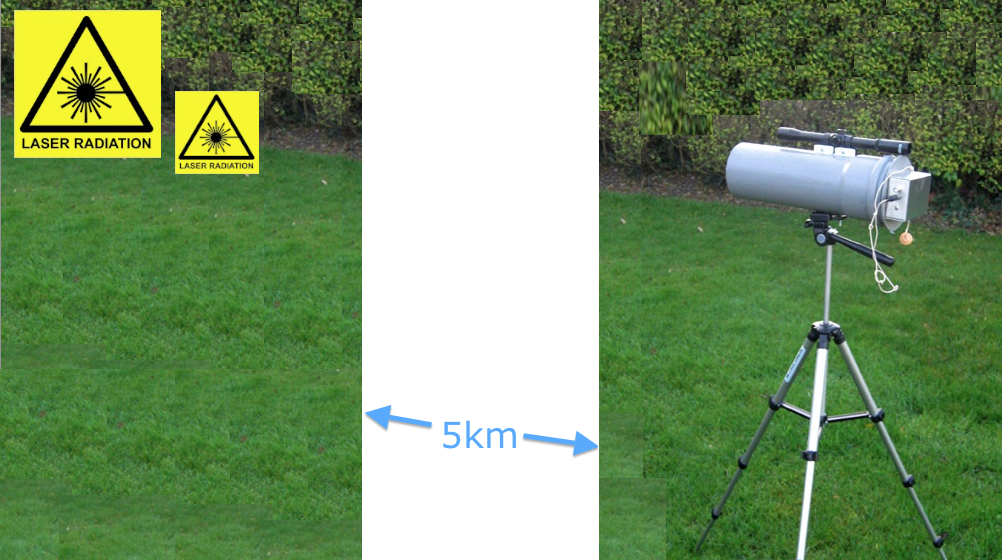} };
		\node [cnode] (Tx1) at (1.2,6.2)  [ ]                            {$x_1$} ;
		\node  [cnode] (Tx2) at (2.9,5.6)  [  ]                            {$x_2$} ;
		\node [qnode] (Rx) at (11.3,4.8) [ ]
		 { \ \ \textrm{Rx} \ \ 
		 } 
			edge [pre,draw=red,very thick]  node[] {} (Tx1);		
		\draw[->,draw=green,very thick]	([yshift=-1cm]Tx2)	-- ([yshift=-1cm]Rx);
	  \end{scope}
	\end{tikzpicture}
	\end{center}
	
	\caption{A real-world multiple access channel $\mcal{N}_1$.}

	\label{fig:real-world-MAC}
	\end{figure}
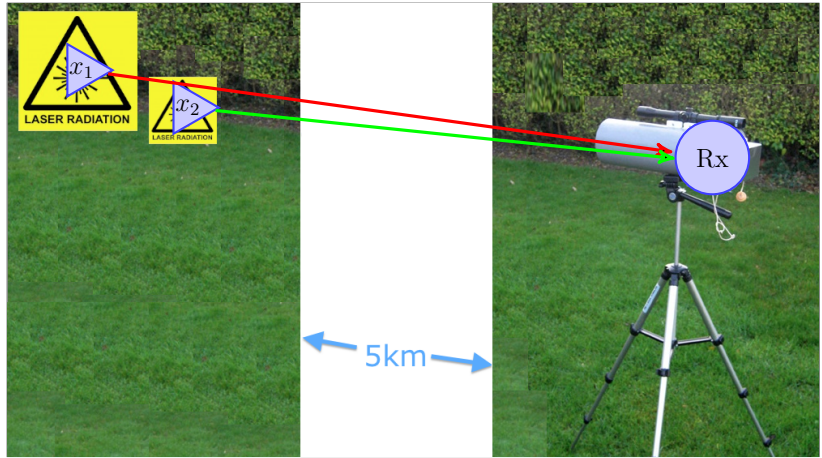

		\end{example}

		The rate pair $(R_1,R_2)=(1,0)$ is achievable
		if we force Sender 2 to always send a constant input.
		The resulting channel between Sender~1 and the receiver
		is a noiseless binary channel.
		The rate $(0,1)$ is similarly achievable if we fix Sender~1's input.
		A natural question is to ask what other rates are achievable for 
		this communication channel.
		
		Note that the model used to describe the above communication scenario
		is very crude and serves only as a first approximation, which we use to 
		illustrate the basic ideas of multiple access communication.
		In Section~\ref{sec:QMAC-model}, we will
		consider more general models for multiple access channels, 
		which allow the channel outputs to be quantum systems.
		In Chapter~\ref{chapter:bosonic}, we will refine the model further 
		by taking into account certain aspects of quantum optics.

	\subsection{Review of classical results}
	
		The multiple access channel is one of the first 
		multiuser communications problems ever considered \cite{S61}.
		It is also one of the rare problems in network information theory
		where a full capacity result is known, i.e., the best known achievable
		rate region matches a proven outer bound.
		The multiple access channel plays an
		important role as a building block for other
		network communication scenarios.
		
		The capacity region of the classical discrete 
		memoryless multiple access channel (DM-MAC) was established 
		by Ahlswede \cite{ahlswede1971multi, ahlswede1974capacity}
		and Liao  \cite{liao-multipleaccess}.
		Consider the classical multiple access channel with two senders
		described by $\mcal{N} = (\mathcal{X}_1 \times \mathcal{X}_2,  p_{Y|X_1X_2}, \mcal{Y} )$.
		The capacity region for this channel 	is given by 
		\begin{equation}
	        		\nonumber
	        		 \mcal{C}_{\textrm{MAC}}(\mcal{N})  
				\equiv  
					\bigcup_{ p_{X_{1}}\!, p_{X_{2}}\! }
					\left\{ 
						(R_1,R_2) \in \mathbb{R}_+^2 
						\left| 
						\begin{array}{rcl}					
					            R_1             &\leq&    I(X_1;Y|X_2) \\
					            R_2             &\leq&    I(X_2;Y|X_1) \\
					            R_1+R_2   &\leq&    I(X_1X_2;Y)
					           \end{array}
					          \right.
					 \right\},
	        \end{equation}
	        where $p_{X_1} \in \mcal{P}(\mcal{X}_1)$,
	         $p_{X_2} \in \mcal{P}(\mcal{X}_2)$ 
	        and the mutual information quantities are taken with 
	        respect to the joint input-output distribution
	        \be
	        		p_{X_1X_2Y}(x_1,x_2,y)  \equiv p_{X_1}(x_1)p_{X_2}(x_2)p_{Y|X_1X_2}(y|x_1,x_2).
		\ee
		Note that the input distribution is
		chosen to be a product distribution $p_{X_1}p_{X_2}$, 
		which reflects  the assumption that the two senders 
		are spatially separated and act independently.
		We can calculate the exact capacity region of any multiple access
		channel by evaluating the mutual information expressions for all 
		possible input distributions and taking the union.
		
		\bigskip

	\setcounter{example} {\value{exContinuedTMP}}

	\begin{example}[continued]
	
		The capacity region for the multiple access channel $\mcal{N}_1$ 
		described in Example~\ref{ex:lasersCMAC} is given by:
		\begin{align}
			C_{MAC}(\mcal{N}_1)  = 
					\left\{ 
						(R_1,R_2) \in \mathbb{R}_+^2 
						\left| 
						\begin{array}{rcl}					
					            R_1             &\leq&    1 \\
					            R_2             &\leq&    1 \\
					            R_1+R_2   &\leq&    1.5
					           \end{array}
					          \right.
					 \right\}.			
		\end{align}
		To see how the rate pair $(1,0.5)$ can be achieved
		consider an encoding strategy where each sender
		generates codebooks according to the uniform
		probability distribution and the receiver
		decodes the messages from Sender~2 first,
		followed by the messages from Sender~1.
		The effective channel from Sender~2 to the receiver
		when the input of Sender~1 is unknown corresponds to a 	
		symmetric binary erasure channel with erasure probability $\frac{1}{2}$.
		This is because when the receiver's output is ``$\shPulse$''
		or ``$\longPulse$'' there is no ambiguity about what was sent.
		The output  ``$\medPulse$'' could arise in two
		different ways, so we treat it as an erasure.
		The capacity of this channel is 0.5 bits per channel use
		\cite[Example 14.3.3]{CT91}.
		Assuming the receiver correctly decodes the codewords
		from Sender~2, the resulting channel from Sender~1 to
		the receiver is a binary noiseless channel which has
		capacity one.
		To achieve the rate pair $(0.5,1)$ we must generate 
		codebooks at the appropriate rates and use the opposite 
		decoding order.
		The capacity region is illustrated in the following figure.
	\end{example}

		\begin{figure}[h]
			\begin{center}
				\includegraphics[width=7.5cm]{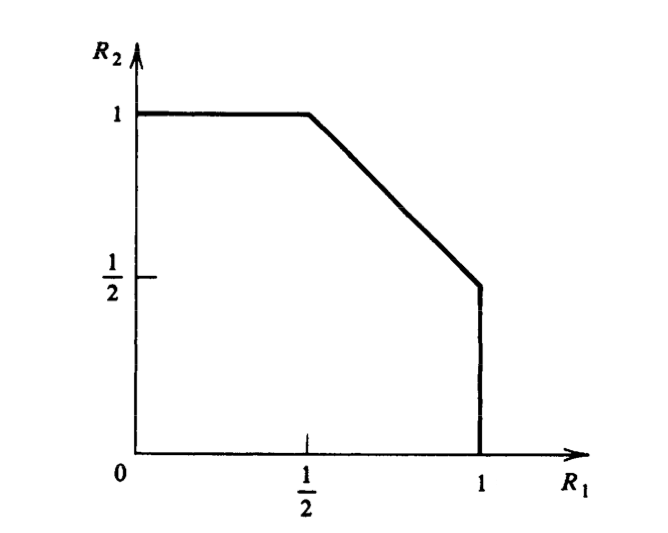} 
				\caption{The capacity region of the 
						adder channel. 
						}
			\end{center}
			\label{fig:MAC-cap-region}
		\end{figure}
 
		The above example 	illustrates the key aspect of the 
		multiple access channel problem: the trade off between
		the communication rates of the senders.
		
		\bigskip

	\subsection{Quantum multiple access channels}	
			\label{sec:QMAC-model}

		The communication model used to  evaluate 
		the capacity in  Example~\ref{ex:lasersCMAC} 	is classical.
		We modelled the detection of 
		light intensity in a classical way 
		and ignored details of the quantum
		measurement process.
		%

		The capacity result of Ahlswede and Liao is therefore 
		a result which depends on the classical model which we used.
		Better communication rates might be possible if we choose
		to model the quantum degrees of freedom in the communication channel.
		In Example~\ref{ex:PtoPcVScq}, we saw how the 
		\emph{quantum} analysis of the detection aspects of the 
		communication protocol can lead to improved communication rates
		for point-to-point channels.
		In this chapter, we pursue the study of quantum decoding
		strategies in the \emph{multiple access} setting.
		
		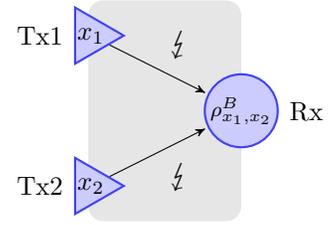
\begin{wrapfigure}{r}{0pt}
		\ \ 
		\begin{tikzpicture}[node distance=2.0cm,>=stealth',bend angle=45,auto]

		  \begin{scope}
			\node [cnode] (Tx1) [ label=left:Tx1     ]                            {\footnotesize $x_1$} ; 
			\node [cnode] (Tx2) [ label=left:Tx2, below of=Tx1]                {\footnotesize $x_2$}; 
			\node [qnode] (Rx) [ label=right:Rx, right of=Tx1,yshift=-10mm] {\scriptsize $\rho^B_{x_1,x_2}$} 
				edge  [pre]             node[swap]  {$\lightning$}    (Tx1)
				edge  [pre]             node		   {$\lightning$}         (Tx2) ;
		  \end{scope}
		  \begin{pgfonlayer}{background}
		    \filldraw [line width=4mm,join=round,black!10]
		      ([xshift=-3mm]Tx1.north -| Tx1.east) rectangle ([xshift=+3mm]Tx2.south -| Rx.west);
		  \end{pgfonlayer}
		\end{tikzpicture}
			\caption{A quantum multiple access channel with two senders.
				      The output of the channel are conditional quantum
				      states $\mcal{N}^{B}(x_1,x_2)\equiv \rho^B_{x_1,x_2}$.
				}
			\end{wrapfigure}
		
		A classical-quantum multiple access channel 
		is defined as the most general map with two classical inputs 
		and one quantum output:
		\[
			\left(
				\mcal{X}_1 \!\times\! \mcal{X}_2, 
				\mcal{N}^{X_1X_2\to B}\!(x_1,x_2)\equiv  \rho^B_{x_1,x_2}, 
				\mcal{H}^B
			\right)\!\!.
		\]
		
		Our intent is to quantify the communication rates that are possible
		for classical communication from each of the two senders to the receiver.
		The main difference with the classical case is that the decoding
		operation we will use is a quantum measurement (POVM).
		We have to find the \emph{rate region}
		for pairs $(R_1,R_2)$ such that the following interconversion 
		can be achieved:
		\be
			n \cdot \mcal{N}^{X_1X_2 \to B}  
			\ \ \overset{ ( 1 -\epsilon)}{\longrightarrow} \ \ 
			nR_1 \cdot  [c^1 \to c]
			\ + \
			nR_2 \cdot  [c^2 \to c].
		\ee
		The above expression states that $n$ instances of the channel
		can be used to carry $nR_1$ {\bf c}lassical bits from Sender~1 to the receiver
		(denoted $[c^1 \to c]$) and $nR_2$ bits from Sender~2 to the receiver
		(denoted $[c^2 \to c]$).
		The communication protocol succeeds with probability $(1-\epsilon)$ for
		any $\epsilon >0$ and sufficiently large $n$.

		The problem of classical communication over a
		classical-quantum multiple-access channel 
		was solved by Winter \cite{winter2001capacity}.
		He provided single-letter formulas for the capacity region,
		which can be computed as an optimization over the choice 
		of input distributions for the senders.
		%
		We will discuss Winter's result and proof techniques 
		in Section~\ref{sec:mac-succ-decoding}.
		
		Note that there exist other quantum multiple access communication
		scenarios that can be considered.
		The bosonic multiple access channel
		was studied in \cite{yen2005multiple}.
		%
		The transmission of quantum information over a quantum
		multiple access channel was considered in 
		\cite{yard2005capacity,thesis2005yard,ieee2008yard}.		
%
		The quantum multiple access problem has also been considered
		in the entanglement-assisted setting \cite{itit2008hsieh,xu2011sequential}.
		In this chapter, as in the rest of the thesis, we restrict 
		our attention to the problem of classical communication
		over classical-quantum channels.
		

	\subsection{Information processing task}

		To show that a certain rate pair $(R_1,R_2)$ is achievable we must
		construct an end-to-end coding scheme that the two senders and
		the receiver can employ to communicate with each other.
		In this section we specify precisely the different 
		steps involved in the transmission process.
		
		Sender~1 will send a message $m_1$ chosen from the message set
		$\mathcal{M}_1\equiv \left\{  1,2,\ldots,|\mathcal{M}_1|\right\} $ where 
		$|\mathcal{M}_1|=2^{nR_{1}}$.
		Sender~2 similarly chooses a message $m_2$ from a message set 
		$\mathcal{M}_2 \equiv \left\{  1,2,\ldots,|\mathcal{M}_2|\right\}  $ 
		where $|\mathcal{M}_2|=2^{nR_{2}}$. 
		Senders~1 and 2 encode their messages as codewords 
		$x_1^{n}\!\left(  m_1\right)\in \mathcal{X}_1^n$
		and $x_2^{n}\!\left(  m_2\right) \in \mathcal{X}_2^n$,
		which are then input to the channel.

		The output of the channel is an $n$-fold tensor product state of the form:
		\be
			\mathcal{N}^{\otimes n}\!\left( x_1^{n}(m_1), x_2^{n}(m_2) \right)
			\equiv			
			\rho_{x_2^{n}\left(  m_1\right),  x_2^{n}\left(  m_2\right)  }^{B^{n}} \ \  \in  \mcal{D}(\mathcal{H}^{B^n}).
		\ee

		In order to recover the messages $m_1$ and $m_2$, the
		receiver performs a positive operator valued measure (POVM) 
		$\left\{ \Lambda_{m_1,m_2}\right\}_{m_1\in \mathcal{M}_1,m_2\in \mathcal{M}_2}$ 
		on the output of the channel $B^n$.
		We denote the measurement outputs as $M^{\prime}_1$ and $M^{\prime}_2$. 		
		An error  occurs whenever 
		the receiver measurement outcomes differ from the messages that
		were sent.
		The overall probability of error for message pair $(m_1,m_2)$ is
		\begin{align*}
			p_{e}\!\left(  m_1,m_2\right)   
				&  \ \ \equiv \ \ 
				\Pr\left\{  
					(M^{\prime}_1,M^{\prime}_2)\neq(m_1,m_2) 
				\right\} \\
				&  \ \ = \ \ 
				\text{Tr}\!
				\left[  
					\left(  I-\Lambda_{m_1,m_2} \right)
					\rho_{x_2^{n}\left(  m_1\right)  x_2^{n}\left(  m_2\right)  }^{B^n} 
				\right],
		\end{align*}
		where the measurement operator $\left(  I-\Lambda_{m_1,m_2}\right)$ 
		represents the complement of the correct decoding outcome.

		    \begin{definition}
			An $(n,R_1,R_2,\epsilon)$ code for the multiple access channel 
			consists of two codebooks 
			$\{x^n_1(m_1)\}_{m_1\in \mathcal{M}_1}$
			and 
			$\{x^n_2(m_2)\}_{m_2\in \mathcal{M}_2}$,
			and a decoding POVM
			$\left\{ \Lambda_{m_1,m_2}\right\}$,$m_1\in \mathcal{M}_1$,$m_2\in \mathcal{M}_2$,
			such that the average probability of error 
			$\overline{p}_{e}$ is bounded from above by~$\epsilon$:%
			\begin{align}
				\overline{p}_{e}  
				&  \  \equiv \ 
					\frac{1}{|\mathcal{M}_1||\mathcal{M}_2|}
						\sum_{m_1,m_2}p_{e}\!\left(  m_1,m_2\right) 
				 \leq \epsilon.
			\end{align}

		    \end{definition}  

		A rate pair $\left(  R_{1},R_{2}\right)  $ is \textit{achievable} if there exists
		an $\left(  n,R_{1}-\delta,R_{2}-\delta,\epsilon\right) $ quantum multiple access 
		channel code for all $\epsilon,\delta>0$ and sufficiently large $n$. 
		The \textit{capacity region} 
		$\mcal{C}_{\text{MAC}}(\mcal{N})$ is the closure of the set of all achievable rates.

	\subsection{Chapter overview}

		
			Suppose we have a two-sender classical-quantum multiple 
			access channel and the two messages $m_1$ and $m_2$ were sent.
			This chapter studies the different decoding strategies that can 
			be used by the receiver in order to decode the messages.
			
			The technique used by Winter to prove the achievability of
			the rates in the capacity region of the quantum multiple
			access channel 
			is called \emph{successive decoding}.
			In this approach, the receiver can achieve one of the
			corner points of the rate region by
			decoding the messages in the order ``$m_1 \to m_2|m_1$''.
			In doing so, the best possible rate $R_2$ is achieved, 
			because the receiver will have the side information of $m_1$, 
			and by extensions $x^n_1(m_1)$, when decoding the 
			message $m_2$.
			This approach is also referred to as 	\emph{successive cancellation}
			for channels with continuous variable inputs and additive
			white Gaussian noise (Gaussian channels)
			where the first decoded signal can be \emph{subtracted} from the 
			received signal. 
			%
			The other corner point can be achieved by decoding
			in the opposite order ``$m_2 \to m_1|m_2$''.
			These codes can be combined with \emph{time-sharing}
			and \emph{resource wasting} to achieve all other
			points in the rate region.
			We will discuss this strategy in further detail
			in Section~\ref{sec:mac-succ-decoding} below.
			
			Another approach is to use simultaneous decoding 
			which requires no time-sharing.
			We denote the simultaneous decoding of the messages
			$m_1$ and $m_2$ as ``$(m_1,m_2)$''.
			As far as the QMAC problem is concerned the two approaches
			yield equivalent achievable rate regions. However, if the QMAC code is to be used as 
			part of a larger protocol (like a code for the interference channel 
			for example) then the simultaneous decoding approach is
			much more powerful.

			The main contribution in this chapter is Theorem~\ref{thm:sim-dec-two-sender}
			in Section~\ref{sec:mac-simult-decoding},
			which shows that simultaneous decoding for the classical-quantum
			multiple access channel with two senders is possible.
			This result and the techniques developed for its proof
			will form the key building blocks for the subsequent chapters in this thesis.
			We will also comment on the difficulties in extending the simultaneous
			decoding approach to more than two senders (Conjecture~\ref{conj:sim-dec}).
			In Section~\ref{sec:mac-rate-splitting}, we will briefly discuss 
			a third coding strategy for the QMAC called \emph{rate-splitting}.


\section{Successive decoding}
	\label{sec:mac-succ-decoding}


	Winter found a single-letter formula for the capacity of the classical-quantum 
	multiple access channel with $M$ senders \cite{winter2001capacity}. 
	We state the result here for two senders.

	\begin{theorem}[Theorem 10 in  \cite{winter2001capacity}]
		\label{thm:cqmac-capacity}
		The capacity region for the classical-quantum multiple
		access channel $(\mathcal{X}_1 \times \mathcal{X}_2,  \rho_{x_1,x_2}^{B}, \mathcal{H}^{B} )$
		is given by 
		\begin{equation}
			\label{eq:QMACunion}
	        		 \mathcal{C}_{\textrm{MAC}} 
				=  
					\bigcup_{ p_{X_{1}}\!, p_{X_{2}}\! }
					\!\!\! \{ (R_1,R_2) \in \mathbb{R}_+^2 | 
					  \text{ \emph{Eqns.} \eqref{winterThmEqnsOne}-\eqref{winterThmEqnsThree} } \} 
		\end{equation}
		\vspace{-5mm}
	        \bea 
	            R_1             &\leq&    I(X_1;B|X_2)_\theta,  \label{winterThmEqnsOne} \\
	            R_2             &\leq&    I(X_2;B|X_1)_\theta,  \label{winterThmEqnsTwo} \\
	            R_1+R_2   &\leq&    I(X_1X_2;B)_\theta, \label{winterThmEqnsThree}
	        \eea
	        where the information quantities are taken with respect to the 
	        classical-quantum state:
		\be	\label{eqn:state-in-WinterThm}
			\theta^{X_{1}X_{2}B}
			 \ \equiv \ 
		       	\sum_{x_1,x_2} 
				p_{X_{1}}\!\left(x_{1}\right)p_{X_{2}}\!\left( x_{2}\right) \ 
				\ketbra{x_1}{x_1}^{X_1} 
				\otimes 
				\ketbra{x_2}{x_2}^{X_2} 
				\otimes 
				\rho^{B}_{x_1,x_2}.
		\ee	        	        
	\end{theorem}	



		\begin{figure}[hbt]
		\begin{center}
			\includegraphics[width=0.6\textwidth]{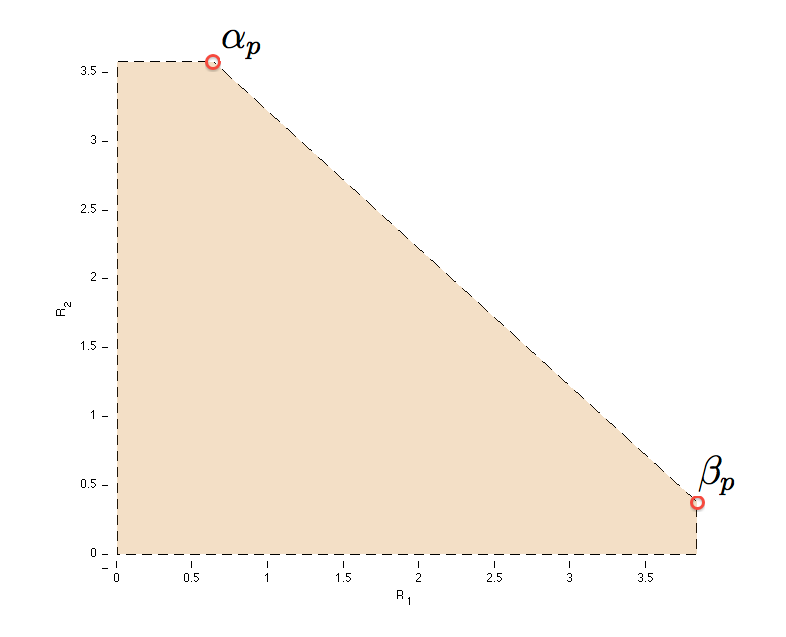}

			\caption{ \small
				The rates achievable by successive decoding correspond to the dominant vertices
				of the rate region $\alpha_p$ and $\beta_p$. Rates in between these points 
				can be achieved by \emph{time-sharing} between the strategies for the two corners.
				}
			\label{fig:MAC-seq}
		\end{center}
		\end{figure}

		For a given choice of input probability distribution
		$p \equiv p_{X_{1}},p_{X_{2}}$,
		the achievable rate region, $\mcal{R}(\mcal{N},p)$, 
		has the form of a pentagon bounded by the three inequalities 
		in equations \eqref{winterThmEqnsOne}-\eqref{winterThmEqnsThree} and two rate positivity conditions.
		The two dominant vertices of this rate region have coordinates 
		$\alpha_{p}
		\equiv (I(X_1;B)_\theta, I(X_2;B|X_1)_\theta)$ and 
		$\beta_{p}
		\equiv (I(X_1;B|X_2)_\theta, I(X_2;B)_\theta)$
		and correspond to two alternate successive decoding strategies.
		The portion of the line $R_1 + R_2 = I(X_1X_2;B)_\theta$ which lies
		in between the points $\alpha_p$ and $\beta_p$ will be referred to 
		as the \emph{dominant facet}.


		In order to show achievability of the entire rate region, 
		Winter proved that each of the corner points of the region is achievable.
		By the use of \emph{time-sharing} we can achieve any point
		on the dominant facet of the region, 
		and we can use \emph{resource wasting} to achieve all the points on the interior of the region.
		It follows that the entire rate region is achievable.
		We show some of the details of Winter's proof below.

	\begin{proof}[Proof sketch.]
		We will use a random coding approach for the codebook
		construction and point-to-point decoding measurements
		based on the conditionally typical projectors. 

		Fix the input distribution $p
		=p_{X_1}(x_1)p_{X_2}(x_2)$ 
		and choose the rates so that they correspond to the
		rate point $\alpha_{p}$:
		\be
			R_1 = I(X_1;B)_\theta - \delta, \qquad R_2 = I(X_2;B|X_1)_\theta - \delta.
		\ee
		
		\prfsec{Codebook construction}
		Randomly and independently generate $2^{nR_{1}}$
		sequences $x_{1}^{n}\!\left(  m_1 \right)$, $m_1 \in\left[1: 2^{nR_{1}}\right]$,
		according to  $\prod\limits_{i=1}^{n}p_{X_{1}}\!\left(  x_{1i}\right)$.
		Similarly generate randomly and independently the codebook 
		$\{ x_{2}^{n}\!\left(  m_2 \right) \}$, $m_2 \in\left[1: 2^{nR_{2}}\right]$
		according to $\prod\limits_{i=1}^{n}p_{X_{2}}\!\left(  x_{2i}\right)$.

		\prfsec{Decoding}
		When the message pair $(m_1,m_2)$ is sent,
		the output of the channel will be $\rho_{x_1^n(m_1)  ,x_2^{n}\left(  m_2\right)  }$.
		Let $\Pi_{\rho_{x_1^n(m_1)  ,x_2^{n}\left(  m_2\right)},\delta}^{n}$ be 
		the conditionally typical projector for that state.
		In order to define the other typical projectors necessary
		for the decoding, we define  the following expectations of the output state:
		\begin{align*}
			\bar{\rho}_{x^n_1(m_1)} &  \equiv
				\sum_{x^n_2}p_{X^n_2}\!\left(  x^n_2\right)  
					\rho_{x^n_1(m_1),x^n_2}
				=  \bigotimes_{i=1}^n \!
					\left(
				\sum_{\mu}p_{X_2}\!\left(  \mu \right)  
					\rho_{x_{1i}(m_1),\mu}
					\right) \\[-2mm]
				&= \ExpXtwo \left\{ \rho_{x^n_1(m_1),X^n_2} \right\},	
												\\[2mm]
			\bar{\rho}^{\otimes n} \! &  \equiv
					\sum_{x_1^n,x^n_2}
						p_{X_1^n}\!\left(  x^n_1\right)  
						p_{X_2^n}\!\left(  x^n_2\right)  
						\rho_{x^n_1,x^n_2}
				=  \bigotimes_{i=1}^n \!
					\left(
				\sum_{\tau,\mu}
					p_{X_1}\!\left(  \tau \right)  
					p_{X_2}\!\left(  \mu \right)  
					\rho_{\tau, \mu}
					\right) \\[-2mm]
				&=  \ExpXonetwo \!\! \left\{ \rho_{X^n_1,X^n_2} \right\} .									
		\end{align*}		
		The state $\bar{\rho}_{x^n_1(m_1)}$ corresponds to the receiver's output
		if he treats the codewords of Sender~2 as noise to be averaged over.
		The state $\bar{\rho}^{\otimes n}$ corresponds to the average output 
		state for a random code constructed according to $p_{X_1}p_{X_2}$.
		Let $\Pi_{\bar{\rho}_{x_1^n(m_1)  },\delta}^{n}
		\equiv \Pi_{\bar{\rho}_{x_1^n(m_1)  },\delta}^{B^n}$ be the
		conditionally typical projector for $\bar{\rho}_{x_1^n(m_1)  }$
		and let 
		$\Pi_{\bar{\rho}}^{n} \equiv \Pi_{\bar{\rho}^{\otimes n},\delta}^{B^n}$ 
		be the typical projector for the
		state $\bar{\rho}^{\otimes n}$.

			To achieve the rates of  $\alpha_{p}$,  
			the receiver will decode the messages in the order 
			``$m_1 \to m_2|m_1$'' using a successive decoding procedure.
			The first step is to use a quantum instrument  
			$\big\{  \Upsilon^\alpha_{m_1}\big\}$ which acts
			as follows on any state defined on $B^n$:
			\be
				\Upsilon^\alpha \ : \ 
				\psi^{B^n} 
				\longrightarrow
				\sum_{m_1} \ 
				\ketbra{m_1}{m_1}^{M_1}
				\otimes
				\left(
				\frac{ 	\sqrt{\Lambda^\alpha_{m_1}} \psi^{B^n} \sqrt{\Lambda^\alpha_{m_1}}  }
					{  \Tr\!\left[ \Lambda^\alpha_{m_1} \psi^{B^n} \right] }
				\right)^{B^{\prime n} }.
			\ee
			The POVM operators 
			$\big\{  \Lambda^\alpha_{m_1}\big\}$ are 
			constructed using the typical projector sandwich
			\be
				 \Pi_{\bar{\rho}}^{n} \ 
				 \Pi_{\bar{\rho}_{x_1^n(m_1)  },\delta}^{n}	 \ 
				 \Pi_{\bar{\rho}}^{n},
			\ee
			and normalized using the square root  measurement approach
			in order to satisfy $\Lambda^\alpha_{m_1} \geq 0$, $\sum_{m_1} \Lambda^\alpha_{m_1} = I$.
			The purpose of the quantum instrument is to extract the message $m_1$
			and store it in the register $M_1$, but also leave behind a system
			in $B^{\prime n}$ which can be processed further.
			
			An error analysis similar to that of the HSW theorem shows that 
			the quantum instrument $\big\{  \Upsilon^\alpha_{m_1}\big\}$ will
			correctly decode the message $m_1$ with high probability.
			This is because we chose the rate for the $m_1$ codebook
			to be $R_1 = I(X_1;B)_\theta - \delta$.
			Furthermore, it can be shown using the \emph{gentle operator lemma for ensembles}
			(Lemma~\ref{lem:gentle-operator}), 
			that the state which remains in the system $B^{\prime n}$
			is negligibly disturbed in the process.
			
			The receiver will then perform a second measurement 
			to recover the message $m_2$.
			%
			%
			The second measurement is a POVM 
			$\big\{  \Lambda^\alpha_{m_2|m_1}\big\}$
			constructed from the projectors
			\be
				 \Pi_{\bar{\rho}_{x_1^n(m_1)  },\delta}^{n}	 \ 
				 \Pi_{\rho_{x_1^n(m_1)  ,x_2^{n}\left(  m_2\right)}}^{n} \
				 \Pi_{\bar{\rho}_{x_1^n(m_1)  },\delta}^{n}, 
			\ee
			and appropriately normalized.
			Note that this measurement is chosen 
			conditionally on the codeword $X_1^n(m_1)$ that Sender~1 input to the channel.
			This is because, when the correct message $m_1$ is decoded in the first step,
			the receiver can infer the codeword which Sender~1 input to the channel.
			%
			%
			Thus, after the first step, the effective channel from Sender~2 to the receiver is
			\be
				(X_1^n, x_2^n) 	\to	(X_1^n, \rho_{X_1^n,x_2^n}^{B^n} ),
			\ee
			where $X_1^n$ is a random variable distributed according
			to $\prod_{i=1}^n p_{X_1}$.
			This is a setting in which the quantum packing lemma can be applied.
			By substituting $U^n=X_1^n$ and $X^n=X_2^n$ into
			Lemma~\ref{lem:q-packing}, we conclude that if we choose the rate
			to be $R_2 = I(X_2;B|X_1)_\theta - \delta$,
			then the message $m_2$ will be decoded correctly
			with high probability.


			The rate point $\beta_{p}$ corresponds to the alternate decode ordering
			where the receiver decodes the message $m_2$ first and $m_1$ second.
			All other rate pairs in the region can be obtained from the 
			corner points $\alpha_{p}$ and $\beta_{p}$ by
			using  \emph{time-sharing} and 
			\emph{resource wasting}.			
		\end{proof}


		Note that one of the key ingredients in the proof
		was the use of Lemma~\ref{lem:gentle-operator},
		which guarantees that the act of decoding $m_1$
		does not disturb the state too much.
		This step of our quantum decoding procedure may be counterintuitive at a first glance,
		since quantum mechanical measurements are usually 
		described as processes in which the quantum system is disturbed.
		Any retrieval of data from a quantum system inevitably disturbs the state
		of the system, so the second measurement, which the receiver performs 
		on the system $B^{\prime n}$,
		may fail if the first measurement has disturbed the state too much.
		The \emph{gentle measurement lemma} guarantees 
		that very little information disturbance to the state 
		occurs when there is one measurement outcome that is very likely.
		When the state of the receiver is $\rho_{x_1^n,x_2^n}^{B^n}$,
		we can be almost certain that the outcome of the quantum instrument 
		$\{ \Upsilon^\alpha_{m_1} \}$ is going to be $m_1$.
		Therefore, this process 
		leaves the state in $B^{\prime n}$ only slightly disturbed.

		The proof technique in Theorem~\ref{thm:cqmac-capacity} generalizes
		to the case of the $M$-sender MAC, which has 
		$M!$ dominant vertices. 
		Each vertex corresponds to one permutation of the decode ordering.


\section{Simultaneous decoding}
								\label{sec:mac-simult-decoding}

		Another approach for achieving the capacity of the multiple access channel, 
		which does not use time-sharing, is simultaneous decoding.
		%
		%
		In the classical version of this decoding strategy,
		the receiver will report $(m_1,m_2)$ if he finds a unique pair of codewords $X_1^n(m_1)$
		and $X_2^n(m_2)$ which are jointly typical with the output of the channel $Y^n$:
		\be
			\left( 
				X_1^n(m_1), \ X_2^n(m_2), Y^n
			\right)
			\ \in \  \mcal{J}^{(n)}_\epsilon(X_1,X_2,Y).
		\ee
		%
		Assuming the messages $m_1$ and $m_2$ are sent, 
		we categorize the different kinds of \emph{wrong message} decode errors
		that may occur.
		\begin{equation}%
		\begin{tabular}
		[c]{c|c|c}\hline\hline
		error & $\hat{M}_1$ & $\hat{M}_2$\\
		\hline\hline
		\eqref{eq:err-one}	 & $\ast$ & $m_2$\\
		\eqref{eq:err-two}	 & $m_1$ & $\ast$\\
		\eqref{eq:err-both}	 & $\ast$ & $\ast$\\ 
		\hline\hline
		\end{tabular}
						\label{eq:three-errors}%
		\end{equation}
		The $\ast$ in the above table denotes any message other than
		the one which was sent.		
		The analysis of the \emph{classical} simultaneous decoder
		uses the properties of the jointly typical sequences 
		and the randomness in the codebooks.
		Recall that a multi-variable sequence is jointly typical if and only
		if all the sequences in the subsets of the variables are jointly typical.
		Thus, the condition $\left( X_1^n(m_1), \ X_2^n(m_2), Y^n \right)
		\in \mcal{T}^{(n)}_\epsilon(X_1,X_2,Y)$ implies that:
		\begin{align}
					\left( 	X_1^n(m_1), Y^n \right) 
					&\in 
						\mcal{T}^{(n)}_\epsilon(X_1,Y), \\
					\left( 	X_2^n(m_2), Y^n \right) 
					&\in 
						\mcal{T}^{(n)}_\epsilon(X_2,Y), \\
					Y^n
					&\in 
						\mcal{T}^{(n)}_\epsilon(Y).
		\end{align}
		Starting from these conditions, it is straightforward
		to bound the probability of the different decoding error events 
		using the properties of the jointly typical sequences \cite{el2010lecture}.			
		
		In the quantum case, we can similarly identify three different error terms,
		the probabilities of which can be bounded by using the properties
		of the conditionally typical projectors.
		If we can construct a quantum measurement operator that ``contains'' 
		all the typical projectors so that we can obtain the appropriate averages
		of the output state in the error analysis, 
		then we would have a proof that simultaneous decoding is possible.

		If only things were so simple!
		The construction of a simultaneous decoding POVM turns out
		to be a difficult problem.
		Despite being built out of the same typical projectors,
		the operator constructed according to
		\be
			\Lambda_{m_1,m_2}  \ \ 
			\propto \ \ 
			\Pi_{\bar{\rho}_{x_2^{n}\left(  m_2\right)  }}^{n}
			\Pi_{\bar{\rho}_{x_1^{n}\left(  m_1\right)  }}^{n}
			\Pi_{\rho_{x_1^{n}\left(  m_1\right),x_2^{n}\left(  m_2\right)  }}^{n}
			\Pi_{\bar{\rho}_{x_1^{n}\left(  m_1\right)  }}^{n}
			\Pi_{\bar{\rho}_{x_2^{n}\left(  m_2\right)  }}^{n},
			\label{eqn:mtwo-outside}
		\ee
		is different from the operator
		\be
			\Lambda^\prime_{m_1,m_2}  \ \
			\propto \ \ 
			\Pi_{\bar{\rho}_{x_1^{n}\left(  m_1\right)  }}^{n}			
			\Pi_{\bar{\rho}_{x_2^{n}\left(  m_2\right)  }}^{n}
			\Pi_{\rho_{x_1^{n}\left(  m_1\right),x_2^{n}\left(  m_2\right)  }}^{n}
			\Pi_{\bar{\rho}_{x_2^{n}\left(  m_2\right)  }}^{n}
			\Pi_{\bar{\rho}_{x_1^{n}\left(  m_1\right)  }}^{n},
			\label{eqn:mone-outside}
		\ee
		because the different typical projectors do not commute in general.		
		In fact, there is very little we can say about the relationship
		between the subspaces spanned by the two averaged
		typical projectors: 
			$\Pi_{\bar{\rho}_{x_1^{n}\left(  m_1\right)  }}^{n}$ 
		and
			$\Pi_{\bar{\rho}_{x_2^{n}\left(  m_2\right)  }}^{n}$.
		This is a problem because, for one of the error terms in the analysis, we would like 
		to have $\Pi_{\bar{\rho}_{x_2^{n}\left(  m_2\right)  }}^{n}$
		on the ``outside'' as in \eqref{eqn:mtwo-outside} so that we can use Property~\ref{eqn:bound-on-size}  
        of typical
		projectors to obtain a factor $2^{nH(B|X_2)}$.
		For another error term, we want 
		$\Pi_{\bar{\rho}_{x_1^{n}\left(  m_1\right)  }}^{n}$ 
		to be on the outside as in \eqref{eqn:mone-outside} in order to be 
		able to do the averaging
		in the alternate order to obtain a term of the form $2^{nH(B|X_1)}$.
		Thus it would seem, and originally it seemed so to my colleagues and me,
		that the construction of a simultaneous decoding POVM for which 
		we can bound the probability of all error events might be a difficult task.
		
		Quantum simultaneous decoding actually \emph{is} possible,
		and this is what we will show in this section
		for the case of the multiple access channel with two senders.
		Our proof techniques do not generalize readily to quantum multiple
		access channels with more than two independent senders.
		At the end of this section we will formulate Conjecture~\ref{conj:sim-dec} regarding the 
		existence of a simultaneous decoder for three-sender multiple access channels,
		which will be required for the proof of Theorem~\ref{thm:quantum-HK-region}
		in the next chapter.

	\begin{figure}
	\begin{center}
		\includegraphics[width=0.6\textwidth]{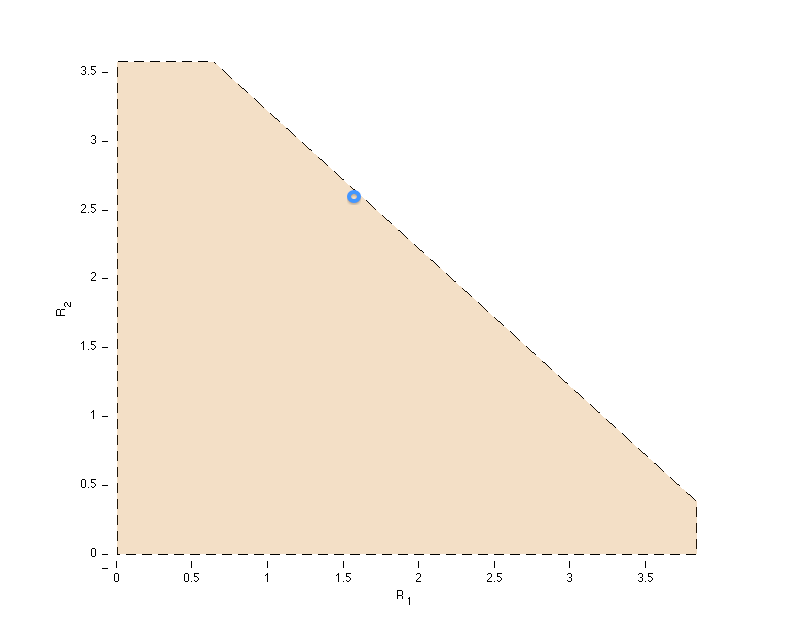}

		\caption{ \small
			Simultaneous decoding strategy.
			Simultaneous decoding of the two messages is more powerful than successive decoding,
			because it allows us to achieve any rate pair $(R_1,R_2)$ of the capacity region 
			without the need for time-sharing.
			}
		\label{fig:MAC-simult}
	\end{center}
	\end{figure}

		\begin{theorem}[Two-sender quantum simultaneous decoding]
		\label{thm:sim-dec-two-sender}
			Let  $(\mathcal{X}_1 \times \mathcal{X}_2,  \rho_{x_1,x_2}^{B}, \mathcal{H}^{B} )$
			be a quantum multiple access channel 
			with two senders and a single receiver,
			and let $p=p_{X_1}p_{X_2}$ be a choice for the input code distribution.
			Let 
			  $\{X_1^n(m_1)\}_{m_1 \in \{1, \dots, |\mcal{M}_1|\}}$ 
			and
			  $\{X_2^n(m_2)\}_{m_2 \in \{1, \dots, |\mcal{M}_2|\}}$
			be random codebooks generated according to the product distributions
			 $p_{X_1^n}$ and $p_{X_2^n}$.
			There exists a simultaneous decoding POVM 
			$\left\{  \Lambda_{m_1,m_2}\right\}_{m_1\in \mathcal{M}_1,m_2\in\mathcal{M}_2}$,
			with  expected average probability of error bounded from above by $\epsilon$
			for all $\epsilon,\delta>0$ and sufficiently large $n$,
			provided the rates $R_1,R_2$ satisfy the inequalities
		        \bea 
		            R_1             &\leq&    I(X_1;B|X_2)_\theta,  \label{myThmEqnsOne} \\
		            R_2             &\leq&    I(X_2;B|X_1)_\theta,  \label{myThmEqnsTwo} \\
		            R_1+R_2   &\leq&    I(X_1X_2;B)_\theta, \label{myThmEqnsThree}
		        \eea
		        where the state $\theta^{X_1X_2B}$ is defined in \eqref{eqn:state-in-WinterThm}.
		\end{theorem}

		The main difference between the coding strategy 
		employed by Winter in the proof of  Theorem~\ref{thm:cqmac-capacity}
		and Theorem~\ref{thm:sim-dec-two-sender} above is that the latter
		does not require the use of time-sharing.
		Using the simultaneous decoding approach we can achieve any 
		of the rates in the QMAC capacity region 
		using a \emph{single} codebook,
		whereas time-sharing requires us to switch between 
		the two codebooks for the vertices.
		This distinction is minor in the context of the multiple access channel
		problem, but it will become important in situations where there 
		are multiple receivers as in the compound multiple access channel
		and the interference channel.
		Note that Sen gave an alternate proof of Theorem~\ref{thm:sim-dec-two-sender}
		using a different approach \cite{S11a}.

		\begin{proof}[Proof of Theorem~\ref{thm:sim-dec-two-sender}]
		The proof proceeds by random coding arguments using the
		properties of projectors onto the typical subspaces of the output states
		and the \emph{square root} measurement.

		Consider some choice $p=p_{X_1}(x_1)p_{X_2}(x_2)$
		for the input distributions.

		\prfsec{Codebook construction} 
		Randomly and independently generate $2^{nR_{1}}$
		sequences $x_{1}^{n}\!\left(  m_1 \right)$, $m_1 \in\left[1: 2^{nR_{1}}\right]$,
		according to  $\prod\limits_{i=1}^{n}p_{X_{1}}\!\left(  x_{1i}\right)$.
		Similarly, generate randomly and independently the codebook 
		$\{ x_{2}^{n}\!\left(  m_2 \right) \}$, $m_2 \in\left[1: 2^{nR_{2}}\right]$,
		according to $\prod\limits_{i=1}^{n}p_{X_{2}}\!\left(  x_{2i}\right)$.
		
		\medskip
		\prfsec{POVM construction}
		In order to lighten the notation, the channel output will be denoted
		with the shorthand  $\rho_{m_1,m_2}\equiv \rho_{x_1^n(m_1)  ,x_2^{n}\left(  m_2\right)  }$,
		when the inputs to the channel are $x_1^n(m_1) $ and 	$x_2^n(m_2)  $.
		Let $\PIMAConetwo \equiv \Pi_{\rho_{x_1^n(m_1)  ,x_2^{n}\left(  m_2\right)
		},\delta}^{n}$ be the conditionally typical projector for that state.
		Consider the following averaged output states:
		\begin{align}
			\bar{\rho}_{x_1} &  \equiv
				\sum_{x_2}p_{X_2}\!\left(  x_2\right)  \rho_{x_1,x_2},\label{eq:rho_x} \\
			\bar{\rho}_{x_2} &  \equiv
				\sum_{x_1}p_{X_1}\!\left(  x_1\right)  \rho_{x_1,x_2},\label{eq:rho_y} \\
			\bar{\rho} &  \equiv
				\sum_{x_1,x_2}p_{X_1}\!\left(  x_1\right)  p_{X_2}\!\left(  x_2\right) \rho_{x_1,x_2}.\label{eq:rho}
		\end{align}		
		Let $\PIMACone \equiv \Pi_{\bar{\rho}_{x_1^n(m_1)  },\delta}^{n}$ be the
		conditionally typical projector for the tensor product state 
		$\bar{\rho}_{m_1} \equiv \bar{\rho}_{x_1^n(m_1)  }$ defined by (\ref{eq:rho_x}) for $n$ uses of the
		channel.
		Let $\PIMACtwo \equiv \Pi_{\bar{\rho}_{x_2^{n}\left(  m_2\right)  },\delta}^{n}$ 
		be the conditionally typical projector for the tensor product state 
		$\bar{\rho}_{m_2} \equiv  \bar{\rho}_{x_2^{n}\left(  m_2\right)  }$ 
		defined by (\ref{eq:rho_y}) 
		and finally let $\Pi_{\bar{\rho},\delta}^{n}$ be the typical
		projector for the 
		state $\bar{\rho}^{\otimes n}$ defined by~(\ref{eq:rho}).

		The detection POVM\ $\left\{  \Lambda_{m_1,m_2}\right\}  $ has the following form:
		\vspace{-2mm}
		\begin{align}
		\Lambda_{m_1,m_2} &  \!\! \equiv \!
		\left(  
			\sum_{m_1^{\prime},m_2^{\prime}}P_{m_1^{\prime},m_2^{\prime}}
		\right)^{\!\!\!-\frac{1}{2}}
		\!\! P_{m_1,m_2}
		\left(
			\sum_{m_1^{\prime},m_2^{\prime}}P_{m_1^{\prime},m_2^{\prime}}
		\right)^{\!\!-\frac{1}{2}}\!\!, 
		\nonumber
		\end{align}
		where 
		\begin{align}
		%
		P_{m_1,m_2} &  \equiv 
			\Pi_{\bar{\rho},\delta}^{n}  \
			\PIMACone \
			 \PIMAConetwo  \
			\PIMACone \
			\Pi_{\bar{\rho},\delta}^{n},				 \label{eq:proj-sandwitch} 
		\end{align}
		is a positive operator which consists of three typical projectors ``sandwiched'' together.
		Observe that the layers of the sandwich go from the more general ones on the outside
		to the more specific ones on the inside.
		Observe also that the conditionally typical projector $\PIMACtwo$ is not included.



		The average error probability of the code is given by:%
		\begin{equation}
			\overline{p}_{e}
			\equiv
			\frac{1}{|\mathcal{M}_1||\mathcal{M}_2|}\sum_{m_1,m_2}
				\text{Tr}\left[ 
					\left(I-\Lambda_{m_1,m_2}\right)  
					\rho_{m_1,m_2}
				\right].
			\label{eq:avg-error-prob}%
		\end{equation}
		
		The first step in our error analysis is to make a substitution of the output state 
		$\rho_{m_1,m_2}$ 
		with a smoothed version:
		\be
			\tilde{\rho}_{m_1,m_2} \equiv \PIMACtwo \rho_{m_1,m_2} \PIMACtwo.
			\label{smoothed-rho}
		\ee
		We do this to ensure that we will have the operator $\PIMACtwo$ 
		inside the trace when we perform the averaging.
		The term \emph{smoothing} refers to the fact that we are now coding for a different
		channel which has all of the $\PIMACtwo$-atypical subspace removed,
		i.e., we remove the ``spikes'' (the large eigenvalues).

		We can use the inequality
		\be
		\text{Tr}\!\left[  \Lambda\rho\right]  
		\leq
		\text{Tr}\!\left[ \Lambda \sigma\right]  
		+\left\Vert \rho-\sigma\right\Vert _{1}
		\label{eqn:tr-trick-copy}
		\ee
		from Lemma~\ref{lem:trace-inequality},
		which holds for all  operators such that   $0\leq \rho, \sigma, \Lambda \leq I$,
		in order to bound the \emph{smoothing penalty} which we incur
		as a result of the substitution.

		After the substitution step \eqref{eq:avg-error-prob} and the use of \eqref{eqn:tr-trick-copy},
		we obtain the following bound on the 
		probability of error:
		%
		\begin{align}
			\overline{p}_{e}
			\leq &
			\frac{1}{|\mathcal{M}_1||\mathcal{M}_2|}\!\sum_{m_1,m_2}\!\!\Bigg[  
				\text{Tr}\!\left[  
					\left(I-\Lambda_{m_1,m_2}\right)  
					\tilde{\rho}_{m_1,m_2}
				\right] 
			 	+ \Vert \tilde{\rho}_{m_1,m_2} - \rho_{m_1,m_2} \Vert_1
			\Bigg].
			\label{eq:error-bound-pre-HN}%
		\end{align}
		
		The next step is to use the Hayashi-Nagaoka operator inequality 
		\cite{hayashi2003general}  (Lemma~\ref{lem:HN-inequality}):
		\[
			I-\left(  S+T\right)^{-\frac{1}{2}}S\left(  S+T\right)  ^{-\frac{1}{2}}
			\leq
			2\left(  I-S\right)  +4T.
		\]
		Choosing $S=P_{m_1,m_2}, \ T=\sum_{\left(  m_1^{\prime},m_2^{\prime}\right)
		\neq\left(  m_1,m_2\right)  }P_{m_1^{\prime},m_2^{\prime}}$, 
		we apply the above operator inequality to bound the average error
		probability of the first term in  (\ref{eq:error-bound-pre-HN}) as:
		\begin{align}
			\overline{p}_{e}
			\leq &
			\frac{1}{|\mathcal{M}_1||\mathcal{M}_2|}\!\sum_{m_1,m_2}\!\!\Bigg[  
					2\text{Tr}\!\left[  
					\left(I-P_{m_1,m_2}\right)  \tilde{\rho}_{m_1,m_2} 
					\right] 			\label{eq:error-bound-HN}  \\
			 &  \!\!\!
			 	\qquad \qquad \qquad \qquad +4\hspace{-0.9cm}\sum_{\left(  m_1^{\prime},m_2^{\prime}\right)  \neq\left(m_1,m_2\right)  } 
				\hspace{-0.8cm} \text{Tr}\!\left[  
					P_{m_1^{\prime},m_2^{\prime}}\tilde{\rho}_{m_1,m_2} 
					\right]  
					+ 
						\Vert \tilde{\rho}_{m_1,m_2} \!- \!\rho_{m_1,m_2} \Vert_1
			\Bigg].   \nonumber
		\end{align}
		The three terms in the summation have an intuitive interpretation.
		The first term corresponds to the case when the output state 
		 is non-typical, the second term describes the probability of a 
		 wrong message being decoded,
		 and the third term accounts for the \emph{smoothing penalty}
		 which we have to pay for using a code designed for
		 the channel $\tilde{\rho}_{m_1,m_2}$ on the channel $\rho_{m_1,m_2}$.

		We apply a random coding argument to bound the expectation of 
		the average error probability in \eqref{eq:error-bound-HN}.
		%
		We compute the expected value of the error terms with respect
		to the random choice of codebook: $\{X_1^n(m_1)\}$, $\{X_2^n(m_2)\}$.
		Recall that in our shorthand notation, the codewords are not indicated.
		Thus when we say $\mathop{\mathbb{E}}_{X^n_1,X^n_2} \rho_{m_1,m_2}$,
		we really mean $\mathop{\mathbb{E}}_{X^n_1,X^n_2} \rho_{X^n_1(m_1),X^n_2(m_2)}$.

		A bound on the first term in \eqref{eq:error-bound-HN}  follows from the following argument:
		{\allowdisplaybreaks
		\begin{align}
	  	\mathop{\mathbb{E}}_{X^n_1,X^n_2} \text{Tr}&\!\left[  P_{m_1,m_2}\tilde{\rho}_{m_1,m_2} \right] 
		 = \nonumber \\
		& =  \mathop{\mathbb{E}}_{X^n_1,X^n_2} \text{Tr}\!\left[  
			\Pi_{\bar{\rho},\delta}^{n} \ 
			\PIMACone \ 
			\PIMAConetwo \
			\PIMACone \ 
			\Pi_{\bar{\rho},\delta}^{n} \ \
			\  \PIMACtwo  \rho_{m_1,m_2} \PIMACtwo
			\right]  \nonumber\\
		&   \geq \mathop{\mathbb{E}}_{X^n_1,X^n_2} \text{Tr}\!\left[  
		\PIMAConetwo \rho_{m_1,m_2} \right]  \nonumber \\
			& \qquad \qquad \ \ \ 
			-\mathop{\mathbb{E}}_{X^n_1,X^n_2}
			\left\Vert 
			\PIMACtwo \rho_{m_1,m_2 } \PIMACtwo  -\rho_{m_1,m_2} 
			\right\Vert _{1}\nonumber\\
			& \qquad \qquad \ \ \  
			-\mathop{\mathbb{E}}_{X^n_1,X^n_2} 
			\left\Vert 
			\PIMACavg  \rho_{m_1,m_2}  \PIMACavg  -\rho_{m_1,m_2}  \right\Vert _{1}\nonumber\\
			& \qquad \qquad \ \ \ 
			-\mathop{\mathbb{E}}_{X^n_1,X^n_2}
			\left\Vert 
			\PIMACone \rho_{m_1,m_2} \PIMACone-\rho_{m_1,m_2}
			\right\Vert _{1}\nonumber\\
		&   \geq \mathop{\mathbb{E}}_{X^n_1,X^n_2} \text{Tr}\!\left[  
		\PIMAConetwo \rho_{m_1,m_2} \right]  \nonumber - 6\sqrt{\epsilon} \\
		&   \geq1-\epsilon-6\sqrt{\epsilon}. \label{eq:first-error-chain}%
		\end{align}%
		}%
		%
		The first inequality follows from \eqref{eqn:tr-trick-copy} (Lemma~\ref{lem:trace-inequality}) 
		applied three times.
		The second inequality follows from Lemma~\ref{lem:gentle-operator}
		and the properties of the 
		conditionally typical projectors:		
		\eqref{eqn:QMACtypSupAvgOne}, \eqref{eqn:QMACtypSupAvgTwo}
		and \eqref{eqn:QMACtypSupAvgBoth} given in Appendix~\ref{apdx:quantum-typicality}.
		The last inequality follows from equation \eqref{eqn:QMACtypSup}.

		The same reasoning is used to obtain a bound the expectation of the smoothing-penalty
		(the third term  in \eqref{eq:error-bound-HN}).
		\begin{align}
			\mathbb{E}_{X^n_1,X^n_2}  \Vert \tilde{\rho}_{m_1,m_2}  -  \rho_{m_1,m_2} \Vert_1 
			& = 
				\mathbb{E}_{X^n_1,X^n_2}  \Vert \PIMACtwo \rho_{m_1,m_2} \PIMACtwo \ -  \  \rho_{m_1,m_2} \Vert_1  \nonumber \\
			& \leq  
				2\sqrt{\epsilon}.  \label{eqn:smoothing-penality-eps}
		\end{align}

		The main part of the error analysis consists 
		of obtaining a bound on the second term in \eqref{eq:error-bound-HN}.
		This term corresponds to the probability that a wrong message pair
		is decoded by the receiver. 
		We split this term into  three parts, each representing a different type of decoding error:
		\vspace{-2mm}
		\begin{align}
		&\hspace{-4mm} 
		\sum_{\left(  m_1^{\prime},m_2^{\prime}\right)  \neq\left(  m_1,m_2\right)  }
		\hspace{-9mm}
		\Tr\!\left[  P_{m_1^{\prime},m_2^{\prime}}
		\tilde{\rho}_{m_1,m_2}
		\right]  
		= \nonumber \\
		& \qquad \qquad  =
			\sum_{m_1^{\prime}\neq m_1}\Tr\!\left[  P_{m_1^{\prime},m_2}
			\tilde{\rho}_{m_1,m_2}
			\right]  
			\tag{E1} \label{eq:err-one}\\
		& \qquad \qquad \qquad 
			+\sum_{m_2^{\prime}\neq m_2}\Tr\!\left[P_{m_1,m_2^{\prime}}
			\tilde{\rho}_{m_1,m_2}
			\right] \tag{E2} \label{eq:err-two} \\
		& \qquad \qquad \qquad 
			+ \hspace{-6mm} \sum_{m_1^{\prime}\neq m_1,m_2^{\prime}\neq m_2}
			\hspace{-6mm} \Tr\!\left[  P_{m_1^{\prime},m_2^{\prime}}
								\tilde{\rho}_{m_1,m_2}
			\right]. \tag{E12}
			\label{eq:err-both}%
		\end{align}%


		We will bound each of these terms in turn.
		
		\bigskip
		\prfsec{Bound on (E1) }
		The expectation over the random choice of codebook for the error term \eqref{eq:err-one}
		is as follows:
		\vspace{-1mm}
		\begin{align}
		 \mathop{\mathbb{E}}_{X^n_1,X^n_2} \left\{ 
		 	\eqref{eq:err-one}
			\right\} & 
			=
			\mathop{\mathbb{E}}_{X^n_1,X^n_2} 
		\big\{  \sum_{m_1^{\prime}\neq m_1}
		\text{Tr}\left[  P_{m_1^{\prime},m_2}
		\tilde{\rho}_{m_1,m_2}
		\right]  \big\}  \nonumber\\
		&  
		\overset{\mdingone}{=} 
		\!\!\!\sum_{m_1^{\prime}\neq m_1}\mathop{\mathbb{E}}_{X^n_2}\left\{  \text{Tr}\left[
		\mathop{\mathbb{E}}_{X^n_1}\left\{  
		P_{m_1^{\prime},m_2} \right\}  
		\mathop{\mathbb{E}}_{X^n_1}\left\{
		  \tilde{\rho}_{m_1,m_2}
		  \right\} 
		 \right]
		\right\}  \nonumber\\
		&  
		= 
		\!\!\!\sum_{m_1^{\prime}\neq m_1}\mathop{\mathbb{E}}_{X^n_2}\left\{  \text{Tr}\left[
		\mathop{\mathbb{E}}_{X^n_1}\left\{  
		P_{m_1^{\prime},m_2} \right\}  
		\mathop{\mathbb{E}}_{X^n_1}\left\{
		 	\PIMACtwo
		  \rho_{m_1,m_2}
		 	\PIMACtwo
		  \right\} 
		 \right]
		\right\}  \nonumber\\
		&  
		= 
		\!\!\!\sum_{m_1^{\prime}\neq m_1}\mathop{\mathbb{E}}_{X^n_2}\left\{  \text{Tr}\left[
		\mathop{\mathbb{E}}_{X^n_1}\left\{  
		P_{m_1^{\prime},m_2} \right\}
		\ 
  		 	\PIMACtwo \!
		\mathop{\mathbb{E}}_{X^n_1}\left\{
		  \rho_{m_1,m_2}
		  \right\} 
  		 	\PIMACtwo
		 \right]
		\right\}  \nonumber\\
		&  
		\overset{\mdingtwo}{=}
		 \!\!\!\sum_{m_1^{\prime}\neq m_1}\mathop{\mathbb{E}}_{X^n_1X^n_2}\left\{  \text{Tr}\left[
			P_{m_1^{\prime},m_2} 
		 	\PIMACtwo
			\bar{\rho}_{m_2}
			\PIMACtwo \right]
		\right\}  \nonumber \\
		&  
		\overset{\mdingthree}{\leq}
		 2^{-n\left[  H\left(B|X_2\right) -\delta\right]}  
		\!\! \sum_{m_1^{\prime}\neq m_1}\mathop{\mathbb{E}}_{X^n_1X^n_2}\left\{  \text{Tr}\left[
		P_{m_1^{\prime},m_2} 
		 	\PIMACtwo \right]
		\right\}  \nonumber 
		\end{align}%
		\vspace{-1mm}%
		Equation~\dingone follows because the codewords 
		for $m_1^{\prime}$ and $m_1$ are independent.
		Equality~\dingtwo comes from the definition of the averaged code state 
		$\bar{\rho}_{m_2} \equiv  \bar{\rho}_{x_2^{n}\left(  m_2\right)  }$.
		The inequality~\dingthree follows from the bound
		\[
		\PIMACtwo \bar{\rho}_{m_2} \PIMACtwo 
		\leq
		2^{-n\left[  H\left(B|X_2\right) -\delta\right]  }
		\PIMACtwo.
		\]

		We focus our attention on the expression inside the trace:
		\begin{align}
		\!\!\!\!\!\text{Tr}\left[
		P_{m_1^{\prime},m_2}  \ 
		 	\PIMACtwo \right] \nonumber 
		& = 
		\text{Tr}\left[
		\Pi_{\bar{\rho},\delta}^{n}  \
		\PIMAConepr \
		\PIMAConeprtwo  \
		\PIMAConepr \
		\Pi_{\bar{\rho},\delta}^{n} 
		\ \ \PIMACtwo		
		 \right] \nonumber \\
		& \overset{\mdingfour}{=} 
		\text{Tr}\left[
		\PIMAConepr \
		\Pi_{\bar{\rho},\delta}^{n} \
		\PIMACtwo		\ 
		\Pi_{\bar{\rho},\delta}^{n}  \
		\PIMAConepr  \
		\PIMAConeprtwo
		 \right] \nonumber \\		
		& \overset{\mdingfive}{\leq}
		\text{Tr}\left[
		\PIMAConeprtwo 
		 \right]. \nonumber 
		\end{align}
		In the first step we substituted the definition of $P_{m_1,m_2}$ from
		equation \eqref{eq:proj-sandwitch}.
		Equality \dingfour follows from the cyclicity of trace.
		%
		Inequality \dingfive follows from
		\vspace*{-1mm}
		\begin{align}
		\PIMAConepr \PIMACavg \PIMACtwo \PIMACavg \PIMAConepr 
		 \leq 
		 \PIMAConepr \PIMACavg \PIMAConepr
		 \leq 
		 \PIMAConepr
		 \leq
		 I. \label{positive-op-anihilation}
		\end{align}

		Next, we obtain the following bound on the expected probability 
		of the term \eqref{eq:err-one}:
		\begin{align}
		\!\!\mathop{\mathbb{E}}_{X^n_1,X^n_2}
		\!\!\!\!\!\!\left\{ \eqref{eq:err-one}
		\right\}  
		& \leq
		2^{-n\left[  H\left(  B|X_2\right)-\delta\right]  } 
		\!\!\!\!
		\sum_{m_1^{\prime}\neq m_1}
		\mathop{\mathbb{E}}_{X^n_1,X^n_2}
		\left\{  \Tr\!\left[  \PIMAConeprtwo 
		\right]  \right\}  \nonumber\\
		&  \overset{\mdingsix}{\leq}
		2^{-n\left[  H\left(  B|X_2\right)-\delta\right]  }
		\sum_{m_1^{\prime}\neq m_1}
		2^{n\left[  H\left(  B|X_1X_2\right)+\delta\right]  }\nonumber\\
		&  \leq
		|\mcal{M}_1|\  2^{-n\left[  I\left(X_1;B|X_2\right)  -2\delta\right]  }.
		\label{eq:err-one-bound}
		\end{align}
		Inequality  \dingsix 
		follows from the bound%
		\begin{align}
			\text{Tr}\{
			\PIMAConetwo
			\}\leq2^{n\left[  H\left(  B|X_1X_2\right)  +\delta\right]  }%
			\nonumber
		\end{align}
		on the rank of a conditionally typical projector.

		\prfsec{Bound on (E2) }	
		We employ a different argument to bound the probability of the second error term 
		\eqref{eq:err-two} based on the following fact
		\begin{align}
			\PIMAConetwo 
			& \leq 
				2^{n[H(B|X_1X_2) + \delta] } 
				\PIMAConetwo
				\rho^B_{m_1,m_2}
				\PIMAConetwo \nonumber \\
			& =
				2^{n[H(B|X_1X_2) + \delta] } 
				\sqrt{\rho^B_{m_1,m_2}}
				\PIMAConetwo
				\sqrt{\rho^B_{m_1,m_2}}
				\nonumber \\
			& \leq
				2^{n[H(B|X_1X_2) + \delta] } 
				\rho^B_{m_1,m_2},
				\label{eqn:proj-trick-in-prf}
		\end{align}
		which we refer to as the \emph{projector trick} \cite{GLM10}.
		The first inequality is the standard lower bound on the eigenvalues of
		$\rho^B_{m_1,m_2}$ expressed as an operator upper bound on the projector $\PIMAConetwo$.
		The equality follows because the state and its typical projector commute.
		The last inequality follows from $0 \leq \PIMAConetwo \leq I$.
		
		We now proceed to bound the expectation of the error term \eqref{eq:err-two}.
		{\allowdisplaybreaks
		\begin{align}
			\mathop{\mathbb{E}}_{X^n_1,X^n_2}
			\!\!
			\Big\{ \eqref{eq:err-two}
			\Big\} 
			&= 
			\mathop{\mathbb{E}}_{X^n_1,X^n_2}
		\!\!
		\left\{ 
		\sum_{m_2^{\prime}\neq m_2}\text{Tr}\left[P_{m_1,m_2^{\prime}}
			\tilde{\rho}_{m_1,m_2}
			\right] 
		\right\} \nonumber \\
		%
		&  = \!\!\!\sum_{m_2^{\prime}\neq m_2}\mathop{\mathbb{E}}_{X^n_1}\left\{  \text{Tr}\left[
		\mathop{\mathbb{E}}_{X^n_2}\left\{  
		P_{m_1,m_2^{\prime}} \right\}  
		\mathop{\mathbb{E}}_{X^n_2}\left\{
		  \tilde{\rho}_{m_1,m_2}
		  \right\} 
		 \right]
		\right\}  \nonumber\\
		&  = \!\!\!\sum_{m_2^{\prime}\neq m_2}
		\!\!\!\mathop{\mathbb{E}}_{X^n_1}\left\{  \text{Tr}\left[
		\mathop{\mathbb{E}}_{X^n_2}\left\{  
			\Pi_{\bar{\rho},\delta}^{n}  \
			\PIMACone \
			 \PIMAConetwopr \
			 \PIMACone \
			 \Pi_{\bar{\rho},\delta}^{n}  
		\right\}  
		\mathop{\mathbb{E}}_{X^n_2}\left\{
		  \tilde{\rho}_{m_1,m_2}
		  \right\} 
		 \right]
		\right\}  \nonumber\\
		&  = \!\!\!\sum_{m_2^{\prime}\neq m_2}
		\!\!\!\mathop{\mathbb{E}}_{X^n_1}\left\{  
		\text{Tr}\!\left[
		\Pi_{\bar{\rho},\delta}^{n}
		\mathop{\mathbb{E}}_{X^n_2}\!\!\left\{  
			  \PIMACone\PIMAConetwopr\PIMACone
		\right\}\!
		\Pi_{\bar{\rho},\delta}^{n}
		\!\!\mathop{\mathbb{E}}_{X^n_2}\!\!\left\{
		  \tilde{\rho}_{m_1,m_2}
		  \right\} 
		 \right]
		\!\right\}  \nonumber
		\end{align}}

		We focus our attention on the first expectation inside the trace:
		\begin{align}
		\mathop{\mathbb{E}}_{X^n_2}\!\!\left\{  
			  \PIMACone\PIMAConetwopr\PIMACone
		\right\}  		  
		& \ \ 
		\overset{ \mdingone}{
		\leq 
		}
		2^{n[H(B|X_1X_2) + \delta] } 
		\mathop{\mathbb{E}}_{X^n_2}\!\!\left\{  
			  \PIMACone \rho^B_{m_1,m_2^{\prime}} \PIMACone
		\right\}  		  \nonumber \\
		& \ \ 
		= 2^{n[H(B|X_1X_2) + \delta] } 
		\PIMACone
		\mathop{\mathbb{E}}_{X^n_2}\!\!\left\{  
			   \rho^B_{m_1,m_2^{\prime}} 
		\right\} 
		\PIMACone
		\nonumber \\
		& \ \ 
		= 2^{n[H(B|X_1X_2) + \delta] } 
		\PIMACone
		\bar{\rho}_{m_1}
		\PIMACone
		\nonumber \\
		& \ \ 
		\overset{ \mdingtwo}{
		\leq 
		}
		2^{n[H(B|X_1X_2) + \delta] } 
		2^{-n[H(B|X_1) - \delta] } 
		\PIMACone
		\nonumber \\
		& \ \ 
		= 2^{-n[I(X_2;B|X_1)- 2\delta] } 
		\PIMACone.
		\nonumber
		\end{align}
		In inequality \dingone we used the projector trick from \eqref{eqn:proj-trick-in-prf}.
		Inequality \dingtwo follows from the properties of the conditionally
		typical projector $\PIMACone$.

		Substituting back into the expression for the error bound,
		we obtain:
		{\allowdisplaybreaks
		\begin{align}
		 	\hspace{-3mm}
			\mathop{\mathbb{E}}_{X^n_1,X^n_2}
			\!\!
			\{ \eqref{eq:err-two}
			\}  
		& \leq
		2^{-n[I(X_2;B|X_1)- 2\delta] } 
		\!\!\!\!\sum_{m_2^{\prime}\neq m_2} \!\!
		\!\!\text{Tr}\!\left[
		\Pi_{\bar{\rho},\delta}^{n}\!
			  \PIMACone
		\Pi_{\bar{\rho},\delta}^{n}
		  \tilde{\rho}_{m_1,m_2}
		 \right]
		\nonumber \\
		&  =
		2^{-n[I(X_2;B|X_1)- 2\delta] } 
		\!\!\!\sum_{m_2^{\prime}\neq m_2}
		\!\!\text{Tr}\!\left[
		\Pi_{\bar{\rho},\delta}^{n}\!
			  \PIMACone
		\Pi_{\bar{\rho},\delta}^{n}
		 \PIMACtwo \rho_{m_1,m_2} \PIMACtwo
		 \right]
		\nonumber \\
		&  =
		2^{-n[I(X_2;B|X_1)- 2\delta] } 
		\!\!\!\sum_{m_2^{\prime}\neq m_2}
		\!\!\text{Tr}\!\left[
		\PIMACtwo
		\Pi_{\bar{\rho},\delta}^{n}\!
			  \PIMACone
		\Pi_{\bar{\rho},\delta}^{n}
		 \PIMACtwo \rho_{m_1,m_2} 
		 \right]
		\nonumber \\
		&  
		\overset{\mdingthree}{\leq}
		2^{-n[I(X_2;B|X_1)- 2\delta] } 
		\!\!\!\sum_{m_2^{\prime}\neq m_2}
		\!\!\text{Tr}\!\left[
		 \rho_{m_1,m_2} 
		 \right]
		\nonumber \\ 
		& 
		 \leq
		2^{-n\left[  I\left(  X_2;B|X_1\right)  -2\delta\right]  }
		|\mathcal{M}_2|.  \label{eq:err-two-bound}
		\end{align}}%
		\noindent%
		Inequality \dingthree follows from an argument analogous to
		\eqref{positive-op-anihilation}.

		\prfsec{Bound on (E12) }
		We use a slightly different argument in order
		to bound the probability of the third error term: %
		{\allowdisplaybreaks
		\begin{align}
		\mathop{\mathbb{E}}_{X^n_1,X^n_2} 
		\!\!
		\Big\{\eqref{eq:err-both}
		\Big\}   
		& = 
		\mathop{\mathbb{E}}_{X^n_1,X^n_2}
		\left\{
		\sum_{m_1^{\prime}\neq m_1,m_2^{\prime}\neq m_2}
			 \Tr\!\left[  
			 	P_{m_1^{\prime},m_2^{\prime}}
				\tilde{\rho}_{m_1,m_2}
			\right] 
		\right\}  \nonumber \\
		&  
		\overset{\mdingone}{=} 
		\sum_{m_1^{\prime}\neq m_1,m_2^{\prime}\neq m_2}
		\mathop{\mathbb{E}}_{X^n_2}\left\{  \text{Tr}\left[
		\mathop{\mathbb{E}}_{X^n_1}\left\{  
		P_{m_1^{\prime},m_2^{\prime}} \right\}  
		\mathop{\mathbb{E}}_{X^n_1}\left\{
		  \tilde{\rho}_{m_1,m_2}
		  \right\} 
		 \right]
		\right\}  \nonumber\\
		&  
		\overset{\mdingtwo}{=}
		\sum_{m_1^{\prime}\neq m_1,m_2^{\prime}\neq m_2}
		\mathop{\mathbb{E}}_{X^n_2}\left\{  \text{Tr}\left[
		\mathop{\mathbb{E}}_{X^n_1}\left\{  
		P_{m_1^{\prime},m_2^{\prime}  } \right\}  
		 	\PIMACtwo
			\bar{\rho}_{m_2}
		 	\PIMACtwo
		 \right]
		\right\}  \nonumber \\
		&  
		\overset{\mdingthree}{\leq}
		\sum_{m_1^{\prime}\neq m_1,m_2^{\prime}\neq m_2}
		\mathop{\mathbb{E}}_{X^n_2}\left\{  \text{Tr}\left[
		\mathop{\mathbb{E}}_{X^n_1}\left\{  
		P_{m_1^{\prime},m_2^{\prime}  } \right\}  
			\bar{\rho}_{m_2}
		 \right]
		\right\}  \nonumber \\
		&  
		=
		\sum_{m_1^{\prime}\neq m_1,m_2^{\prime}\neq m_2}
		  \text{Tr}\left[
		\mathop{\mathbb{E}}_{X^n_1,X^n_2}\left\{  
		P_{m_1^{\prime},m_2^{\prime}  } \right\}  
		\mathop{\mathbb{E}}_{X^n_2}\left\{
			\bar{\rho}_{m_2}
		\right\}
		 \right]
		 \nonumber \\
		&  
		= \!\!\!\!\!\!
		\sum_{m_1^{\prime}\neq m_1,m_2^{\prime}\neq m_2}
		  \text{Tr}\left[
		\mathop{\mathbb{E}}_{X^n_1,X^n_2}\!\!\!\left\{  
		P_{m_1^{\prime},m_2^{\prime}  } \right\}  
		\bar{\rho}^{\otimes n}
		 \right]
		 \nonumber \\
		&  
		= \!\!\!\!\!\!
		\sum_{m_1^{\prime}\neq m_1,m_2^{\prime}\neq m_2}
		\!\!\!\!\!
		\mathop{\mathbb{E}}_{X^n_1,X^n_2}\!\!\!\left\{  \text{Tr}\left[
			\Pi_{\bar{\rho},\delta}^{n}  \
			\PIMAConepr \
			 \PIMAConeprtwopr \
			 \PIMAConepr \
			 \Pi_{\bar{\rho},\delta}^{n}  
			\ \ 
			\bar{\rho}^{\otimes n}
		 \right]
		\right\}  \nonumber \\
		&  
		\overset{\mdingfour}{\leq}
		2^{-n[H(B)-\delta]}
		\!\!\!\!\!\!\!\!\!\!\!\!
		\sum_{m_1^{\prime}\neq m_1,m_2^{\prime}\neq m_2}
		\!\!\!\!\!
		\mathop{\mathbb{E}}_{X^n_1,X^n_2}\!\!\left\{  \text{Tr}\left[
			\PIMAConepr \
			 \PIMAConeprtwopr \
			 \PIMAConepr \
			 \Pi_{\bar{\rho},\delta}^{n}  
		 \right]
		\right\}  \nonumber \\
		&  
		\overset{\mdingfive}{\leq}
		2^{-n[H(B)-\delta]}
		\sum_{m_1^{\prime}\neq m_1,m_2^{\prime}\neq m_2}
		\mathop{\mathbb{E}}_{X^n_1,X^n_2}\left\{  \text{Tr}\left[
			 \PIMAConetwopr \
		 \right]
		\right\}  \nonumber \\
		&  
		\overset{\mdingsix}{\leq}
		2^{-n[H(B)-\delta]}
		2^{n[H(B|X_1X_2) + \delta ] }
		\sum_{m_1^{\prime}\neq m_1,m_2^{\prime}\neq m_2}
		\ 1
		 \nonumber \\
		& \leq 
		|\mathcal{M}_1||\mathcal{M}_2|\ 2^{-n\left[  I\left(  X_1X_2;B\right)-2\delta\right]  }.
		\label{eqn:sum-rate-bound}
		\end{align}}%
Equality \dingone follows from the independence of	the codewords.
To obtain equality \dingtwo we take the $X_1^n$ expectation over the state.
Inequality \dingthree follows from
$\PIMACtwo \
\bar{\rho}_{m_2}  \
\PIMACtwo
=
\sqrt{ \bar{\rho}_{m_2} } \
\PIMACtwo 
\sqrt{ \bar{\rho}_{m_2} }
\leq
\bar{\rho}_{m_2}$.
Inequality \dingfour is obtained by using the cyclicity of trace to
surround the state $\bar{\rho}^{\otimes n}$ by its typical projectors
and then using the property 
$\Pi_{\bar{\rho},\delta}^{n} 
\bar{\rho}^{\otimes n}
\Pi_{\bar{\rho},\delta}^{n} 
\leq
2^{-n[H(B)-\delta]}
\Pi_{\bar{\rho},\delta}^{n}$ of the average output-typical projector.
Inequality \dingfive follows from
$
 \PIMAConepr \
 \Pi_{\bar{\rho},\delta}^{n}  \
  \PIMAConepr 
  \leq 
   \PIMAConepr
   \leq I$.
Finally, inequality \dingsix follows from the bound on 
the rank of the  conditionally typical projector.

		\bigskip
		Combining the bounds from equations 
		\eqref{eq:first-error-chain},
		\eqref{eq:err-one-bound},
		\eqref{eq:err-two-bound}, 
		\eqref{eqn:sum-rate-bound} and the smoothing penalty
		from \eqref{eqn:smoothing-penality-eps}, 
		we get the following bound on the expectation
		of the average error probability:%
		\begin{multline*}
		\mathop{\mathbb{E}}_{X_1^{\prime n},X_2^{\prime n}}\!\!
		\Big\{  \overline{p}_{e}\Big\}
		\leq2\left(  \epsilon+6\sqrt{\epsilon}\right) + 2\sqrt{\epsilon} \\
		+4\bigg[  |\mathcal{M}_1|\ 2^{-n\left[  I\left(
		X_1;B|X_2\right)  -2\delta\right]  }+|\mathcal{M}_2|\ 2^{-n\left[  I\left(  X_2;B|X_1\right) 
		-2\delta\right]  }\\ 
		+ |\mathcal{M}_1||\mathcal{M}_2|\ 2^{-n\left[
		I\left(  X_1X_2;B\right)  -2\delta\right]  }\bigg]  .
		\end{multline*}
		Thus, we can choose the message sets sizes to be 
		$|\mcal{M}_1|  =2^{n\left[  R_{1}-3\delta\right]  }$,
		and 
		$|\mcal{M}_2|   =2^{n\left[  R_{2}-3\delta\right]  }$,
		the expectation of the average error probability vanishes whenever the
		rates $R_{1}$ and $R_{2}$ obey the inequalities:
		\begin{align*}
		R_{1}-\delta   &< I\left(  X_1;B|X_2\right), \\
		R_{2}-\delta   &< I\left(  X_2;B|X_1\right),\\
		R_{1}+R_{2}-4\delta &  <I\left(  X_1X_2;B\right).
		\end{align*}
		If the probability of error of a random code vanishes, then
		there must exist a particular code with vanishing average error probability,
		and given that $\delta>0$ is an arbitrarily small number, the bounds
		in the statement of the theorem follow.
		\end{proof}

		\bigskip

		We now state a corollary regarding the 
		``coded time-sharing''  approach to the MAC problem \cite{HK81,el2010lecture}.
		The main idea is to introduce an auxiliary random variable $Q$
		distributed according to $p_{Q}(q)$ and use the probability 
		distribution $p_{Q}(q)p_{X_1|Q}(x_1|q)p_{X_2|Q}(x_2)$ for the codebook construction.
		First we generate a random sequence $q^{n} \sim \prod_i^n p_{Q}(q_i)$,
		 and then pick the codeword sequences $x_1^{n}$ and $x_2^{n}$
		according to the distributions 
		$p_{X_1^{n}|Q^{n}}(x_1^{n}|q^{n}) \equiv \prod_{i=1}^n p_{X_1|Q}(x_{1i}|q_i)$ 
		and 
		$p_{X_2^{n}|Q^{n}}(x_2^{n}|q^{n})\equiv \prod_{i=1}^n p_{X_2|Q}(x_{2i}|q_i)$.
		%

		\begin{corollary}[Coded time-sharing for QMAC]
		\label{cor:two-sender-QSD} 
		Suppose that the rates $R_{1}$ and $R_{2}$ satisfy
		the following inequalities:%
		\begin{align}
		R_{1} &  \leq I\left(  X_1;B|X_2Q\right)_\theta  ,\\
		R_{2} &  \leq I\left(  X_2;B|X_1Q\right) _\theta,\\
		R_{1}+R_{2} &  \leq I\left(  X_1X_2;B|Q\right)_\theta,
		\end{align}
		where the entropies are with respect to a state $\theta^{QX_1X_2B}$ of the following form:
		\be
		\sum_{x_1,x_2,q}\!\!
		p_{Q}(q)
		p_{X_1|Q}\!\left(  x_1|q\right)
		p_{X_2|Q}\!\left(  x_2|q\right) 
		\ketbra{q}{q}^{Q}
		\!\otimes\!
		\ketbra{x_1}{x_1}^{X_1}
		\!\otimes\!
		\ketbra{x_2}{x_2}^{X_2}
		\!\otimes\!
		\rho_{x_1,x_2}^{B}.
		\ee
		Then 
		there exists a 
		corresponding simultaneous decoding POVM $\left\{  \Lambda_{m_1,m_2}\right\}  $
		such that the expectation of the average probability of error is bounded above
		by $\epsilon$ for all $\epsilon>0$ and sufficiently large $n$.
		\end{corollary}

		The proof of Corollary~\ref{cor:two-sender-QSD} proceeds exactly
		as the proof of Theorem~\ref{thm:sim-dec-two-sender},
		but all the typical projectors are chosen conditionally on $Q^{n}$, 
		and we take the expectation over $Q^n$ in the error analysis.
		%
		The statement of the QMAC capacity rates
		using coded time-sharing will be important for 
		the results in Chapter~\ref{chapter:IC}.

\subsection{Conjecture for three-sender simultaneous decoding}
			
		%
	
	We now state our conjecture regarding the existence of a quantum simultaneous
	decoder for a classical-quantum multiple access channel with three senders.
	We focus on the case of three senders, because this is the form
	that will be required in Section~\ref{sec:HK} for the achievability proof of 
	the quantum Han-Kobayashi achievable rate region \cite{HK81,S11a}.

	\begin{conjecture}[Three-sender quantum simultaneous decoder]\label{conj:sim-dec} \ \\
	\noindent
	Let 
	$(\mcal{X}_1 \times\mcal{X}_2 \times\mcal{X}_3,\ \rho_{x_1,x_2,x_3} , \  \mcal{H}^{B} )$
	be a 
	classical-quantum multiple access channel with three senders.
	Let $p_{X_1}, p_{X_2}$ and $p_{X_3}$ be distributions on the inputs. 
	Define the following random code:
	let $\{X_1^n(m_1)\}_{m_1 \in \{1, \dots, |\mcal{M}_1|\}}$ be an
	independent random codebook distributed according to the
	product distribution $p_{X_1^n}$ and similarly and independently let $\{X_2^n(m_2)\}_{m_2 \in \{1, \dots, |\mcal{M}_2|\}}$ and $\{X_3^n(m_3)\}_{m_3 \in \{1, \dots, |\mcal{M}_3|\}}$ be independent random 
	codebooks distributed according to product distributions $p_{X_2^n}$ and $p_{X_3^n}$. 
	%
	Suppose that the rates of the codebooks obey the following inequalities:%
	\begin{align*}
	R_{1}  &  \leq I\left(  X_1;B|X_2X_3\right)  _{\rho},\\
	R_{2}  &  \leq I\left(  X_2;B|X_1X_3\right)  _{\rho},\\
	R_{3}  &  \leq I\left(  X_3;B|X_1X_2\right)  _{\rho},\\
	R_{1}+R_{2}  &  \leq I\left(  X_1X_2;B|X_3\right)  _{\rho},\\
	R_{1}+R_{3}  &  \leq I\left(  X_1X_3;B|X_2\right)  _{\rho},\\
	R_{2}+R_{3}  &  \leq I\left(  X_2X_3;B|X_1\right)  _{\rho},\\
	R_{1}+R_{2}+R_{3}  &  \leq I\left(  X_1X_2X_3;B\right)  _{\rho},
	\end{align*}
	where the Holevo information quantities are with respect to the following
	classical-quantum state:%
	\begin{align}
	\label{eq:in-out-3mac}
	\rho^{X_1X_2X_3B} & \equiv  \sum_{x_1,x_2,x_3}p_{X_1}\!\left(  x_1\right)  p_{X_2}\!\left(  x_2\right)
	p_{X_3}\!\left(  x_3\right)  \times  \\[-3mm] 
	& \qquad \qquad \quad 
	\left\vert x_1\right\rangle \!\! \left\langle x_1\right\vert
	^{X_1}\otimes\left\vert x_2\right\rangle \!\!\left\langle x_2\right\vert ^{X_2} 
	\otimes\left\vert x_3\right\rangle\!\! \left\langle x_3\right\vert ^{X_3}\otimes
	\rho_{x_1,x_2,x_3}^{B}. \nonumber 
	\end{align}
	Then there exists a simultaneous decoding POVM $\left\{  \Lambda_{m_1,m_2,m_3}\right\}_{m_1,m_2,m_3}$ 
	such that the expectation of the average probability of error is bounded above by $\epsilon$
	for all $\epsilon>0$ and sufficiently large $n$:%
	\[
	\mathbb{E}\!\left\{\!  \frac{1}{|\mcal{M}_1||\mcal{M}_2||\mcal{M}_3|}\sum_{m_1,m_2,m_3}%
	\!\!\!\!\!\! \Tr\!\left[
	  \left(  I-\Lambda_{m_1,m_2,m_3}\right)  
	  \rho_{X_1^{n}\left(m_1 \right)  ,X_2^{n}\left(  m_2\right)  ,X_3^{n}\left(  m_3\right)  }
	  \right]  \!
	  \right\}
	\!\leq\epsilon,
	\]
	where the expectation is with respect to $X_1^{n}$, $X_2^{n}$, and $X_3^{n}$.
	\end{conjecture}

	The importance of this conjecture stems from the fact that 
	it might be broadly useful for ``quantizing'' other results from 
	classical multiuser information theory \cite{FHSSW11}. 
	Indeed, many coding theorems in classical network information
	theory exploit a simultaneous decoding approach (sometimes known as jointly typical decoding)
	\cite{el2010lecture}. Also, Dutil and Hayden have recently put forward a related conjecture
	known as the ``multiparty typicality'' conjecture \cite{D11}, and it is likely that a proof of
	Conjecture~\ref{conj:sim-dec} could aid in producing a proof of the multiparty typicality conjecture
	or vice versa. 
	The notion of a multiparty quantum typicality also appears
	in the problem of universal state merging \cite{bjelakovic2011universal}.
	Recent progress towards the proof of this conjecture can be found in \cite{PranabProof}.

	
	The conjecture naturally extends to $M$-senders, but we have described 
	the three-sender case because this is the form that will be required 
	for the Han-Kobayashi strategy discussed in Section~\ref{sec:HK}.

\section{Rate-splitting}
	\label{sec:mac-rate-splitting}

		Rate-splitting is another approach for achieving the rates of
		the classical multiple access channel capacity region
		\cite{urbanke-rate-splitting} which generalizes readily to the quantum setting 
		using the successive decoding approach in \cite{winter2001capacity}.
		
		\begin{lemma}[Quantum rate-splitting]
			For a given $p=p_{X_1},p_{X_2}$, any rate pair 
			$(R_1,R_2)$ that lies in between the two corner points 
			of the MAC rate region $\alpha_p$ and $\beta_p$
			can be achieved if Sender 2 splits her message $m_2$ into 
			two parts $m_{2u}$ and $m_{2v}$ and encodes them with a \emph{split codebook} 
			and a mixing function
			$\left(\{u^n(m_{2u})\}_{m_{2u}},\{v^n(m_{2v})\}_{m_{2v}},f\right)$.
			The receiver decodes the messages in the order $m_{2u} \to m_1|m_{2u} \to m_{2v}|m_1m_{2u}$
			using successive decoding.
			The total rate for Sender~2 is the sum $R_2 = R_{2u}+R_{2v}$.
		\end{lemma} 
		
		The rate-split codebook consists of two random codebooks
		generated from $p_U$ and $p_V$ and a mixing function such that $f(U,V)=X_2$
		\cite{urbanke-rate-splitting}\footnote{
		Alternately, the mixing can be performed using a
		\emph{switch random variable}, S,
		which is a \emph{shared randomness} resource (denoted $[cc]$)
		between Sender 2 and the receiver \cite{rimoldi2001generalized}.
		}%
		.
		The rate splitting coding strategy for the two sender  quantum multiple access channel
		consists of a successive decoding strategy for the following three channels:
		\begin{align} 
			(U^n, V^n, X^n_1, X^n_2 ) 	&\rightarrow \rho_{X^n_1,X^n_2}^{B^n},	\\
			(U^n, V^n, X^n_1, X^n_2 ) 	&\rightarrow ( U^n, \rho_{X^n_1,X^n_2}^{B^n}) , \\
			(U^n, V^n, X^n_1, X^n_2 ) 	&\rightarrow ( U^n, X_1^n,  \rho_{X^n_1,X^n_2}^{B^n}).
		\end{align}
		The codebooks are constructed with the following rates:
		\begin{align}
		 	R_{2u} 	&=   I(U;B) -\delta,	\\
			R_{1}	&=   I(X_1;B|U) -\delta, \\
			R_{2v}	&=   I(V;B| UX_1 ) -\delta.
		\end{align}
		Observe that the resulting rate pair $(R_1, R_2)=(R_1,R_{2u}+R_{2v})$
		is close to the \emph{dominant facet} of the rate region, which
		is defined as $R_1+R_2=I(X_1X_2|B)$, since:
		\begin{align*}
		  R_1 + R_2 
		    	 & =   R_{2u} + R_1+R_{2v} \\
			 & = I(U;B) -\delta + I(X_1;B|U) -\delta +   I(V;B| UX_1 ) -\delta  \\
			 & = I(X_1X_2|B) - 3\delta.
		\end{align*}
		By varying the choice of the distributions $p_U$ and $p_V$
		and choosing the rates rates of the split-codebooks appropriately,
		we can achieve all the rates of the dominant facet, and therefore
		all the rates of the region.

		\vspace{-1mm}
		The choice of rate split $R_{2u} \leftrightarrow R_{2v}$ 
		depends on the properties of the channel for which we are coding.
		This dependence limits the usefulness of the rate-splitting strategy 
		in situations where there are multiple receivers.
		In general, we cannot choose the rates of the split codebooks such
		that they will be optimal for two receivers.
		Receiver 1 whose output is the system $\rho_{x_1,x_2}^{B_1}$ 
		would want the rates of the codebooks to be set 
		at $(R_{2u},R_{2v}) = (I(U;B_1),  I(V;B_1| UX_1 ) )$,
		whereas Receiver 2, with outputs $\rho_{x_1,x_2}^{B_2}$ 
		would want to set $(R_{2u},R_{2v}) = (I(U;B_2),  I(V;B_2| UX_1 ) )$.
		We will comment on this further in the next chapter. \\

\section{Example of a quantum multiple access channel}

We now show an example of a simple quantum multiple access channel
for which we can compute the capacity region.







\medskip

\begin{example}						\label{ex:2QMAC} 

	Consider the channel
	%
	that takes two binary variables $x_1$ and $x_2$ as inputs and outputs 
	one of the four ``BB84'' states. The following table shows the channel 
	outputs for the different possible inputs.
	\begin{align*}
		\begin{tabular}{c || ll}
		 			& $x_1 = 0$ 		& $x_1 = 1$		\\
		\hline
		\hline
		$x_2 = 0$		& $\ket{0}^B$		& $\ket{+}^B$	\\
		$x_2 = 1$		& $\ket{-}^B$		& $\ket{1}^B$		
		\end{tabular}
	\end{align*}
	The classical-quantum state on which we evaluate information quantities is
	\begin{align*}
	\rho^{X_1X_2B} 
	 	& \equiv
	\sum_{x_1,x_2=0}^{1} 
	 p_{X_1}\!\left( x_1\right)  
	 p_{X_2}\!\left(  x_2\right)  
	\left\vert x_1\right\rangle \left\langle x_1\right\vert ^{X_1}
	\otimes
	\left\vert x_2\right\rangle \left\langle x_2\right\vert ^{X_2}
	\otimes
	\psi_{x_1,x_2}^{B},
	\end{align*}
	where $\psi_{x_1,x_2}^{B}$ is one of 
	$\ketbra{0}{0}$, $\ketbra{1}{1}$, $\ketbra{+}{+}$ or 
	$\ketbra{-}{-}$
	depending on the choice of the input bits $x_1$ and $x_2$. 
	The conditional entropy 
	$H\left(  B|X_1X_2\right)  _{\rho}$ vanishes for this state because the
	state is pure when conditioned on the classical registers $X_1$ and $X_2$.
	We choose
	$p_{X_1}\!\left(  x_1\right)$ and $p_{X_2}\!\left(x_2\right)  $ 
	to be the uniform distribution. 
	This gives the following state on $X_1$, $X_2$, and $B$:
	\begin{align*}
		\rho^{X_1X_2B} 
		&=
		\frac{1}{4}\bigg[
			\ketbra{00}{00} \!\otimes\! \ketbra{0}{0}
			+
			\ketbra{01}{01} \!\otimes\! \ketbra{-}{-}
			 \ifthenelse{\boolean{BOOKFORM}}{ \\[-2mm] & \qquad \qquad  }{}
			+
			\ketbra{10}{10} \!\otimes\! \ketbra{+}{+}			
			+
			\ketbra{11}{11} \!\otimes\! \ketbra{1}{1}	
		\bigg].
	\end{align*}
	From this state we can calculate the reduced 
	density matrix $\rho^{X_2B}=\Tr_{X_1}[ \rho^{X_1X_2B} ]$ by taking the partial
	trace over the $X_1$ system:
	\begin{align*}
		\rho^{X_2B} 
		 & =
		\frac{1}{2}\bigg[
			\ketbra{0}{0}^{X_2} \!\otimes \tfrac{1}{2}\!
				\left( \ketbra{0}{0} + \ketbra{+}{+} \right)^{B}
			 \ifthenelse{\boolean{BOOKFORM}}{ \\[-3mm] & \qquad \qquad  }{}
			+
			\ketbra{1}{1}^{X_2} \!\otimes \tfrac{1}{2}\!
				\left( 	\ketbra{-}{-} +  \ketbra{1}{1}  \right)^{B}
		\bigg],
	\end{align*}
	from which we can determine that
	the conditional entropy $H\left(  B|X_2\right)_{\rho}$ takes its maximum value of
	$H_{2}\!\left(  \cos^{2}\left(  \pi/8\right)  \right)  $ when $p_{X_1}\left(  x_1\right)$
	and $p_{X_2}\left(x_2\right)$ are uniform.

	Taking the partial trace over $X_2$ we obtain the state
	\begin{align*}
		\rho^{X_1B} 
		&=
		\frac{1}{2}\bigg[
			\ketbra{0}{0}^{X_1} \!\otimes \tfrac{1}{2}\! 
				\left( \ketbra{0}{0} + \ketbra{-}{-} \right)^{B}
			+
			\ketbra{1}{1}^{X_1} \!\otimes \tfrac{1}{2}\! 
				\left( 	\ketbra{+}{+} +  \ketbra{1}{1}  \right)^{B}
		\bigg],
	\end{align*}
	from which we can observe that $H(B|X_1)=H_{2}\!\left(  \cos^{2}\left(  \pi/8\right)  \right)$.

	Thus, the capacity region for this channel is:
	\begin{align*}
	R_{1} &  \leq H_{2}\!\left(  \cos^{2}\!\left(  \pi/8\right)  \right) \approx 0.6009, \\
	R_{2} &  \leq H_{2}\!\left(  \cos^{2}\!\left(  \pi/8\right)  \right) \approx 0.6009, \\
	R_{1}+R_{2} &  \leq 1.
	\end{align*}

\begin{figure}[ptb]
\begin{center}
	\includegraphics[width=0.6\textwidth]{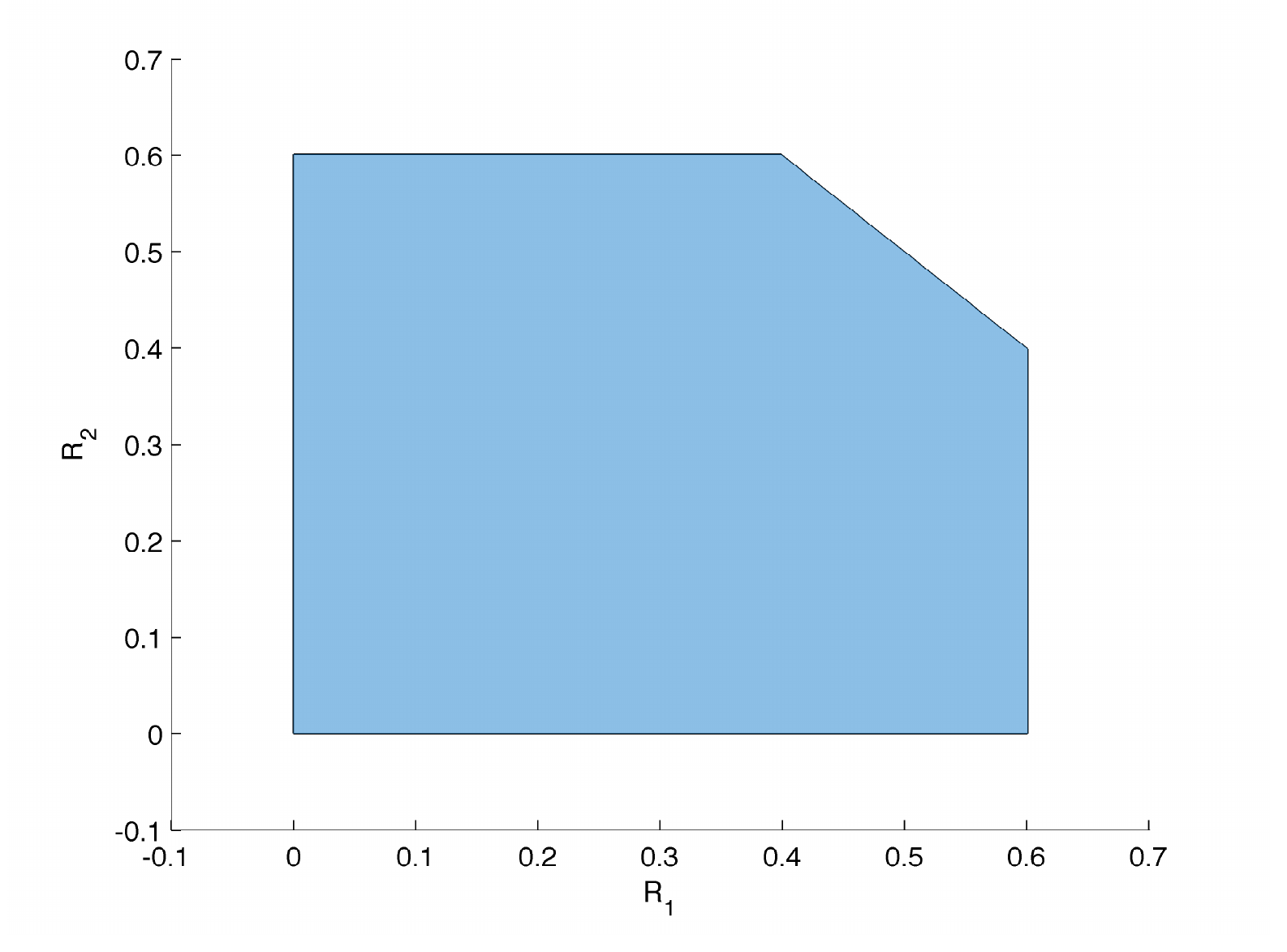}
\end{center}
\caption{The capacity region for the multiple access channel in Example~\ref{ex:2QMAC}. }%
\label{fig:Q2MAC_example}%
\end{figure}

\end{example}


\section{Discussion}

This concludes our exposition on the quantum multiple access channel.
The techniques used in the proof of Theorem~\ref{thm:sim-dec-two-sender}
are the tools 
that will be used throughout the remainder of this thesis.
We review them here for the convenience of the reader and
in order to highlight them in isolation from the technicalities 
in the proof of Theorem~\ref{thm:sim-dec-two-sender}.

The first idea is the POVM construction with layered typical projectors:
\be
    \PIMACavg \
    \PIMACone \
    \PIMAConetwo \
    \PIMACone \ 
    \PIMACavg.
\ee
We call this a \emph{projector sandwich}.
Observe that the more specific projectors are on the inside.
Each of the projectors seems to be necessary in some part of the proof,
and this layering of the projectors ensures that the averaging
can be performed.

The second idea that makes the quantum simultaneous decoder
possible is the \emph{state smoothing trick}, which is to 
perform the error analysis with the unnormalized state:
\be
    \tilde{\rho}_{m_1,m_2} \ \ \equiv \ \  \PIMACtwo \ \rho_{m_1,m_2} \ \PIMACtwo,
    \label{smoothed-rho-bis}
\ee
which is close to the original state, but has the $X_2^n(m_2)$ non-typical
parts of it trimmed off.

The third idea is to use equation \eqref{propertyBBB}
in order to obtain the bound
\begin{align}
    \PIMAConetwo 
    & \ \ \leq \ \ 
        2^{n[H(B|X_1X_2) + \delta] } 
        \rho^B_{m_1,m_2}.
\end{align}
We will call this the \emph{projector trick} \cite{GLM10,S11a,FHSSW11}.

Because of the \emph{ah hoc} nature of the proof of the two-sender simultaneous
decoder, the ideas from the two-sender case cannot be applied to show
that simultaneous decoding of three or more messages is possible.
The techniques used in the proof are sufficiently general for the 
analysis of many problems of quantum network information theory:
quantum interference channels (Chapter~\ref{chapter:IC}),
quantum broadcast channels (Chapter~\ref{chapter:BC}),
and quantum relay channels (Chapter~\ref{chapter:RC}).
%


\chapter{Interference channels}
											\label{chapter:IC}

	In an ideal world, when a sender and a receiver wish to communicate,
	the only obstacle they face is the presence of the background noise.
	Real-world communication scenarios, however,
	often involve multiple senders and multiple receivers
	sending information at the same time and in a shared
	communication medium.
	The receivers have to contend not only
	with the background noise but  also with the interference 
	caused by the other transmissions.
	The interference channel (IC) is a model for the
	effects of this \emph{crosstalk},
	which occurs 
	whenever a communication channel is shared.

	%

	%


\section{Introduction}
												\label{sec:QIC-intro}

	Interference is a big problem for all modern multiuser communication systems.
	In order to avoid interference, techniques such as 
	frequency division multiple access (FDMA)
	and time division multiple access (TDMA)
	can be used to ensure that the senders never transmit at the same time
	and in the same frequency band.
	Another approach is to use code division multiple access (CDMA)
	and allow users to transmit
	at the same time, but their signal power is 
	randomly spread over large sections of the spectrum
	so as to make it look like white noise.
	%

	Rather than treating the interference as noise, a receiver could
	instead decode the interfering signal and then ``subtract'' it
	from the received signal in order to reduce (or even remove) the interference.
	%
	We call this approach \emph{interference cancellation},
	and such strategies are the main theme of this chapter.

	Note that the interference channel problem differs from the 
	multiple access channel problem since in this case
	the  multiple access communication is not \emph{intended}.
	A receiver in the interference channel problem
	is not \emph{required} to decode the interfering 
	messages, but he will be able to achieve better communication rates
	if he does so.
	All the decoding strategies discussed in this 
	chapter use some form of interference cancellation as 
	part of the decoding strategy.

	\subsection{Applications}
	
		The interference channel is an excellent model for many 
		practical communication scenarios where medium contention is an issue.

		\begin{example}[Next-generation WiFi routers]			\label{ex:nex-gen-wifi}
			Consider two 
			neighbours who want to connect to their respective WiFi routers.
			Suppose that the communication happens in the same frequency band (radio channel).
			Suppose further that the neighbours' laptops are located such that 
			they are close to their neighbour's WiFi router and far from their own.
			In such a situation, the \emph{interference} signal
			will be stronger than their own signal.
			%
			%
			Because the interference signal is ``masking'' the intended signal,
			it would be possible for the neighbours to decode it,
			and then \emph{cancel} its effects.
			Thus, we see that it can be to a neighbour's advantage to decode 
			wireless packets which are not intended for him.
			%
			%
			Decoding messages not intended for us 
			can increase the communication rate from the intended sender.
			Note that to implement such a strategy in practice would require 
			a re-engineering of the physical layer of transmission protocols.
		\end{example}

		Interference also plays an important role in
		digital subscriber line (DSL) internet connections.
		The twisted pair copper wires of the telephone system
		were not originally designed to carry high frequency
		and high bandwidth signals,
		and so there is a significant amount of  crosstalk on the wires 
		\emph{en route} to the phone company premises.
		%
		Cross-channel interference is in fact the current
		limiting factor which imposes speed limits on the 
		order of 30Mb/s.
		The next generation VDSL technology
		includes the \texttt{G.vector} standard,
		which is essentially an interference
		 cancellation
		scheme for a vector additive white Gaussian channel \cite{ginis2002vectored,oksman2010tu}.
		The use of the new \texttt{G.vector} VDSL standard for interference 
		mitigation will allow speeds of up to 100Mb/s to the home.

		Interestingly, Shannon's first paper on multiuser communication channels
		was on ``Two-way communication channels'', 
		which can model the simultaneous transmission of information
		in both directions over a phone line \cite{S61}.
		Shannon anticipated the importance of NEXT (near-end crosstalk) and 
		FEXT (far-end crosstalk) to communication systems fifty years in advance.
		Clearly, he was a man ahead of his times!

	\subsection{Review of classical results}


	%
	The seminal papers by Carleial \cite{Carleial78} and 
	Sato \cite{Sato77} defined the interference channel problem 
	in its present form and established many of the fundamental results.
	Finding the capacity region of the general discrete memoryless interference channel (DMIC) is 
	still an open problem, but there are certain special cases where the 
	capacity can be calculated.
	For channels with ``strong'' \cite{sato1981capacity} 
	and ``very strong'' \cite{carleial1975case} interference,
	the full capacity region can be calculated. 
	The capacity-achieving decoding strategies for both of the above
	special cases require the receivers to completely decode 
	the interfering messages.


	For an arbitrary interference channel, it may only be possible to \emph{partially}
	decode the interfering signal.
	The Han-Kobayashi rate region $\mcal{R}_{\text{HK}}$,
	which is achieved by using partial interference cancellation,
	is the best known achievable rate region
	for the general discrete memoryless interference channel~\cite{HK81}. 
	Recently, Chong, Motani and Garg used a different encoding scheme 
	to obtain an achievable rate region, $\mcal{R}_{\text{CMG}}$, which contains
	the Han-Kobayashi rate region \cite{CMG06}.
	Soon afterwards Kramer proposed a compact description of the Han-Kobayashi 
	rate region, $\mcal{R}_{\text{HK}}^c$, which involved fewer constraints \cite{K06}.
	Han and Kobayashi published a comment
	regarding the Fourier-Motzkin elimination procedure
	used to derive the bounds \cite{kobayashi2007further},
	but the question remained whether the above rate regions
	are all equivalent or whether one is strictly larger
	than the others.
	%
	The matter was finally settled by Chong, Motani, Garg and Hesham El Gamal,
	who showed that all three rate regions are in fact equivalent:
	\be
	\mcal{R}_{\text{HK}} \equiv 
	 \mcal{R}_{\text{CMG}} \equiv
	\mcal{R}_{\text{HK}}^c,
	\ee
	when the union is taken over all possible input distributions \cite{CMGE08}.

	There has been comparatively less work on proving outer bounds
	on the capacity region for general discrete memoryless 
	interference channels \cite{Sato77,Carleial83}.
	%


		\ifthenelse{\boolean{BOOKFORM}}
		{
		\vspace{-2mm}
		}
		{}

	\subsection{Quantum interference channels}

		\ifthenelse{\boolean{BOOKFORM}}
		{
		\vspace{-2mm}
		}
		{}

		In this chapter, we apply and extend insights from classical 
		information theory to the study of the quantum interference channel (QIC):
		%
		\be
		(\cX_1 \times \cX_2, \ 
			\mathcal{N}^{X_1X_2 \to B_1B_2}\!\left(x_1,x_2\right) 
			\equiv 
			\rho^{B_1B_2}_{x_1,x_2}, \ 
			\mathcal{H}^{B_1} \otimes \mathcal{H}^{B_2}),
		\ee
		which is a model for a general communication network with two classical inputs and
		a quantum state $\rho^{B_1B_2}_{x_1,x_2}$ as output.
		%
%
		The classical-quantum interference channel 
		can model physical systems such as fibre-optic cables and 
		free space optical communication channels \cite{GSW11bosonic}.

		\begin{wrapfigure}{r}{0pt}
	\begin{tikzpicture}[node distance=1.8cm,>=stealth',bend angle=45,auto]

	  \begin{scope}
		\node [cnode] (ICTx1) [ label=left:Tx1   ]                            {\footnotesize $x_1$};
		\node [cnode] (ICTx2) [ label=left:Tx2, below of=ICTx1,yshift=+5mm]		{\footnotesize $x_2$};
		\node [qnode] (ICRx1) [ label=right:Rx1, right of=ICTx1]	{\tiny $\rho^{B_1}_{x_1,x_2}$}
			edge  [pre]             		node[swap]  	{$\lightning$}	(ICTx1)
			edge  [pre,draw=red]								(ICTx2);
		\node [qnode] (ICRx2) [ label=right:Rx2, right of=ICTx2] {\tiny $\rho^{B_2}_{x_1,x_2}$}
			edge  [pre,draw=red]								(ICTx1)
			edge  [pre]             node		   {$\lightning$}		(ICTx2) ;
	  \end{scope}
	  \begin{pgfonlayer}{background}
	    \filldraw [line width=4mm,join=round,black!10]
	      ([xshift=-3mm,yshift=+2mm]ICTx1.north -| ICTx1.east) rectangle ([xshift=+3mm]ICRx2.south -| ICRx2.west);
	  \end{pgfonlayer}

	\end{tikzpicture}

	\caption{\small 
	The quantum interference channel $\rho^{B_1B_2}_{x_1,x_2}$.
	}
	\label{fig:qic}
		\end{wrapfigure}
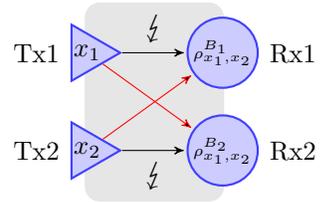


%

	We fully specify a \textit{cc-qq} interference channel by the set of output states it produces 
	$\left\{\rho^{B_1B_2}_{x_1,x_2}\right\}_{x_1\in \cX_1, x_2 \in \cX_2}$
	for each possible combination of inputs. 
	Since Receiver~1 does not have access to the $B_2$ part of the state $\rho_{x_1,x_2}^{B_{1}B_{2}}$,
	we model his state as 
	$\rho_{x_1,x_2}^{B_{1}} = \text{Tr}_{B_2}\!\!\left[  \rho_{x_1,x_2}^{B_{1}B_{2}}  \right]$,
	where $\text{Tr}_{B_2}$ denotes the partial trace over Receiver~2's system. 
	Similarly, the output state for Receiver~2 is given by 	
	$\rho_{x_1,x_2}^{B_{2}} = \text{Tr}_{B_1}\!\!\left[  \rho_{x_1,x_2}^{B_{1}B_{2}}  \right]$.

	A classical interference channel with transition probability function  
	$p(y_1,y_2|x_1,x_2)$ is a special case of 
	the \textit{cc-qq} channel where the output states are of the form
	$\rho^{B_1B_2}_{x_1,x_2} = \sum_{y_1, y_2} p(y_1,y_2|x_1,x_2) 
	\proj{y_1}^{B_1}\!\otimes\!\proj{y_2}^{B_2}$ where
	$\{\ket{y_1}\}$ and $\{\ket{y_2}\}$ are orthonormal bases of $\cH^{B_1}$ and $\cH^{B_2}$.

	\subsection{Information processing task}

		The task of communication over an interference channel can be described as follows.
		Using $n$ independent uses of the
		channel, the objective is for Sender~1 to communicate with Receiver~1 at a rate $R_1$
		  and for Sender~2 to communicate with Receiver~2 at a rate $R_2$.

		If there exists an $(n,R_1,R_2,\epsilon)$-code for the classical-quantum 
		interference channel, then the following conversion is possible:
		\be
			n \cdot \mcal{N}^{X_1X_2 \to B_1B_2}  
			\ \ \overset{ ( 1 -\epsilon)}{\longrightarrow} \ \ 
			nR_1 \cdot  [c^1 \to c^1]
			\ + \
			nR_2 \cdot  [c^2 \to c^2]. \nonumber
		\ee		
		Note that we are only interested in the communication 
		rates from the sender to the intended receiver,
		and we ignore the communication capacity of the
		crosslinks: $[c^1 \to c^2]$ and $[c^2 \to c^1]$.
		%
		%

		%
		More specifically, Sender~1 chooses a message $m_1$ from a message set
		$\mathcal{M}_1\equiv \left\{  1,2,\ldots,|\mathcal{M}_1|\right\} $ where $|\mathcal{M}_1|=2^{nR_{1}}$, 
		and Sender~2 similarly chooses a message $m_2$ from a message set 
		$\mathcal{M}_2 \equiv \left\{  1,2,\ldots,|\mathcal{M}_2|\right\}  $ 
		where $|\mathcal{M}_2|=2^{nR_{2}}$. 
		Senders~1 and 2 encode their messages as codewords $x_1^{n}\!\left(  m_1\right)\in \mathcal{X}_1^n$
		and $x_2^{n}\!\left(  m_2\right) \in \mathcal{X}_2^n$ respectively,
		which are then input to the channel.
		%
		The output of the channel is an $n$-fold tensor product state of the form:
		\be
			\mathcal{N}^{\otimes n}\!\left( x_1^{n}(m_1), x_2^{n}(m_2) \right)
			\equiv			
			\rho_{x_2^{n}\left(  m_1\right),  x_2^{n}\left(  m_2\right)  }^{B_{1}^{n}B_{2}^{n}} \ \  
			\in  \mcal{D}(\mathcal{H}^{B_1^n B_2^n}).
		\ee



\begin{figure}[ptb]
\begin{center}
\includegraphics[width=0.6\textwidth]{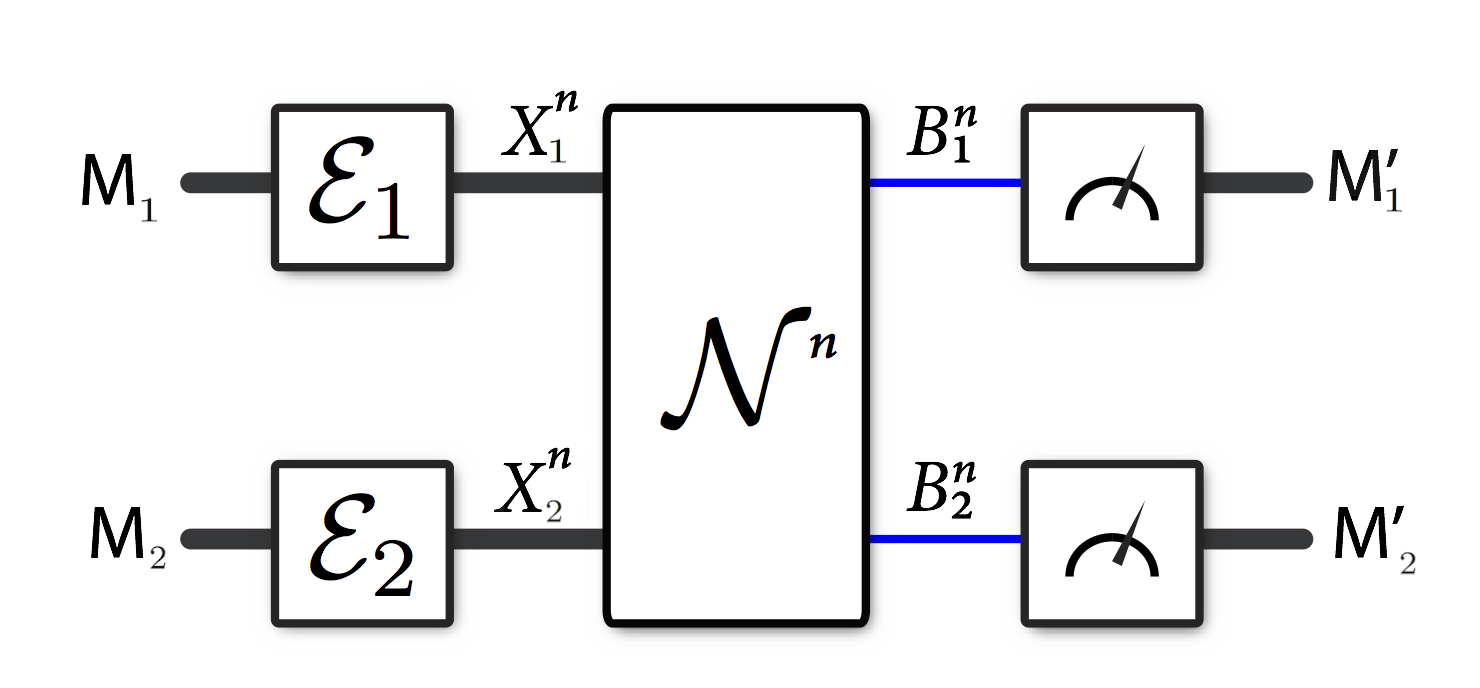}
\end{center}
\caption{
Diagram showing the parts of a classical-quantum interference channel 
\emph{code} for $n$ copies of the channel.
Sender~1 selects a message $m_1$ to transmit (modeled by a random variable $M_1$),
and Sender~2 selects a message $m_2$ to transmit (modeled by $M_2$). Each sender
encodes their message as an $n$-symbol codeword suitable for transmission
over the channel. The receivers each perform a 
quantum measurement in order to decode the messages that their partner sender transmitted.}
\label{fig:info-task}%
\end{figure}

		\noindent
		To decode the message $m_1$ intended for him,
		Receiver~1 performs a positive operator-valued measure (POVM) 
		$\left\{ \Lambda_{m_1}\right\}  _{m_1\in\left\{  1,\ldots,|\mathcal{M}_1|\right\}  }$ on
		the system $B_{1}^n$, the output of which we denote $M^{\prime}_1$. 
		For all $m_1$, $\Lambda_{m_1}$ is a positive semidefinite operator and $\sum_{m_1}\Lambda_{m_1}=I$.
		Receiver~2 similarly performs a POVM 
		$\left\{  \Gamma_{m_2}\right\} _{m_2\in\left\{  1,\ldots ,|\mathcal{M}_2|\right\}  }$
		  on the system  $B_{2}^n$,
		and the random variable associated with this outcome is denoted $M^{\prime}_2$.
				

		An error occurs whenever Receiver~1's measurement outcome is different from the message sent by Sender~1
		($M^{\prime}_1 \neq m_1$) or Receiver~2's measurement outcome is different from the message sent by Sender~2 ($M^{\prime}_2 \neq m_2$).
		%
		The overall probability of error for message pair $(m_1,m_2)$ is
		\begin{align*}
			p_{e}\!\left(  m_1,m_2\right)   
				&  \equiv
				\Pr\left\{  
					(M^{\prime}_1,M^{\prime}_2)\neq(m_1,m_2) 
				\right\} \\
				&  =
				\text{Tr}\!
				\left\{  
					\left(  I-\Lambda_{m_1}\otimes\Gamma_{m_2}\right)
					\rho_{x_2^{n}\left(  m_1\right)  x_2^{n}\left(  m_2\right)  }^{B_{1}^{n}B_{2}^{n}} 
				\right\},
		\end{align*}
		where the measurement operator $\left(  I-\Lambda_{m_1}\otimes\Gamma_{m_2}\right)$ represents
		the complement of the correct decoding outcome.

		    \begin{definition}
			An $(n,R_1,R_2,\epsilon)$ code for the interference channel consists
			of two codebooks 
			$\{x^n_1(m_1)\}_{m_1\in \mathcal{M}_1}$
			and 
			$\{x^n_2(m_2)\}_{m_2\in \mathcal{M}_2}$,
			and two decoding POVMs  
			$\left\{ \Lambda_{m_1}\right\}_{m_1\in \mathcal{M}_1}$ 
			and 
			$\left\{  \Gamma_{m_2}\right\}_{m_2\in \mathcal{M}_2}$,
			such that the average probability of error 
			$\overline{p}_{e}$ is bounded from above by $\epsilon$:%
			\begin{align}
				\overline{p}_{e}  \!
				&  \!\equiv \!
					\frac{1}{|\mathcal{M}_1||\mathcal{M}_2|}\sum_{m_1,m_2}p_{e}\!\left(  m_1,m_2\right) 
				 \leq \epsilon.
			\end{align}

		    \end{definition}

%
		
		A rate pair $\left(  R_{1},R_{2}\right)  $ is \textit{achievable} if there exists
		an $\left(  n,R_{1}-\delta,R_{2}-\delta,\epsilon\right)  $ quantum interference channel code
		for all $\epsilon,\delta>0$ and sufficiently large $n$. The channel's \textit{capacity
		region} 
		  is the closure of the set
		of all achievable rates.

	\subsubsection{Interference channel as two disinterested MAC sub-channels}
	
		The quantum interference channel described by 
		$(\mathcal{X}_1 \times \mathcal{X}_2,  \rho_{x_1,x_2}^{B_1B_2}, 
		\mathcal{H}^{B_1} \otimes \mathcal{H}^{B_2})$
		induces two quantum multiple access (QMAC) sub-channels.
		More specifically \MACone is the channel to Receiver~1 given by
		$(\mathcal{X}_1 \times \mathcal{X}_2,  \rho_{x_1,x_2}^{B_1}= 
		\mathop{\textrm{Tr}}_{B_2}\!\!\left\{\rho_{x_1,x_2}^{B_1B_2}\right\}, \mathcal{H}^{B_1} )$,
		and \MACtwo 
		is the channel to Receiver~2 defined by $(\mathcal{X}_1 \times \mathcal{X}_2,  \rho_{x_1,x_2}^{B_2}, \mathcal{H}^{B_2} )$.
		Thus, one possible coding strategy for the interference channel is to build a codebook for each 
		multiple access channel that is decodable for \emph{both} receivers. 
		For this reason, the coding theorems which we developed for quantum multiple access channels 
		in Chapter~\ref{chapter:MAC} will play an important role in this chapter.
		
		Note however that the IC \emph{problem specification} 
		does not require that Receiver 1 
		be able to decode $m_2$ correctly 
		nor does it specify that Receiver 2 needs to be able to decode
		the message sent by Sender 1 correctly,  
		though most interesting coding strategies involve at 
		least partial decoding of the crosstalk messages.			
		If we take the logical \texttt{and} of the 
		two MAC subtasks, i.e., we require both receivers to be able 
		to decode the messages from both senders,
		then this communication task
		is known as the \emph{compound multiple access channel} problem 
		\cite{Ahlswede1974}.



\subsection{Chapter overview}

	In this chapter, we use the theorems from Chapter~\ref{chapter:MAC}
	for quantum multiple access channels to prove coding 
	theorems for quantum interference channels.

	In Section~\ref{sec:capacity-results}, we prove capacity theorems 
	for two special cases of the interference channel.
	In Theorem~\ref{thm:carleial} we calculate the capacity region of the quantum
	interference channel with ``very strong''  interference (see Definition~\ref{def:very-strong})
	using the successive decoding strategy from
	Theorem~\ref{thm:cqmac-capacity}.
	In Theorem~\ref{thm:strong-int}, we prove the capacity of the channels with ``strong'' 
	interference (see Definition~\ref{def:strong}) using the simultaneous decoding strategy 
	derived in Theorem~\ref{thm:sim-dec-two-sender}.
	%

	In Section~\ref{sec:HK} we discuss the quantum
	Han-Kobayashi coding strategy, 
	where the messages of the senders are split into two parts 
	so that 
	the receivers can perform partial interference cancelation~\cite{HK81}.
	The quantum Han-Kobayashi coding strategy (Theorem~\ref{thm:quantum-HK-region}) requires 
	the use of quantum simultaneous decoding for multiple access channels 
	with three senders which we described in Conjecture~\ref{conj:sim-dec}.
	%
	
	The main contribution of this chapter is to show that the rates of
	the Han-Kobayashi rate region can be achieved without the need for Conjecture~\ref{conj:sim-dec}.
	We will show this in Section~\ref{sec:QCMGvia2MAC},
	where we present an achievability proof for the quantum Chong-Motani-Garg 
	rate region which only uses the two-message simultaneous 
	decoding technique from Theorem~\ref{thm:sim-dec-two-sender}.
	%
	%
	Recall that the Chong-Motani-Garg region is equivalent to the Han-Kobayashi region.
	%
	%
	%
	%
	%

	Note that the achievability of the quantum Chong-Motani-Garg
	rate region was first proved by Sen in \cite{S11a}
	using a different error analysis technique based on an \emph{intersection projector}
	and a careful analysis of the geometric properties of the CMG rate region.
	%
	The alternate proof given in Section~\ref{sec:QCMGvia2MAC}
	uses the simultaneous decoding techniques developed in 
	Section~\ref{sec:mac-simult-decoding} and an interesting geometric 
	argument  by Eren \Sasoglou \ \cite{sasoglu2008successive}.

	The arguments in Section~\ref{sec:QCMGvia2MAC} show that
	we can reduce the decoding requirements from three-message simultaneous
	decoding to two-message simultaneous decoding and still achieve all
	the rates in the Han-Kobayashi rate region.
	Perhaps, it might be 	possible to remove the need for a simultaneous decoder altogether.
	Can the Han-Kobayashi rate region be achieved using only successive decoding?
	In Section~\ref{sec:rate-succ-decoding}, we discuss the difference between 
	interference channel codes (both classical and quantum)
	based on successive decoding and those based on simultaneous decoding.
	In particular, we show that rate-splitting strategies
	based on successive decoding are not a good choice for interference channel codes,
	contrary to what has been claimed elsewhere
	\cite{sasoglu2008successive,yagi2011multi}.

	Finally, we obtain Theorem~\ref{thm:sato-weaker},
	which is a quantum analogue of Sato's outer bound
	for the interference channel.


\section{Capacity results for special cases}
	\label{sec:capacity-results}

	In this section, we consider decoding strategies
	where the receivers decode the messages 
	from both senders.
	We show that this decoding strategy is optimal for the special cases 
	of the interference channel with  ``very strong'' and  	``strong'' interference.
	

%
%

%
%

\subsection{Very strong interference case}
  		\label{sec:rate-succ-decoding-IC}

		%
		%
		If we use a successive decoding strategy at both receivers,
		and calculate the best possible rates that are compatible 
		with both receivers' ability to decode, we obtain an achievable rate region.
		Consider the decoding strategy where Receiver~1 decodes
		in the decode order $m_2 \to m_1|m_2$ and Receiver~2 decodes in the 
		order $m_1 \to m_2|m_1$.
		In this case, we know that the messages are decodable for Receiver~1 
		provided $R_1 \leq I(X_1; B_1|X_2)$ and $R_2 \leq I(X_2; B_1)$. 
		Receiver~2 will be able to decode provided $R_1 \leq I(X_1; B_2)$
		and $R_2 \leq I(X_2; B_2|X_1)$.
		Thus, the rate pair 
		$R_1 \leq  \min\{ I(X_1; B_1|X_2), I(X_1;B_2) \}$, 
		$R_2 \leq  \min\{I(X_2; B_1), I(X_2; B_2|X_1)\}$ 
		is achievable for the interference channel.

	On the other hand, the rate $R_1 \leq I(X_1;B_1|X_2)$ is the 
	optimal rate Receiver~1 could possibly achieve, 
	since this rate corresponds the message $m_1$
	being decoded second \cite{winter2001capacity}.
	%
	Similarly the rate $R_2 \leq I(X_2;B_2|X_1)$  is an upper bound
	on the rates achievable between Sender~2 and Receiver~2.

	We now define a special class of interference channels,
	where the achievable rate region obtained using the above successive 
	decoding strategy matches the outer bound.
	
		\begin{definition}[Very strong interference]	\label{def:very-strong}
		An interference channel with \emph{very strong} interference \cite{carleial1975case},
		is such that for all input distributions $p_{X_1}$ and $p_{X_2}$,
		\begin{align}	
			I\left(  X_{1};B_{1}|X_{2}\right)    &  \leq I\left(  X_{1}%
			;B_{2}\right), \label{eq:VSI-1}\\
			I\left(  X_{2};B_{2}|X_{1}\right)  &  \leq I\left(  X_{2}%
			;B_{1}\right).  \label{eq:VSI-2}%
		\end{align}
		\end{definition}

		The information inequalities in (\ref{eq:VSI-1})-(\ref{eq:VSI-2}) imply that the
		interference is so strong, that it is possible for each receiver to decode the
		other sender's message before decoding the message intended for him. 
		These conditions
		are a generalization of Carleial's conditions for a classical Gaussian
		interference channel~\cite{carleial1975case,el2010lecture}.

		Thus, we can calculate the exact capacity region
		for the special case of the classical-quantum interference channel
		with very strong interference.

	\begin{figure}[htb]
	\begin{center}
		\includegraphics[width=0.6\textwidth]{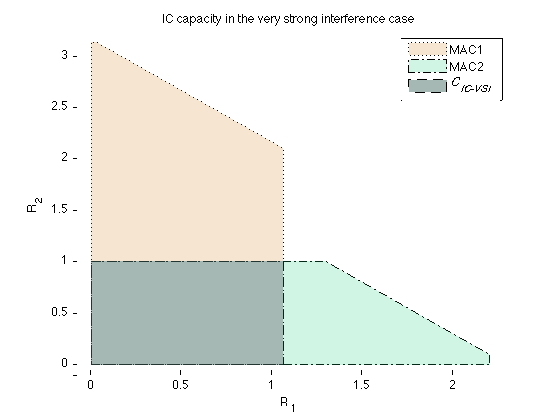}
		\caption{ 
				The capacity region for a \emph{cc-qq} quantum
				interference channel which satisfies the
				``very strong'' interference conditions \eqref{eq:VSI-1}
				and \eqref{eq:VSI-2}.
				The figure also shows the capacity regions
				for the multiple access channel problems 
				associated with each receiver: \MACone and \MACtwo\!\!.
				The capacity region for the IC corresponds to their intersection.
				}
		\label{fig:ICcap-just-very-strong}
	\end{center}
	\end{figure}

		\begin{theorem}[Channels with very strong interference]
		\label{thm:carleial}
		The channel's capacity region is
		given by:
		\be
		\bigcup_{p_Q,p_{X_{1}|Q},p_{X_{2}|Q} }
		\left\{
		(R_1,R_2) \in \mathbb{R}_+^2 
		\left| 
		\begin{array}{rcl}
			R_{1}  &  \leq & I\left(  X_{1};B_{1}|X_{2}Q\right)_{\theta},\\
			R_{2}  &  \leq & I\left(  X_{2};B_{2}|X_{1}Q\right)_{\theta}		
		\end{array}
		\right.
		\right\},
		\ee
		where the 
		mutual information quantities are calculated 
		with respect to a state $\theta^{QX_1X_2B}$
		 of the form:
		\be
		\sum_{x_1,x_2,q}\!\!\!
		p_{Q}(q)
		p_{X_1|Q}\!\left(  x_1|q\right)
		p_{X_2|Q}\!\left(  x_2|q\right) 
		\ketbra{q}{q}^{Q}
		\!\,\otimes\!\,
		\ketbra{x_1}{x_1}^{X_1}
		\!\,\otimes\!\,
		\ketbra{x_2}{x_2}^{X_2}
		\!\,\otimes\!\,
		\rho_{x_1,x_2}^{B}\!.
		\ee
		
		\end{theorem}

		An intuitive interpretation of this result is the seemingly counterintuitive statement that,
		for channels with very strong interference, 
		the capacity is the same as if there were no interference \cite{carleial1975case}.

		\begin{proof}
		We require the receivers to decode the messages for both senders.
		The average probability of error for the interference channel code is given by:
		\begin{align}
		\overline{p}_{e}  \!
			&  \equiv 
			\frac{1}{|\mathcal{M}_1||\mathcal{M}_2|}
				\sum_{m_1,m_2}
				p_e(m_1,m_2) \nonumber \\
			 &  \overset{\mdingone}{=}
				p_e(m_1,m_2)  \nonumber \\
			&  =
				\text{Tr}\!
				\left[  
					\left(  I-\Lambda_{m_1,m_2}^{B_1^n}\otimes\Gamma^{B_2^n}_{m_1,m_2}\right)
					\rho_{x_2^{n}\left(  m_1\right)  x_2^{n}\left(  m_2\right)  }^{B_{1}^{n}B_{2}^{n}} 
				\right],
				\label{eqn:IC-err-fixed-m}
		\end{align}
		where equality \dingone comes from the symmetry of the codebook construction:
		it is sufficient to perform the error analysis for a fixed message pair $(m_1,m_2)$.
		
		\noindent
		Next, we use the following lemma, which is a kind of operator union bound \cite{FQSW}.
		\begin{lemma} 
			\label{lem:operator-union-bound}
			For any operators  $0 \leq P^A, Q^B \leq I$, we have:
			\be
				(I^{AB} - P^A\!\otimes\!Q^B) \leq (I^{A}\! -\! P^{A})\!\otimes\! I^{B} + I^{A}\!\otimes\!(I^{B} - Q^{B}).
				\label{eqn:operator-union-bound}
			\ee
		\end{lemma}%
		\begin{proof}[Proof of Lemma~\ref{lem:operator-union-bound}]
		Starting from $P^A \leq I$ and $Q^B \leq I$, we obtain $0 \leq (I-P^A)$ and $0 \leq (I-Q^B)$ which
		can be combined to obtain:
		\begin{align*}
			0   & \leq (I-P^A)\otimes(I-Q^B) \\
			     & = I^{AB}  - P^A\otimes I^B  -   I^A\otimes Q^B + P^A \otimes Q^B.
		\end{align*}
		The inequality \eqref{eqn:operator-union-bound} follows by moving 
		the term $P^A \otimes Q^B$ to the left hand side and adding a term $I^{AB}$ to both sides.
		\end{proof}

		When applied to the current problem, the inequality \eqref{eqn:operator-union-bound} gives:		
		\[
			\left(  
				I^{B^n_1B^n_2}-\Lambda_{m_1,m_2}^{B^n_1}\otimes\Gamma_{m_1,m_2}^{B^n_2}
			\right)
			\leq
			\left(  I^{B^n_1} -\Lambda_{m_1,m_2}^{B^n_1}\right) \!\otimes\! I^{B^n_2}
			+ I^{B^n_1} \!\otimes\! \left( I^{B^n_2} - \Gamma_{m_1,m_2}^{B^n_2}\right),
		\]
		which in turn allows us to split expression \eqref{eqn:IC-err-fixed-m} into two terms:
		\begin{align*}
		\overline{p}_{e}  \!
			&  =
				\text{Tr}_{B^n_1B^n_2}\!
				\left[  
					\left(  I-\Lambda^{B^n_1}_{m_1,m_2}\otimes\Gamma^{B^n_2}_{m_1,m_2}\right)
					\rho_{x_2^{n}\left(  m_1\right)  x_2^{n}\left(  m_2\right)  }^{B_{1}^{n}B_{2}^{n}} 
				\right], \\
			& \leq 
				\text{Tr}_{B^n_1B^n_2}\!
				\left[
					\left(  I-\Lambda^{B^n_1}_{m_1,m_2} \right)
					\rho_{x_2^{n}\left(  m_1\right)  x_2^{n}\left(  m_2\right)  }^{B_{1}^{n}B_{2}^{n}} 
				\right]
				+
				\text{Tr}_{B^n_1B^n_2}\!
				\left[ 
					\left(  I-\Gamma^{B^n_2}_{m_1,m_2}\right)
					\rho_{x_2^{n}\left(  m_1\right)  x_2^{n}\left(  m_2\right)  }^{B_{1}^{n}B_{2}^{n}} 
				\right] \\
			& = 				
				\text{Tr}_{B^n_1}\!
				\left[ 
					\left(  I-\Lambda^{B^n_1}_{m_1,m_2} \right)
					\rho_{x_2^{n}\left(  m_1\right)  x_2^{n}\left(  m_2\right)  }^{B_{1}^{n}} 
				\right]
				\ \ \ + 
				\text{Tr}_{B^n_2}\!
				\left[  
					\left(  I-\Gamma^{B^n_2}_{m_1,m_2}\right)
					\rho_{x_2^{n}\left(  m_1\right)  x_2^{n}\left(  m_2\right)  }^{B_{2}^{n}} 
				\right]\!.	
		\end{align*}
		
		Each of the above error terms is associated with the probability of error for one of the receivers.
		The decoding problem for each receiver corresponds to a multiple access channel (MAC)
		problem.
		We can use the successive decoding techniques from  
		Theorem~\ref{thm:cqmac-capacity} 
		to show that the decoding 
		at the rates $R_1 \leq I(X_1;B_1|X_2)$, $R_2 \leq I(X_2;B_2|X_1)$ 	will succeed.
		
		Receiver~1 will decode in the order $m_2 \to m_1|m_2$. 
		During the first decoding step Receiver~1 decodes the interfering message $m_2$
		and we know that this is possible because the rate $R_2 \leq I(X_2;B_1)$, which is
		guaranteed by \eqref{eq:VSI-1}.
		In the second step, Receiver~1 now decodes the message from Sender~1 
		given full knowledge of the transmission of Sender~2,
		which is possible any rate $R_1  \leq I(X_1;B_1|X_2)$. 
		Receiver~2 decodes in the order $m_1 \to m_2|m_1$ in order to use
		full interference cancellation and achieve the rate $R_2 \leq I(X_2;B_2|X_1)$.
		
		The  outer bound  follows from the converse part of Theorem~\ref{thm:cqmac-capacity},
		since the individual rates are optimal in the two MAC sub-channels \cite{carleial1975case}.
	\end{proof}



\begin{example}
	\label{ex:theta-SWAP}
	We now consider an example of a \emph{cc-qq} quantum interference channel with two
	classical inputs and two quantum outputs and calculate its capacity region using 
	Theorem~\ref{thm:carleial} \cite{FHSSW11}.
	The \textquotedblleft$\theta$-SWAP\textquotedblright\  
	channel $\mcal{N}:\{0,1\}^2 \to \mathbb{C}^4$ is described by:
	\begin{align}
	00  &  \rightarrow\left\vert 00\right\rangle ^{B_{1}B_{2}}%
	,\label{eq:patrick-example-1}\\
	01  &  \rightarrow\cos\left(  \theta\right)  \left\vert 01\right\rangle
	^{B_{1}B_{2}}+\sin\left(  \theta\right)  \left\vert 10\right\rangle
	^{B_{1}B_{2}},\\
	10  &  \rightarrow-\sin\left(  \theta\right)  \left\vert 01\right\rangle
	^{B_{1}B_{2}}+\cos\left(  \theta\right)  \left\vert 10\right\rangle
	^{B_{1}B_{2}},\\
	11  &  \rightarrow\left\vert 11\right\rangle ^{B_{1}B_{2}}.
	\label{eq:patrick-example-4}%
	\end{align}

	\begin{figure}[ptb]
	\begin{center}
		\includegraphics[width=0.6\textwidth]{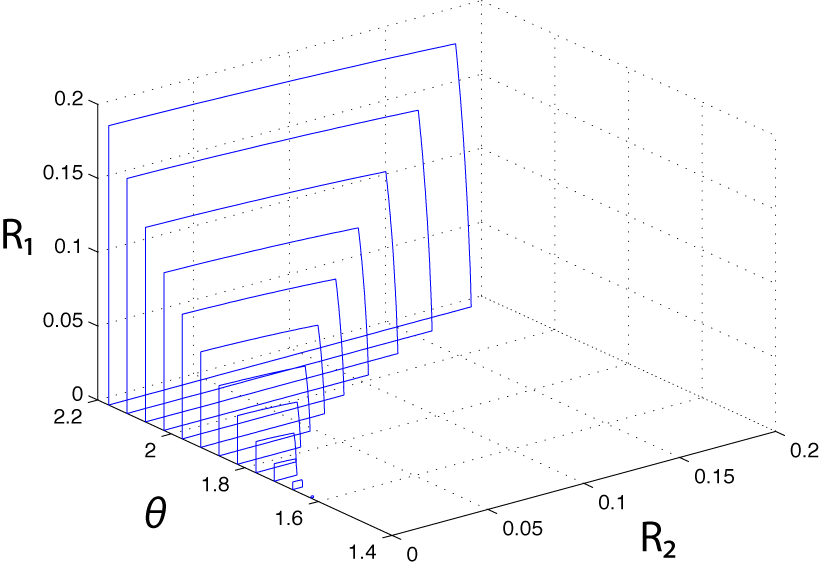}%
		\caption{The capacity region of the \textquotedblleft$\theta$%
		-SWAP\textquotedblright\ interference channel for various values of $\theta$
		such that the channel exhibits \textquotedblleft very strong\textquotedblright%
		\ interference. The capacity region is largest when $\theta$ gets closer to
		2.18, and it vanishes when $\theta=\pi/2$ because the channel becomes a full
		SWAP (at this point, Receiver~$i$ gets no information from Sender~$i$, where
		$i\in\left\{  1,2\right\}  $). }%
		\label{fig:patrick-example}%
	\end{center}
	\end{figure}
	We would like to determine an interval for the parameter $\theta$ for which
	the channel exhibits \textquotedblleft very strong\textquotedblright%
	\ interference. In order to do so, we need to consider classical-quantum
	states of the following form:%
	\be
	\rho^{X_1 X_2 B_{1}B_{2}}\equiv\sum_{x_1,x_2=0}^{1}p_{X_1}\!\left(  x_1\right)  p_{X_2}\!\left(
	x_2\right)  \left\vert x_1\right\rangle \!\!\left\langle x_1\right\vert ^{X_1}%
	\otimes \left\vert x_2\right\rangle\!\!\left\langle x_2\right\vert ^{X_2}\otimes
	\psi_{x_1,x_2}^{B_{1}B_{2}},\label{eq:patrick-example-cq-state}%
	\ee
	where $\psi_{x_1,x_2}^{B_{1}B_{2}}$ is one of the pure output states in
	(\ref{eq:patrick-example-1})-(\ref{eq:patrick-example-4}). We should then check
	whether the conditions in (\ref{eq:VSI-1})-(\ref{eq:VSI-2}) hold for all
	distributions $p_{X_1}\!\left(  x_1\right)  $ and $p_{X_2}\!\left(  x_2\right)  $. We can
	equivalently express these conditions in terms of von Neumann entropies as
	follows:%
	\begin{align*}
	H\!\left(  B_{1}|X_2\right)  _{\rho}-H\!\left(  B_{1}|X_1 X_2\right)  _{\rho} &  \leq
	H\!\left(  B_{2}\right)  _{\rho}-H\!\left(  B_{2}|X_1\right)  _{\rho},\\
	H\!\left(  B_{2}|X_1\right)  _{\rho}-H\!\left(  B_{2}| X_1 X_2\right)  _{\rho} &  \leq
	H\!\left(  B_{1}\right)  _{\rho}-H\!\left(  B_{1}|X_2\right)  _{\rho},
	\end{align*}
	and thus, it suffices to calculate six entropies for states of the form in
	(\ref{eq:patrick-example-cq-state}). After some straightforward calculations,
	we find that:
	\begin{align*}
	\!\!\!\!H\!\left(  B_{1}|X_1 X_2\right)  _{\rho} \!&  \!=H\!\left(  B_{2}|X_1 X_2\right)  _{\rho
	}=\left(  p_{X_1}\!\!\left(  0\right)  p_{X_2}\!\!\left(  1\right)  +p_{X_1}\!\!\left(
	1\right)  p_{X_2}\!\!\left(  0\right)  \right)  H_{2}\!\left(  \cos^{2}\!\!\left(
	\theta\right)  \right)\!, \\ 
	H\!\left(  B_{1}\right)_{\rho}   &  =H_{2}\!\left(  p_{X_1}\!\left(  0\right)  +\left(
	p_{X_1}\!\left(  1\right)  p_{X_2}\!\left(  0\right)  -p_{X_1}\!\left(  0\right)
	p_{X_2}\!\left(  1\right)  \right)  \sin^{2}\left(  \theta\right)  \right)  ,\\
	H\!\left(  B_{2}\right)_{\rho}   &  =H_{2}\!\left(  p_{X_2}\!\left(  0\right)  +\left(
	p_{X_1}\!\left(  0\right)  p_{X_2}\!\left(  1\right)  -p_{X_1}\!\left(  1\right)
	p_{X_2}\!\left(  0\right)  \right)  \sin^{2}\left(  \theta\right)  \right)  ,\\
	H\!\left(  B_{2}|X_1\right)_{\rho}   &  =p_{X_1}\!\left(  0\right)  H_{2}\!\left(  p_{X_2}\!\left(
	1\right)  \cos^{2}\left(  \theta\right)  \right)  +p_{X_1}\!\left(  1\right)
	H_{2}\!\left(  p_{X_2}\!\left(  0\right)  \cos^{2}\left(  \theta\right)  \right)
	,\\
	H\!\left(  B_{1}|X_2\right)_{\rho}   &  =p_{X_2}\!\left(  0\right)  H_{2}\!\left(  p_{X_1}\!\left(
	1\right)  \cos^{2}\left(  \theta\right)  \right)  +p_{X_2}\!\left(  1\right)
	H_{2}\!\left(  p_{X_1}\!\left(  0\right)  \cos^{2}\left(  \theta\right)  \right)  , 
	\end{align*}
	where $H_{2}\!\left(  p\right)  $ is the binary entropy function. We numerically
	checked for particular values of $\theta$ whether the conditions
	(\ref{eq:VSI-1})-(\ref{eq:VSI-2}) hold for all distributions $p_{X_1}\!\left(
	x_1\right)  $ and $p_{X_2}\!\left(  x_2\right)  $, and we found that they hold when
	$\theta\in\left[  0.96,2.18\right]  \cup\left[  4.10,5.32\right]  $ (the latter
	interval in the union is approximately a shift of the first interval by $\pi
	$). The interval $\left[  0.96,2.18\right]  $ contains $\theta=\pi/2$, the
	value of $\theta$ for which the capacity should vanish because the
	transformation is equivalent to a full SWAP (the channel at this point has
	\textquotedblleft too strong\textquotedblright\ interference). We 
	compute the capacity region given in Theorem~\ref{thm:carleial}\ for several
	values of $\theta$ in the interval $\theta\in\left[  \pi/2,2.18\right]  $ (it
	is redundant to evaluate for other intervals because the capacity region is
	symmetric about $\pi/2$ and it is also equivalent for the two $\pi$-shifted
	intervals $\left[  0.96,2.18\right]  $ and $\left[  4.1,5.32\right]  $).
	Figure~\ref{fig:patrick-example}\ plots these capacity regions for several
	values of $\theta$ in the interval $\left[  \pi/2,2.18\right]  $.%
\end{example}


\subsection{Strong interference case}
\label{sec:strong-int}

%

		The 
		simultaneous decoder from 
		Theorem~\ref{thm:sim-dec-two-sender} allows us to calculate the 
		capacity region for the following special case of the quantum interference channel.
		
		\begin{definition}[Strong interference]	\label{def:strong}
		A quantum interference channel with \emph{strong} interference \cite{sato1981capacity,costa1987capacity} 
		is one for which the following conditions hold:
		\begin{align}
			I\left(  X_{1};B_{1}|X_{2}\right)    &  \leq  I\left(  X_{1}	;B_{2}|{X_{2}}\right), \label{eq:SI-1}\\
			I\left(  X_{2};B_{2}|X_{1}\right)  &  \leq    I\left(  X_{2}	;B_{1}|{X_{1}}\right),  \label{eq:SI-2}%
		\end{align}
		for all input distributions $p_{X_1}$ and $p_{X_2}$.
		\end{definition}

		\begin{figure}[htb]
		\begin{center}
			\includegraphics[width=0.6\textwidth]{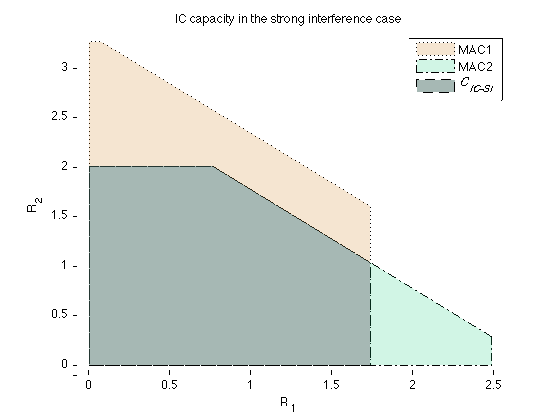}
			\caption{ 
					The capacity region for a \emph{cc-qq} quantum
					interference channel which satisfies the
					``strong'' interference conditions \eqref{eq:SI-1}
					and \eqref{eq:SI-2}.
					The figure also shows the capacity regions
					for the multiple access channel problems 
					associated with each receiver: \MACone and \MACtwo\!\!.
					The capacity region corresponds to the intersection.
					}
			\label{fig:ICcap-just-strong}
		\end{center}
		\end{figure}

		\begin{theorem}[Channels with strong interference]
		\label{thm:strong-int}
		The channel's capacity region is:
		\be
		\bigcup_{p_Q,p_{X_{1}|Q}, \atop p_{X_{2}|Q} } \!\!
		\left\{
		(R_1,R_2) \in \mathbb{R}_+^2 
		\left| \!
		\begin{array}{rcl}
		R_{1}  &  \leq & I\!\left(  X_{1};B_{1}|X_{2}Q\right)_{\theta},\\
		R_{2}  &  \leq & I\!\left(  X_{2};B_{2}|X_{1}Q\right)_{\theta}, \\
		R_1 + R_{2}  &  \leq & 
			\min\left\{ \begin{array}{c}\!\! I\!\left(X_{1}X_{2};B_{1}|Q\right)_\theta \!\! \\
							         \!\! I\!\left(X_{1}X_{2};B_{2}|Q\right)_{\theta} 
						\!\! \end{array}	 \right\} \!\!\!
		\end{array}
		\right.
		\right\},
		\label{ineqs:stong-int}
		\ee
		where the 
		mutual information quantities are calculated 
		with respect to a state $\theta^{QX_1X_2B}$
		 of the form:
		\be
		\!\!\sum_{x_1,x_2,q}\!\!
		p_{Q}(q)
		p_{X_1|Q}\!\left(  x_1|q\right)
		p_{X_2|Q}\!\left(  x_2|q\right) 
		\ketbra{q}{q}^{Q}
		\!\,\otimes\!\,
		\ketbra{x_1}{x_1}^{X_1}
		\!\,\otimes\!\,
		\ketbra{x_2}{x_2}^{X_2}
		\!\,\otimes\!\,
		\rho_{x_1,x_2}^{B}.
		\ee
		
		\end{theorem}
		
		The capacity region is the intersection of the MAC rate regions
		for the two receivers which corresponds to the condition that we
		choose the rates such that each receiver can decode both $m_1$ and $m_2$.
		See Figure~\ref{fig:ICcap-just-strong}.		
		
		\begin{proof}
			The first part of the proof is analogous to the proof of Theorem~\ref{thm:carleial}
			for the interference channel with very strong interference.
			We use Lemma~\ref{lem:operator-union-bound} to split the error analysis
			for the interference channel decoding task into two multiple access channel
			decoding tasks, one for each receiver.
			
			The key difference with Theorem~\ref{thm:carleial} is that 
			for the strong interference case, we require the decoders to use the simultaneous decoding approach 
			from Theorem~\ref{thm:sim-dec-two-sender} and \emph{coded time-sharing} codebooks as described in
			Corollary~\ref{cor:two-sender-QSD}.
			The rate pairs described by the inequalities in  \eqref{ineqs:stong-int} are decodable by
			both receivers. 
			Therefore, these rates are achievable for the interference channel problem.
			
			The proof of the outer bound for Theorem~\ref{thm:strong-int} follows from the outer bound 
			in Theorem~\ref{thm:cqmac-capacity} and an argument similar to the one used in
			the classical case \cite{costa1987capacity} (see also \cite[page 6--13]{el2010lecture}).
		\end{proof}


\section{The quantum Han-Kobayashi rate region}
														\label{sec:HK}


		For general interference channels, the Han-Kobayashi coding strategy 
		gives the best known achievable rate region \cite{HK81} and involves
		\emph{partial} decoding of the interfering signal.
		%
		%
		%
		%
		Instead of using a standard codebook 
		to encode her message $m_1$,
		Sender~1 splits  her message into two parts: a \emph{personal} message $m_{1p}$ 
		and a common message $m_{1c}$.
		%
		Assuming that Receiver~1 is able to decode both of these messages,
		the net rate from Sender~1 to Receiver~1
		will be the sum of the rates of the split codebooks: $R_1 = R_{1p}+R_{1c}$.
		The benefit of using a split codebook\footnote{
			Note that the Han-Kobayashi strategy 
			is also referred to as a \emph{rate-splitting} in the literature.
			In this document we reserve this term \emph{rate-splitting}
			for the use of a split codebook and successive decoding
			as in \cite{urbanke-rate-splitting} and \cite{rimoldi2001generalized}.
		}, is that Receiver~2 can 
		decode Sender~1's common message $m_{1c}$ and 
		achieve a better communication rate
		by using interference cancellation.
		Because only part of the interfering message is used, we call 
		this \emph{partial} interference cancellation.
		Sender~2 will also split her message $m_2$ into two parts: $m_{2p}$ and $m_{2c}$.
		%

	\paragraph{Codebook construction:}



		Consider the auxiliary random variables $Q,U_1,W_1,U_2,W_2$ 
		and the class of Han-Kobayashi probability distributions, 
		$\mathcal{P}_{HK}$, which factorize as
		 $p_{HK}(q,u_1,w_1,x_1,u_2,w_2,x_2)$ $= p(q)$$p(u_1|q)$ $p(w_1|q)p(x_1|u_1,w_1)
		 p(u_2|q)p(w_2|q)p(x_2|u_2,w_2)$,
		where $p(x_1|u_1,w_1)$ and $p(x_2|u_2,w_2)$ are degenerate
		probability distributions that correspond to deterministic functions $f_1$ and $f_2$,
		$f_i\colon\mathcal{U}_i\times\mathcal{W}_i \to \mathcal{X}_i$, which are used to 
		combine the values of $U$ and $W$ to produce a symbol $X$ suitable
		as input to the channel.
		
		We generate the random codebooks in the following manner:
		\begin{itemize}
		\item
		Randomly and independently generate
		a sequence $q^n$ according to $\prod\limits_{i=1}^{n}p_{Q}\!\left(  q_{i}\right)$.

		\item 
		Randomly and independently generate $2^{nR_{1c}}$
		sequences $w_{1}^{n}\!\left(  m_{1c} \right)$, $m_{1c} \in\left[1: 2^{nR_{1c}}\right]$ 
		conditionally on the sequence $q^n$ 
		according to  $\prod\limits_{i=1}^{n}p_{W_{1}|Q}\!\left(  w_{1i}|q_i\right)$.

		\item 
		Randomly and independently generate $2^{nR_{1p}}$
		sequences $u_{1}^{n}\!\left(  m_{1p} \right)$, $m_{1p} \in\left[1: 2^{nR_{1p}}\right]$ 
		conditionally on the sequence $q^n$ 
		according to  $\prod\limits_{i=1}^{n}p_{U_{1}|Q}\!\left(  u_{1i}|q_i\right)$.

		\item 
		Apply the function $f_1$ symbol-wise to the codewords
		$w_{1}^{n}\!\left(  m_{1c} \right)$
		and
		$u_{1}^{n}\!\left(  m_{1p} \right)$
		to obtain the codeword
		$x_{1}^{n}\!\left(  m_{1c}, m_{1p} \right)$.
		
		\item
		We generate the common and personal codebooks for Sender~2 in a similar fashion
		and combine them using $f_2$ to obtain $x_2^n(m_{2c}, m_{2p})$.

		\end{itemize}

	\paragraph{Decoding:}

		When the split codebooks are used for the interference channel,
		we are effectively coding for an interference network with four inputs
		and two outputs.
		We can think of the decoding performed by each of the
		receivers as two multiple access channel (MAC) decoding subproblems.
		We will denote the achievable rate regions for the MAC sub-problems
		as $\MAConeRHK$ and $\MACtwoRHK$.
		The task for Receiver~1 is to decode
		the messages $(m_{1p},m_{1c},m_{2c})$, and thus
		the sub-task $\MAConeRHK$ corresponds to a three-sender multiple
		access channel, the rate region for which is described by
		seven inequalities on the rate triples $(R_{1p},R_{1c},R_{2c})$.
		The decoding task for Receiver~2, $\MACtwoRHK$, 
		is similarly described 
		by seven inequalities on the rates $(R_{1c},R_{2c},R_{2p})$.

		We perform Fourier-Motzkin elimination on the 
		inequalities of the MAC rate regions for the two receivers
		in order to eliminate the variables  
		$R_{1p}$, $R_{1c}$, $R_{2p}$ and $R_{2c}$ and replacing them 
		with the sum variables
		\be
			R_{1}  = R_{1p}+R_{1c}, \qquad 
			R_{2}  = R_{2p}+R_{2c}.
		\ee
		At each step in the Fourier-Motzkin elimination process, we use the 
		information theoretic properties in order to eliminate redundant inequalities. 
		The result is the Han-Kobayashi rate region.

		\begin{theorem}[Quantum Han-Kobayashi rate region]
		\label{thm:quantum-HK-region}
		Consider the region: 
		        \begin{equation}
		        		\nonumber
		        		\HKR^o(\mcal{N}) 
					\equiv 
						\bigcup_{ p_{\text{HK}} \in \mathcal{P}_{\text{HK}} \atop f_1,f_2 } 
						\{ (R_1,R_2) \in \bbR^2 | \text{ \emph{Eqns. (HK1) - (HK9) }} \} 
		        \end{equation}%
		        \vspace{-11mm}
		        \begin{align}
		            R_1 		&\leq		I(U_1W_1;B_1|W_2Q)   							\tag{HK1} \\
		            R_1 		&\leq		I(U_1;B_1|W_1W_2Q) + I(W_1;B_2|U_2W_2Q)	            \tag{HK2} \\
		            R_2 		&\leq		 \qquad \qquad  \qquad \qquad \ \ \ \ I(U_2W_2;B_2|W_1Q)    							\tag{HK3} \\
		            R_2 		&\leq		 I(W_2;B_1|U_1W_1Q) + I(U_2;B_2|W_1W_2Q) 		\tag{HK4} \\
		            R_1 + R_2	&\leq		I(U_1W_1W_2;B_1|Q) +  I(U_2;B_2|W_1W_2Q)		\tag{HK5}\\
		            R_1 + R_2	&\leq		  I(U_1;B_1|W_2W_1Q) +I(U_2W_2W_1;B_2|Q) 	\tag{HK6} \\
		            R_1 + R_2	&\leq		I(U_1W_2;B_1|W_1Q)   + I(U_2W_1;B_2|W_2Q) 		\tag{HK7} \\
		            2R_1 + R_2	&\leq		I(U_1;B_1|W_1W_2Q) +  I(U_2W_1;B_2|W_2Q)
		            					\nonumber \\ 			& 		\quad 
								\!\!+ \! I(U_1W_1W_2;B_1|Q)  	 \tag{HK8} \\
		            R_1 + 2R_2	&\leq		I(U_1W_2;B_1|W_1Q) + I(U_2;B_2|W_2W_1Q) 
		            					\nonumber \\  	&		\qquad \qquad	\qquad \qquad \quad \ \
								+  I(U_2W_2W_1;B_2|Q)		\tag{HK9} 
		        \end{align}
		        where the information theoretic quantities are taken with respect 
			to a state $\theta^{U_{1}U_{2}W_{1}W_{2}B_{1}B_{2}}$ of the form:
			\begin{align}
			& 
				\sum_{q,u_{1},u_{2},\atop w_{1},w_{2}	}
				\hspace*{-2mm}
				p_{Q}\!\left(  q\right)  
				p_{U_{1}|Q}\!\left(  u_{1}|q\right)  
				p_{U_{2}|Q}\!\left(  u_{2}|q\right)  
				p_{W_{1}|Q}\!\left(  w_{1}|q\right)  
				p_{W_{2}|Q}\!\left(  w_{2}|q\right) 
				\left\vert q\right\rangle \!\! \left\langle q\right\vert ^{Q} 	\!\otimes\! \nonumber  \\[-2mm]
			&  \quad \ 
				\!\otimes\!
				\left\vert u_{1}\right\rangle \!\! \left\langle u_{1}\right\vert ^{U_{1}}
				\!\otimes\!
				\left\vert u_{2}\right\rangle \!\! \left\langle u_{2}\right\vert ^{U_{2}} 
				\!\otimes\!	
				\left\vert w_{1}\right\rangle \!\! \left\langle w_{1}\right\vert ^{W_{1}}
				\!\otimes\!
				\left\vert w_{2}\right\rangle \!\! \left\langle w_{2}\right\vert ^{W_{2}}
				\!\otimes\!
				\rho_{f_{1}\left(  u_{1},w_{1}\right)  ,f_{2}\left(  u_{2},w_{2}\right)  }^{B_{1}B_{2}}
				\nonumber 
			\end{align}
			is an achievable rate region provided Conjecture~\ref{conj:sim-dec} holds.			



		\end{theorem}

		\begin{figure}
		[htb]
		\begin{center}
		\includegraphics[
		%
		width=0.5\textwidth
		]%
		{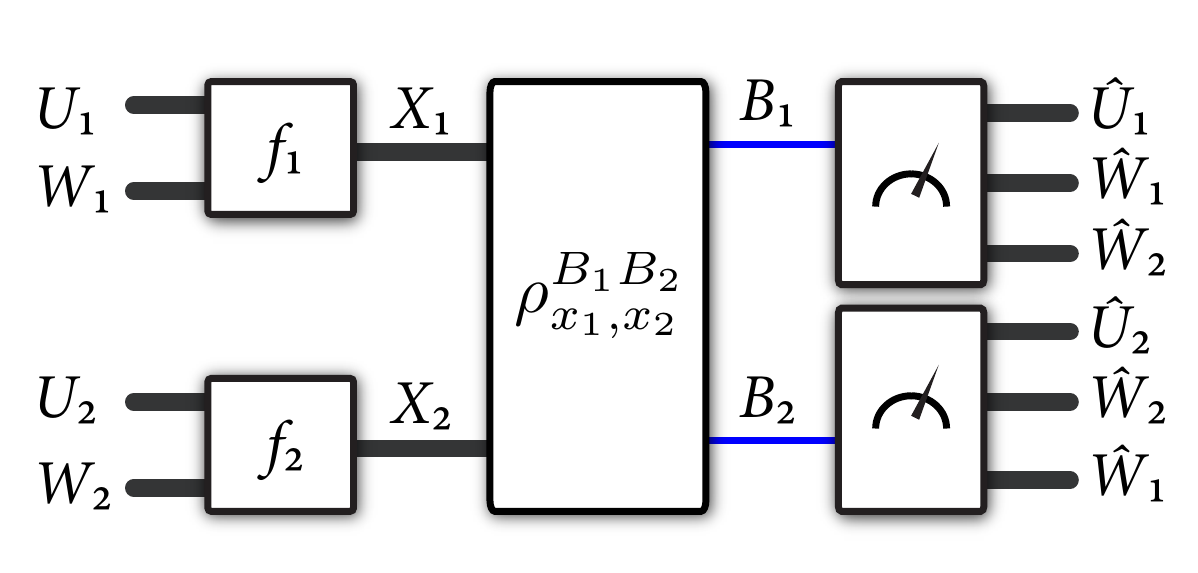}%
		\caption{The random variables used in the 
		Han-Kobayashi coding strategy. Sender~1 selects codewords
		according to a \textquotedblleft personal\textquotedblright\ random variable
		$U_{1}$ and a \textquotedblleft common\textquotedblright\ random variable
		$W_{1}$. She then acts on $U_{1}$ and $W_{1}$ with some deterministic function $f_{1}$ that
		outputs a variable $X_{1}$ which serves as a classical input to the
		interference channel. Sender~2 uses a similar encoding. Receiver~1 performs a
		measurement to decode both variables of Sender~1 and the common random
		variable $W_{2}$ of Sender~2. Receiver~2 acts similarly. The advantage of
		this coding strategy is that it makes use of interference in the channel by
		having each receiver partially decode what the other sender is transmitting.
		Theorem~\ref{thm:quantum-HK-region} gives the rates that are achievable assuming that
		Conjecture~\ref{conj:sim-dec}\ holds.}%
		\label{fig:han-kob-code}%
		\end{center}
		\end{figure}

		Each of the inequalities (HK1)-(HK9) describes 
		some limit imposed on the personal or common rates of the two senders.
		For example, (HK1) corresponds to the maximum rate at which $m_{1p}$ 
		and $m_{1c}$ can be decoded by Receiver~1 \emph{given} that he has
		already decoded $m_{2c}$.
		Other inequalities correspond to mixed bounds, 
		in which one of the terms comes from a constraint on Receiver~1
		and the other from a constraint on Receiver~2.
		An example of this is (HK2) which comes from the bound 
		on Receiver~1's ability to decode $m_{1p}$ (given $m_{1c}$ and $m_{2c}$)
		and a bound from Receiver~2's ability to decode $m_{1c}$
		(given $m_{2c}$ and $m_{2p}$)\footnote{
										%
										Receiver~2 is
										not required to decode the common
										message of Sender~1, but the 
										Han-Kobayashi strategy does require this
										condition despite the fact there could be
										no interference cancellation benefits 
										for doing so, given that Receiver~2 has
										already decoded the messages $m_{2c}$ and $m_{2p}$.
										This should serve as a hint that the Han-Kobayashi
										decoding requirements can be relaxed. 
										We will discuss this further in the next 
										section.  }.
		
		%
		%
		Note that the original description of the rate region given
		by Han and Kobayashi in \cite{HK81} and later in \cite{kobayashi2007further}
		contained two extra inequalities.
		Chong \emph{et al.} showed that these extra inequalities are redundant,
		and so the best description of $\HKR$ involves only nine inequalities as above
		\cite{CMGE08}.

		\begin{proof}		
		The proof is in the same spirit as the original result of Han and Kobayashi \cite{HK81}.
		%
		The first step is to use the Lemma~\ref{lem:operator-union-bound} to obtain:
		\[
		\begin{split}
			\left(  
				I^{B^n_1B^n_2}-\Lambda_{m_{1p},m_{1c},m_{2c}}^{B^n_1}\otimes\Gamma_{m_{1c},m_{2c},m_{2p}}^{B^n_2}
			\right) \\
			\leq 
			\left(  I^{B^n_1} -\Lambda_{m_{1p},m_{1c},m_{2c}}^{B^n_1}\right) \!\otimes & I^{B^n_2}
			+ I^{B^n_1} \!\otimes\! \left( I^{B^n_2} - \Gamma_{m_{1c},m_{2c},m_{2p}}^{B^n_2}\right),
		\end{split}
		\]
		which allows us to bound the error analysis for the interference channel task
		in terms of the error analysis for two MAC sub-channels.
		Our result is conditional on Conjecture~\ref{conj:sim-dec} 
		for the construction of the decoding POVMs for each MAC sub-channel:
		$\left\{\Lambda_{m_{1p},m_{1c},m_{2c}} \right\}$
		for Receiver~1, and $\left\{  \Gamma_{m_{1c},m_{2c},m_{2p}}  \right\} $ 
		for Receiver~2.
		\end{proof}

	\bigskip

	At the very least, observe that Theorem~\ref{thm:quantum-HK-region} depends
	on Conjecture~\ref{conj:sim-dec} for its proof.
	While we do not doubt that the conjecture will ultimately turn out
	to be true, the fact remains that our result is conditional on an
	unproven conjecture, which is somewhat unsatisfactory.
	
	In order to remedy this shortcoming, we searched for other approaches
	which could be used to prove that the rates of the quantum Han-Kobayashi
	rate region are achievable.
	First, we proved that the quantum Han-Kobayashi rate region is achievable
	for a special class of interference channels where the output states commute.
	We also derived an achievable rate region described in terms of 
	min-entropies \cite{Renner05PHD,TomamichelPHD}, 
	which is in general smaller than the Han-Kobayashi 
	rate region. These results are well documented in \cite{FHSSW11}.
	Another approach which we studied 
	is the use of a rate-splitting and successive decoding approach in order to achieve
	the rates of the Han-Kobayashi rate region.
	We attempted to adapt the results of \Sasoglou \ in \cite{sasoglu2008successive},
	which claimed, erroneously, that the rate-splitting strategy can be used
	in order to achieve the Chong-Motani-Garg (CMG) rate region.
	Recall that the Chong-Motani-Garg rate region is equivalent to
	the Han-Kobayashi rate region \cite{CMGE08}.
	In fact, as we will see shortly, the Chong-Motani-Garg approach
	is simply
	a specific coding strategy to carry out the Han-Kobayashi partial interference
	cancellation idea.

	The analysis in \cite{sasoglu2008successive} is in two parts.
	The first part is a geometric argument,
	henceforth referred to as the \emph{\Sasoglou \ argument},
	which shows that there is a many-to-one mapping between the
	rates of the split codebooks $(R_{1p},R_{1c},R_{2c},R_{2p} )$,
	and the resulting rates $(R_1,R_2)$ for the interference channel task.
	In the second part of the analysis, \Sasoglou \ describes a strategy for the use of rate-splitting and successive
	decoding for the common message.
	The common-message codebook for one sender is split so 
	as to accommodate one of the receivers 
	assuming the common-message codebook of the other sender is not split.
	However, if both users split their common-message codebooks, the rates cannot be chosen,
	in general, so as to achieve all the rates of the Chong-Motani-Garg rate region.
	We will comment on this further in Section~\ref{sec:rate-succ-decoding}.
	
	While rate-splitting and successive decoding turned out to be a dead end in our quest
	for the quantum Hon-Kobayashi region,
	the \emph{\Sasoglou \ argument} and the use of two-sender
	simultaneous decoding turns out to be sufficient 
	in order to show the achievability of the quantum Chong-Motani-Garg rate region.
	%
	This will be the subject of Section~\ref{sec:QCMGvia2QMACproof} below.

	%


\section{The quantum Chong-Motani-Garg rate region}
											\label{sec:QCMGvia2MAC}

	The achievability of the quantum Chong-Motani-Garg (CMG) rate region 
	was recently proved by Sen using novel geometric ideas for the 
	``intersection subspace'' of projectors and a ``sequential decoding'' technique~\cite{S11a}.
	In this section we will describe the CMG coding strategy 
	and state Sen's result in Theorem~\ref{thm:QCMG}.
	In Section~\ref{sec:QCMGvia2MAC}, we will provide an alternate proof of this result 
	based on the \Sasoglou \ argument \cite{sasoglu2008successive} 
	and the two-sender simultaneous decoding techniques from Theorem~\ref{thm:sim-dec-two-sender}.
	

	The differences between the Chong-Motani-Garg coding strategy 
	and the Han-Kobayashi coding strategy are: (1) 
	the different way the senders' codebooks are constructed 
	and (2) the relaxed decoding requirements for the two receivers.
	We discuss these next.
	
	\paragraph{Codebook construction:}

		The codebooks  are constructed using 
		the \emph{superposition coding} technique,
		which was originally developed by Cover 
		in the context of the classical broadcast channel~\cite{C72}.
		The idea behind this encoding strategy is to first generate a 
		set of \emph{cloud centers} for each common message
		and then choose the satellite codewords for the personal messages 
		relative to the cloud centers.

		Let $Q,W_1,W_2$ be auxiliary random variables and let
		$\mcal{P}_{\textrm{CMG}}$ be the class of probability density functions
		which factorize as 
		$p_{\textrm{CMG}}(q,w_1,x_1,w_2,x_2)$ $= p(q)$ $p(w_1|q)$ $p(x_1|w_1,q)$ $p(w_2|q)$ $p(x_2|w_2,q)$.
		To construct the codebook we proceed as follows:
		\begin{itemize}
		\item
		First  randomly and independently generate
		a sequence $q^n$ according to $\prod\limits_{i=1}^{n}p_{Q}\!\left(  q_{i}\right)$.

		\item 
		Randomly and independently generate $2^{nR_{1c}}$
		sequences $w_{1}^{n}\!\left(  m_{1c} \right)$, \ $m_{1c}~\in~\left[1: 2^{nR_{1c}}\right]$ 
		conditionally on the sequence $q^n$ 
		according to  $\prod\limits_{i=1}^{n}p_{W_{1}|Q}\!\left(  w_{1i}|q_i\right)$.

		\item
		Next, for each message $m_{1c}$, we randomly and independently generate $2^{nR_{1p}}$  
		conditional codewords  $x_{1}^{n}\!\left(  m_{1p} | m_{1c} \right)$,  $m_{1p} \in\left[1: 2^{nR_{1p}}\right]$,
		$m_{1c} \in [2^{nR_{1c} }]$
		according to  the product conditional probability distribution $\prod\limits_{i=1}^{n}p_{X_{1}|W_1Q}\!\left(  x_{1i}|w_{1i}(m_{1c}),q_i\right)$.

		\item
		We generate the common and personal codebooks for Sender~2 in a similar fashion.
		First generate $\{ w_2^n(m_{2c}) \}$, ${m_{2c} \in [2^{nR_{2c} }] }$ 
		according to $\prod^n p_{W_2|Q}$ 
		and then generate $\{ x_2^n(m_{2p}| m_{2c} ) \}$, $m_{2p} \in [2^{nR_{2p} }]$, $m_{2c} \in [2^{nR_{1c} }] $
		conditionally on $w^n_2(m_{2c})$ according to $\prod^n p_{X_2|W_2Q}$.

		\end{itemize}

	\paragraph{Decoding for the MAC subproblems:}
		%
		%
		The decoding task for each of the receivers is associated with 
		a multiple access channel subproblem.
		We will denote the achievable rate regions for the MAC sub-problems
		for a fixed input distribution $\pCMG \in \mcal{P}_{\textrm{CMG}}$
		as $\MAConeR\cNpCMG$ and $\MACtwoR\cNpCMG$.

		Consider the decoding task for Receiver~1.
		The messages to be decoded are $(m_{1p},m_{1c},m_{2c})$,
		while the effects of the message $m_{2p}$ superimposed on top of
		the codeword for $m_{2c}$ are considered as noise to be averaged over.
		The desired achievable rate region $\MAConeR \cNpCMG$ 
		is defined as follows:
	        \be
	        		\!\!\!\MAConeR\cNpCMG \  \triangleq  
					\!\!\!\!\!\!\!\!\!\!\bigcup_{p(x_1|w_1,q)p(w_1|q) \atop p(x_2|w_2,q)p(w_2|q)p(q)} \!\!\!\!\!\!\!\!\!\!\!\!
					\{ (R_{1p},R_{1c},R_{2c}) \in \mathbb{R}_+^3 | \text{ Eqns (a1)-(d1) below} \} \nonumber 
	        \ee%
		\begin{align}
			R_{1p}  & \leq I\left(  X_{1};B_{1}|W_{1}W_{2}Q\right)  & \triangleq I(a_1), 			\tag{a1} \label{daRegionOne} \\
			R_{1p}+R_{1c}  & \leq I\left(  X_{1};B_{1}|W_{2}Q\right)  &\triangleq I(b_1),  			\tag{b1}\\
			R_{1p}+R_{2c}  & \leq I\left(  X_{1}W_{2};B_{1}|W_{1}Q\right)  &\triangleq I(c_1),  	\tag{c1}\\
			R_{1p}+R_{1c}+R_{2c}  & \leq I\left(  X_{1}W_{2};B_{1}|Q\right) &\triangleq I(d_1) 	\tag{d1}.
		\end{align}
		The mutual information quantities are calculated with respect to the following state:
		\begin{align}
				\!\!\!\! & \sum_{q,w_{1},\atop  x_{1}, w_{2}	} \!\!
				p
				\!\left(  q\right)  
				p
				\!\left(  w_{1}|q\right)  
				p
				\!\left(  x_1| w_{1},q\right)  
				p
				\!\left(  w_{2}|q\right) \times  \\[-4mm]
				& \qquad \qquad 
				\ketbra{q}{q}^{Q}
				\otimes
				\ketbra{w_1}{w_1}^{W_1}
				\otimes
				\ketbra{x_1}{x_1}^{X_1}
				\otimes
				\ketbra{w_2}{w_2}^{W_2}
				\otimes
				\rho_{x_1,w_2}^{B_1},
			\nonumber
		\end{align}
		where 
		\be
			\rho_{x_1,w_2}^{B_1}
			\equiv
			 \sum_{x_2} p\!\left(  x_2| w_{2}\right)   \Tr_{B_2}\!\left[ \rho_{x_1,x_2}^{B_1B_2} \right]
		\ee
		is the effective code state for Receiver~1. 
		It is the average over the random variable $X_2$ (since we treat $m_{2p}$ as noise) 
		and the partial trace over the degrees of freedom associated with Receiver~2.

		%

		The rate region for Receiver~2 is similarly described by:
		\be
	        		\MACtwoR\cNpCMG \ \  \triangleq
					\!\!\!\!\!\bigcup_{p(x_1|w_1,q)p(w_1|q)\atop p(x_2|w_2,q)p(w_2|q)p(q)} \!\!\!\!\!
					\{ (R_{2p},R_{2c},R_{1c}) \in \mathbb{R}_+^3 | \text{ Eqns (a2)-(d2) below} \} 
		\ee
		\begin{align}
			R_{2p}  & \leq I\left(  X_{2};B_{2}|W_{1}W_{2}Q\right)  & \triangleq I(a_2), 			\tag{a2}\\
			R_{2p}+R_{2c}  & \leq I\left(  X_{2};B_{2}|W_{1}Q\right)  & \triangleq I(b_2), 			\tag{b2}\\
			R_{2p}+R_{1c}  & \leq I\left(  X_{2}W_{1};B_{2}|W_{2}Q\right)  & \triangleq I(c_2), 		\tag{c2}\\
			R_{2p}+R_{2c}+R_{1c}  & \leq I\left(  X_{2}W_{1};B_{2}|Q\right) & \triangleq I(d_2),	\tag{d2} 
		\end{align}
		with respect to a code state in which the variable $X_1$ is treated
		as noise and a partial trace over the system $B_1$ is performed.

		Observe that the above MAC rate regions are described only by four inequalities,
		rather than by seven inequalities like the multiple access channel 
		with three senders (cf.~Conjecture~\ref{conj:sim-dec}).
		Two of the rate constraints do not appear because we are using the 
		superposition encoding technique and always decode $m_{1c}$ before $m_{1p}$.
		A third inequality can be dropped if we recognize that Receiver~1 
		is not \emph{really} interested in decoding $m_{2c}$; he is only decoding $m_{2c}$ 
		to serve as side information which will help him decode the messages $m_{1c}$ and 
		$m_{1p}$ intended for him. 
		%
		%
		This is called \emph{relaxed decoding},
		and allows us to drop the constraint associated the
		decoding of $m_{2c}$ after $m_{1c}$ and $m_{1p}$ \cite{CMG06}.
		The relaxed decoding approach cannot be applied directly to the quantum case,
		and so a different decoding strategy is required~\cite{S11a}. 
		We postpone the discussion about the decoding strategies of 
		the receivers until the end of this section.
 

		We are now in a position to describe the Chong-Motani-Garg rate region
		$\CMGR$, which is obtained by combining the constraints from $\MAConeR$ and $\MACtwoR$.
	  	Recall that, for the interference channel problem, we are interested in the 
		\emph{total} rates achievable between each sender and the corresponding receiver.
		For Receiver~1, we have a net rate of $R_1 = R_{1c}+R_{1p}$ and similarly for 
		Receiver~2 we have $R_2 = R_{2c}+R_{2p}$.
		Consider the projection ${\mathbf \Pi}$
		which takes the 4-tuple of rates $(R_{1p},R_{1c}, R_{2c},  R_{2p})$ to the 
		space of net rates $(R_1,R_2)$:
		\be \colvec{R_1 \\ R_2} = 		
			\left[
			\begin{array}{c}
				R_{1p}+R_{1c} \\
				R_{2p}+R_{2c} 
			\end{array}
			\right]
			=
			\underbrace{
			\left[
			\begin{array}{cccc}
			  1 & 1  & 0 & 0 \\
			  0 & 0  & 1 & 1
			\end{array}
			\right]
			}_{\mathbf \Pi}
			\colvec{ \ R_{1p}\  \\ R_{1c} \\ R_{2c} \\ R_{2p} }.
			\label{eqn:FM_elim}
		\ee
		The Chong-Motani-Garg rate region for the interference channel
		is obtained by taking the union over all input distributions of
		the intersection between the two MAC rate regions,
		followed by the projection ${\mathbf \Pi}$ to obtain:
		\be
			\CMGR(\mcal{N})
			\equiv 
			{\mathbf \Pi}
			\left(
				\bigcup_{\pCMG \in \mcal{P}_{\textrm{CMG}} }
				\MAConeR\cNpCMG \ 
				\cap \ 
				\MACtwoR\cNpCMG
			\right).
			\label{eqn:CMG-as-UNION}
		\ee

		
		Equivalently, it is possible to compute the intersection of the 
		two MAC rate regions by performing Fourier-Motzkin elimination
		on the inequalities from equations (a1)-(d1) and (a2)-(d2).
		By taking all possible combinations of the inequalities in the two MAC subproblems,
		we obtain the equivalent set of inequalities in the two dimensional space $(R_1,R_2)$.
		The resulting achievable rate region has the following form:
		\begin{theorem}[Quantum Chong-Motani-Garg rate region \cite{S11a}]
				\label{thm:QCMG}
		The following rate region is achievable for the quantum interference channel:
	        \begin{align}
	        		\CMGR(\mcal{N}) 
				\ & \triangleq   \!\!\!\!\!\!\!
					\bigcup_{p(x_1|w_1,q)p(w_1|q) \atop p(x_2|w_2,q)p(w_2|q)p(q)} \!\!\!\!\!\!\!\!
					\{ (R_1,R_2) \in \mathbb{R}_+^2 | \text{ Eqns. (CMG1)-(CMG9) hold. } \} 
	        \end{align}
	        \vspace{-3mm}
	        \begin{align}
	            R_1 		&\leq		I(X_1;B_1| W_2 Q)   
	            					\tag{CMG1} \\
	            R_1 		&\leq		I(X_1; B_1|W_1W_2Q) + I(X_2W_1; B_2|W_2Q)  
	            					\tag{CMG2} \\
	            R_2 		&\leq		I(X_2;B_2| W_1 Q) 
	            					\tag{CMG3} \\
	            R_1 		&\leq		I(X_1W_2; B_1|W_1Q) + I(X_2; B_2|W_1W_2Q)  
	            					\tag{CMG4} \\
	            R_1 + R_2	&\leq		I(X_1W_2; B_1Q) + I (X_2; B_2|W_1W_2Q) 
	            					\tag{CMG5} \\
	            R_1 + R_2	&\leq		I (X_1; B_1|W_1W_2Q) + I (X_2W_1; B_2Q)  
	            					\tag{CMG6} \\
	            R_1 + R_2	&\leq		I (X_1W_2; B_1|W_1Q) + I(X_2W_1; B_2|W_2Q)  
	            					\tag{CMG7} \\
	            2R_1 + R_2	&\leq		I (X_1W_2; B_1|Q) + I (X_1; B_1|W_1W_2Q) + I(X_2W_1;B_2|W_2Q)  
	            					\tag{CMG8} \\
	            R_1 + 2R_2	&\leq		I (X_2; B_2|W_1W_2Q) + I (X_2W_1; B_2|Q) +  I (X_1W_2; B_1|W_1Q) 
	            					\tag{CMG9}
	        \end{align}
		        where the information theoretic quantities are taken with respect 
			to a state of the form $\theta^{QW_{1}X_1W_{2}X_2B_{1}B_{2}}\equiv$
			\begin{align}
									\label{eq:CMG-code-state}
			& 	\hspace*{-4mm}  
				\sum_{q,w_{1}, w_{2}, \atop x_1,x_2}
				\hspace*{-2mm}
				p_{Q}\!\left(  q\right)  
				p_{W_{1}|Q}\!\left(  w_{1}|q\right)  
				p_{W_{2}|Q}\!\left(  w_{2}|q\right) 
				p_{X_{1}|W_1Q}\!\left(  x_{1}|w_1,q\right)  
				p_{X_{2}|W_2Q}\!\left(  x_{2}|w_2,q\right)  \nonumber  \\[-2mm]
			&  \qquad \quad 
				\left\vert q\right\rangle \! \left\langle q\right\vert ^{Q}
				\!\otimes\!
				\left\vert w_{1}\right\rangle \! \left\langle w_{1}\right\vert ^{W_{1}}
				\!\otimes\!
				\left\vert w_{2}\right\rangle \! \left\langle w_{2}\right\vert ^{W_{2}}
				\!\otimes\!
				\left\vert x_{1}\right\rangle \! \left\langle x_{1}\right\vert ^{X_{1}}
				\!\otimes\!
				\left\vert x_{2}\right\rangle \! \left\langle x_{2}\right\vert ^{X_{2}} 
				\!\otimes\!
				\rho_{x_1,x_2}^{B_{1}B_{2}}. \nonumber
			\end{align}

		\end{theorem}
		
		The classical CMG rate region is known to be equivalent 
		to the Han-Kobayashi rate region \cite{CMGE08}.
		Thus, Sen's achievability proof for the rates of the
		Chong-Motani-Garg rate region is also a proof of the
		quantum Han-Kobayashi rate region.

		\subsubsection{Quantum relaxed decoding}

		Let us consider more closely the relaxed decoding approach 
		that is employed by Receiver~1 in the \emph{classical case}.
		The decoding strategy for Receiver~1 is to use jointly typical decoding
		and search the codebooks 
		$\{w_1^n(m_{1c})\}$, $\{x_1^n(m_{1p}|m_{1c})\}$ and  $\{w_2^n(m_{2c})\}$
		for messages $(m_{1c},m_{1p},\hat{m}_{2c})$ such that
		\[
			\bigg( 
				w_1^n(m_{1c}), \ x_1^n(m_{1p}|m_{1c}), \ w_2^n(\hat{m}_{2c}), \
				Y_1^n
			\bigg)
			\ \in \  \mcal{J}^{(n)}_\delta(W_1,X_1,W_2,Y_1).
		\]
		If such messages are found, the decoder will output $m_1 = (m_{1c},m_{1p})$.
		This decoding is \emph{relaxed} because the above condition
		can be satisfied for some $\hat{m}_{2c}$ which is not necessarily the correct $m_{2c}$ 
		transmitted by Sender~2.
		
		The use of the relaxed decoding strategy allows us to drop the
		following constraint:
		\be
			R_{2c} \leq I(W_2;B_1|W_1X_1),
		\ee
		which corresponds to the message $m_{2c}$ being decoded last,
		given the side information of $m_{1c}$ and $m_{1p}$.

		The relaxed decoding strategy does not generalize readily to the case where a 
		quantum decoding is to be performed \cite{S11a}. 
		For each message triple $(m_{1c},m_{1p},m_{2c})$,
		we could define the measurement $\{ \Lambda_{m_{1c},m_{1p},m_{2c} } \}$,
		but how does one combine the measurement operators 
		$\{ \Lambda_{m_{1c},m_{1p},\hat{m}_{2c} } \}$, $\hat{m}_{2c} \in [2^{nR_{2c} }]$ 
		to form a ``relaxed measurement''?
		Indeed, the usual quantum measurements we use are ones that 
		``ask specific questions'' and for which one outcome is more likely than the others.
		This allows us to use the gentle operator lemma which tells us that
		the our measurement disturbs the system only marginally.

		Sen sidestepped the difficulty of asking a ``vague'' question by
		using two different decoding strategies depending on which 
		rates we want to achieve.
		Receiver~1 will either decode $m_{2c}$ or ignore it altogether. 
		The set of achievable rates for Receiver~1
		$(R_{1p},R_{1c},R_{2c}) \in \mathbb{R}_+^3$
		obtained by Sen is described as follows: 
		\[
	        \begin{array}{rcl}
			R_{2c} 		&\leq& I(W_2; B_1 | X_1 ),		\\
			R_{1p} 		&\leq& I(X_1;B_1|W_1W_2),		\\
			R_{1c}+R_{1p} &\leq& I(X_1;B_1|W_2), 		\\
			R_{2c}+R_{1p} &\leq& I(X_1W_2;B_1|W_1), 		\\
	    \!\!\!R_{1c}+R_{2c}+R_{1p} &\leq& I(X_1W_2;B_1),
		\end{array}
		\ \  \textrm{OR}  
	        \begin{array}{rcl}
			R_{2c} 		&\geq& I(W_2; B_1 | X_1 ),	\\
			R_{1p} 		&\leq& I(X_1:B_1|W_1),		\\
			R_{1c}+R_{1p} &\leq& I(X_1;B_1). 		
		\end{array}
		\]
		Note that the region is not convex. 
		To achieve the rates on the left hand side,
		Sen developed a novel three-sender simultaneous decoding measurement.
		The rates on the right hand side correspond to a \emph{disinterested MAC} problem,
		in which the message $m_{2c}$ will not be decoded.
		%
		%
		After taking the intersection of the achievable rate regions
		for Receiver~1 and Receiver~2 and applying the projection as in \eqref{eqn:CMG-as-UNION},
		Sen obtained a region which is equivalent to the quantum CMG rate region \cite{S11a}.
		
		In the next section we will describe another route to prove the 
		achievability of the quantum CMG 
		rate region.
		We will show that the use of three-sender simultaneous decoding is not necessary.
		Each of the receivers will use one of \emph{three} different decoding strategies
		that only require two-sender simultaneous decoding and,
		in combination, these decoding strategies achieve all the rates
		$(R_1,R_2) \in \CMGR \cNpCMG$.
		

		%

	%
	%
		%
		

\section{Quantum CMG rate region via two-sender simultaneous decoding}
		\label{sec:QCMGvia2QMACproof}

		In the original Han-Kobayashi paper \cite{HK81} and the subsequent Chong-Motani-Garg papers
		\cite{CMG06, CMGE08}, the decoding strategy is to use the three-sender simultaneous decoder.
		This strategy allows for \emph{all possible} interference cancellation scenarios.
		An example of a \emph{specific} decoding strategy would be to decode 
		the interference message $m_{2c}$ simultaneously with $m_{1c}$
		and then decode $m_{1p}$ last using the side information from both common messages.
		We denote this $(m_{1c},m_{2c}) \to m_{1p}|m_{1c}m_{2c}$.
		Another example would be to decode $m_{1p}$ and $m_{2c}$ simultaneously after 
		having decoded $m_{1c}$ first: $m_{1c} \to (m_{1p}|m_{1c}, \ m_{2c}|m_{1c})$.
		%
		%
		Simultaneous decoding is a catchall strategy that subsumes all of the above specific strategies.
		However, as we saw in Chapter~\ref{chapter:MAC}, 
		the existence of a simultaneous decoder for a general three-sender QMAC
		is still an open problem (Conjecture~\ref{conj:sim-dec}).
		%
		It would therefore be desirable to find some specific quantum decoding strategy (or a set of strategies like in \cite{S11a}),
		which can be used to achieve all the rates of the quantum CMG rate region.
		%
		%

		In this section, we will extend the geometrical argument presented in \cite{sasoglu2008successive}, 
		to do away with the need for the simultaneous decoding of three messages.
		We will show that the quantum two-sender simultaneous decoder from 
		Theorem~\ref{thm:sim-dec-two-sender} 
		is sufficient to achieve the quantum Han-Kobayashi rate region.
		

		Observe that in equation \eqref{eqn:FM_elim} only the sum rate $R_{1c}+R_{1p}$ is of importance
		for Receiver~1. The relative values of $R_{1c}$ and $R_{1p}$ are not important ---
		only their sum (provided that all the inequalities (a1)-(a4) are satisfied).
		This fact implies that we are allowed a certain freedom in the way we choose the rates of the codebooks
		for the interference channel.
		We define this freedom more formally as follows:
		
		\medskip 
		\begin{definition}[Rate moving operation]
																	\label{def:rate-moving}
			Let $\pCMG$ 
			be the probability distribution used to construct CMG codebooks.
			Let  $\mcal{C}$ and $\mcal{C}^\prime$ be two codebooks with rates
			\begin{align}
				\mcal{C}  			 & : \ \ (R_{1p}, R_{1c}, R_{2c},R_{2p}) \\
				\mcal{C}^\prime 	& : \ \  (R_{1p}+\delta_1, R_{1c}-\delta_1, R_{2c}-\delta_2,R_{2p}+\delta_2),
			\end{align}
			such that the rates of both codebooks satisfy all the inequalities (a1)-(d1) and (a2)-(d2),
			then they achieve the same rate pair $(R_1,R_2) \in \CMGR\cNpCMG$.
			Such a transformation of rate tuples is called a \emph{rate moving} operation.
			%
		\end{definition}
		
		In words, we say that to achieve the rate pair $(R_1,R_2)$ for the interference channel,
		we are free to \emph{move} the rate points so as to decrease the 
		common rates and increase the personal rates.
		Intuitively, such a transformation is interesting because decreasing the common rates
		will make the decoding task easier overall, since \emph{both} receivers have to decode 
		the common messages whereas only a single receiver needs to decode the personal part.
		The idea for this \emph{rate moving} operation is due to Eren \Sasoglou \ \cite{sasoglu2008successive}.

		\bigskip

		To show the achievability of the Chong-Motani-Garg rate region, $\CMGR(\mcal{N})$,
		it is sufficient to show that we can achieve points on the boundary of the region,
		which we will denote as $\delCMGR(\mcal{N})$.
		In fact, it is sufficient to achieve points on the non-vertical, non-horizontal
		boundary of the rate region which we will denote 
		$\delCMGRpr(\mcal{N}) \subseteq \delCMGR(\mcal{N})$.
		This region is illustrated in Figure~\ref{fig:cmg-heptagon}~(b).
		We refer to the facets that make up the $\delCMGRpr(\mcal{N})$ as the 
		\emph{dominant facets} of the CMG rate region in analogy with
		the dominant facet of the multiple access channel capacity region.
		
		We now state the main theorem of this section:
	
		\begin{theorem}[The dominant facets of the QCMG are achievable]
				\label{thm:quantum-CMG-achievable}
			Any rate pair $(R_1,R_2) \in \delCMGRpr\cNpCMG$ of 
			the non-horizontal, non-vertical facets of the CMG rate region
			is achievable for the quantum  interference channel
			$\left(  \mathcal{X}_{1}\times\mathcal{X}_{2}\ ,\ \mcal{N}(x_1,x_2)\equiv\rho_{x_{1},x_{2}}^{B_1B_2} \ , 
			\ \mathcal{H}^{B_1} \otimes \mathcal{H}^{B_2}\right)$.
		\end{theorem}
		
		As a corollary of the above theorem, we can say that 
		the quantum Chong-Motani-Garg rate region is achievable.
		Any point in the interior of the CMG rate region $\CMGR\cNpCMG$,
		is \emph{dominated} by some point on the 
		non-vertical, non-horizontal dominant facets of the boundary $\delCMGRpr\cNpCMG$.
		Therefore, we can achieve all other points of the rate region 
		 by resource wasting.

		\begin{figure}[hptb]
			\begin{center}
			\subfigure[The CMG achievable rate region. ]{
				\includegraphics[width=0.45\textwidth]{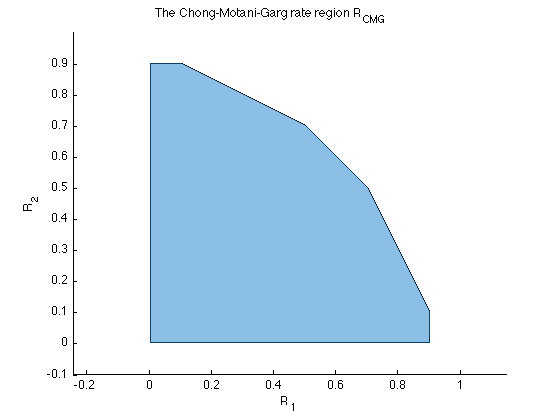}
			}\ \ \ \ %
			\subfigure[The non-horizontal, non-vertical dominant facets of the CMG rate region,
					$\delCMGRpr$, which are achievable by two-sender simultaneous 
					decoding, are shown in bold.]{
				\includegraphics[width=0.45\textwidth]{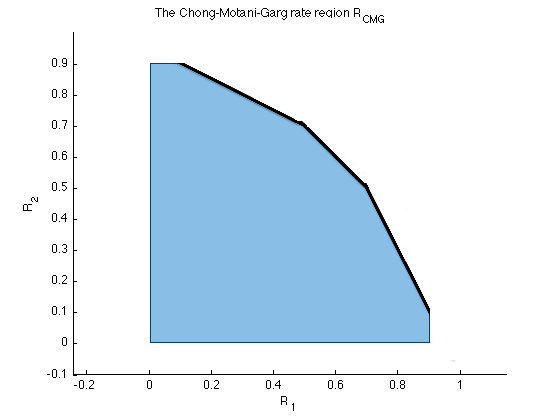}
			}
			\caption{\small 
				The CMG achievable rate region for a given input distribution $p(q)p(w_1,x_1|q)p(w_2,x_2|q)$
				in general has the shape of a heptagon.
				The region is bounded by the two rate positivity conditions
				and each of the other facets corresponds to one of the inequalities (CMG1)-(CMG9).}
			\label{fig:cmg-heptagon}%
			\end{center}
		\end{figure}

		The proof of Theorem~\ref{thm:quantum-CMG-achievable} 
		is somewhat long, so we have broken it up into several lemmas.
		Below we give a brief sketch of the steps involved:
		\begin{itemize}
%
%
					
			\item		In Section~\ref{sec:geometry-of-CMG-region}, we will discuss the geometry 
					of the achievable rate regions $\MAConeR\cNpCMG$
					and $\MACtwoR\cNpCMG$ for the two receivers.
					We state Lemma~\ref{lemma:geometry-of-MAConeR}, which
					identifies the relative placement of the inequalities (a1)-(d1)
					by using the properties of mutual information quantities $I(a_1)$ 
					through $I(d_1)$.					
			
			\item		In Section~\ref{subsec:sasoglu-argument}, we will show that
					any rate pair $(R_1,R_2) \in \delCMGRpr$ 
					can be achieved using codebooks with rates that 
					lie either on the (a) or (c) planes of the MAC rate regions.
					To show this statement, we will prove 
					Lemma~\ref{lem:sasoglou-moving-bd-points}
					which describes a procedure in which we use 
					\emph{rate moving} to transfer any rate point on the (b) or (d) planes 
					to an equivalent rate point on the (a) or (c) planes.
					
			\item		In Section~\ref{subsec:two-simult-decoding-for-ac}, we prove
					that 	the receivers can use two-sender quantum simultaneous decoding
					to achieve any rate on the planes (a) and (c).
					More precisely, there are three possible decode orderings that
					may be used.
					%
					Lemma~\ref{lem:sasoglou-succ-canc-for-ac} shows that the following
					three decoding strategies (shown for Receiver~1)
					are \emph{sufficient} 
					to achieve the rates in the CMG rate region:
					\begin{description}
						\item[Case a:]	 $(m_{1c},m_{2c} ) \to m_{1p}|m_{1c}m_{2c}$,
						\item[Case c:]	 $m_{1c}  \to (m_{1p}|m_{1c}, \ m_{2c}|m_{1c})$,
						\item[Case c':]	 $m_{1c}  \to m_{1p}|m_{1c}$.
					\end{description}

					%
					%
					
		\end{itemize}

	%
	
	%
	%
	%


	\subsection{Geometry of the CMG rate region}			\label{sec:geometry-of-CMG-region}

		For a general input distribution $\pCMG$, 
		the CMG rate region $\CMGR\cNpCMG$ and the two MAC subproblem rate regions could 
		take on different shapes
		depending on the relative values of the mutual information quantities 
		$I(a_1)$, $I(b_1)$, $I(c_1)$, $I(d_1)$, $I(a_2)$, $I(b_2)$, $I(c_2)$ and $I(d_2)$.
		
		In his paper \cite{sasoglu2008successive}, \Sasoglou \  develops a powerful intuition 
		for dealing with the polyhedra that describe their boundaries 
		$\delCMGR\cNpCMG$, $\delMAConeR\cNpCMG$ and $\delMACtwoR\cNpCMG$.
		Define the two-dimensional facets $a_1,b_1,c_1, d_1$ 
		which make up the region boundary.
		Each facet is a subset of the plane in $\mathbb{R}^3$ associated with 
		the equality condition of inequalities (a1), (b1), (c1) and (d1),
		which correspond to the rate constraints of Receiver 1.
		The boundary of the region $\MAConeR\cNpCMG$ can be written as 
		$\delMAConeR\cNpCMG = a_1 \cup b_1 \cup c_1 \cup d_1$.


		%
		%
	        We can visualize the three dimensional rate region $\MAConeR\cNpCMG$  
	        as in Figure \ref{fig:shash1} below.
		
		\begin{figure}[hptb]
			\begin{center}
			\includegraphics[width=0.6\textwidth]{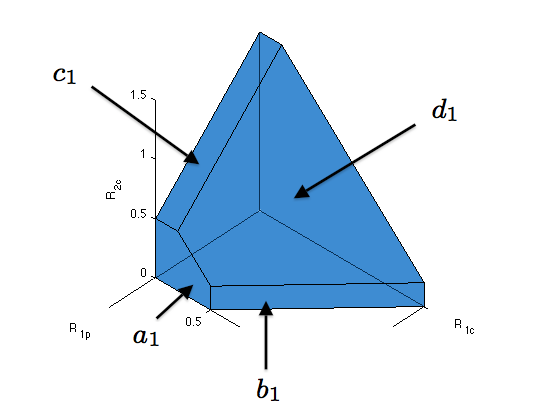}
			\caption{The achievable rate region $\MAConeR\cNpCMG$ and its bounding facets
			$a_1,b_1,c_1$, and $d_1$.
			Each surface is associated with the equality condition in one of the equations
			(a1), (b1), (c1) and (d1) from page \pageref{daRegionOne}.  }%
			\label{fig:shash1}%
			\end{center}
		\end{figure}


		This shape of the rate region is governed by 
		the information-theoretic quantities on the right hand side of equations (a1) through (d1).
		The following relations establish the geometry of the rate-region $\MAConeR\cNpCMG$
		which hold for any input distribution.

		\begin{lemma}[Geometry of  $\MAConeR\cNpCMG$]  \label{lemma:geometry-of-MAConeR}
			The information-theoretic quantities from equations (a1), (b1), (c1) and (d1)
			satisfy the following inequalities:
			\bea
				I(a_1) \leq		& I(b_1)	&	  \leq I(d_1),		\label{polymatr1} \\
				I(a_1) \leq	 	&I(c_1)& 	  \leq I(d_1),		\label{polymatr2} \\
				I(a_1)+I(d_1)  \ \ \  	&\leq& 	\ \ \ I(b_1) + I(c_1).		 	\label{polymatr3}
			\eea
		\end{lemma}

		Geometrically $I(a_1) \leq I(b_1)$ indicates that the plane containing $b_1$ 
		intersects the plane containing $a_1$ in the positive octant.
		Similarly $I(b_1) \leq I(d_1)$ indicates that  the plane containing $d_1$ intersects the plane 
		containing $b_1$ inside $\mathbb{R}_+^3$.
		Equation \eqref{polymatr2} dictates that the plane containing $c_1$ intersects the plane
		containing $a_1$ and that the plane containing $d_1$ intersects the plane of $c_1$.
		Finally, equation \eqref{polymatr3} states that $I(a_1)+I(d_1)  \leq I(b_1) + I(c_1)$, which means that
		the rate constraint on the sum $2R_{1p}+R_{1c}+R_{2c}$ obtained by adding (a1) and (d1) is tighter
		than the rate constraint  obtained by adding (b1) and (c1).
		If we define the sets $A=\{1p,1c \}$ and $B=\{ 1p, 2c \}$ and $\rho(X)$ to be the information-theoretic 
		quantities of the right hand side, then equation 
		\eqref{polymatr3} 
		has a super-modular polymatriod structure $ \rho(A \cap B) + \rho(A \cup B) \leq \rho(A) + \rho(B)$.
		The proof of Lemma~\ref{lemma:geometry-of-MAConeR}
		is given in Appendix~\ref{apdx:CMG-MAC-geom}.

	\subsection{\Sasoglou \ argument}		\label{subsec:sasoglu-argument}

		Let the rate pair $(R_1,R_2) \in \delCMGRpr\cNpCMG$ be part of the non-horizontal,
		non-vertical boundary of the two dimensional rate region $\CMGR\cNpCMG$.
		This rate pair is associated (non-uniquely) to a pair of points 
		$P_1=(R_{1p},R_{1c},R_{2c})$ and $P_2=(R_{2p},R_{2c},R_{1c})$
		on the boundaries of  the respective regions 
		$\MAConeR\cNpCMG$ and $\MACtwoR\cNpCMG$.

		\bigskip
		%
		\begin{claim} 
			\label{lem:opt-rates-on-bdrs}
		If the two-dimensional rate pair $(R_1,R_2) \in \delCMGRpr\cNpCMG$ 
		is the projection of the 
		points $P_1=(R_{1p},R_{1c},R_{2c})$ and $P_2=(R_{2p},R_{2c},R_{1c})$
		via the mapping in \eqref{eqn:FM_elim},
		%
		then $P_1 \in \delMAConeR\cNpCMG$ and $P_2 \in \delMACtwoR\cNpCMG$.
		\end{claim}

		Suppose that this were not the case ---
		that is, we assume that at least one of the points, $P_i$ is not 
		on the boundary of its region $\partial\mathcal{R}^i_{\textrm{CMG}}\cNpCMG$.
		Suppose, for a contradiction, 
		that $P_i$ is in the interior of $\mathcal{R}^i_{\textrm{CMG}}\cNpCMG$,
		then 	there must exist a ball of achievable rates of size $\delta$ around $P_i$.
		This means that we would
		be able to increase the private rate to $R'_{ip}=R_{ip}+\delta$ for some 
		$\delta > 0$.
		The resulting point $P^\prime_i$ will be still be achievable so long as
		we stay within the region $\mathcal{R}^i_{\textrm{CMG}}\cNpCMG$.
		However, such a $\delta$ displacement leads to an increase the sum rate 
		$R^\prime_i = R^\prime_{ip} + R^\prime_{ic}  =  R_{ip}+\delta + R_{ic}= R_i+\delta$.
		This 	contradicts our initial assumption that $(R_1,R_2) \in \delCMGRpr\cNpCMG$.
		%
		Therefore, Claim~\ref{lem:opt-rates-on-bdrs} must be true,
		and this means that it is sufficient to show how to achieve 
		all the rates on the boundary of the rate regions
		$\partial\mathcal{R}^i_{\textrm{CMG}}\cNpCMG = a_i \cup b_i \cup c_i \cup d_i$.

		\bigskip

		A priori, we have to consider all possible starting combinations of the points 
		$P_i \in a_i \cup b_i \cup c_i \cup d_i$.
		However, using the \emph{rate moving} operation (Definition~\ref{def:rate-moving}), we can move any
		point in $b_i \cup d_i \setminus a_i \cup c_i$ to an equivalent  point in $a_i \cup c_i$
		as illustrated in Figure~\ref{fig:shash2}.
		\begin{figure}[hptb]
			\begin{center}
			\includegraphics[width=0.5\textwidth]{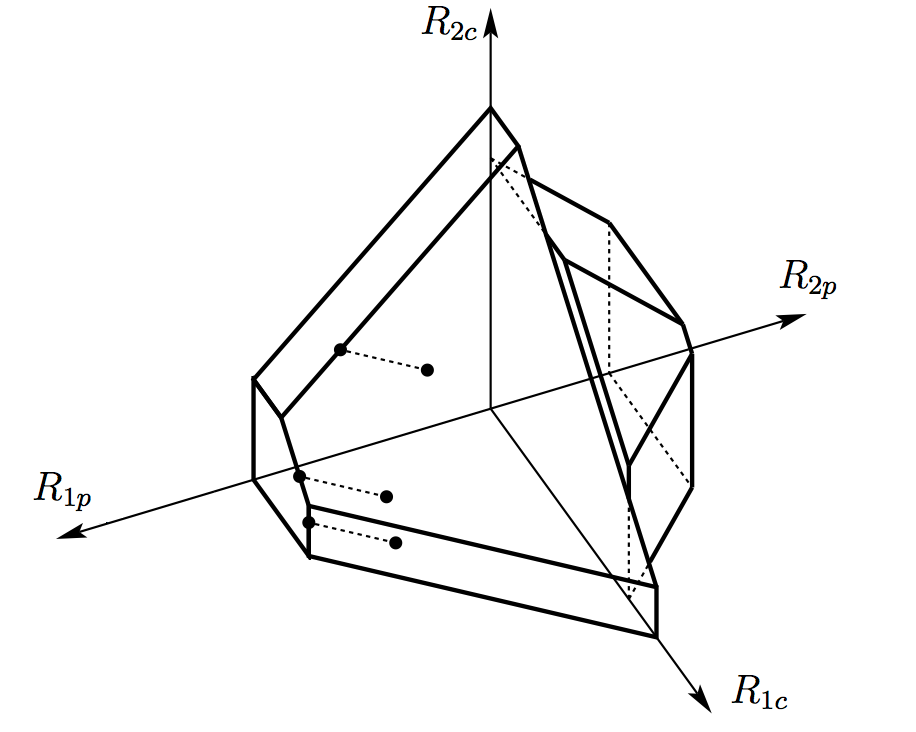}
			\caption{Moving points on the $b_1$ and $d_1$ facets to equivalent
					points on $a_1$ and $c_1$.}
			\label{fig:shash2}%
			\end{center}
		\end{figure}

		\begin{lemma}[Moving points \cite{sasoglu2008successive}]		
										\label{lem:sasoglou-moving-bd-points}
			Any point $P_i$ that lies on one of the planes $b_i \cup d_i \setminus a_i \cup c_i$
			can be converted to a different point $P'_i$ on one of the  
			planes $a_i \cup c_i$,  while leaving the sum rate $(R_1,R_2)$ unchanged.
		\end{lemma}

		In order to be precise, we have to study the effects of the rate 
		moving operation on both points $P_1$ and $P_2$ simultaneously.
		This is because the \emph{same} rates $R_{1c}$ and $R_{2c}$ appear in the common 
		coordinates of both $P_1$ and $P_2$.
		The reasoning behind the proof of Lemma~\ref{lem:sasoglou-moving-bd-points} 
		is reminiscent of the argument used to prove Claim~\ref{lem:opt-rates-on-bdrs}.
		The details are given in Appendix~\ref{sec:sasoglou-moving-bd-points}.

		Lemma~\ref{lem:sasoglou-moving-bd-points} is important because in the 
		next section we will show
		how to achieve the rates in the facets $a_i$ and $c_i$ using 
		two-sender quantum simultaneous decoding.
		This means that we can construct a decoder that achieves all the rates
		for the quantum Chong-Motani-Garg rate region without the need for a three sender
		simultaneous decoder from Conjecture~\ref{conj:sim-dec}.

%

	\subsection{Two-message simultaneous decoding is sufficient for the rates of the facets $a_i$ and $c_i$}		
											\label{subsec:two-simult-decoding-for-ac}
	
		In this section we show how to achieve the rates on the $a_1$ and 
		$c_1$ facets using only two-sender simultaneous decoding.


		\begin{lemma}[Two-simultaneous decoding for $a$ and $c$ planes] \label{lem:sasoglou-succ-canc-for-ac}
            		Fix an input distribution $\pCMG\in \mcal{P}_{\textrm{CMG}}$
			and let the rate pair $(R_1,R_2) \in \delCMGR\cNpCMG $ come from the rate triples 
			$P_1=(R_{1p},R_{1c},R_{2c}) \in \delMAConeR\cNpCMG $ and $P_2=(R_{2p},R_{2c},R_{1c})\in \delMACtwoR\cNpCMG $
			such that
			\be
				(P_1, P_2) \in a_1 \cup c_1 \times a_2 \cup c_2.
			\ee
			Then the rate $(R_1,R_2)$ is achievable for the QIC using 
			two-sender quantum simultaneous decoding.
		\end{lemma}

		\begin{proof}
		Our analysis is similar to \cite{sasoglu2008successive},
		but we are not going to use a rate-splitting strategy.

		{\bf Achieving points in a:} 
			Consider a point $P_1 \in a_1$, which implies
			\begin{align}
				R_{1p}  & = I\left(  X_{1};B_{1}|W_{1}W_{2}Q\right) , \label{eq:R1ptight} \\
				R_{1p}+R_{1c}  & \leq I\left(  X_{1};B_{1}|W_{2}Q\right)  , \\
				R_{1p}+R_{2c}  & \leq I\left(  X_{1}W_{2};B_{1}|W_{1}Q\right) ,  \\
				R_{1p}+R_{1c}+R_{2c}  & \leq I\left(  X_{1}W_{2};B_{1}|Q\right).
			\end{align}
			
			We can subtract equation \eqref{eq:R1ptight} from the inequalities below it to obtain
			a new set of inequalities
			\begin{align}
				R_{1p}  & = I\left(  X_{1};B_{1}|W_{1}W_{2}Q\right),   \\
				R_{1c}  & \leq  I\left(  W_{1};B_{1}|W_{2}Q\right) 				\label{eqn:caseAone}
							= I\left(  X_{1};B_{1}|W_{2}Q\right)-  I\left(  X_{1};B_{1}|W_{1}W_{2}Q\right)  , \\
				R_{2c}  & \leq I\left(  W_{2};B_{1}|W_{1}Q\right)  				\label{eqn:caseAtwo}
							= I\left(  X_{1}W_{2};B_{1}|W_{1}Q\right)-I\left(  X_{1};B_{1}|W_{1}W_{2}Q\right) ,  \\
				R_{1c}+R_{2c}  & \leq  I\left(  W_{1} W_{2};B_{1}|Q\right) 			\label{eqn:caseAthree}
							= I\left(  X_{1}W_{2};B_{1}|Q\right)- I\left(  X_{1};B_{1}|W_{1}W_{2}Q\right).
			\end{align}			
			
			Looking at equations \eqref{eqn:caseAone}-\eqref{eqn:caseAthree} we see that 
			the rates $(R_{1c},R_{2c})$  have the form of a MAC rate region 
			with inputs $W_1 \sim p(w_1|q), W_2 \sim p(w_2|q)$ and output $B_1$.
			We will perform the decoding in the following order at Receiver~1: $(W_{1},W_{2}) \to X_1|W_1W_2$.

						Consider the quantum channel 
			\be
				w_1, w_2 	\to \rho_{w_1,w_2}^{B_1},			\label{chan:cloud-centers}
			\ee
			where $\rho_{w_1,w_2}^{B_1}$ is defined as the average output
			state assuming superposition encoding of the random variables
			$x_1$ and $x_2$ will be performed:
			\be
				\rho_{w_1,w_2}^{B_1}  \equiv \sum_{x_1} \sum_{x_2} p(x_1|w_1) p(x_2|w_2) \rho_{x_1,x_2}^{B_1}.
			\ee
			
			The decoding strategy for Receiver~1 when the rates are on the facet $a_1$
			correspond to the use of the two-message simultaneous decoder 
			(Theorem~\ref{thm:sim-dec-two-sender})
			on the channel shown in \eqref{chan:cloud-centers}.

			After the common parts have been decoded, 
			Receiver~1 will use a conditional HSW decoder to decode the message encoded in $X_1$.

		{\bf Achieving points in c:} 
			Consider a point $P_1 \in c_1$, which implies that the constraint on the $R_{1p}+R_{2c}$ 
			inequality is tight.
			\begin{align}
				R_{1p}  & \leq I\left(  X_{1};B_{1}|W_{1}W_{2}Q\right) , \label{eqn:ctighta} \\
				R_{1p}+R_{1c}  & \leq I\left(  X_{1};B_{1}|W_{2}Q\right)  , \label{eqn:ctightb} \\
				R_{1p}+R_{2c}  & = I\left(  X_{1}W_{2};B_{1}|W_{1}Q\right) , \label{eqn:ctightc}  \\
				R_{1p}+R_{1c}+R_{2c}  & \leq I\left(  X_{1}W_{2};B_{1}|Q\right). \label{eqn:ctightd} \\
			\intertext{ 	If we subtract \eqref{eqn:ctightc} from \eqref{eqn:ctightd} 
					we obtain the following equivalent set of inequalities. }
				R_{1p}  			& \leq I\left(  X_{1};B_{1}|W_{1}W_{2}Q\right) , \label{eqn:ctighta2} \\
				R_{1p}+R_{1c}  	& \leq I\left(  X_{1};B_{1}|W_{2}Q\right)  , \label{eqn:ctightb2} \\
				R_{1p}+R_{2c}  	& = I\left(  X_{1}W_{2};B_{1}|W_{1}Q\right) , \label{eqn:ctightc2}  \\
				R_{1c}  			& \leq I\left(W_{1};B_{1}|Q\right) 
								 =  I\left(  X_{1}W_{2};B_{1}|Q\right)-I\left(  X_{1}W_{2};B_{1}|W_{1}Q\right)  
								 \label{eqn:ctightd2}
			\end{align}

			The constraint on the sum rate $R_{1p}+R_{1c}$ imposed by equation \eqref{eqn:ctightb2} 
			is less tight than the sum rate constraint obtained by adding equations 
			\eqref{eqn:ctighta2} and \eqref{eqn:ctightd2},				
			therefore we will drop equation \eqref{eqn:ctightb2}  from the remainder of the argument.
			The accuracy of this statement can be verified starting from $I(W_1;W_2|B_1) \geq 0$
			and rearranging the terms.
			See Appendix~\ref{apdx:redundant-ineq} for the details.

			The decoding strategy depends on the position of the point 
			$P_1$ lying within the $c_1$ plane. 
			We will treat two cases separately.
			\begin{description}
				
				\item[Case $c$:]
						Suppose $R_{1p}$ is such that: 
						\be
							I(X_1;B_1|W_1Q) \leq R_{1p}.   	 \label{eqn:lower-bnd-assume}
						\ee
						If we subtract this lower bound on $R_{1p}$ from equation \eqref{eqn:ctightc2}
						we can obtain an upper bound on $R_{2c}$.
						We also have an upper bound on $R_{1p}$ from \eqref{eqn:ctighta2} 
						and a bound on the sum rate $R_{1p}+R_{2c}$ from \eqref{eqn:ctightc2}.
						This gives us the following rate constraints:
						\begin{align}
							R_{1p}  			& \leq I\left(  X_{1};B_{1}|W_{1}W_{2}Q\right) , \tag{\ref{eqn:ctighta2} }  \\
							R_{2c}  			& \leq I\left(  W_{2};B_{1}|X_{1}Q\right) 
								= \eqref{eqn:ctightc2} -\eqref{eqn:lower-bnd-assume}  , \label{eqn:ctightc3}  \\
							R_{1p}+R_{2c}  	& = I\left(  X_{1}W_{2};B_{1}|W_{1}Q\right).  \tag{\ref{eqn:ctightc2}} \\
							R_{1c}			&  \leq I\left(W_{1};B_{1}|Q\right)  
						\end{align}
						
						\Sasoglou \ recognizes the rate constraints on $(R_{1p}, R_{2c})$ 
						in equations \eqref{eqn:ctighta2}, \eqref{eqn:ctightc3}
						and \eqref{eqn:ctightc2} to correspond to the dominant facet of a
						MAC rate region for a channel with inputs 
						$X_1 \sim p(x_1|w_1,q), W_2 \sim p(w_2|q)$ and output $(W_1,B_1)$. 
						In other words we have a special channel where $W_1$ is available
						as side information for Sender~1 and Receiver~1.
						The decode order is given by: $W_1 \to (X_1|W_1,\ W_2|W_1)$.

			To achieve rates on the plane $c_1$, Receiver~1 will first use a standard HSW decoder
			to decode the message $m_{1c}$ encoded in $W_1$
			and then apply the simultaneous decoding as stated in the following lemma:
			\end{description}%
			
			\begin{lemma}[Conditional simultaneous decoding]
			
				Let $\{ w_1^n(\ell_1) \}_{\ell_1 \in [2^{nR_{1\alpha} }] }$ be a codebook generated
				according to $\prod^n p_{W_1}$,
				and let $\{ x_1^n(m_1| w_1^n(\ell_1) ) \}_{m_1 \in [2^{nR_{1\beta} }] , \ell_1 \in [2^{nR_{1\alpha} }] }$
				be a conditional codebook generated according to $\prod^n p_{X_1|W_1}$.
				Similarly for Sender 2, we define codebooks 
				$\{ w_2^n(\ell_2) \}_{\ell_2 \in [2^{nR_{2\alpha} }] }$ 
				and another  $\{ x_2^n(m_2| w_2^n(\ell_2) ) \}_{m_2 \in [2^{nR_{2\beta} }] , \ell_2 \in [2^{nR_{1\alpha} }] }$
				generated according to $\prod^n p_{W_2}$ and $\prod^n p_{X_2|W_2}$.
				Suppose these codebooks are used on $n$ copies of the quantum multiple access 
				channel $\rho_{x_1,x_2}$, resulting in the map:
				\be
					(W_1^n, X_1^n,W_2^n, X_2^n) 
					\ \ \longrightarrow \ \
					\rho^n_{X_1^n|W_1^n,X_2^n|W_2^n}.
				\ee
				
				Consider the case where $W_1^n$ is known to the receiver, 
				and $X_2^n$ is considered as noise (averaged over).
				This situation corresponds to the following map:
				\be
					(W_1^n, X_1^n,W_2^n) 
					\ \ \longrightarrow \ \
					(W_1^n, \rho^n_{X_1^n|W_1^n,W_2^n} ),
				\ee
				where we defined $\rho^n_{X_1^n|W_1^n,W_2^n} \equiv 
					 \mathop{\mathbb{E}}_{X_2^n} \rho^n_{X_1^n|W_1^n,X_2^n|W_2^n}$,
				or in terms of the channel outputs:
				\be
					\rho^n_{X_1^n|W_1^n,W_2^n} 
					= 
					\bigotimes_{i=1}^n \left( 
						\sum_{x_2} 
							p_{X_2|W_2}(x_2|W_{2i}) 
							\rho_{X_{1i},x_2}.	
						\right).
				\ee
				
				An achievable rate region for the pair $(R_{1\beta}, R_{2\alpha})$ is described
				by: 
				\begin{align}
					R_{1\beta} 			  &\leq I(X_1;B|W_1W_2), \\
					 R_{2\alpha} 			 &\leq I(W_2;B|X_1W_1) = I(W_2;B|X_1), \\
					R_{1\beta} + R_{2\alpha}  &\leq I(X_1W_2;B|W_1),
				\end{align}
				where the mutual information quantities are with respect to the state:
				\be
					\!\!\theta^{W_1X_1W_2B}
					\equiv \!\!\!\!\!
					\sum_{w_1,x_1,w_2} \!\!\!\!\!
					p(w_1,x_1)p(w_2) 
					\ketbra{w_1}{w_1}^{W_1} 
					\!\otimes\!\,
					\ketbra{x_1}{x_1}^{X_1} 
					\!\otimes\!\,					
					\ketbra{w_2}{w_2}^{W_2} 
					\!\otimes\!\,					
					\rho^{B}_{x_1,w_2}.
				\ee
			\end{lemma}
			
			\begin{proof}
				The proof is similar to the two-sender MAC simultaneous decoding   
				from Theorem~\ref{thm:sim-dec-two-sender}.
			\end{proof}

			\begin{description}
				\item[Case $c^\prime$:]
						Now  suppose that $R_{1p} \leq I(X_1;B_1|W_1Q)$,
						then the trivial successive decoding  strategy is sufficient.
						Receiver 1 will decode in the order $W_1 \to X_1$.

					The decoding for is done sequentially using HSW decoding.
					Receiver~1 decodes the message $m_{1c}$ first,
					followed by $m_{1p}$.
					The decoding in this case is similar to the successive decoding
					used in Theorem~\ref{thm:cqmac-capacity}.
					The interfering messages $m_{2c}$ and $m_{2p}$ are treated as noise.%
			\end{description}%
		\end{proof}

		Thus we see that the combination of Lemma~\ref{lemma:geometry-of-MAConeR},
		Lemma~\ref{lem:sasoglou-moving-bd-points}, and Lemma~\ref{lem:sasoglou-succ-canc-for-ac}
		shows that the quantum Chong-Motani-Garg rate region is achievable using
		only two-sender simultaneous decoding.

%

%
%
%
%


\section{Successive decoding strategies for interference channels}
\label{sec:rate-succ-decoding}
	
	We report on some results concerning achievable
	rate regions for the interference channel that use the successive decoding
	approach.

	\subsection{Time-sharing strategies}
	
	In Section \ref{sec:mac-succ-decoding} on the multiple access channel, we 
	saw that a successive decoding strategy can be used to achieve all the rates
	on the dominant vertices of the rate region.
		Recall that for a fixed choice of encoding distribution 
		$p\equiv p_{X_1}(x_1)p_{X_2}(x_2)$, the two-sender
		QMAC capacity region has the shape of a pentagon with 
		two extreme points $\alpha_p \equiv (I(X_1;B), I(X_2;B|X_1))$ 
		and $\beta_p \equiv (I(X_1;B|X_2), I(X_2;B))$,
		which correspond to the rates achievable by successive decoding
		in two different orders.
	To achieve the rates in the convex hull of these points,
	we can use time-sharing between different codes achieving these rates.
	%
	%

	\medskip

	\begin{definition}[Time-sharing]
		%
		%
		Given two codebooks $\cC_1$ and $\cC_2$ with rates 
		corresponding to rate points $\alpha_p$ and $\beta_p$
		and a desired rate point $P 
		\in \mathrm{conv}(\alpha_p, \beta_p)$,
		we will have
		\be
			P=t \alpha_p + (1-t)\beta_p,
		\ee
		for some $t \in \mathbb{R}$,
		which we call the \emph{time-sharing} parameter.
		We can achieve the rates of a point $P^* \approx P$
		if we use the rational time-sharing parameter $t^* \approx t$, $t^*\equiv \frac{M}{N} \in \mathbb{Q}$
		and the following strategy:
		during each $N$ block-uses of the channel, use codebook $\cC_1$
		during $M$ of them and during the remaining $N-M$ uses of the channel,
		use codebook~$\cC_2$.
	\end{definition}


	The time-sharing strategy is not well-adapted for the interference channel.
	This is because the rates of the corner points of the achievable rate regions
	for the two receivers are not necessarily the same.
	The time-sharing strategy that works for one of the receivers
	might not work for the other one.



It is however possible to use successive decoding strategies for an interference channel in the following way. We start by considering a strategy where both receivers are asked to decode both messages, i.e., we are dealing with the compound multiple access channel. Such a strategy defines an achievable
rate region known as the ``successive decoding inner bound'' 
for the interference channel (cf. page 6-7 of Ref.~\cite{el2010lecture}).
%

	Consider all possible decode orderings that could 
	be used by the two receivers:
	\be
        \begin{array}{ll}	 
		\pi_1:  m_2 \to m_1|m_2,  		&\qquad 		\pi_2: m_2, \\
		\pi_1: m_2  \to m_1|m_2,  		&\qquad 		\pi_2: m_1 \to m_2|m_1, \\
		\pi_1: m_1,    				&\qquad 		\pi_2: m_1 \to m_2|m_1, \\		
		\pi_1: m_1,    				&\qquad 		\pi_2: m_2. 
        	\end{array}
	\ee
	Using each of these, we can achieve rates arbitrarily close to the following points:
        \begin{align}        	 
		P_1		&= (I(X_1; B_1|X_2), \min\{I(X_2; B_1), I(X_2; B_2) \} ), \label{eq:sd-1}\\
		P_2		&= (\min\{I(X_1; B_1|X_2), I(X_1;B_2)\}, \nonumber\\
		& \,\,\,\,\,\,\,\,\,\,\,\min\{I(X_2;B_1), I(X_2;B_2|X_1)\}), \label{eq:sd-2}\\
		P_3		&= (\min\{I(X_1;B_1), I(X_1;B_2)\}, I(X_2;B_2|X_1)),\label{eq:sd-3} \\
		P_4		&=(I(X_1; B_1), I(X_2; B_2)). \label{eq:sd-4}
        	\end{align}        
	We can use time-sharing between these different 
	codes for the interference channel to obtain all other 
	rates in $\mathrm{conv}(P_1,P_2,P_3,P_4)$.
	This achievable rate region is illustrated in the RHS of Figure \ref{fig:jt-succ-decoder}.

	\begin{figure}[ptb]
	\begin{center}
	\includegraphics[width=0.75\textwidth]{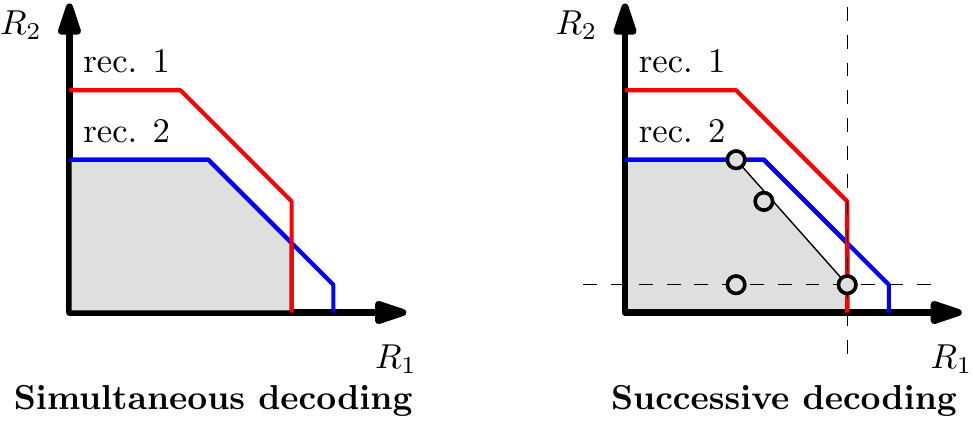}%
	\caption{These plots show achievable rates regions for the interference 
	channel for simultaneous decoding and successive decoding strategies with 
	fixed input distributions. Using a simultaneous decoding strategy, it is 
	possible to achieve the intersection of the two regions of the corresponding 
	multiple access channels. Using a successive decoding strategy, we obtain 
	four achievable rate points that correspond to the possible decoding orders 
	for the two multiple access channels. The solid red and blue lines outline the different 
	multiple access channel achievable rate regions, and the shaded gray areas 
	outline the achievable rate regions for the two different decoding strategies.
	}%
	\label{fig:jt-succ-decoder}%
	\end{center}
	\end{figure}

	\subsection{Split codebook strategies}
							\label{sec:rate-succ-decoding-rs}

		We can improve the successive decoding region described in 
		Section~\ref{sec:rate-succ-decoding} if we use split codebooks.
		Inspired by the Han-Kobayashi strategy we make the senders split their 
		messages into two parts:
		the messages of Sender~1 will be $m_{1p}$ and $m_{1c}$,
		and the messages of Sender~2 will be  $m_{2p}$ and $m_{2c}$.
		As in the Han-Kobayashi strategy, the use of the split codebooks 
		induces two three-sender multiple access channels.
		Receiver~1 is required to decode the set of messages 
		$m_{1p}, m_{1c}$ and $m_{2c}$ using successive decoding,
		and there are six different decode orderings he can use.

		Let the decoding ordering of Receiver 1 be represented by a permutation $\pi_1$
		on the set three elements $\{ 1p, 1c, 2c \}$.
		For example, the successive decoding of 
		the messages in the order $m_{2c} \to m_{1c}|m_{2c} \to m_{1p}|m_{1c}m_{2c}$
		will be denoted as the permutation $\pi_1=(2c,1c,1p)$.		

		We can naturally use all  $6 \times 6$ pairs of decoding orders to obtain a set of achievable 
		rate pairs. 

		\begin{proposition}
			Consider the rate point $P$ associated with the decode ordering $\pi_1$ for Receiver~1 and
			$\pi_2$ for Receiver~2:
			\be
				P=\left(R^{(1)}_{1p} + \min\{R^{(1)}_{1c},R^{(2)}_{1c}\}, \  \min\{R^{(1)}_{2c},R^{(2)}_{2c}\} + R^{(2)}_{2p} \right), 
				\nonumber
			\ee
			where the rate constraints for Receiver $j$ satisfy
			%
		        \begin{align}        	 
				R^{(j)}_{\pi_j(1)}		&\leq I(X_{\pi_j(1)};B_j), \\
				R^{(j)}_{\pi_j(2)}		&\leq I(X_{\pi_j(2)};B_j|X_{\pi_j({1})}), \\
				R^{(j)}_{\pi_j(3)}		&\leq 
					\left\{ 	
						\begin{array}{ll}
							I(X_{\pi_j(3)};B_j|X_{\pi_j({1})}X_{\pi_j({2})}) \ \ \ 
								&\mathrm{  if } \ \pi_j(3)={jc} \ \ \mathrm{ or }\ \  \pi_j(3)={jp} 	\\
							\infty, \ \ \
								& \mathrm{  otherwise }
						\end{array}
					\right.
		        	\end{align}        
			The rate pair $P$ is achievable for the quantum interference channel,
			for all permutations $\pi_1$ of the set of indices  $(1p,1c,2c)$ 
			and for all permutations $\pi_2$ of the set  $(2p,2c,1c)$.
		\end{proposition}

	\begin{figure}[htb]
	\begin{center}
	\includegraphics[width=\textwidth]{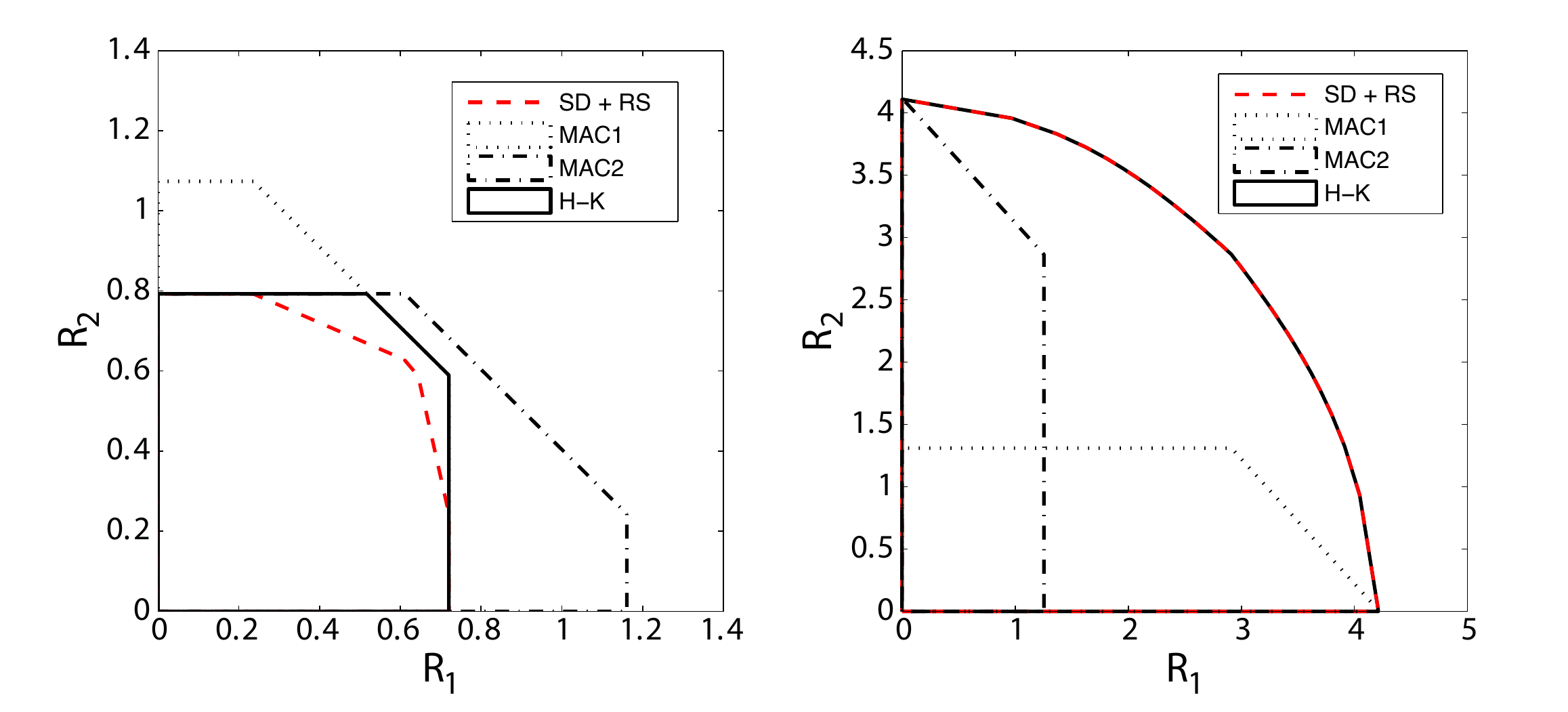}%
	\caption{These two figures plot rate pairs that the senders and receivers in a classical
	Gaussian interference channel can achieve using successive decoding and rate-splitting (SD+RS). The figures
	compare these rates with those achievable by the Han-Kobayashi (HK) coding strategy, while also plotting
	the regions corresponding to the two induced multiple access channels to each
	receiver (MAC1 and MAC2). The LHS figure
	demonstrates that, for a particular choice of signal to noise (SNR)
	 and interference to noise (INR)\cite{etkin2007gaussian}
	 parameters (SNR1 = 1.7, SNR2 = 2, INR1 = 3.4, INR2 = 4),
	successive decoding with rate-splitting does not perform as well as the Han-Kobayashi strategy. The RHS figure
	demonstrates that, for a different choice of parameters 
	(SNR1 =  343, SNR2 = 296, INR1 = 5, INR2 = 5),
	the two strategies perform equally well.}
	\label{fig:jt-succ-decoder-rate-splitting}
	\end{center}
	\end{figure}


	The rate region described by the convex hull of the points $P$
	is generally smaller than the Han-Kobayashi region as 
	illustrated in Figure~\ref{fig:jt-succ-decoder-rate-splitting}.
	Note that the split-codebook and successive decoding strategy 
	works pretty well in the low interference regime.
	An interesting open problem is whether we can achieve all rates 
	of the Han-Kobayashi region
	by splitting each sender's message
	into more than two parts and using only successive decoding.

	In particular, we want to know whether the capacity of the interference
	channel with strong interference can be achieved using only successive
	decoding. 
	Alternately, it would be interesting to prove that successive decoding
	is \emph{not} sufficient in order to achieve all the capacity in the 
	strong interference regime for any number of splits and any possible 
	decode order.
	We know that the time-sharing, 
	rate-splitting \cite{sasoglu2008successive} and
	generalized time-sharing  \cite{yagi2011multi}
	strategies do not work for the interference channel,
	but is it possible to show a negative result for all 
	successive decoding strategies?
	This question is explored further in \cite{FS12noRSforIC}.


\section{Outer bound}

		We will close this chapter by giving a simple outer bound 
		for the capacity of general quantum interference channels
		analogous to the classical result by Sato \cite{Sato77}.

		\begin{theorem}[Quantum Sato outer bound\cite{SavovICprojectECSE}]
		\label{thm:sato-weaker} Consider the Sato region defined as follows:
		\begin{equation}
		\mathcal{R}_{\text{Sato}}(\mathcal{N})\triangleq  \!\!\!\!
		\bigcup_{p_Q(q)p_{1}(x_{1}|q)p_{2}(x_{2}|q)} \!\!\!\!
		\{(R_{1},R_{2})
		\in \mathbb{R}_+^2 | \text{ Eqns \eqref{GsatoOne}-\eqref{GsatoThree} below}  \ \}, 
		\label{region:Gsato}%
		\end{equation}
		\vspace{-5mm}
		\begin{align}
		R_{1}  &  \leq I(X_{1};B_{1}|X_{2}Q)_{\theta},				\label{GsatoOne}\\
		R_{2}  &  \leq I(X_{2};B_{2}|X_{1}Q)_{\theta},				\label{GsatoTwo}\\
		R_{1}+R_{2}  &  \leq I(X_{1}X_{2};B_{1}B_{2}|Q)_{\theta}.		\label{GsatoThree}
		\end{align}
		The entropic quantities are with respect to the state $\theta^{QX_{1}X_{2}B_{1}B_{2}} \equiv$
		\be
		\sum_{q,x_{1},x_{2}}
		p_Q(q)p_{1}(x_{1}|q)p_{2}(x_{2}|q) \ 
		|q\rangle\langle q|^{Q}\otimes 
		|x_{1}\rangle\langle x_{1}|^{X_{1}}\otimes|x_{2}\rangle\langle
		x_{2}|^{X_{2}}\otimes\rho_{x_{1}x_{2}}^{B_{1}B_{2}}.
		\ee
		Then the region $\mathcal{R}_{\text{Sato}}(\mcal{N})$ 
		is an outer bound on the
		capacity region of the quantum interference channel.
		\end{theorem}

		This proof follows from the observation 
		that any code for the quantum interference channel also gives codes for three
		quantum multiple access channel subproblems: one for Receiver~1, another
		for Receiver~2, and a third for the two receivers considered together.
		We obtain the outer bound in Theorem~\ref{thm:sato-weaker}
		by using the outer bound on the quantum multiple access channel rates 
		from Theorem~\ref{thm:cqmac-capacity} for each of these channels.

%


\section{Discussion}

In this chapter we saw how the coding techniques and theorems
which we obtained in Chapter~\ref{chapter:MAC} can be applied
to prove coding theorems for the quantum interference channel.

The key takeaway is that interference is not noise,
and that it can be advantageous to the receivers to decode
messages in which they are not interested.
For Receiver~1, knowing the other user's transmissions
allows him to increase the rate at which he can decodel
going from  $I(X_1;B_1) 		=  H(B_1) - H(B_1|X_1)$
to the improved rate of $I(X_1;B_1|X_2) 	=  H(B_1|X_2) - H(B_1|X_1X_2)$.


Because some of our results concerned special cases of the 
interference channel problem, it is worthwhile to review our
overall progress towards the characterization of the 
capacity region of the general quantum interference channel $\ICcap$.
For general interference channels we have:
\[
	\mathcal{R}_{\mathrm{succ}}(\mcal{N})
	\ \subsetneq  	
	\mathcal{R}_{\mathrm{sim}}(\mcal{N})
	\ \subsetneq  	
	\mathcal{R}^o_{\mathrm{HK}}(\mcal{N})
	\equiv
	\mathcal{R}_{\mathrm{CMG}}(\mcal{N})
	\subseteq 
	\ICcap
	 \subseteq 
	\mathcal{R}_{\mathrm{Sato}}(\mcal{N}).
\]

In the special case of the interference channel with very strong
interference, the rate region achievable by successive decoding
achieves the capacity $\mathcal{R}_{\mathrm{succ}}(\mcal{N}) = \ICcap$.
In the special case of strong interference, the 
rate region achievable by simultaneous decoding 
is optimal $\mathcal{R}_{\mathrm{sim}}(\mcal{N}) = \ICcap$.


An interesting research question would be to investigate whether splitting
the messages into more than two parts, that is, turning the two-user IC into
a multiple-input multiple-output (MIMO) IC, can improve on the rates that are achievable using the
Han-Kobayashi strategy.

In this chapter, we used the superposition coding technique to construct 
the codebooks for the CMG coding strategy.
%
We will use this technique again in the next chapter in the context of 
the quantum broadcast channel.
%

%

%



\def\PIavg{\Pi^n_{\bar{\rho},\delta}}
\def\PIone{\Pi_{u^n_1(\ell_1)}}
\def\PIonepr{\Pi_{u^n_1(\ell_1^{\prime})}}

\def\PIsigUone{\Pi_{\sigma_{u^n_1(\ell_1)}} }
\def\PIsigUonePr{\Pi_{\sigma_{u^n_1(\ell_1^\prime)}} }

\def\rhoNulul{\rho_{\ell_1,\ell_2} }
\def\rhouu{\rho^{B_1B_2}_{u_1,u_2}}
\def\rhoNuu{\rho^{B_1}_{u^n_1,u^n_2}}
\def\omgu{\omega^{B_1}_{u_1}}
\def\omgNul{\omega_{\ell_1}} 
\def\omgNulpr{\omega_{\ell'_1} } 

\def\ExpUone{ \mathop{\mathbb{E}}_{U_1} }
\def\ExpUtwo{ \mathop{\mathbb{E}}_{U_2} }
\def\ExpUboth{ \mathop{\mathbb{E}}_{U_1,U_2} }
\def\ExpUtwoGone{ \mathop{\mathbb{E}}_{U_2|U_1} }

\def\AtypGua{\mcal{A}^{(m_a)}_{Z_{1p}|u_a,\delta}}
\def\AtypGuaun{\mcal{A}^{(m_a)}_{Z_{1p}|u_au^n_2,\delta}}

\def\ExpXW{\mathop{\mathbb{E}}_{\scriptscriptstyle X^{\!n}\!\!,W^{\!n}}}
\def\ExpW{\mathop{\mathbb{E}}_{\scriptscriptstyle W^{\!n}}}
\def\ExpXGW{\mathop{\mathbb{E}}_{\scriptscriptstyle X^{\!n}\!|\!W^{\!n}}}

\def\PIrhoW{\Pi_{\!W^{\!n}\!(m)} }
\def\PIrhoWone{\Pi_{\!W^{\!n}\!(1)} }

\chapter{Broadcast channels}
								\label{chapter:BC}

	How can a broadcast station communicate separate messages to two 
	receivers using a single antenna?
	The two message streams must somehow be ``mixed'' during the 
	encoding process so that the transmitted codewords will contain
	the information intended for both receivers.
	In this chapter we apply two codebook construction ideas from the 
	chapter on interference channels to build codebooks for the quantum broadcast channel.

	%
	The Chong-Motani-Garg construction used \emph{superposition encoding} to encode 
	a ``personal''  message (satellite codeword) on top of a ``common'' message (cloud center).
	In Section~\ref{sec:superposition-coding} we will use the superposition coding technique to 
	encode a ``personal'' message for one of the receivers on top of a ``common'' message for both receivers.
	Such a choice of encoding is well suited for broadcast channels where one of the receivers'
	signals is stronger than the other. We can pick the rate of the common message 
	so as to be decodable by the receiver with the weaker reception,
	and use the left-over capacity to the better receiver to transmit a personal message for him.
	The superposition coding technique was originally developed in this context \cite{C72}.
	
	Another approach to constructing the mixing of the information streams is
	to use two separate codebooks and an arbitrary mixing function that combines them
	as in the Han-Kobayashi coding strategy.
	%
	The Marton coding scheme presented in Section~\ref{sec:marton} uses this approach.
		
	%
	%

	%



\section{Introduction}

	The general broadcast communication scenario with two receivers
	involves the transmission 	of up to three separate information streams.
	To illustrate the communication problem, 
	consider the situation described in Figure~\ref{fig:BC-mt-royal} where
	the television station wants to transmit multiple streams of television 
	programming to two separate receivers.

	\begin{figure}[hb]
	\begin{center}
	
	\begin{tikzpicture}[node distance=2.0cm,>=stealth',bend angle=45,auto,scale=1.0, every node/.style={scale=1.0}]

	  \begin{scope}
		\node [draw=black,anchor=south west,inner sep=0] (bglayer) at (0,0) {\includegraphics[width=13cm]{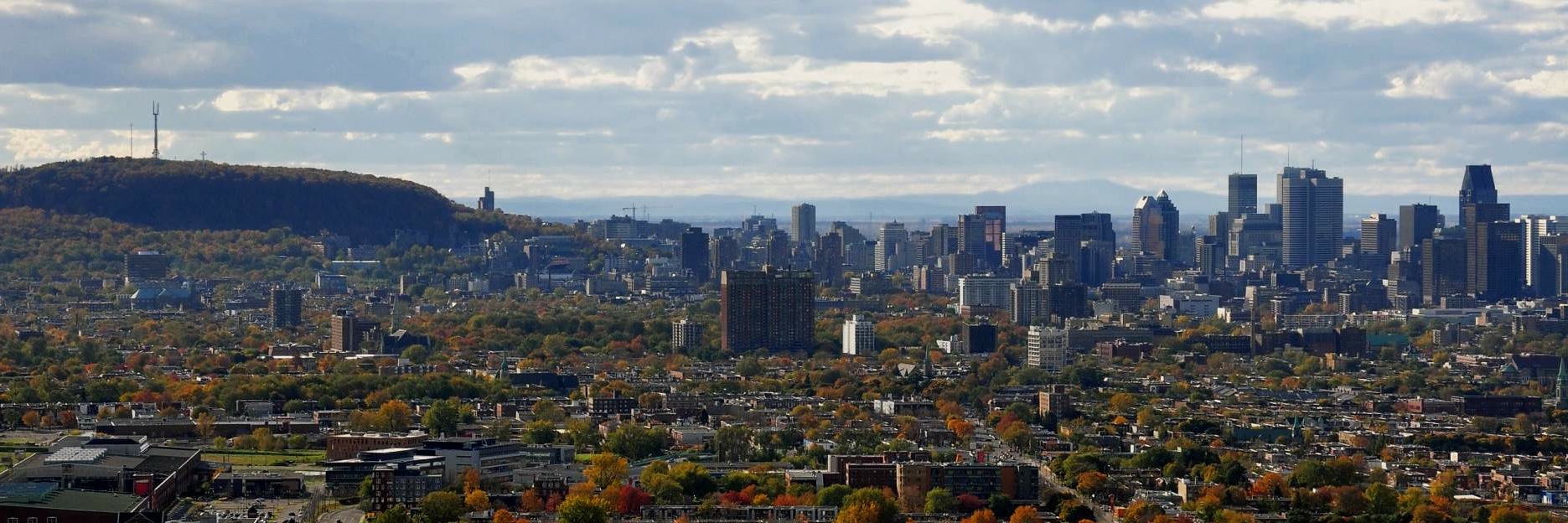} };
		\node [cnode] (Tx) at (1.2,3.3)  [ label=left:Tx,    ]                            {\footnotesize $x$} ; 
		\node [cnode] (Rx1) at (6.4,1.7) [ label=below:${\color{white}\textrm{Rx1}}$] {$y_1$} 
			edge [pre,draw=green,line width=2pt]  node[] {} (Tx);
		\node [cnode] (Rx2) at (10.5,2.1) [ label=below:${\color{white}\textrm{Rx2}}$] {$y_2$} 
			edge [pre,draw=green,line width=2pt]  node[] {} (Tx);		
		\draw[->,draw=red,line width=2pt]	([yshift=-1cm]Tx)	-- ([yshift=-1.5mm]Rx1.west);
		\draw[->,draw=blue,line width=2pt]	([yshift=+3cm]Tx)	-- ([yshift=+2mm]Rx2.west);
	  \end{scope}
	\end{tikzpicture}
	\end{center}
	\caption{	The broadcast channel. The sender wishes to transmit three separate
			information streams: an English language TV station for Receiver~1,
			a French language TV station for Receiver~2 and a
			weather TV station which is of interest to both receivers.}
	\label{fig:BC-mt-royal}
	\end{figure}
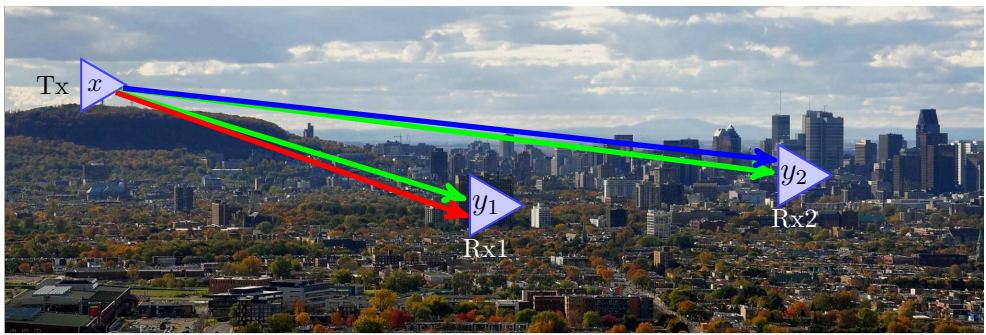

	Suppose that in each block, the antenna has to transmit 
	a common message $m \in [1:2^{nR}]$ intended for both receivers
	and personal messages $m_1 \in [1:2^{nR_1}]$ and $m_2 \in [1:2^{nR_2}]$
	each intended for one of the receivers.
	The task is therefore described by the following resource transformation: 
	\[
		n \cdot \mcal{N}^{X \to Y_1Y_2}  
		\ \ \overset{ ( 1 -\epsilon)}{\longrightarrow} \ \ 
		nR_1 \cdot  [c \to c^1]
		\ + \
		nR \cdot  [c \to c^1c^2]
		\ + \
		nR_2 \cdot  [c \to c^2].
	\]
	What are the achievable rate triples $(R_1,R,R_2)$ for this communication task?
	
	Note that the everyday usage of the word \emph{broadcast} presumes
	that only a common message is to be transmitted to all receivers.
	%
	If only a common message is to be transmitted, that is, we are looking for rates
	of the form $(0,R,0)$, the broadcast channel problem reduces to the compound 
	point-to-point channel problem and the capacity is given
	by the 
	minimum of the rates achievable for the receivers.
	In order to make the problem interesting from the information theory perspective,
	we have to consider the case where at least one personal message is to be transmitted.

	\subsection{Previous work}

	A wide body of research exists in classical information theory
	on the study of broadcast channels.
	An excellent review of this research is presented in \cite{cover1998comments}.
	The broadcast channel is also covered in textbooks \cite{CT91,el2010book,el2010lecture}.
	In the classical case, 
	two of the best known strategies  
	for transmitting information over broadcast channels are 
	superposition coding \cite{C72,B73,korner1977general} and Marton over-binning using correlated
	auxiliary random variables \cite{M79}.
	Sections \ref{sec:superposition-coding} and \ref{sec:marton} of this chapter 
	are dedicated to the generalization
	of these coding strategies to classical-quantum broadcast channels.


	

	\subsection{Quantum broadcast channels}

		Previous work on quantum broadcast channels includes \cite{YHD2006,guha2007classical,DHL10}.
		In \cite{YHD2006}, the authors consider both classical and quantum communication over quantum broadcast channels
		and prove a superposition coding inner bound similar to our Theorem~\ref{thm:sup-coding-inner-bound}.
		There has also been research on quantum broadcast channels in two other settings:
		quantum-quantum channels \cite{DHL10} and bosonic broadcast channels \cite{guha2007classical}.
		The Marton rate region for the quantum-quantum broadcast channel was developed in \cite{DHL10}.
		The authors use \emph{decoupling techniques} \cite{FQSW,AHS07,dupuis2010phd} in order to 
		show the Marton achievable rate region with no common message
		for \emph{quantum} communication\footnote{
			Note that the well known \emph{no cloning} theorem of quantum information precludes 
			the possibility of a quantum common message: $[q \to q^1q^2]$,
			where the quantum information of some system controlled by the 
			sender is faithfully transferred to \emph{two} receivers.
			%
			See \cite{YHD2006} for more comments on this issue.				
		}.

		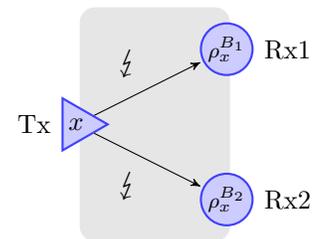
\begin{wrapfigure}{r}{0pt}
		\begin{tikzpicture}[node distance=2.0cm,>=stealth',bend angle=45,auto]

		  \begin{scope}
			\node [cnode] (Tx) [ label=left:Tx,yshift=-10mm   ]                            {\footnotesize $x$};
			\node [qnode] (Rx1) [ label=right:Rx1, right of=Tx,yshift=+10mm]	{\scriptsize $\rho^{B_1}_{x}$}
				edge  [pre]             node[swap]  		{$\lightning$}    (Tx);
			\node [qnode] (Rx2) [ label=right:Rx2, right of=Tx,yshift=-10mm] {\scriptsize $\rho^{B_2}_{x}$}
				edge  [pre]             node		   {$\lightning$}         (Tx) ;
		  \end{scope}
		  \begin{pgfonlayer}{background}
		    \filldraw [line width=4mm,join=round,black!10]
		      ([xshift=-2mm]Rx1.north -| Tx.east) rectangle ([xshift=+2mm]Rx2.south -| Rx2.west);
		  \end{pgfonlayer}

		\end{tikzpicture}
		\caption{A quantum broadcast channel $\rho^{B_1B_2}_x$. }
		\label{fig:QBC}
		\end{wrapfigure}

		We define a classical-quantum-quantum broadcast channel as the triple:%
		\begin{equation}
			(\mcal{X},\mcal{N}(x)\equiv \rho_{x}^{B_{1}B_{2}}, \mcal{H}^{B_1B_2}),	
			\label{eq:cqq-broadcast}%
		\end{equation}
		where $x$ is a classical letter in an alphabet $\mathcal{X}$ and 
		$\rho_{x}^{B_{1}B_{2}}$ is a density operator on the tensor product Hilbert
		space for systems $B_{1}$ and $B_{2}$. The model is such that when the sender
		inputs a classical letter $x$, Receiver~1 obtains system $B_{1}$, and
		Receiver~2 obtains system $B_{2}$.
		Since Receiver~1 does not have access to the $B_2$ part of the state $\rho_{x}^{B_{1}B_{2}}$,
		we model his state as $\rho_{x}^{B_{1}} = \text{Tr}_{B_2}\!\!\left[  \rho_{x}^{B_{1}B_{2}}  \right]$,
		where $\text{Tr}_{B_2}$ denotes the partial trace over Receiver~2's system.

	\subsection{Information processing task}

		The task of communication over a broadcast channel 
		is to use $n$ independent instances of the channel in order to communicate
		classical information to Receiver~1 at a rate~$R_1$, to Receiver~2 at a rate $R_2$,
		and to both receivers at a rate $R$.
		More specifically, the sender chooses a triple of messages $(m_1,m,m_2) \in[1:2^{nR_1}]\times[1:2^{nR}]\times[1:2^{nR_2}]$, 
		and encodes these messages into an $n$-symbol codeword $x^{n}\!\left( m_1,m,m_2\right)\in \mathcal{X}^n$
		suitable as input for the $n$ channel uses.

		The output of the channel is a quantum state 
		of the form:
		\begin{align}
			\mathcal{N}^{\otimes n}\!\left( x^{n}(m_1,m,m_2) \right)
			& \equiv			
			\rho_{x^{n}\left(  m_1,m, m_2\right)  }^{B_{1}^{n}B_{2}^{n}} \ \  \in  \ \mcal{D}(\mathcal{H}^{B_1^n B_2^n}),
		\end{align}
		  where
		  $
		  \rho_{x^{n}}^{B_{1}^{n}B_{2}^{n}} \equiv \rho_{x_1}^{B_{11}B_{21}} \otimes
		  \cdots \otimes \rho_{x_n}^{B_{1n}B_{2n}}.
		  $
		To decode the common message $m$ and the message $m_1$ intended specifically for him,
		Receiver~1 performs a POVM
		$\left\{ \Lambda_{m_1,m}\right\}$, $m_1\in\left[  1,\ldots,|\mathcal{M}_1|\right]$,
		$m\in\left[  1,\ldots,|\mathcal{M}|\right]$, on
		the system $B_{1}^n$, the output of which we denote $(M^{\prime}_1,M^{\prime})$. 
		%
		Receiver~2 similarly performs a POVM 
		$\left\{  \Gamma_{m,m_2}\right\}$, $m_2\in\left\{  1,\ldots ,|\mathcal{M}_2|\right\}$,
		$m\in\left[  1,\ldots,|\mathcal{M}|\right]$ on the system  $B_{2}^n$,
		and his outcome is denoted $(M^{\prime\prime},M^{\prime\prime}_2)$.
		
		An error occurs whenever either of the receivers decodes one 
		of the messages incorrectly.
		%
		%
		%
		The probability of error for a particular message triple $(m_1,m,m_2)$ is
		\begin{align*}
			p_{e}\!\left(  m_1,m,m_2\right)   
				\equiv
				\text{Tr}\!
				\left\{  
					\left(  I-\Lambda_{m_1,m}\otimes\Gamma_{m,m_2}\right)
					\rho_{x^{n}\left(  m_1,m,m_2\right)  }^{B_{1}^{n}B_{2}^{n}} 
				\right\},
		\end{align*}
		where the measurement operator $\left(  I-\Lambda_{m_1,m}\otimes\Gamma_{m,m_2}\right)$ represents
		the complement of the correct decoding outcome.

		    \begin{definition}%
			An $(n,R_1,R,R_2,\epsilon)$ classical-quantum broadcast channel code  consists
			of a codebook
			$\{x^n(m_1,m,m_2)\}$,
			$m_1\in \mathcal{M}_1$, $m\in \mathcal{M}$, $m_2\in \mathcal{M}_2$
			and two decoding POVMs  
			$\left\{ \Lambda_{m_1,m}\right\}_{m_1\in \mathcal{M}_1,m\in \mathcal{M}}$ 
			and 
			$\left\{  \Gamma_{m,m_2}\right\}_{m\in \mathcal{M},m_2\in \mathcal{M}_2}$
			such that the average probability of error 
			$\overline{p}_{e}$ is bounded from above as
			\begin{align}
				\overline{p}_{e}  
				&  \equiv 
					\frac{1}{|\mathcal{M}_1| |\mathcal{M}| |\mathcal{M}_2|}\sum_{m_1,m,m_2}p_{e}\!\left(  m_1,m,m_2\right) %
				 \leq \epsilon.
			\end{align}

		    \end{definition}  
		
		We say that a rate pair $\left(  R_{1},R,R_{2}\right)  $ is \textit{achievable} if there exists
		an $\left(  n,R_{1}-\delta,R-\delta,R_{2}-\delta,\epsilon\right)  $ quantum broadcast channel code
		for all $\epsilon,\delta>0$ and sufficiently large $n$. 

		A broadcast channel code with \emph{no common message}
		is a special case of the above communication task where
		the rate of the common message is set to zero: $(n,R_1,0,R_2,\epsilon)$.
		Alternately, we could choose not to send a personal message for Receiver~2 
		and obtain codes of the form $(n,R_1,R,0,\epsilon)$,
		which is known as the broadcast channel with a \emph{degraded message set} 
		\cite{korner1977general}.

\subsection{Chapter overview}


	In this chapter, we derive two achievable rate regions 
	for classical-quantum broadcast channels
	by exploiting the error analysis techniques developed
	in the context of quantum multiple access channels
	(Chapter~\ref{chapter:MAC}) and quantum interference channels (Chapter~\ref{chapter:IC}).
	%
	
	In Section~\ref{sec:superposition-coding}, we establish the achievability of the 
	rates in the superposition coding rate region (Theorem~\ref{thm:sup-coding-inner-bound}).
	We use a quantum simultaneous decoder at one of the receivers.
	Yard {\it et al}. independently proved the quantum superposition coding inner bound \cite{YHD2006},
	but our proof is arguably simpler and more in the spirit of its classical analogue~\cite{el2010lecture}.

	In Section~\ref{sec:marton} we prove that the quantum Marton rate region with no common message is achievable
	(Theorem~\ref{thm:marton-no-common}).
	The Marton coding scheme is based on the idea of \emph{over-binning} and using correlated
	auxiliary random variables \cite{M79}.
	The sub-channels to each receiver are essentially point-to-point,
	but it turns out that the \emph{projector trick} technique seems to be necessary
	in our proof.
	The Marton coding scheme gives the best known achievable rate region for 
	the classical-quantum broadcast channel.
	




\section{Superposition coding inner bound}	\label{sec:superposition-coding} 

		
		One possible strategy for the broadcast channel is to send a message
		at a rate that is low enough that both receivers are able to decode.
		Furthermore, if we assume that Receiver~1 has a better reception signal, 
		then the sender can encode a further message \emph{superimposed} on top of the common message
		that Receiver 1 will be able to decode \emph{given} the common message.
		The sender encodes the common message at rate $R$ using a codebook
		generated from a probability distribution $p_W(w)$ 
		and the additional message for Receiver~1 at rate $R_1$ using a conditional
		codebook with distribution $p_{X|W}(x|w)$.
		This is known as the superposition coding strategy \cite{C72,B73}.
		%
		%
		%

		\begin{theorem}[Superposition coding inner bound] 	\label{thm:sup-coding-inner-bound}
		Let $W$ be an auxiliary random variable, let 
		$p=p_{X|W}(x|w)p_W(w)$ be an arbitrary code distribution and let 
		$(\mcal{X}, \rho_x^{B_1B_2}, \mcal{H}^{B_1B_2} )$ be a classical-quantum broadcast channel.
		The superposition coding rate region $\mcal{R}_{\mathrm{SC}}(\mcal{N},p)$
		consists of all rate pairs $\left(  R_{1},R\right)$ such that:
		\begin{align}
			R_1 & \leq I(X; B_1 | W)_\theta, \\
			R  &\leq I(W; B_2)_\theta, \\
			R_1 + R  &\leq I(X;B_1)_\theta,
		\end{align}
		is achievable for the quantum broadcast channel.
		The information quantities are with respect to a state $\theta^{WXB_1B_2} $ of the form:
		\be
			\sum_{w,x} 
				p_W(w)p_{X|W}(x|w) \
				\ketbra{w}{w}^{W}
				\otimes
				\ketbra{x}{x}^{X}
				\otimes
				\rho_x^{B_1B_2}.\label{eq:code-state-QBC}
		\ee		
		\end{theorem}
%
%
%
%
%
%

		The superposition coding strategy allows us to construct
		codes for the broadcast channel of the form $(n,R_1,R,0,\epsilon)$,
		which have no personal message for Receiver~2.
		The task is therefore described as follows:
		\be
			n \cdot \mcal{N}^{X \to B_1B_2}  
			\ \ \overset{ ( 1 -\epsilon)}{\longrightarrow} \ \ 
			nR_1 \cdot  [c \to c^1]
			\ + \
			nR \cdot  [c \to c^1c^2],
		\ee
		where $[c \to c^1c^2]$ denotes the noiseless transmission of one bit 
		to both receivers.

		\begin{proof}
		The new idea in the proof is to exploit superposition coding
		and a quantum simultaneous decoder for the decoding of
		the first receiver \cite{C72,B73} 
		instead of the quantum successive decoding used in \cite{YHD2006}.
		We use a standard HSW\ decoder for the second receiver \cite{H98,SW97}.
		%

	\noindent
	\textbf{Codebook generation.} 
		We randomly and independently generate
		$2^{nR}$ sequences $w^{n}\!\left(  m\right)$ according to the product
		distribution $
		\prod\limits_{i=1}^{n}p_{W}\!\left(w_{i}\right)$.
		For each sequence $w^{n}\!\left(  m\right)$, we then randomly and
		conditionally independently generate $2^{nR_1}$ sequences 
		$x^{n}\!\left(m_{1},m\right)$ according to the product distribution: 
		$
		\prod\limits_{i=1}^{n}p_{X|W}\!\left(  x_{i}|w_{i}\!\left(  m\right)  \right)$.
		%

	\noindent
	\textbf{POVM\ Construction for Receiver~1}. 
	 	We now describe the POVM that 
		Receiver~1 employs in order to decode the transmitted messages. 
		First consider the state we obtain from (\ref{eq:code-state-QBC}) by tracing over
		the $B_{2}$ system:%
		\[
		\rho^{WXB_{1}}=\sum_{w,x}p_{W}\!\left(  w\right)  \ p_{X|W}\!\left(  x|w\right)
		\ \left\vert w\right\rangle\!\left\langle w\right\vert ^{W}\otimes\left\vert
		x\right\rangle\!\left\langle x\right\vert ^{X}\otimes\rho_{x}^{B_{1}}.
		\]
		Consider the following two averaged states:
		\begin{align*}
		\sigma_{w^n}^{B_{1}^n}
			&\equiv 
			\sum_{x^n}p_{X^n|W^n}\!\left(  x^n |w^n\right)  
						\rho^{B_1^n}_{x^n}
			=
			\bigotimes_{i=1}^n
						\left(
							\sum_{x}p_{X|W}\!\left(  x|w_i\right)\rho_{x}^{B_{1}}
						\right)
					\ifthenelse{\boolean{BOOKFORM}}{ \\[-1mm] & = }{=}
					\mathbb{E}_{X^n|w^n}\!\left\{ \rho_{X^n}^{B^n_1} \right\},	
													\\
		\bar{\rho}^{\otimes n}
			&\equiv 
			\sum_{w_n,x^n}p_{W^n}(w^n)p_{X^n|W^n}\!\left(  x^n |w^n\right)  
						\rho^{B_1^n}_{x^n}
			=
			\bigotimes_{i=1}^n
						\left(
							\sum_{w,x}p
							(w)p
							\!\left(  x|w\right)\rho_{x}^{B_{1}}
						\right)
					\ifthenelse{\boolean{BOOKFORM}}{ \\[-1mm] & = }{=} 
					\mathop{\mathbb{E}}_{W^n, X^n} \!\!\left\{ \rho_{X^n}^{B^n_1}  \right\}.	
			\end{align*}	
		We now introduce the following shorthand notation to denote the conditionally typical projectors
		with respect to the output state $\rho_{X^{n}\left(m_{1},m\right)}^{B_{1}^{n}}$ and
		the two averaged states defined above: 
		\begin{align*}
		\Pi_{X^{n}\left(  m_{1},m\right)  }    \equiv\Pi_{\rho_{X^{n}\left(
		m_{1},m\right)  },\delta}^{B_{1}^{n}}, \qquad
		\PIrhoW     \equiv\Pi_{\sigma_{W^{n}\left(
		m\right)  },\delta}^{B_{1}^{n}}, \qquad 
		\Pi   \equiv\Pi_{\rho,\delta}^{B_{1}^{n}}.
		\end{align*}
		Receiver~1 will decode using a POVM\ $\{\Lambda_{m_{1},m}\}$ 
		defined as the square root measurement:
		\begin{align}
			\Lambda_{m_{1},m} &  \equiv 
			\left(  
				\sum_{k_1,k}P_{k_{1},k}
			\right)^{\!\!\!-\frac{1}{2}}
			P_{m_{1},m}
			\left(
				\sum_{k_1,k}P_{k_{1},k}
			\right)^{\!\!-\frac{1}{2}}, \label{eq:square-root-POVM-generic-QBC} 
		\end{align}
		based on the following positive operators:
		\be
		P_{m_{1},m}\equiv\Pi\ \PIrhoW %
		\ \Pi_{X^{n}\left(  m_{1},m\right)  }\ \PIrhoW \ \Pi. \label{eq:rec-1-POVM}
		\ee
		
		Note the \emph{projector sandwich} structure with the more specific projectors on the inside.
		We have seen this previously in the construction of the simultaneous decoder POVM for the
		quantum multiple access channel.

	\noindent
	\textbf{POVM\ Construction for Receiver~2}. 		
		Consider now the state in equation (\ref{eq:code-state-QBC}) from the point of view 
		of Receiver~2. 
		If we trace over the $X$ and $B_{1}$
		systems, we obtain the following state:%
		\begin{align*}
			\rho^{WB_{2}}  
			=
			\sum_{w}  p_{W}\left(  w\right)  
			\left\vert w\right\rangle \! \left\langle w\right\vert ^{W}\otimes\sigma_{w}^{B_{2}},
		\end{align*}
		where $\sigma_{w}^{B_{2}}\equiv \sum_{x}p_{X|W}\left(x|w\right)  \rho_{x}^{B_{2}}$.
		Define also the state
		\be 
			\bar{\rho} \equiv \sum_{w,x}p_W(w)p_{X|W}\!\left(  x|w\right)\rho_{x}^{B_{2}}.
		\ee

		The second receiver uses a standard square root measurement:
		\begin{align}
			\Lambda_{m} &  \equiv 
			\left(  
				\sum_{k}P_{k}
			\right)^{\!\!\!-\frac{1}{2}}
			P_{m}
			\left(
				\sum_{k}P_{k}
			\right)^{\!\!-\frac{1}{2}}, \label{eq:square-root-POVM-generic-QBC2} 
		\end{align}		
		based on the following positive operators:
		\be
			P_{m}^{B_{2}^{n} }
			=\Pi_{\bar{\rho},\delta}^{B_{2}^{n} }\ 
			\Pi_{\sigma_{W^{\!n}\!(m)},\delta}^{B_{2}}\ 
			\Pi_{\bar{\rho},\delta}^{B_{2}^{n}}, \label{eq:rec-2-POVM}
		\ee
		where the above projectors are typical projectors defined with respect to
		the states 
		  $\sigma^{B_{2}^{n}}_{W^{n}\left(m\right)}$ 
		  and $\bar{\rho}^{\otimes n}$.

	\noindent
	\textbf{Error analysis for Receiver~1}.		
		We now analyze the expectation of the average error probability for the first
		receiver with the POVM defined in \eqref{eq:square-root-POVM-generic-QBC}:
		\begin{multline*}
		\qquad 
		\ExpXW \left\{  \frac{1}{M_{1}M_{2}}\sum_{m_{1},m}\text{Tr}\left\{  \left(  I-\Gamma_{m_{1},m}^{B_{1}^{n}}\right)
		\rho_{X^{n}\left(  m_{1},m\right)  }^{B_{1}}\right\}  \right\}  \\
		= \ \ \frac{1}{M_{1}M_{2}}\sum_{m_{1},m}\ExpXW \left\{
		\text{Tr}\left\{  \left(  I-\Gamma_{m_{1},m}^{B_{1}^{n}}\right)
		\rho_{X^{n}\left(  m_{1},m\right)  }^{B_{1}}\right\}  \right\}.
		\qquad
		\end{multline*}
		Due to the above exchange between the expectation and the average and the
		symmetry of the code construction (each codeword is selected randomly and
		independently), it suffices to analyze the expectation of the average error
		probability for the first message pair $\left(  m_{1}=1,m=1\right)  $,
		i.e., the last line above is equal to
		$
		  \ExpXW \left\{  \text{Tr}\left\{  \left(  I-\Gamma
		_{1,1}^{B_{1}^{n}}\right)  \rho_{X^{n}\left(  1,1\right)  }^{B_{1}}\right\}
		\right\}  .
		$
		Using the Hayashi-Nagaoka operator inequality 
		(Lemma~\ref{lem:HN-inequality} on page \pageref{lem:HN-inequality}),
		we obtain the following upper bound on this term:
		\begin{align}
		 \ExpXW \left\{  
		 	\text{Tr}\left[  
				\left(  I-\Gamma_{1,1}^{B_{1}^{n}}\right)  
				\rho_{X^{n}\left(  1,1\right)  }^{B_{1}}
				\right]
			\right\}
		& \leq 
		2  \ExpXW \left\{  \text{Tr}\left\{  \left(
		I-P_{1,1}\right)  \rho_{X^{n}\left(  1,1\right)  }^{B_{1}}\right\}
		\right\} \nonumber \\
		& \qquad 
		+4 \!\!\!\!\!\!\!\! \sum_{\left( m_{1},m\right)  \neq\left(  1,1\right)  }\!\!\!\!\!\!\!\!\!\!
		\ExpXW \!\!\left\{  \text{Tr}\left\{  P_{m_{1},m} \ \ 
		\rho_{X^{n}\left(  1,1\right)  }^{B_{1}}\right\}  \right\}. \label{eq:after-HN} 
		\end{align}
		
		We begin by bounding the term in the first line above. Consider the following
		chain of inequalities:%
		\begin{align*}
		\ExpXW \!\! \left\{  \text{Tr}\left\{  \Pi_{1,1}^{\prime}
		\rho_{X^{n}\left(  1,1\right)  }^{B_{1}}\right\}  \right\}   
		&  =
		\ExpXW \!\!
		\left\{  \text{Tr}\left\{  \Pi\ \PIrhoWone
		\Pi_{X^{n}\left(  1,1\right)  }\ \PIrhoWone\ \Pi
		\ \rho_{X^{n}\left(  1,1\right)  }^{B_{1}}\right\}  \right\}  \\
		&  \geq\ExpXW \!\!\left\{  \text{Tr}\left\{  \Pi_{X^{n}\left(
		1,1\right)  }\ \ \rho_{X^{n}\left(  1,1\right)  }^{B_{1}}\right\}  \right\}
		\\
		&  \ \ \ \ \ -\ExpXW \!\! \left\{  \left\Vert \rho_{X^{n}\left(
		1,1\right)  }^{B_{1}}-\Pi\ \rho_{X^{n}\left(  1,1\right)  }^{B_{1}}%
		\ \Pi\right\Vert _{1}\right\}  \\
		&  \ \ \ \ \ -\ExpXW \!\! \left\{  \left\Vert \rho_{X^{n}\left(1,1\right)  }^{B_{1}}
		-\PIrhoWone\ \rho_{X^{n}\left(1,1\right)  }^{B_{1}}\ \PIrhoWone\right\Vert _{1}\right\}
		\\
		&  \geq1-\epsilon-4\sqrt{\epsilon},
		\end{align*}
		where the first inequality follows from the inequality%
		\begin{equation}
		\text{Tr}\left\{  \Lambda\rho\right\}  \leq\text{Tr}\left\{  \Lambda
		\sigma\right\}  +\left\Vert \rho-\sigma\right\Vert _{1}, \label{eq:trace-inequality}
		\end{equation}
		which holds for all $\rho$, $\sigma$, and $\Lambda$ such that $0\leq
		\rho,\sigma,\Lambda\leq I$. 
		The second inequality follows from the \emph{gentle operator lemma for 
		ensembles} (see Lemma~\ref{lem:gentle-operator}) and the properties of typical projectors
		  for sufficiently large $n$.

		We now focus on bounding the second term of (\ref{eq:after-HN}). We can expand this term
		as follows:%
		\begin{align}
		\sum_{\left( m_{1},m\right)  \neq\left(  1,1\right)  }\!\!\!\!\!\!\!\!\!
		& \ExpXW \left\{  \text{Tr}\left\{  P_{m_{1},m}
		\rho_{X^{n}\left(  1,1\right)  }^{B_{1}}\right\}  \right\} \nonumber \\
		&= 
			\sum_{m_{1}\neq1}\ExpXW \left\{  \text{Tr}\left\{  P_{m_{1},1}\ 
			\rho_{X^{n}\left(  1,1\right)  }^{B_{1}}\right\}  \right\} \qquad \qquad  \label{eq:ERR1QBC}  
            \tag{$\mathbf{E1}$}\\[-1mm]
		&\qquad \qquad +\sum_{\substack{m_{1}, \\ m \neq1}}\ExpXW \left\{
		\text{Tr}\left\{  P_{m_{1},m}\ \rho_{X^{n}\left(  1,1\right)
		}^{B_{1}}\right\}  \right\}. \label{eq:ERR2QBC} \tag{$\mathbf{E2}$}
		\end{align}

		We will now compute the expectation of the first the term, $\mathbf{(E1)}$,
		with respect to the code randomness:
		{\allowdisplaybreaks
		\begin{align*}
		\ifthenelse{\boolean{BOOKFORM}}{
		& \!\!\!\!\!\!\ExpXW \left\{
			\mathbf{(E1)}
			\right\} 	 = \\
		}
		{
		\ExpXW \left\{
			\mathbf{(E1)}
			\right\} 		
		}
		&=    \sum_{m_{1}\neq1}
			\ExpXW \left\{  \text{Tr}\left\{
			P_{m_{1},1}\ \rho_{X^{n}\left(  1,1\right)  }^{B_{1}}\right\}
			\right\}  \\
		&  =\sum_{m_{1}\neq1}\ExpXW   \text{Tr}\left\{
			\Pi\ \PIrhoWone\ \Pi_{X^{n}\left(  m_{1},1\right)  }
			\ \PIrhoWone\ \Pi\ \rho_{X^{n}\left(  1,1\right)  }^{B_{1}}\right\}    \\
		& \leq
			2^{n\left[  H\left(  B_{1}|WX\right)  +\delta\right]  }  \sum_{\!m_{1}\neq1}\!
			\ExpXW \!
			\left\{  \text{Tr}\!\left[  \Pi\ \PIrhoWone 
			\rho_{X^{n}\left(  m_{1},1\right)  }
			\PIrhoWone 
			\Pi\ \rho_{X^{n}\left(  1,1\right)  }^{B_{1}}\right] \! \right\}  \\
		&  =
			2^{n\left[  H\left(  B_{1}|WX\right)  +\delta\right]  } 
			\sum_{\!m_{1}\neq1}
			\ExpW
			\bigg\{  \text{Tr}[   
			\PIrhoWone
			\ExpXGW \left\{  \rho_{X^{n}\left(m_{1},1\right)  }\right\}  \ 
			\PIrhoWone\  \\[-5mm]
			& \hspace{74mm} \Pi\ \ExpXGW 
			\left\{  \rho_{X^{n}\left(  1,1\right)  }^{B_{1}}\right\} 
			\Pi\ ] \bigg\} \\[3mm]
		&  =
		2^{n\left[  H\left(  B_{1}|WX\right)  +\delta\right]  } 
		\sum_{m_{1}\neq 1} \!\!
		\ExpW \!\!
		\left\{  \text{Tr}\left\{  
			\Pi\ 
			\PIrhoWone  \sigma_{W^{n}\left(  1\right)  }\ \PIrhoWone\ 
			\Pi\ \sigma_{W^{n}\left(  1\right)  }
		\right\}  \right\}  \\
		&  \leq
		2^{n\left[  H\left(  B_{1}|WX\right)  +\delta\right]  } \,2^{-n\left[ H\left(  B_{1}|W\right)  -\delta\right]  }
		\sum_{m_{1}\neq1} \!\!
		\ExpW \!\!
		\left\{  \text{Tr}\left\{  \Pi\ 
		\PIrhoWone\ \Pi\ \sigma_{W^{n}\left(  1\right)  }\right\}  \right\}  \\
		&  \leq
		2^{n\left[  H\left(  B_{1}|WX\right)  +\delta\right]  }\,
		2^{-n\left[H\left(  B_{1}|W\right)  -\delta\right]  }\ 
		\sum_{m_{1}\neq1}
		\ExpW \!\!
		\left\{  \text{Tr}\left\{  \sigma_{W^{n}\left(  1\right)  }\right\}
		\right\}  \\
		&  \leq
		2^{-n\left[  I\left(  X;B_{1}|W\right)  -2\delta\right]  }\ |\mcal{M}_{1}|,
		\end{align*}}%
		The first inequality is due to the \emph{projector trick} inequality 
		which states that:
		\begin{align}
			\Pi_{X^{n}\left(  m_{1},1\right)  }
		& \leq
			2^{n\left[  H\left(  B_{1}|WX\right)  +\delta\right]  }\,\rho_{X^{n}\left( m_1,1\right)  }^{B_{1}}. \label{eq:projector-trick}
		\end{align}
		The second inequality follows from the properties of typical projectors:
		\be
			\PIrhoWone  \sigma_{W^{n}\left(  1\right)  }\ \PIrhoWone \leq 2^{-n\left[H\left(  B_{1}|W\right)  -\delta\right]  }\PIrhoWone.
		\ee

		We now consider the expectation of the second term $\mathbf{(E2)}$ with respect to the random choice 
		of codebook.
		{\allowdisplaybreaks
		\begin{align*}
		\ExpXW \left\{
			\mathbf{(E2)}
			\right\}
		&=
		\sum_{\substack{m_{1}, \\ m \neq1}} 
		\ExpXW \!\!
		\left\{  \text{Tr}\left\{
		P_{m_{1},m}\ \rho_{X^{n}\left(  1,1\right)  }^{B_{1}}\right\}
		\right\}  \\
		&  = 
		\sum_{\substack{m_{1}, \\ m \neq1}} 
		\ExpXW \!\!\!
		\left\{  \text{Tr}\!
		\left[  \Pi \Pi_{W^{\!n}(m)  } \Pi_{X^{n}\left(  m_{1},m\right)  } \PIrhoW \Pi\ \rho_{X^{n}\left(1,1\right)  }^{B_{1}}\right]  \right\}  \\
		&  = 
		\sum_{\substack{m_{1}, \\ m \neq1}}
		\text{Tr}\left[  \ExpXW 
		\bigg\{  \PIrhoW \ \Pi_{X^{n}\left(  m_{1},m\right)  }\ \PIrhoW \right\}  \\[-5mm]
		& \hspace{47mm}  \Pi 
		\ExpXW \left\{  \rho_{X^{n}\left(  1,1\right) }^{B_{1} } \right\}  \Pi\  \bigg]  \\
		&  = 
		\sum_{\substack{m_{1}, \\ m \neq1}} 
		\text{Tr}\left\{  \ExpXW \!\!\left\{  \PIrhoW  \Pi_{X^{n}\left(  m_{1},m\right)  } \PIrhoW \right\}  
		\ \Pi  \bar{\rho}^{\otimes n} \Pi  \right\} \\
		&  \leq
		2^{-n\left[  H\left(  B_{1}\right)  -\delta\right]  } 
		\sum_{\substack{m_{1}, \\ m \neq1}} 
		\text{Tr}\!\!\left[ \! \ExpXW \!\!\! \left\{  
			\PIrhoW  \Pi_{X^{n}\left(  m_{1},m\right)  }
			\PIrhoW \right\} \! \Pi 
		\right]  \\
		&  =
		2^{-n\left[  H\left(  B_{1}\right)  -\delta\right]  } 
		\sum_{\substack{m_{1}, \\ m \neq1}} \!\!
		\ExpXW \!   \text{Tr}\left[  
		\Pi_{X^{n}\left(  m_{1},m\right)  } \PIrhoW 
		\Pi \PIrhoW \right]   
						\end{align*}
		\begin{align*}
		&  \leq
		2^{-n\left[  H\left(  B_{1}\right)  -\delta\right]  }
		\sum_{m	\neq1,\ m_{1}}\ExpXW 
		\left\{  \text{Tr}\left\{  \Pi_{X^{n}\left(  m_{1},m\right)  }\right\}  \right\}  \\
		&  \leq
		2^{-n\left[  H\left(  B_{1}\right)  -\delta\right]  }\ 
		2^{n\left[	H\left(  B_{1}|WX\right)  +\delta\right]  }\ 
		|\mcal{M}_{1}| |\mcal{M}_{2}|\\
		&  =
		2^{-n\left[  I\left(  WX;B_{1}\right)  -2\delta\right]  }\  |\mcal{M}_{1}| |\mcal{M}_{2}| \\
		& =
		2^{-n\left[  I\left(  X;B_{1}\right)  -2\delta\right]  }\ |\mcal{M}_{1}| |\mcal{M}_{2}|.
		\end{align*}}%
		The equality $I(WX;B_1)=I(X;B_1)$ follows from
		the way the codebook is constructed (the quantum Markov chain $W-X-B$).
		This completes the error analysis for the first receiver.

	\noindent
	\textbf{Error analysis for Receiver~2}.		
		The proof for the second receiver is analogous to the point-to-point
		HSW theorem. 
		%
		The following bound holds for the 
		expectation of the average error probability for the second receiver if $n$ is
		sufficiently large:%
		\begin{align*}
		 \ExpXW &  \left\{  \frac{1}{|\mcal{M}_{2}|}\sum_{m}\text{Tr}%
		\left\{  \left(  I-\Lambda_{m}^{B_{2}^{n}}\right)  \rho_{X^{n}\left(
		m_{1},m\right)  }^{B_{2}^{n}}\right\}  \right\}  \\
		&  \ \ =\ExpW \left\{  \frac{1}{|\mcal{M}_{2}|}\sum_{m}\text{Tr}\left\{
		\left(  I-\Lambda_{m}^{B_{2}^{n}}\right)  \ExpXGW \left\{
		\rho_{X^{n}\left(  m_{1},m\right)  }^{B_{2}^{n}}\right\}  \right\}
		\right\} \\
		&  \ \ =\ExpW \left\{  \frac{1}{|\mcal{M}_{2}|}\sum_{m}\text{Tr}\left\{
		\left(  I-\Lambda_{m}^{B_{2}^{n}}\right)  \sigma_{W^{n}\left(
		m\right)  }^{B_{2}^{n}}\right\}  \right\}  \\
		&  \ \ \leq 2\left(  \epsilon+2\sqrt{\epsilon}\right)  +4 \left[  2^{-n\left[
		I\left( W;B_{2}\right)  -2\delta\right]  }\  |\mcal{M}_{2}|\right]  .
		\end{align*}
		

		Putting everything together, the joint POVM\ performed by both receivers is of
		the form:
		$
		\Gamma_{m_{1},m}^{B_{1}^{n}}\otimes\Lambda_{m'}^{B_{2}^{n}},
		$
		and the expectation of the average error probability for both receivers is
		bounded from above as%
		\begin{align*}
		&  \!\!\!\!\!\!\!\!\!\!\ExpXW   \frac{1}{|\mcal{M}_{1}| |\mcal{M}_{2}|}\sum_{m_{1},m%
		}\text{Tr}\left\{  \left(  I-\Gamma_{m_{1},m}^{B_{1}^{n}}\otimes
		\Lambda_{m}^{B_{2}^{n}}\right)  \rho_{X^{n}\left(  m_{1},m\right)
		}^{B_{1}^{n}B_{2}^{n}}\right\}   \\
		&  \leq\ExpXW \left\{  \frac{1}{|\mcal{M}_{1}| |\mcal{M}_{2}|}\sum_{m_{1},m%
		}\text{Tr}\left\{  \left(  I-\Gamma_{m_{1},m}^{B_{1}^{n}}\right)
		\rho_{X^{n}\left(  m_{1},m\right)  }^{B_{1}^{n}}\right\}  \right\} \\
		&  \ \ \ \ \ \ +\ExpXW \left\{  \frac{1}{|\mcal{M}_{1}| |\mcal{M}_{2}|}%
		\sum_{m_{1},m}\text{Tr}\left\{  \left(  I-\Lambda_{m}^{B_{2}^{n}%
		}\right)  \rho_{X^{n}\left(  m_{1},m\right)  }^{B_{2}^{n}}\right\}
		\right\} \\
		&  \leq 4  \epsilon+12\sqrt{\epsilon}  + 4\ \left[
		2^{-n\left[  I\left(  W;B_{2}\right)  -2\delta\right]  }\  |\mcal{M}_{2}|\right] \\
		& \qquad 4\left[  2^{-n\left[
		I\left(  X;B_{1}|W\right)  -2\delta\right]  }\ |\mcal{M}_{1}| +2^{-n\left[  I\left(
		X;B_{1}\right)  -2\delta\right]  }\ |\mcal{M}_{1}| |\mcal{M}_{2}| \right] ,
		\end{align*}
		where the first inequality uses the operator union bound from Lemma~\ref{lem:operator-union-bound}:
		\be
		I^{B_{1}^{n}B_{2}^{n}}-\Gamma_{m_{1},m}^{B_{1}^{n}}\otimes\Lambda_{m%
		}^{B_{2}^{n}} \nonumber 
		\leq\left(  I^{B_{1}^{n}B_{2}^{n}}-\Gamma_{m_{1},m}%
		^{B_{1}^{n}}\otimes I^{B_{2}^{n}}\right)  +\left(  I^{B_{1}^{n}B_{2}^{n}%
		}-I^{B_{1}^{n}}\otimes\Lambda_{m}^{B_{2}^{n}}\right).
		\ee
		Thus, as long
		as the sender chooses the message sizes $|\mcal{M}_{1}|$ and $|\mcal{M}_{2}|$ such that
		$|\mcal{M}_{1}|    \leq2^{n\left[  I\left(  X;B_{1}|W\right)  -3\delta\right]  }$,
		$|\mcal{M}_{2}|\leq2^{n\left[  I\left(  W;B_{2}\right)  -3\delta\right]  }$,
		and $|\mcal{M}_{1}| |\mcal{M}_{2}|   \leq2^{n\left[  I\left(  X;B_{1}\right)  -3\delta\right]  }$,
		then there exists a particular code with asymptotically vanishing average
		error probability in the large $n$ limit.
		\end{proof}

		%

		%

		Taking the union over all possible choices of input
		distribution $p_{WX}(w,x)$ gives us the superposition coding inner bound:
		$\mcal{R}_{\mathrm{SC}}(\mcal{N}) \equiv \bigcup_{p_{WX}} \mcal{R}_{\mathrm{SC}}(\mcal{N},p_{WX})$.



\section{Marton coding scheme}		\label{sec:marton}

	We now prove that the Marton inner bound is achievable
	for quantum broadcast channels.
	%
	The Marton scheme depends on auxiliary random variables $U_1$ and $U_2$, 
	\emph{binning}, and  the properties of strongly\footnote{
			The notion of \emph{strong} typicality or \emph{frequency} typicality 
			differs from the entropy typicality we have used until now.
			See \cite[Section 14.2.3]{wilde2011book}.} 
 	typical sequences and projectors. 
	\begin{theorem}[Marton inner bound] \label{thm:marton-no-common}
		Let 
		$\{\rho^{B_1B_2}_x \}$ 
		be a classical-quantum broadcast channel and let $x=f(u_1,u_2)$ be a deterministic function.
		The following rate region is achievable:
		\bea
			R_1 	&\leq&	I(U_1; B_1)_\theta,  \nonumber \\
			R_2	&\leq&	I(U_2; B_2)_\theta,  \\ 
		R_1+R_2	&\leq&	I(U_1; B_1)_\theta + I(U_2; B_2)_\theta - I(U_1; U_2)_\theta,  \nonumber
		\eea			
		where the information quantities are with respect to the state:
		$$
			\theta^{U_1U_2B_1B_2} 
			=\!
			\sum_{u_1,u_2} \!
				p(u_1,u_2) 
				\ketbra{u_1}{u_1}^{U_1}
				\otimes
				\ketbra{u_2}{u_2}^{U_2}
				\otimes
				\rho_{f(u_1,u_2)}^{B_1B_2}.
		$$
	\end{theorem}
	
	The coding scheme in Theorem~\ref{thm:marton-no-common} is a broadcast channel
	code with no common message: $(n,R_1,0,R_2,\epsilon)$.
	The information processing task is described by:
	\be
		n \cdot \mcal{N}^{X \to B_1B_2}  
		\ \ \overset{ ( 1 -\epsilon)}{\longrightarrow} \ \ 
		nR_1 \cdot  [c \to c^1]
		\ + \
		nR_2 \cdot  [c \to c^2].
	\ee

\begin{proof} Consider the classical-quantum broadcast channel 
	$\{\mcal{N}(x) \equiv \rho^{B_1B_2}_x \}$,
	and a deterministic mixing function: $f: \mcal{U}_1 \times \mcal{U}_2  \to \mcal{X}$.	
	Using the mixing function as a pre-coder to the broadcast channel $\mcal{N}$,
	we obtain a channel $\mcal{N'}$ defined as:
	\be 
	  \mcal{N}'(u_1,u_2) 
	  \equiv \rho^{B_1B_2}_{f(u_1,u_2)}  
	  \equiv \rhouu.
	\ee

\noindent	
\textbf{Codebook construction}. Define two auxiliary indices $\ell_1 \in [1:L_1]$, $L_1= 2^{n[I(U_1;B_1)-\delta]}$ 
	and  $\ell_2 \in [1:L_2]$, $L_2 = 2^{n[I(U_2;B_2)-\delta]}$.
	For each $\ell_1$ generate an i.i.d.~random sequence $u_1^n(\ell_1)$ according to $p_{U_1^n}(u_1^n)$.
	Similarly we choose $L_2$ random i.i.d.~sequences $u_2^n(\ell_2)$ according to $p_{U_2^n}(u_2^n)$.
	%
	Partition the sequences $u_1^n(\ell_1)$ into $2^{nR_1}$ different bins $B_{m_1}$.
	Similarly, partition the sequences $u_2^n(\ell_2)$ into $2^{nR_2}$ bins $C_{m_2}$.
	For each message pair $(m_1,m_2)$,
	the sender selects a sequence $\big(u_1^n(\ell_1),u_2^n(\ell_2) \big) \in  \left( B_{m_1} \times C_{m_2} \right) \ \cap  \
	\mcal{A}^n_{p_{U_1U_2},\delta}$,
	such that each sequence is taken from the appropriate bin and the sender demands that they are 
	strongly jointly typical   
    and otherwise declares failure.
%
	The codebook $x^n(m_1,m_2)$ is deterministically
	constructed from $\big(u_1^n(\ell_1),u_2^n(\ell_2) \big)$ by applying the function $x_i=f(u_{1i},u_{2i})$.

	%
	

	%
	
\noindent
\textbf{Transmission}. Let $(\ell_1,\ell_2)$ denote the pair of indices of the joint sequence $(u_1^n(\ell_1), u_2^n(\ell_2))$ 
	which was chosen as the codeword for message $(m_1,m_2)$.
	%
	Expressed in terms of these indices the output of the channel is
	\be
		\rho_{u_1^n(\ell_1),u_2^n(\ell_2)}^{B_1^nB_2^n}
			=
				\bigotimes_{i \in [n]} \rho_{f(u_{1i}(\ell_1),u_{2i}(\ell_2))}^{B_1B_2} 
			\equiv \rho_{\ell_1,\ell_2}.
	\ee
	Define the following average states for Receiver~1:
	\be
	  \omgu  \equiv  \sum_{u_2} p_{U_2|U_1}(u_2|u_1) \rho^{B_1}_{u_1,u_2}, \qquad \label{avg-both-single-copy} 
	  \bar{\rho}  \equiv  \sum_{u_1}  p(u_1) \omgu.
	  \ee

\noindent
\textbf{Decoding}. 
	The detection POVM for Receiver 1, \ $\left\{  \Lambda_{\ell_1}\right\}_{\ell_1 \in [1,\ldots,L_1]}$,
	is constructed by using the square-root measurement 
	as in \eqref{eq:square-root-POVM-generic} based on the following combination
	of strongly typical projectors: 
\begin{align}
	%
	\Pi_{\ell_1}^{\prime} 
		  \equiv 
		\PIavg \ \PIone \ \PIavg.
	\end{align}
	%
	%
    The outcome of the measurement will be denoted $L^\prime_1$.
	 The projectors $\PIone$ and $\PIavg$ are 
	defined with respect to the states 
	$\omega_{u_1^n(\ell_1)  }$ and $\bar{\rho}^{\otimes n}$ 
	given in (\ref{avg-both-single-copy}).
	%
	Note that we use {\it strongly} typical projectors in this case 
    as defined in 
    \cite[Section 14.2.3]{wilde2011book}.
	Knowing $\ell_1$ and the binning scheme, Receiver~1 can deduce
	the message $m_1$ from the bin index.
	Receiver~2 uses a similar decoding strategy to obtain 
	$\ell_2$ and infer $m_2$.

\noindent	
\textbf{Error analysis}. An error occurs if one (or more) of the following events occurs.
	\begin{description}
	\item[$\mathbf{(E0)}$:] An encoding error occurs whenever there is no jointly typical sequence in $B_{m_1} \times C_{m_2}$
				for some message pair $(m_1,m_2)$.
	\item[$\mathbf{(E1)}$:] A decoding error occurs at Receiver 1 if $L_1^\prime \neq \ell_1$.
	\item[$\mathbf{(E2)}$:] A decoding error occurs at Receiver 2 if $L_2^\prime \neq \ell_2$.
	\end{description}

	The probability of an encoding error $\mathbf{(E0)}$ is bounded like in the 
	classical Marton scheme  	\cite{M79,el2010lecture,cover1998comments}. 
	To see this, we use Cover's counting argument \cite{cover1998comments}. 
	The probability that two random sequences $u_1^n$, $u_2^n$ chosen according to the 
	marginals are jointly typical is $2^{-nI(U_1;U_2)}$ and since there
	are $2^{n[I(U_1;B_1)-R_1]}$ and 
	$2^{n[I(U_2;B_2)-R_2]}$ sequences in each bin, 
	the expected number of jointly typical sequences
	that can be constructed from each combination of bins is
	\be
		2^{n[I(U_1;B_1)-R_1]} 2^{n[I(U_2;B_2)-R_2]}2^{-nI(U_1;U_2)}.
	\ee
	Thus, if we choose $R_1 + R_2 +\delta \leq I(U_1; B_1) + I(U_2; B_2) - I(U_1; U_2)$,
	 then the expected number of
    strongly jointly typical sequences in $B_{m_1} \times C_{m_2}$ is much larger than one.

	To bound the probability of error event $\mathbf{(E1)}$,
	we use the Hayashi-Nagaoka operator inequality (Lemma~\ref{lem:HN-inequality}):
	\begin{align}
		\Pr(\mathbf{E1})  
		& =
		\frac{1}{L_1} \sum_{\ell_1}
		\Tr\left[  (I - \Lambda_{\ell_1})  \rhoNulul  \right] \nonumber \\
		& \leq
		\frac{1}{L_1} \sum_{\ell_1}
		\bigg(
			2\underbrace{ \Tr\!\left[  (I - \PIavg \PIone \PIavg )  \rhoNulul  \right]}_{(T1) }  
			\nonumber \\[-2mm]
		& \qquad \qquad \qquad\qquad\qquad \ \ \ \  	+ 
			4 \underbrace{ \sum_{\ell'_1\neq \ell_1} \Tr\!\left[   \PIavg \PIonepr \PIavg   \rhoNulul  \right] }_{(T2)}
		\bigg).  \nonumber  \\[-7mm] \nonumber 
	\end{align}%
	Consider the following lemma \cite[Property 14.2.7]{wilde2011book}.
	\begin{lemma}	\label{lemma-avg-typ-proj-works}
	When $u_1^n(\ell_1)$ and $u_2^n(\ell_2)$ are strongly jointly typical,
    the state $\rhoNulul$ 
    is well supported by both the averaged
    and conditionally typical projector in the sense that:
	$
		\Tr\!\big[ \PIavg \ \rhoNulul \big] \geq 1 - \epsilon, \ \forall \ell_1, \ell_2,
	$
	and 
	$
		\Tr\!\left[ \PIone \ \rhoNulul \right] \geq 1 - \epsilon, \  \forall \ell_2,
	$
	\end{lemma}
	
	
	To bound the first  term (T1), we use the  following  argument:
	\begin{align}
	 1 - (T1) 
		& =
		\Tr\!\left[  \PIavg \PIone \PIavg \  \rhoNulul  \right]  \nonumber \\
		& =
		 \Tr\!\left[   \PIone \ \PIavg   \rhoNulul  \PIavg  \right] \nonumber  \\		
		& \geq
		 \Tr\!\left[   \PIone \ \rhoNulul     \right] 
			- \|  \PIavg \rhoNulul \PIavg - \rhoNulul \|_1   \nonumber \\	
		& \geq
		 (1 - \epsilon) - 2\sqrt{\epsilon},
	\end{align}
	where the inequalities follow from \eqref{eq:trace-inequality} and
	Lemma~\ref{lemma-avg-typ-proj-works}.
	This use of Lemma~\ref{lemma-avg-typ-proj-works} demonstrates why 
	the Marton coding scheme selects the 
	sequences $u_1^n(\ell_1)$ and $u_2^n(\ell_2)$ such that they are strongly jointly typical.
	
	To bound the second term, we begin by applying a variant of the projector trick from 
	\eqref{eq:projector-trick}. For what follows, note that the expectation $\ExpUboth$
	over the random code 
	is with respect to the product distribution $p_{U_1^n}(u_1^n) p_{U_2^n}(u_2^n)$:%
	{\allowdisplaybreaks
	\begin{align*}
		 \ExpUboth\!\!\left\{  (T2) \right\}  
		& =
		\ExpUboth \left\{ \sum_{\ell'_1\neq \ell_1} 
			\Tr\!\left[   \PIavg \Pi_{U^n_1(\ell_1^\prime)}  \PIavg \   \rhoNulul  \right]  \right\} \\
		%
		& \leq
		2^{n[H(B_1|U_1)+\delta]} \ExpUboth \left\{ \sum_{\ell'_1\neq \ell_1} 
			\Tr\!\left[    \PIavg \ \omgNulpr  \ \PIavg   \ \rhoNulul  \right]	 \right\}.\\
	\intertext{We continue the proof using averaging over the choice of codebook and
		  	the properties of typical projectors:}
		& = 
		2^{n[H(B_1|U_1)+\delta]} 
		\ExpUtwo \sum_{\ell'_1\neq \ell_1} 
			\Tr\!\left[ \PIavg \  \ExpUone\!\left\{ \omgNulpr \right\} \  \PIavg \ExpUone\!\left\{\rhoNulul\right\}   \right]	 \\			
		& = 
		2^{n[H(B_1|U_1)+\delta]} 
		\ExpUtwo \sum_{\ell'_1\neq \ell_1} 
			\Tr\!\left[ \PIavg \  \bar{\rho} \  \PIavg \ \ExpUone\!\left\{\rhoNulul\right\}    \right]	 \\		
		& \leq 
		2^{n[H(B_1|U_1)+\delta]} 2^{-n[H(B_1)- \delta]} 
		\ExpUboth \sum_{\ell'_1\neq \ell_1} 
			\Tr\!\left[  \PIavg \ \rhoNulul   \right]	 \\
		& \leq 
		2^{n[H(B_1|U_1)+\delta]} 2^{-n[H(B_1)- \delta]} 
		\ExpUboth \sum_{\ell'_1\neq \ell_1} 1 \\
		& \leq
		|\mcal{L}_1| \  2^{-n[I(U_1;B_1) - 2\delta]}.
	\end{align*}}%
	Therefore, if we choose $2^{nR_1} = |\mcal{L}_1| \leq  2^{n[I(U_1;B_1) - 3\delta]}$, the probability of 
	error will go to zero in the asymptotic limit of many channel uses.
	The analysis of the event $\mathbf{(E2)}$ is similar.
	\end{proof}

	\section{Discussion}

		We established two achievable rate regions for the classical-quantum broadcast channel.
		%
		In each case a fundamentally different coding strategy was used.

		The \emph{superposition coding} strategy is a very
		powerful coding technique for encoding two ``layers'' of messages in 
		the same codeword.
		Recall that the codebooks in the Chong-Motani-Garg coding strategy
		were also constructed using the superposition coding technique.
		In the next chapter, we will use this technique 
		to build codes for the relay channel.
	
		The \emph{binning}
		 strategy used in the Marton scheme is also applicable more widely.
		It can be used every time two
		uncorrelated messages must be encoded into a single codeword.
		From the point of view of Receiver~1, the messages intended
		for Receiver~2 are seen as random noise.
		By using the correlated variables $(U_1,U_2) \sim p(u_1,u_2)$
		to construct the codebooks we can obtain better rates than 
		would be possible if independent codebooks were used.
		This is because the ``noise'' codewords are now correlated with the 
		messages for Receiver~1 and thus helping him with the communication task.
		%

		Note that the above two techniques can be combined to give
		the quantum Marton coding scheme with a common message \cite{TakeokaMartonSup}.



\def\rhoFULLatRE{ \rho_{x^{n}\left(  m_{j},l_{j},l_{j-1}\right)  ,x_{1}^{n}\left(l_{j-1}\right)  }^{B_{1\left(  j\right)  }^{n}} }
\def\rhojRE{\rho_{m_j,\ell_j,\ell_{j-1}}^{B_{1(j)}^n }}
\def\rhobarRE{\bar{\rho}_{\ell_j,\ell_{j-1}}^{1(j)}}

\def\rhoRX{\rho^{(j)}_{m_j\ell_j\ell_{j\!-\!1}}\!\!\!\otimes\!\!\!\;\rho^{(j+1)}_{m_{j\!+\!1}\ell_{j\!+\!1}\ell_{\!j}} }
\def\rhojRX{\rho_{m_j,\ell_j,\ell_{j-1}}^{(j)}}
\def\rhojRXmjpr{\rho_{m^\prime_j,\ell_j,\ell_{j-1}}^{(j)}}
\def\rhojBOTH{\rho_{m_j,\ell_j,\ell_{j-1}}^{(j)}}
\def\rhojjRX{\rho_{m_{j+1},\ell_{j+1},\ell_{j}}^{(j+1)}}
\def\rhojjBOTH{\rho_{m_{j+1},\ell_{j+1},\ell_{j}}^{(j+1)}}
\def\rhobarRX{\bar{\rho}_{\ell_j,\ell_{j-1}}^{(j)}}
\def\rhodbarRX{\dbar{\rho}_{|\ell_{j-1}}^{(j)}}

\def\rhoFULLatRXj{ \rho_{x^{n}\left(  m_{j},l_{j},l_{j-1}\right)  ,x_{1}^{n}\left(l_{j-1}\right)  }^{B_{\left(  j\right)  }^{n}} }
\def\rhoFULLatRXjj{ \rho_{x^{n}\left(  m_{j+1},l_{j+1},l_{j}\right)  ,x_{1}^{n}\left(l_{j}\right)  }^{B_{\left(  j+1\right)  }^{n}} }

\def\dbar#1{\bar{\bar{#1}}}

\def\PIRElj{\Pi_{\sigma_{\ell_j|\ell_{j-1} } }^{} }
\def\PIREljpr{\Pi_{\sigma_{\ell^\prime_j|\ell_{j-1} } }^{} }

\def\PIREavg{\Pi_{\bar{\sigma}_{|\ell_{j-1}} }^{} }

\def\PIjRXmjlj{\Pi_{\rho_{m_j,\ell_j|\ell_{j-1} } }^{(j)} }
\def\PIjRXmjprlj{\Pi_{\rho_{m^\prime_j,\ell_j|\ell_{j-1} } }^{(j)} }
\def\PIjRXmjprljpr{\Pi_{\rho_{m^\prime_j,\ell^\prime_j|\ell_{j-1} } }^{(j)} }
\def\PIjRXlj{\Pi_{\bar{\rho}_{\ell_j|\ell_{j-1} } }^{(j)} }
\def\PIjRXljpr{\Pi_{\bar{\rho}_{\ell^\prime_j|\ell_{j-1} } }^{(j)} }
\def\PIjRXavg{\Pi_{\dbar{\rho}_{|\ell_{j-1} } }^{(j)} }
\def\PIjjRXlj{\Pi_{\tau_{\ell_j } }^{(j+1)} }
\def\PIjjRXljpr{\Pi_{\tau_{\ell^\prime_j } }^{(j+1)} }
\def\PIjjRXavg{\Pi_{\bar{\tau}  }^{(j+1)} }


\def\PRE{P_{\ell_j|\ell_{j-1}}^{B^n_{1(j)} } }
\def\PREpr{P_{\ell_j^\prime|\ell_{j-1}}^{B^n_{1(j)} } }
\def\PRX{P_{m_j,\ell_j|\ell_{j-1}}^{B^n_{(j)} B^n_{(j+1)} } }
\def\PRXljmj{P_{m_j,\ell_j|\ell_{j-1}}^{B^n_{(j)} B^n_{(j+1)} } }
\def\PRXljmjpr{P_{m^\prime_j,\ell_j|\ell_{j-1}}^{B^n_{(j)} B^n_{(j+1)} } }
\def\PRXljprmjpr{P_{m^\prime_j,\ell^\prime_j|\ell_{j-1}}^{B^n_{(j)} B^n_{(j+1)} } }

\def\PRXj{P_{m_j,\ell_j|\ell_{j-1}  }^{B^n_{(j)}} }
\def\PRXjmjpr{P_{m^\prime_j,\ell_j|\ell_{j-1}  }^{B^n_{(j)}} }
\def\PRXjljprmjpr{P_{m^\prime_j,\ell^\prime_j|\ell_{j-1}  }^{B^n_{(j)}} }
\def\PRXjj{P_{\ell_j}^{B^n_{(j+1)}} }
\def\PRXjjpr{P_{\ell^\prime_j}^{B^n_{(j+1)}} }

\def\GAMRE{\Gamma_{\ell_j|\ell_{j-1}  }^{B^n_{1(j)} } }
\def\LAMRX{\Lambda_{m_j,\ell_j|\ell_{j-1}  }^{B^n_{(j)}B^n_{(j+1)} } }
\def\LAMRXprpr{\Lambda_{m^\prime_j,\ell^\prime_j|\ell_{j-1}  }^{B^n_{(j)}B^n_{(j+1)} } }
\def\LAMRXj{\Lambda_{m_j,\ell_j|\ell_{j-1}  }^{B^n_{(j)} } }
\def\LAMRXjj{\Lambda_{\ell_j|\ell_{j-1}  }^{B^n_{(j+1)} } }

\def\ExpXgUXone{ \mathop{\mathbb{E}}_{X^n|U^n\!X_1^n } }
\def\ExpUXone{ \mathop{\mathbb{E}}_{U^n\!X_1^n } }
\def\ExpUXgXone{ \mathop{\mathbb{E}}_{U^n\!X^n|X_1^n } }
\def\ExpUgXone{ \mathop{\mathbb{E}}_{U^n|X_1^n } }
\def\ExpALL{ \mathop{\mathbb{E}}_{U^n\!X^n\!X_1^n } }

\def\ExpXone{ \mathop{\mathbb{E}}_{X_1^n} }

\chapter{Relay channels}
						\label{chapter:RC}
	
	\vspace{-3mm}
	Suppose that a source wishes to communicate with a remote destination
	and that a relay station is available which can decode the messages 
	transmitted by the source during one time slot and \emph{forward} them to the destination
	during the next time slot.
	With the relay's help, the source and the destination can improve communication rates
	because the destination can decode the intended messages in parallel from the channel 
	outputs during two consecutive time slots.
	In this way, useful information is received both from the source and the relay.

	\begin{wrapfigure}{r}{0pt}
	\begin{tikzpicture}[node distance=2.0cm,>=stealth',bend angle=45,auto]

	  \begin{scope}
		\node [cnode] (RCTx) [ label=left:S,yshift=-20mm   ]                            {\footnotesize $x$};
		\node [cnode] (RCRERx) [ label=above:$\ \ \ \ \ \ \textrm{R}$, above of=RCTx,xshift=+8mm ]		{$y_1$}
			edge  [pre]             		node[swap]	  	{$\lightning$}	(RCTx);
		\node [cnode] (RCRETx) [above of=RCTx,xshift=+14mm ]		{\footnotesize $x_1$};
		\node [cnode] (RCRx)  [ label=right:D, below of=RCRERx,xshift=+16mm]	{$y$}
			edge  [pre]             		node	  	{$\lightning$}	(RCTx)
			edge  [pre]				node[swap]  	{$\lightning$}	(RCRETx);
	  \end{scope}
	  \begin{pgfonlayer}{background}
	    \filldraw [line width=4mm,join=round,black!10]
	      ([xshift=+1mm,yshift=+2mm]RCRERx.south -| RCTx.north) rectangle ([xshift=+2.5mm]RCRx.south -| RCRx.west);
	  \end{pgfonlayer}
	  
	\end{tikzpicture}

	\caption{The classical relay channel.
	}
	\label{QRC}
	\end{wrapfigure}
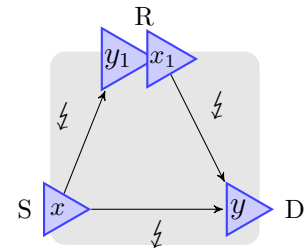
	The discrete memoryless relay channel is 
	a probabilistic model for a communication scenario with
	a \emph{source}, a \emph{destination} and a cooperative \emph{relay} station.
	The channel is modelled as a two-input
	two-output conditional probability distribution 
	\be
		p(y_1,y|x,x_1),
	\ee
	where $x$ is the input of the source,
	$y_1$ and $x_1$ are the received symbol and transmitted symbol of the relay, 
	and  $y$ is the output at the destination.
	This relay channel model is very general and contains 
	many of the other ideas presented in this thesis.
	The transmission of the source towards the relay and the destination
	is a kind of broadcast channel, whereas the decoding at the destination 
	is an instance of the multiple access channel.
	These correspondences can inform our choice of coding strategies,
	but in order to take full advantage of the communication network
	we must build a \emph{relay channel code} which aims to achieve the best
	overall rate from the source to the destination.		

	In this chapter, we will review some of the coding strategies for 
	the classical relay channel and then show that the
	\emph{partial decode-and-forward} strategy can be applied to
	the classical-quantum relay channel.
	%
	Note that we depart from the usual naming conventions for senders and receivers.
	We do so because both the source and the relay act as senders in our scenario,
	so more specific identifiers are necessary.
	%

	\vspace{-1mm}

	\section{Introduction}

		Consider two villages located in a valley that wish to establish a communication
		link between them using a direct link and also with the help of a 
        radio tower on a nearby mountain peak.
		We can setup a relay station on the tower, which decodes the messages from 
		the source village and retransmits them towards the destination village.
		Assuming the villagers only have access to point-to-point communication technologies,
		they now have two obvious options.
		Either they send information on the \emph{direct transmission} link,
		or they use \emph{full relaying}, where all their communication happens
		via the tower.
		In the first case, the tower is not used at all and in the second case
		the direct link is not used at all.
		
		It is worthwhile to examine the exact timing associated with the information
		flow in the latter scenario, since it is the first appearance of a multi-hop communication protocol.
		Let us assume that the source wants to send the string ``\texttt{constitution}'' to the destination.
		Assume that we use codewords of size $n$, and that each character is encoded in a separate
		codeword. The source and the relay have transmit codebooks $\{X_s^n(a)\}, \{X_r^n(a)\}$, $a \in \mcal{ASCII}$.

		The \emph{direct transmission} strategy will make $12n$ uses of the channel.
		The transmissions of the source will be 
		$[ X_s^n(\texttt{c}), X_s^n(\texttt{o}), X_s^n(\texttt{n}),  \ldots, X_s^n(\texttt{n})]$
		in each block.
		The relay will transmit a fixed codeword during this time.
		The destination will simply use a point-to-point decoder to extract the messages.
		The rate achievable using this strategy is given by:
		\be
			R \leq \sup_{p(x),x_1} I(X;Y|X_1 = x_1).
			\label{eq:direct-coding-bd}
		\ee
		
		The \emph{full relaying} strategy will use the channel $13n$ times,
		where the need for an extra block of transmission is introduced 
		by the decoding delay at the relay.
		During the 13 blocks, the transmissions of the source will be 
		$[ X_s^n(\texttt{c}), X_s^n(\texttt{o}), X_s^n(\texttt{n}),  \ldots, X_s^n(\texttt{n}), \emptyset]$,
		whereas the transmissions of the relay are one block behind:
		$[ \emptyset, X_r^n(\texttt{c}), X_r^n(\texttt{o}),  \ldots,$ $X_r^n(\texttt{o}), X_r^n(\texttt{n}) ]$.
		The source simply has no more messages to send during block 13,
		whereas the relay has no information to forward during the first block,
		so both parties will stay silent during these different times.
		The rates that are achievable by this approach are:
		\be
		 	R \leq \sup_{p(x),p(x_1)} \min\{ I(X_1;Y), I(X;Y_1|X_1) \}.
		\ee
		This corresponds to the minimum of the point-to-point capacities
		of the two legs of the transmission.
		Note that the second mutual information term is conditional on $X_1$,
		since the relay knows its own transmit signal.
		%
		%
		%

		Surely a better strategy must exist than the ones described above.
		How can we use both the direct link and the relayed link at the same time?
		

	\subsection{Classical relay channel coding strategies}



		
		%
		%
		Two important families of coding strategies exist for relay channels: 
		\emph{compress and forward} and \emph{decode and forward} \cite{cover1979capacity,el2010lecture}.
		%
		%

		In compress-and-forward strategies, the relay does not try to decode the message
		from his received signal $Y_1^n$, but simply searches for a 
		close sequence $\hat{Y}_1^n$ chosen from a predetermined compression codebook. 
		To continue the example from the previous section,
		suppose that the relay's decoding simply tries to determine whether 
		the transmitted message is a vowel or a consonant.
		This partial information about the message 
		is then forwarded to the destination during the next block, 
		encoded into a codeword $x_1^n(\mathbf{s})$, $\mathbf{s} \in \{ \textrm{consonant}, \textrm{vowel} \}$
		to serve as \textbf{s}ide-information for the decoding at the destination.

		Compress and forward strategies are appropriate in situations where
		the direct link between the source and the destination is stronger
		than the link from the source to the relay.
		In such a situation it would be disadvantageous to require that the messages
		from the source be fully decoded by the relay.
		Still, if the relay decodes \emph{something} and forwards this information
		to the destination,  better rates are achievable than if we simply 
		chose to not use the relay as in the direct coding approach \cite{el2010lecture}.
		%
		%
		%
		%

		In a decode-and-forward strategy, each of the 
		transmitted messages is decoded by the relay and retransmitted during the
		next block.
		Using this strategy, the destination can  
		decode useful information both from the source and the relay.
		In this way we could achieve the maximum possible throughput
		to the destination $I(X,X_1;Y)$.
		
		There are at least three decoding strategies that can be used
		by the destination: backwards decoding, sequential decoding with 
		binning at the relay, or collective decoding of consecutive output blocks of the channel (joint decoding).
		All three decoding techniques for the decode-and-forward strategy achieve 
		the same rate:
		\be
			R \leq \max_{p(x,x_1)} \min\{ \ I(X,X_1;Y), \ I(X;Y_1|X_1) \}.
			\label{eq:dec-fwd-rate}
		\ee
		We will focus on the collective decoding strategy.

		To illustrate the collective decoding strategy let us consider
		again the situation in which the source village is transmitting 
		the string ``\texttt{constitution}'' to the destination village.
		The transmission will take 13 block-uses of the channel.
		Figure~\ref{fig:coherent-dec-fwd} illustrates the flow of information 
		for the character $\texttt{n}$ which happens during the third and fourth
		block-uses of the channel.
		%
		During the third and the fourth transmission blocks,
		the destination has collected the output variables $(Y^n_{(3)},Y^n_{(4)})$
		and will perform a decoding operation on both outputs collectively.
		The rate $I(X,X_1;Y)$ is obtained from the decomposition 
		$I(X,X_1;Y)=I(X;Y|X_1)+I(X_1;Y)$, where the second term
		will come from the probability of making a mistake when
		decoding $x_{1(4)}^n(\texttt{n})$ from $Y^n_{(4)}$ 
		and the first terms comes from the probability 
		of wrongly decoding $x_{(3)}^n(\texttt{n})$ from $Y^n_{(3)}$.
		%
		%

		%

		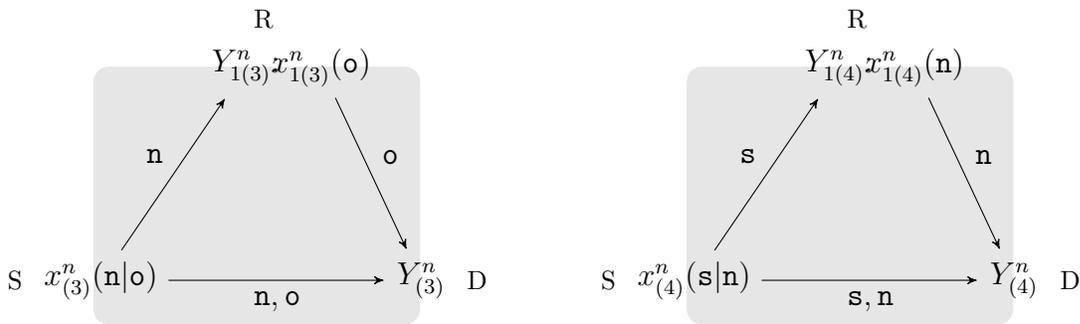
\begin{figure}[h]
		
		\begin{center}
		\subfigure[During block 3, the relay will transmit its codeword ``\texttt{o}'',
				which we assume was received in the previous block.
				The source transmits a codeword $x^n(\texttt{n}|\texttt{o})$
				which is chosen from a \emph{coherent} codebook.]{
		\begin{tikzpicture}[node distance=4cm,>=stealth',bend angle=45,auto]
		  \begin{scope}
			\node [] (RCTx) [ label=left:S,yshift=-20mm   ]                            {$x_{(3)}^n(\texttt{n}|\texttt{o})$};
			\node [] (RCRERx) [ label=above:$\ \ \ \ \textrm{R}$, above right of=RCTx,xshift=-9mm ]		{$Y^n_{1(3)}$:}
				edge  [pre]             		node[swap]	  	{$\texttt{n}$}	(RCTx);
			\node [] (RCRETx) [above right of=RCTx,xshift=+1mm ]		{$x_{1(3)}^n(\texttt{o})$};
			\node [] (RCRx)  [label=right:D,below right of=RCRERx,xshift=-5mm]	{$Y^n_{(3)}$}
				edge  [pre]             		node	  	{$\texttt{n},\texttt{o}$}	(RCTx)
				edge  [pre]				node[swap]  	{$\texttt{o}$}	(RCRETx);
		  \end{scope}
		  \begin{pgfonlayer}{background}
		    \filldraw [line width=4mm,join=round,black!10]
		      ([xshift=+1mm,yshift=+2mm]RCRERx.south -| RCTx.north) rectangle ([xshift=+2.5mm]RCRx.south -| RCRx.west);
		  \end{pgfonlayer}  
		\end{tikzpicture}
		} \ \ \ \ \ \ 
		\subfigure[During block 4, the relay will transmit its codeword for ``\texttt{n}'',
				which we assume was received in the previous block.
				The source transmits a codeword $x^n(\texttt{s}|\texttt{n})$.
                ]{
		\begin{tikzpicture}[node distance=4cm,>=stealth',bend angle=45,auto]
		  \begin{scope}
			\node [] (RCTx) [ label=left:S,yshift=-20mm   ]                            {$x_{(4)}^n(\texttt{s}|\texttt{n})$};
			\node [] (RCRERx) [ label=above:$\ \ \ \ \textrm{R}$, above right of=RCTx,xshift=-9mm ]		{$Y^n_{1(4)}$:}
				edge  [pre]             		node[swap]	  	{$\texttt{s}$}	(RCTx);
			\node [] (RCRETx) [above right of=RCTx,xshift=+1mm ]		{$x_{1(4)}^n(\texttt{n})$};
			\node [] (RCRx)  [label=right:D,below right of=RCRERx,xshift=-5mm]	{$Y^n_{(4)}$}
				edge  [pre]             		node	  	{$\texttt{s},\texttt{n}$}	(RCTx)
				edge  [pre]				node[swap]  	{$\texttt{n}$}	(RCRETx);
		  \end{scope}
		  \begin{pgfonlayer}{background}
		    \filldraw [line width=4mm,join=round,black!10]
		      ([xshift=+1mm,yshift=+2mm]RCRERx.south -| RCTx.north) rectangle ([xshift=+2.5mm]RCRx.south -| RCRx.west);
		  \end{pgfonlayer}  
		\end{tikzpicture}
		}
		\end{center}
		
		\caption{	Information flow in the relay network during the third and fourth transmission blocks
				of the string ``\texttt{constitution}''. 
		}
		\label{fig:coherent-dec-fwd}

		\end{figure}

	Observe that the optimization in \eqref{eq:dec-fwd-rate} is taken over all 
	joint input distributions $p_{XX_1}(x,x_1)$,
	which would seem to contradict the assumption that the source and the relay
	are different parties and cannot synchronize their encoding.
	%
	Recall that in the multiple access channel problem, the assumption that the 
	senders act independently translated to the optimization over
	all product distributions $p_{X_{1}}\!(x_1)p_{X_{2}}\!(x_2)$ in \eqref{eq:QMACunion}.

	The change from  $p_{X}(x)p_{X_1}(x_1)$ to $p_{XX_1}(x,x_1)$
	is allowed because the source uses a \emph{coherent} codebook.
	The codewords for the relay are chosen according to $p_{X_1}(x_1)$,
	%
	 whereas the codewords for the sender are chosen according to
	 $p_{X|X_1}(x|x_1)$ conditional on the codeword of the relay.
	 But how could the source possibly know what the relay will be transmitting
	 during each time instant?
	 No telepathic abilities are necessary --- only optimism.
	The source knows what the relay will be transmitting
	because, if the protocol is working, it should be the codeword 
	from the previous block.
	 %
	%

		%
		%
		The \emph{partial} decode-and-forward strategy differs from the decode-and-forward 
		strategy in that it requires the relay to decode only \emph{part of} 
		the message from the source \cite{cover1979capacity}. 
		The idea is similar to the \emph{partial} interference cancellation
		strategy used by Han and Kobayashi for the interference channel \cite{HK81}, 
		which is its contemporary.
		%
		%
		

%


\subsection{Quantum relay channels}

		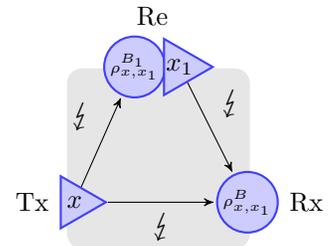
\begin{wrapfigure}{r}{0pt}
		\begin{tikzpicture}[node distance=1.8cm,>=stealth',bend angle=45,auto]

		  \begin{scope}
			\node [cnode] (RCTx) [ label=left:Tx,yshift=-20mm   ]                            {\footnotesize $x$};
			\node [qnode] (RCRERx) [ label=above:$\ \ \ \ \textrm{Re}$, above of=RCTx,xshift=+8mm ]		{\tiny $\rho^{B_1}_{x,x_1}$}
				edge  [pre]             		node[swap]	  	{$\lightning$}	(RCTx);
			\node [cnode] (RCRETx) [above of=RCTx,xshift=+14mm ]		{\footnotesize $x_1$};
			\node [qnode] (RCRx)  [ label=right:Rx, below of=RCRERx,xshift=+15mm]	{\tiny $\rho^{B}_{x,x_1}$}
				edge  [pre]             		node	  	{$\lightning$}	(RCTx)
				edge  [pre]				node[swap]  	{$\lightning$}	(RCRETx);
		  \end{scope}
		  \begin{pgfonlayer}{background}
		    \filldraw [line width=4mm,join=round,black!10]
		      ([xshift=+1mm,yshift=+2mm]RCRERx.south -| RCTx.north) rectangle ([xshift=+2.5mm]RCRx.south -| RCRx.west);
		  \end{pgfonlayer}
		  
		\end{tikzpicture}

		\caption{The quantum relay channel $\rho^{B_1B}_{x,x_1}$. 
		}
		\label{fig:QRC}

		\end{wrapfigure}

	A classical-quantum relay channel $\mcal{N}$ is a map with two classical inputs
	$x$ and $x_1$ and two output quantum systems $B_1$ and $B$.
	For each pair of possible input symbols $(x,x_1)\in \mcal{X} \times \mcal{X}_1$, the channel prepares
	a density operator $\rho^{B_1B}_{x,x_1}$ defined on the
	 tensor-product Hilbert space $\mcal{H}^{B_1}\otimes \mcal{H}^{B}$:
	\be
		\rho^{B_1B}_{x,x_1} \equiv   \mcal{N}^{XX_1\to B_1B}(x,x_1),
		\label{eqn:relay-chan-def}
	\ee
	 where $B_1$ is the relay output and
	$B$ is the destination output.

	\subsection{Chapter overview}

		In this chapter we develop the partial decode-and-forward strategy 
		for classical-quantum relay channels \cite{savov2011partial}. 
		This \emph{partial} decoding at the relay is a more general strategy than
		the \emph{full} decode-and-forward strategy
		in the same way that the partial interference cancellation strategy for
		the interference channel (the Han-Kobayashi strategy)
		was more general than a full interference cancellation strategy.
		%


		Our results are the first extension of the quantum
		 simultaneous decoding techniques used in \cite{FHSSW11,S11a} 
		to multi-hop networks. 
		The decoding is based on a novel ``sliding-window'' quantum measurement
		(see \cite{carleial1982multiple,xie2005achievable}) which 
		involves a collective measurement on two consecutive blocks of the output
		in order to extract information from both the Sender and the relay.

		The next section will describe the coding strategy in more detail
		and state our results.
		The proof is given in Section~\ref{sec:PDFproof}.

\section{Partial decode-and-forward strategy}
				\label{sec:PDF-for-QRC}

		The idea for the code construction is to use a split codebook strategy where the
		source decomposes the message set into the Cartesian product of two 
		 different sets $\mathcal{L}$ and $\mathcal{M}$. We
		can think of the set $\mathcal{L}$ consisting of common messages that both the
		relay and the destination decode, while the set $\mathcal{M}$ consists of
		personal messages that only the destination decodes.

	In the context of our coding strategy,
	we analyze the average probability of error at the relay:
	\begin{align*}
		\bar{p}_{e}^{R}
		\equiv 
			\frac{1}{|\mathcal{L}|}\sum_{\ell_j}					
			\text{Tr}\!
			\left\{  
				\left(  I-\Gamma^{B_{1(j)}^{n}} _{\ell_j} \right)
				\rho_{\ell_j }^{B_{1(j)}^{n}} 
			\right\},
	\end{align*}
	and the average probability of error 
	at the destination:
	\begin{align}
		\label{eqn:avgpe-Rx}
		\bar{p}^{D}_{e}
		\!  \equiv \! 
			\frac{1}{|\mathcal{M}||\mathcal{L}|}\sum_{m_j,\ell_j}
			\text{Tr}\!
			\left[ \! 
				\left(\!  I-\Lambda_{m_j,\ell_j}^{B_{(j)}^{n} B_{(j+1)}^{n} }  \right)
				\rho_{m_j,\ell_j  }^{B_{(j)}^{n} B_{(j+1)}^{n} } 
			\right].
	\end{align}
	The operators $\left(  I-\Gamma_{\ell_j} \right)$ 
	and $\left(  I-\Lambda_{m_j,\ell_j}\right)$
	correspond to the complements of the correct decoding outcomes.

	    \begin{definition}%
		An $(n,R,\epsilon)$ partial decode-and-forward code for the quantum relay channel consists
		of two codebooks 
		$\{x^n(m_j,\ell_j)\}_{m_j\in \mathcal{M}, \ell_j \in \mathcal{L}}$ and
		$\{x_1^n(\ell_j)\}_{\ell_j \in \mathcal{L}}$
		and decoding POVMs  
		$\left\{ \Gamma_{\ell_j}\right\}_{\ell_j\in \mathcal{L}}$ (for the relay)
		and 
		$\left\{  \Lambda_{m_j,\ell_j}\right\}_{m_j\in \mathcal{M}, \ell_j \in \mathcal{L}}$ (for the destination),
		such that the average probability of error
		is bounded from above as 
		$\overline{p}_{e} = \bar{p}^{R}_{e} + \bar{p}^{D}_{e} \leq \epsilon$.
	    \end{definition}  
	
	A rate $R$ is \textit{achievable} if there exists
	an $\left(  n,R-\delta,\epsilon\right)$ quantum relay channel code
	for all $\epsilon,\delta>0$ and sufficiently large $n$.

	The theorem below captures the main result of this chapter.
	
	\begin{theorem}[Partial decode-and-forward inner bound]		\label{thm:PDF-for-QRC}
		Let $\{ \rho_{x,x_1} \}$ be a cc-qq relay channel as in \eqref{eqn:relay-chan-def}.
		Then a rate $R$ is achievable, provided that
		the following inequality holds:
		\be
		R\leq\max_{
		p\!\left(  u,x,x_{1}\right)  }\min\left\{  
		{ 
		\begin{array}{l}
		I\!\left(  XX_{1};B\right)  _{\theta}, \\
		I\!\left( U;B_{1}|X_{1}\right)  _{\theta}+I\!\left(  X;B|X_{1}U\right)  _{\theta} \\
		\end{array}
		}
		\right\},
		\ee
		where the information quantities are with respect to the 
		classical-quantum state
		\begin{equation}
		\theta^{UXX_{1}B_{1}B}\equiv
		\sum_{x,u,x_{1}}
		p\!\left(  u,x,x_{1}\right)
		\left\vert u\right\rangle\!\!
		\left\langle u\right\vert ^{U}\otimes\left\vert x\right\rangle\!\!\left\langle
		x\right\vert ^{X}\otimes\left\vert x_{1}\right\rangle\!\!\left\langle
		x_{1}\right\vert ^{X_{1}}\otimes\rho_{x,x_{1}}^{B_{1}B}.\label{eq:code-state}%
		\end{equation}
	\end{theorem}
	
	Our code construction employs codebooks 
	$\{ x_1^n \}$,
	$\{ u^n \}$, and 	$\{ x^n \}$
	generated according to the distribution $p(x_1)p(u|x_1)p(x|u,x_1)$.
	%
	%
	We split the message for each block into two parts 
	$
	(m,\ell) \in \mathcal{M}\times\mathcal{L}$
	such that we have $R = R_m + R_\ell$.
	The relay fully decodes the message $\ell$ and re-encodes it
	directly (without using binning) in the next block. 
	The destination exploits a ``sliding-window'' decoding strategy
	\cite{carleial1982multiple,xie2005achievable} 
	by performing a collective measurement on two consecutive blocks. In this approach,
	the message pair $(m_j,\ell_j)$ sent during block $j$ is decoded 
	from the outputs of blocks $j$ and $j+1$,
	using an  ``{\sc and}-measurement.''
	%
	%

%

	%
	%
	%


	\section{Achievability proof}
										\label{sec:PDFproof}

	We divide the channel uses into many blocks and build codes
	in a randomized, block-Markov manner within each block. 
	The channel is used for $b$ blocks, each indexed by $j \in \{1,\ldots,b\}$.
	Our error analysis shows that:
	\begin{itemize}

	\item 
	The relay can decode  the message $\ell_{j}$ during  block $j$.

	\item 
	The destination can simultaneously decode $(m_j,\ell_j)$ 
	from a collective measurement on the output systems 
	of blocks $j$  and $j+1$.
	
	\end{itemize}

	The error analysis at the relay is similar
	to that of the Holevo-Schumacher-Westmoreland theorem \cite{H98,SW97}.
	The message $\ell_{j}$ can
	be decoded reliably 
	if the rate $R_{\ell}$ 
	obeys the following inequality:
	\be
	R_{\ell}\leq I\left(  U;B_{1}|X_{1}\right) _{\theta}.
	\label{eqn:bound-from-relay}
	\ee
	%

	The decoding at the destination is a variant of the quantum simultaneous  
	decoder from Theorem~\ref{thm:sim-dec-two-sender}.
	%
	To decode the message $(m_j,\ell_j)$, the destination performs a
	``sliding-window'' decoder, implemented 
	as an ``{\sc and}-measurement'' on the outputs of blocks $j$ and $j+1$.
	%
	This coding technique does not require binning at the relay or backwards decoding
	at the destination \cite{carleial1982multiple,xie2005achievable}.
	
	In this section, we give the details of the coding strategy
	and analyze the probability of error for the destination and the relay.
	


\textbf{Codebook construction}. 
Fix a code distribution $p(u,x,x_1)=$ $p(x_1)p(u|x_1)p(x|x_1,u)$ and
 independently generate a different codebook for each block $j$
 as follows:
\begin{itemize}
\item 
Randomly and independently generate $2^{nR_{\ell}}$
sequences $x_{1}^{n}\!\left(  \ell_{j-1}\right)$, $\ell_{j-1} \in\left[1: 2^{nR_{\ell}}\right]$,
according to  $\prod\limits_{i=1}^{n}p
\!\left(  x_{1i}\right)$.

%

\item 
For each $x_{1}^{n}\!\left(  \ell_{j-1}\right)$, 
randomly and independently generate 
$2^{nR_{\ell}}$  sequences $u^{n}\!\left(  \ell_{j},\ell_{j-1}\right)$,
$\ell_{j} \in\left[1: 2^{nR_{\ell}}\right]$
according to 
$\prod\limits_{i=1}^{n}p
\left(  u_{i}|x_{1i}\!\left(  \ell_{j-1}\right)
\right)  $. 

\item
For each $x_{1}^{n}\!\left(  \ell_{j-1}\right)$
and each corresponding $u^{n}\!\left(  \ell_{j},\ell_{j-1}\right)$,
randomly and independently generate
$2^{nR_{m}}$  sequences 
$x^{n}\!\left(  m_{j},\ell_{j},\ell_{j-1}\right) $, $m_{j} \in\left[1: 2^{nR_{m}}\right]$,
according to the distribution:
$\prod\limits_{i=1}^{n}
p
\big(  x_{i}  | x_{1i}\!\left(  \ell_{j-1}\right)  ,u_{i}\!\left(\ell_{j},\ell_{j-1}\right)  \big)$.

	%

\end{itemize}

\medskip
\textbf{Transmission}. The transmission of $(m_{j},\ell_{j})$ to the destination
happens during blocks $j$ and $j+1$ as illustrated in Figure~\ref{fig:partial-dec-fwd}.
At the beginning of block $j$, we assume
that the relay has correctly decoded the message $\ell_{j-1}$. During
block $j$, the source inputs the new messages $m_{j}$ and $\ell_{j}$, and the
relay forwards the old message $\ell_{j-1}$. That is, their inputs to the channel
for block $j$ are the codewords $x^{n}\!\left(  m_{j},\ell_{j},\ell_{j-1}\right)  $
and $x_{1}^{n}\!\!\left(  \ell_{j-1}\right)  $, leading to the following state at the
channel outputs:%
\[
\rhojBOTH \equiv
\rho_{x^{n}\left(  m_{j},\ell_{j},\ell_{j-1}\right)  ,x_{1}^{n}\left(
\ell_{j-1}\right)  }^{B_{1\left(  j\right)  }^{n}B_{\left(  j\right)  }^{n}}.
\]

During block $j+\!1$, the source transmits $(m_{j+1},\ell_{j+1})$ given $\ell_j$, whereas the
relay sends $\ell_{j}$, leading to the state:
\[
\rhojjBOTH \equiv 
\rho_{x^{n}\left(  m_{j+1},\ell_{j+1},\ell_{j}\right)  ,x_{1}^{n}\left(
\ell_{j}\right)  }^{B_{1\left(  j+1\right)  }^{n}B_{\left(  j+1\right)  }^{n}}.
\]
Our shorthand notation is such that the states are identified
by the messages that they encode, and the codewords are
implicit.
%

	%
	%
	%
	%

		\begin{figure}[htb]
		
		\begin{center}
		\subfigure[During block 2, the relay will transmit its codeword $x_{1(2)}^n(\texttt{o})$.
				We assume ``\texttt{o}'' was correctly decoded by the relay during the previous block.
				The source transmits a codeword $x_{(2)}^n(\texttt{n},\texttt{s},\texttt{o})$.
				]{
		\begin{tikzpicture}[node distance=4cm,>=stealth',bend angle=45,auto,scale=0.73, every node/.style={scale=0.73}]
		  \begin{scope}
			\node [] (RCTx) [ label=left:S,yshift=-20mm   ]                            {$x_{(2)}^n(\texttt{n},\texttt{s},\texttt{o})$};
			\node [] (RCRERx) [ label=above:$\ \ \ \ \textrm{R}$, above right of=RCTx,xshift=-9mm ]		{$Y^n_{1(2)}$:}
				edge  [pre]             		node[swap]	  	{$\texttt{s}$}	(RCTx);
			\node [] (RCRETx) [above right of=RCTx,xshift=+1mm ]		{$x_{1(2)}^n(\texttt{o})$};
			\node [] (RCRx)  [label=right:D,below right of=RCRERx,xshift=-5mm]	{$Y^n_{(2)}$}
				edge  [pre]             		node	  	{$\texttt{n},\texttt{s},\texttt{o}$}	(RCTx)
				edge  [pre]				node[swap]  	{$\texttt{o}$}	(RCRETx);
		  \end{scope}
		  \begin{pgfonlayer}{background}
		    \filldraw [line width=4mm,join=round,black!10]
		      ([xshift=+1mm,yshift=+2mm]RCRERx.south -| RCTx.north) rectangle ([xshift=+2.5mm]RCRx.south -| RCRx.west);
		  \end{pgfonlayer}  
		\end{tikzpicture}
		} \ \ \ \ \ \ 
		\subfigure[During block 3, the relay will transmit its codeword $x_{1(3)}^n(\texttt{s})$,
				which encodes the message $\ell_2=$``\texttt{s}'' transmitted by the source during block~2.
				The source transmits the codeword $x_{(3)}^n(\texttt{t},\texttt{i},\texttt{s})$.
                ]{
		\begin{tikzpicture}[node distance=4cm,>=stealth',bend angle=45,auto,scale=0.73, every node/.style={scale=0.73}]
		  \begin{scope}
			\node [] (RCTx) [ label=left:S,yshift=-20mm   ]                            {$x_{(3)}^n(\texttt{t},\texttt{i},\texttt{s})$};
			\node [] (RCRERx) [ label=above:$\ \ \ \ \textrm{R}$, above right of=RCTx,xshift=-9mm ]		{$Y^n_{1(3)}$:}
				edge  [pre]             		node[swap]	  	{$\texttt{i}$}	(RCTx);
			\node [] (RCRETx) [above right of=RCTx,xshift=+1mm ]		{$x_{1(3)}^n(\texttt{s})$};
			\node [] (RCRx)  [label=right:D,below right of=RCRERx,xshift=-5mm]	{$Y^n_{(3)}$}
				edge  [pre]             		node	  	{$\texttt{t},\texttt{i},\texttt{s}$}	(RCTx)
				edge  [pre]				node[swap]  	{$\texttt{s}$}	(RCRETx);
		  \end{scope}
		  \begin{pgfonlayer}{background}
		    \filldraw [line width=4mm,join=round,black!10]
		      ([xshift=+1mm,yshift=+2mm]RCRERx.south -| RCTx.north) rectangle ([xshift=+2.5mm]RCRx.south -| RCRx.west);
		  \end{pgfonlayer}  
		\end{tikzpicture}
		}
		\end{center}
		
		\caption{	Information flow in the relay network during the second and third transmission blocks
				of the string ``\texttt{co ns ti tu ti on}'' when using the partial decode-and-forward strategy.
				The messages for each block (two characters) are encoded by the Sender using 
				a codebook $x^n(m_j,\ell_j, \ell_{j-1})$ during block $j$.
				The messages pairs $(m_j,\ell_j)$ for the seven uses of the channel are:
				$\{ (\texttt{c},\texttt{o}), (\texttt{n},\texttt{s}), (\texttt{t},\texttt{i}), (\texttt{t},\texttt{u}), (\texttt{t},\texttt{i}),  (\texttt{o},\texttt{n}),
				(\emptyset,\emptyset) \}$
				The source codebook depends on the current message pair $(m_j,\ell_j)$ as well
				as the message $\ell_{j-1}$ of the previous block,
				%
				so the transmitted codewords during the seven blocks are:
				{ \footnotesize
				$\{ x_{(1)}^n(\texttt{c},\texttt{o},\emptyset), 
				x_{(2)}^n(\texttt{n},\texttt{s},\texttt{o}), 
				x_{(3)}^n(\texttt{t},\texttt{i},\texttt{s}), 
				x_{(4)}^n(\texttt{t},\texttt{u},\texttt{i}), 
				x_{(5)}^n(\texttt{t},\texttt{i},\texttt{u}),  
				x_{(6)}^n(\texttt{o},\texttt{n},\texttt{i}),
				x_{(7)}^n(\emptyset,\emptyset,\texttt{n}) \}$} and
				$\{ x_{1(1)}^n(\emptyset), 
				x_{1(2)}^n(\texttt{o}), 
				x_{1(3)}^n(\texttt{s}), 
				x_{1(4)}^n(\texttt{i}), 
				x_{1(5)}^n(\texttt{u}),  
				x_{1(6)}^n(\texttt{i}),
				x_{1(7)}^n(\texttt{n}) \}$.
		}
		\label{fig:partial-dec-fwd}

		\end{figure}

	\subsection{Decoding at the destination}
	
	%
	We now determine a decoding POVM\ that the destination can perform on the output
	systems spanning blocks $j$ and $j+1$. 
	The destination is trying to recover messages $\ell_{j}$ and $m_{j}$ given knowledge
	of $\ell_{j-1}$.

	First let us consider forming decoding operators for block $j+1$.
	Consider the state obtained by tracing over the systems $X$, $U$, and $B_{1}$
	in (\ref{eq:code-state}):%
	\[
	\theta^{X_1B} = 
	\sum_{x_{1}}p\!\left(  x_{1}\right)  \left\vert x_{1}\right\rangle\!\! \left\langle
	x_{1}\right\vert ^{X_{1}}\otimes \tau_{x_{1}}^{B},
	\]
	where $
	\tau_{x_{1}}^{B}\equiv\sum_{u,x} p\!\left(x|x_{1},u\right) p\!\left(  u|x_{1}\right)    \rho_{x,x_{1}}^{B}$.
	Also, let $\bar{\tau}^{B}$ denote the following state:
	$
	\bar{\tau}^{B}\equiv\sum_{x_{1}}p\!\left(  x_{1}\right)  \tau_{x_{1}}^{B}.
	$
	Corresponding to the above states are conditionally typical projectors of the
	following form:%
	\begin{align*}
	\PIjjRXlj & \equiv
	\Pi_{\tau_{x_{1}^{n}(  \ell_{j})}, \delta}^{B_{\left(  j+1\right)  }^{n}}, \qquad 
	\PIjjRXavg   \equiv \Pi_{\bar{\tau}^{\otimes n} \!, \delta}^{B_{\left(  j+1\right)  }^{n}},
	\end{align*}
	which we combine to form the positive operator:
	\begin{align}
		\PRXjj
		&\equiv
		\PIjjRXavg \ \PIjjRXlj \ \PIjjRXavg, \label{eqn:PRX-def-II}
	\end{align}		
	that acts on the output systems $B^{n}_{(j+1)}$ of block $j+1$.

	\smallskip

	Let us now form decoding operators for block $j$. 
	Define the conditional typical
	projector for the state
	$\rhojRX$ 
	as%
	\be
	\PIjRXmjlj
	\equiv \ 
	\Pi_{\rho_{x^{n}(  m_{j},l_{j},l_{j-1})  ,x_{1}^{n}(l_{j-1}), \delta  }}^{B_{\left(  j\right)  }^{n}}.
	\ee	
	The state obtained from
	\eqref{eq:code-state} by tracing over $X$ and $B_{1}$ is
	\[
	\theta^{UX_1B}=
	\sum_{u,x_{1}}p\!\left(  u|x_{1}\right)  p\!\left(  x_{1}\right)  \left\vert
	u\right\rangle\!\! \left\langle u\right\vert ^{U}\otimes\left\vert x_{1}
	\right\rangle \!\!\left\langle x_{1}\right\vert ^{X_{1}}\otimes\bar{\rho}_{u,x_{1}}^{B},
	\]
	where $\bar{\rho}_{u,x_{1}}^{B}\equiv\sum_{x}p\!\left(  x|x_{1},u\right)  \rho_{x,x_{1}}^{B}$.
	We can trace out over $U$ as well to obtain the doubly averaged state
	$\dbar{\rho}_{x_{1}}^{B}\equiv\sum_{u,x}p\!\left(  x|x_{1},u\right)p\!\left(  u|x_{1}\right)  \rho_{x,x_{1}}^{B}$.
	
	The following conditionally typical projectors  will be
	used in the decoding: 
	\begin{align*}
	\PIjRXlj 
	& \equiv
	 \Pi_{\bar{\rho}_{u^{n}(  l_{j},l_{j-1})  ,x_{1}^{n}(	l_{j-1})}\!,\delta  }^{B_{\left(  j\right)  }^{n}}, \ \ \ \ \ \ 
	\PIjRXavg
	 \equiv
	\  \Pi_{\dbar{\rho}_{x_{1}^{n}(  l_{j-1})  },\delta}^{B_{\left(  j\right)  }^{n}%
	}.
	\end{align*}
	We can then form a positive operator 
	``sandwich'':
	\begin{align}
	\!\!
	\PRXj 
	\!\!\!   &\equiv 
	\PIjRXavg \PIjRXlj \PIjRXmjlj \!\!\PIjRXlj \!\!\!\PIjRXavg\!\!\!. \label{eqn:PRX-def-I} 
	\end{align}
	Finally, we combine the positive operators from 
	\eqref{eqn:PRX-def-II} and \eqref{eqn:PRX-def-I} to form the
	``sliding-window'' positive operator:
	\be
		\PRX=\PRXj \otimes \PRXjj,
		\label{eqn:theANDP}
	\ee			
	from which we can build the destination's 
	measurement $\LAMRX$ using the square-root normalization. 
	This measurement is what we call the ``{\sc and}-measurement.''
	
	%

	%
	%

\medskip
\textbf{Error analysis at the destination}. 
	In this section, we prove that the destination can correctly 
	decode the message pair $(m_j,\ell_j)$ by employing the measurement 
	$\{\LAMRX\}$ on the output state $\rhoRX$ spanning blocks $j$ and $j+1$.
	%
	%
	%
	%
	%
	%
	%
	%
	The 
	average probability of error for the destination is given
	in \eqref{eqn:avgpe-Rx}.
	%
	 For now, we consider
	the error analysis for a single message pair $(m_j,\ell_j)$:
	%
	\begin{align}
	 \bar{p}_e^{D}	
	 & \equiv
	\text{Tr}\!\left[ 
		\left(I   - \!\LAMRX \right)  
		\rhoRX
	  \right]. 	\nonumber \\
	& \leq  2\;
	\text{Tr}\left\{  \left(  I 
	-\PRX
	\right)  \ \ 
	\rhoRX
	\right\}   \tag{\textrm{I}}
	\\
	&  \qquad +4 \!\!\!\!\!\!\!\!\!\! 
	\sum_{\left(  \ell_{j}^{\prime},m_{j}^{\prime}\right)  \neq\left(\ell_{j},m_{j}\right)  }
	\!\!\!\!\!\!\!
	\text{Tr}
	\left\{  
	\PRXljprmjpr \ \
	\rhoRX
	\right\}\!,  \tag{\textrm{II}}
	\end{align}
	where we used the Hayashi-Nagaoka inequality (Lemma~\ref{lem:HN-inequality})
	to decompose the error operator $(I-\LAMRX)$ into two
	components:
	(\textrm{I}) a term related to the probability that the correct detector
	 does not ``click'': $(  I-\PRX )$,
	and (\textrm{II}) another term related to the probability that a wrong detector ``clicks'': 
	$\sum_{\left(  \ell_{j}^{\prime},m_{j}^{\prime}\right)} \PRXljprmjpr$,
	$\left(  \ell_{j}^{\prime},m_{j}^{\prime}\right) \neq \left(  \ell_{j},m_{j}\right)$.
	These two errors are analogous to the classical error events 
	in which an output sequence $y^n$ is either not jointly typical with the transmitted codeword
	 or happens to be jointly typical with another codeword.

	We will bound the expectation 
	of the average probability of error $\ExpALL\!\!\left\{ \bar{p}_e^{D} \right\}$
	by bounding the expectation of the average probability 
	for the two error terms: $\ExpALL\!\!\left\{  (\textrm{I}) \right\}$
	and $\ExpALL\!\!\left\{  (\textrm{II}) \right\}$.
	
%
%

	The first term (\textrm{I}) is bounded by using the properties of typical projectors and 
	the operator union bound from Lemma~\ref{lem:operator-union-bound}, 
	which allows us to analyze the errors for the two blocks separately.
	Because $0\leq \PRXj \leq I$ and $0\leq \PRXjj \leq I$, we have:
	\be
		\left(  I-\PRXj \!\!\!\otimes \!\PRXjj \right) 
		 \leq
			\left(  I-\PRXj  \right) 
			+
			\left(  I- \PRXjj \right).
		\label{PDF-union-bound}
	\ee

	We use the definition of $\PRX$ from \eqref{eqn:theANDP} and 
	the inequality \eqref{PDF-union-bound} to obtain:
	\begin{align*}
	&\!\!\!\!\!\!\!\text{Tr}\!\left[  
		\left(  I-\PRX \right)  \ 
		\rhoRX  
	\right]  
	\\[2mm]
	& =
		\text{Tr}\!\left[  
			\left(  I-\PRXj \!\!\!\otimes \!\PRXjj \right) 
			\rhoRX  
		\right]  \\[3mm]
	& \leq
		\underbrace{
		\text{Tr}\!\left[  
			\left(  I-\PRXj  \right) 
			\rhojRX 
		\right] 
		}_{\alpha}
		\underbrace{
		\text{Tr}\!\left[  
			\rhojjRX
		\right]
		}_{=1}	+ 
		 \\
		& \qquad \qquad \qquad \qquad + 
		\underbrace{
		\text{Tr}\!\left[  
			\rhojRX 
		\right] 
		}_{=1}
		\underbrace{			
		\text{Tr}\!\left[  
			\left(  I- \PRXjj \right) 
			\rhojjRX  
		\right]
		}_\beta,
	\end{align*}%
	where we defined the error terms $\alpha$ and $\beta$ associated
	with block $j$ and block $(j+1)$.

We proceed to bound the term $\beta$ as follows:
\begin{align*}
\beta 
&=	
	\text{Tr}\!\left[  
		\left(  I- \PRXjj \right) 
		\rhojjRX  
	\right]  \\
&=
	\text{Tr}\!\left[  
		\left(  I- \PIjjRXavg \PIjjRXlj  \PIjjRXavg \right) 
		\rhojjRX  
	\right]  \\
&=
	1 - 
	\text{Tr}\!\left[  
		\PIjjRXavg \PIjjRXlj  \PIjjRXavg \ 
		\rhojjRX  
	\right]  \\
&\leq
	1 -
	\text{Tr}\!\left[  
	  \PIjjRXlj \ 
		\rhojjRX  
	\right]  \\
&\qquad \qquad \qquad 
  +  \left\| \PIjjRXavg \rhojjRX \PIjjRXavg   - \rhojjRX \right\|_1,
\end{align*}
where the inequality follows from Lemma~\ref{eqn:tr-trick}. 
We will analyze the terms labeled $\alpha$ and $\beta$ separately.

By taking the expectation over the code randomness,
we obtain the upper bound:
\begin{align*}
\ExpALL\!\!\!\left\{ \beta \right\}
&=
	1 - 
	\ExpXone
	\text{Tr}\!\left[  
	  \PIjjRXlj \! 
		\ExpUXgXone \!\!\!\left\{ \rhojjRX   \right\}
	\right]  \\
& \qquad \quad 
  +  \ExpALL \left\| \PIjjRXavg \rhojjRX \PIjjRXavg   - \rhojjRX \right\|_1 \\
&\leq 
	1 - (1-\epsilon) + 2\sqrt{\epsilon}.
\end{align*}
The 
inequality follows from 
$\ExpUXgXone \!\left\{ \rhojjRX   \right\} = \tau_{\ell_j}$,
the 
properties of typical projectors:
$\ExpXone\text{Tr}[  \PIjjRXlj \ \tau_{\ell_j} ] \geq 1- \epsilon$,
$\text{Tr}[  \PIjjRXavg \ \bar{\tau} ] \geq 1- \epsilon$
and Lemma~\ref{lem:gentle-operator}.

The error term $\alpha$ is bounded in a similar fashion.

\medskip
We can split the sum in the second error term (\textrm{II})
as follows:
\begin{align*}
& \!\!\!\!\!\!\!
\sum_{\left(  \ell_{j}^{\prime},m_{j}^{\prime}\right)  \neq\left(\ell_{j},m_{j}\right)  }
\!\!\!\!\!
\text{Tr}\!
\left[
	\PRXljprmjpr \ \ \rhoRX
\right] \\
& \ \ \ \ = 
\underbrace{
\sum_{m_{j}^{\prime}\neq m_{j}}
\!\!\!\!
\text{Tr}\!
\left[
	\PRXljmjpr \ \ \rhoRX
\right] 
}_{(\textrm{A})}
\\
& \qquad \qquad + \!
\underbrace{
\sum_{l_{j}^{\prime}\neq l_{j},\ m_{j}^{\prime}}
\!\!\!\!\!
\text{Tr}\!
\left[ 
	\PRXljprmjpr \ \ \rhoRX
\right]
}_{(\textrm{B})}.
\end{align*}


We now analyze the two terms $(\textrm{A})$ and $(\textrm{B})$ separately.

\paragraph{Matching $\ell_j$, wrong $m_j$.} 
    Assuming $\ell_j$ is decoded correctly, we show that  
    the message $m_j$ will be decoded correctly
    provided $R_m < I(X;B|UX_1)=H(B|UX_1)-H(B|UXX_1)-\delta$.
	We will use the following properties of typical projectors:
	\begin{align}
	\PIjRXmjprlj \!\!\leq  & 2^{n[H(B|UXX_1)+\delta ]}  \rhojRXmjpr, 
			\label{eqn:for-proj-trick}\\
		  \!\!\!\!\PIjRXlj \! \rhobarRX \PIjRXlj   \!\! \leq & 2^{-n[H(B|UX_1)-\delta] }  \PIjRXlj. \label{eqn:for-proj-trick2}
	\end{align}
Consider the first term:%
{\allowdisplaybreaks
\begin{align*}
(\textrm{A}) &=   
\sum_{m_{j}^{\prime}\neq m_{j}}
\!\!\!
\text{Tr}\!
\left[
	\PRXljmjpr \ \ \rhoRX
\right] \\
& = \!\! 
	\sum_{m_{j}^{\prime}\neq m_{j}}
	\!\! \text{Tr}\!
	\left[  
	\left(  
	\PRXjmjpr \! \otimes \! \PRXjj
	\right)  
	\rhojRX \!\otimes \!\rhojjRX
	\right]  
	\\
&  \leq \!\!
	\sum_{m_{j}^{\prime}\neq m_{j}}
	\!\! \text{Tr} \! \left[  
	\PRXjmjpr 
	\! \otimes \! I^{B_{\left(  j+1\right)  }^{n}} 
	\ \ 
	\rhojRX \otimes \rhojjRX
	\right]  
	\\
&  = \!\!
	\sum_{m_{j}^{\prime}\neq m_{j}	}
	\!\! \text{Tr}\!
	\left[  
	\PRXjmjpr
	\ \
	\rhojRX
	\right]
	\\[-3mm]
&  =
	\!\!
	\sum_{m_{j}^{\prime}\neq m_{j}	}
	\!\!
	\text{Tr}\!
	\Bigg[
	\PIjRXavg
	\!
	\underbrace{
	\PIjRXlj
	\!\!\!
	\overbrace{\PIjRXmjprlj \!\!\!\!\! }^{\textrm{\dingone}} 
	\ 
	\PIjRXlj
	\!\!
	}_{\textrm{\dingtwo}}
	\PIjRXavg
	\rhojRX \!
	\Bigg]  
\end{align*}}%
We now upper bound expression \dingone using \eqref{eqn:for-proj-trick}
and take the conditional expectation with respect to $X^n$:
\[
	\ExpXgUXone\!\!\left\{ \rhojRXmjpr \right\} = \rhobarRX,
\]
which is independent of the state $\rhojRX$ since $m^\prime_j \neq m_j$.
The resulting expression in \dingtwo has the state 
$\rhobarRX$ sandwiched between its
typical projector on both sides, and so we can use \eqref{eqn:for-proj-trick2}.
After these steps, we obtain the upper bound:
\begin{align}
 \!\!\!\! \ExpXgUXone\!\!\left\{  (\textrm{A}) \right\}  
&\leq \ 
2^{n\left[   H\left(  B|XUX_{1}\right)  +\delta\right]  }  \ 
2^{-n\left[  H\left(B|UX_{1}\right)  -\delta\right]  } \times \nonumber  \\[-3mm] 
& \quad \quad
\times \ExpXgUXone
	\!\!\sum_{m_{j}^{\prime}\neq m_{j}	}\!\!\!
	\text{Tr}\!
	\left[
		\PIjRXavg
		 \PIjRXlj
		\PIjRXavg
	\
		\rhojRX 	
	\right]   \nonumber  \\[4mm]
&  \leq \ 
	2^{n\left[  H\left(  B|XUX_{1}\right)  +\delta\right]  }
	2^{-n\left[  H\left(B|UX_{1}\right)  -\delta\right]  } 
	\sum_{m_{j}^{\prime}\neq m_{j}	} \text{Tr}\!\left[ \ \rhojRX \ \right]
	 \nonumber  \\[2mm]
& \leq \ 
	|\mcal{M}| \  
	2^{-n\left[  I\left(  X;B|UX_{1}\right)-2\delta\right]  }
	.
	\label{eqn:bound-on-Rm}
\end{align}
The second inequality follows because each operator
inside the trace is positive semidefinite and less than or equal to the identity.

\paragraph{Wrong $\ell_j$ (and thus wrong $m_j$).}
	We obtain the requirement  
	$R \equiv R_\ell + R_m \leq I(XX_1;B)=I(X_1;B)+I(UX;B|X_1)$
	from the ``{\sc and}-measurement'' and the following inequalities:
	\begin{align}
		\ExpALL \Tr[ \PIjjRXlj ] & \leq 2^{n[H(B|X_1)+\delta ]}, \\[1mm]
		 \PIjjRXavg \  \bar{\tau}  \ \ \PIjjRXavg & \leq 2^{-n[H(B)-\delta] } \PIjjRXavg, \\[3mm]
		\ExpALL \Tr[ \PIjRXmjlj ] & \leq 2^{n[H(B|UXX_1)+\delta ]},  \label{eqn:rhobar-proj-size}\\[1mm]
		 \PIjRXavg \  \rhodbarRX  \ \PIjRXavg & \leq 2^{-n[H(B|X_1)-\delta] } \PIjRXavg.
		 \label{eqn:rhodbar-sandwich} 
	\end{align}	
%


Consider the following term:%
\begin{align*}
\!\!(\textrm{B}) 
&=	
	\sum_{\ell_{j}^{\prime}\neq \ell_{j},m_{j}^{\prime}}
	\!\!\!\!\!
	\text{Tr}\!
	\left[ 
		\PRXljprmjpr \ \ \rhoRX
	\right] \\
& =   
	\!\!
	\sum_{\ell_{j}^{\prime}\neq \ell_{j},m_{j}^{\prime}}
	\!\!\!\!\text{Tr}\!
	\left[  
	\left(  
	\PRXjljprmjpr \!\! \otimes \! \PRXjjpr
	\right)  
	\rhoRX
	\right]  
	\\	
& = 
	\sum_{\!\! \ell_{j}^{\prime}\neq \ell_{j},m_{j}^{\prime} \ }
	\!\!
	\underbrace{
	\!\!\!\!\text{Tr}\!
	\left[  
		\PRXjljprmjpr 
		\rhojRX
	\right] 
	}_{(\textrm{B1})}
	\underbrace{
	\!\text{Tr}\!
	\left[  	
		\PRXjjpr
		 \rhojjRX
	\right]  
	}_{(\textrm{B2})}.
\end{align*}
We want to calculate the expectation of the term $(\textrm{B})$ with respect to
the code randomness $\ExpALL$.   
The random variables in different blocks are independent, and
so we can analyze the expectations
of the factors $(\textrm{B1})$ and $(\textrm{B2})$ separately. 

Consider first the calculation in block $j$, which leads
to the following bound on the expectation of the factor $(\textrm{B1})$:

{\allowdisplaybreaks 
\begin{align*}
\ExpALL \!  \left\{ (\textrm{B1}) \right\}  
&= 
	\ExpALL \! \left\{ 
	\text{Tr}\!
	\left[  
		\PRXjljprmjpr 
		\rhojRX
	\right] 
	\right\} \\
& = \ExpALL
	\text{Tr}\left[
	\begin{array}[c]{l}%
		\PIjRXljpr
		\PIjRXmjprljpr
		\PIjRXljpr \times \\
		\qquad  \qquad \quad
		\PIjRXavg
		\rhojRX 
		\PIjRXavg
	\end{array}
	\right]
	\\
& = 
\ExpXone
	\text{Tr}\left[
	\begin{array}[c]{l}
		\displaystyle \ExpUXgXone \{
		\PIjRXljpr
		\PIjRXmjprljpr
		\PIjRXljpr \} \times \\
		\qquad  \quad 
		\PIjRXavg
		\underbrace{
		\displaystyle \ExpUXgXone \!\! \left\{ \
		\rhojRX  \
		 \right\}
		 }_{\ \textrm{\dingthree}  }
		  \PIjRXavg
	\end{array}
	\right]
	\\
& = 
\ExpXone
	\text{Tr}\left[
	\begin{array}
	[c]{l}%
	\displaystyle \ExpUXgXone \{
	\PIjRXljpr
	\PIjRXmjprljpr
	\PIjRXljpr \} \times \\
	\qquad  \qquad  \qquad
	\underbrace{
	\PIjRXavg
	\rhodbarRX
	 \PIjRXavg
	 }_{\textrm{\ \dingfour}}
	\end{array}
	\right]
	\\	
& \leq 
2^{-n\left[  H\left(  B|X_{1}\right)  -\delta\right]  }
\!\!\!\!\!\ExpALL \!\!\!\!\!\!\!
	\text{Tr}\!\left[ \!\!
	\begin{array}[c]{l}%
	\PIjRXljpr \!\!
	\PIjRXmjprljpr \!\!
	\PIjRXljpr  \!
		\PIjRXavg
	\!\!\!\!
	\end{array}
	\right]
	\\	
& \leq 
2^{-n\left[  H\left(  B|X_{1}\right)  -\delta\right]  }
\!\!\!\!\!\ExpALL \!\!\!
	\text{Tr}\left[
	\PIjRXmjprljpr
	\right]
	\\	
& \leq
	2^{-n\left[  H\left(  B|X_{1}\right)  -\delta\right]  }
	\!\!\!\ExpALL \!\!\!
	2^{n\left[H\left(  B|X_{1}UX\right)  +\delta\right]  }\\
& =
	2^{-n\left[  I\left(  UX;B|X_{1}\right)  -2\delta\right]  }.
\end{align*}}%
The result of the expectation in \dingthree is $\rhodbarRX$,
and we can bound the expression in  \dingfour using \eqref{eqn:rhodbar-sandwich}.
The first inequality follows because all the other terms in the trace are positive semidefinite 
operators less than or equal to the identity. The final inequality follows from \eqref{eqn:rhobar-proj-size}.

Now we consider the expectation of the second term:
{\allowdisplaybreaks 
\begin{align*}
 \ExpALL  \!  \left\{ (\textrm{B2}) \right\}  
& = 
	\ExpALL \!\!
	 \left\{  \text{Tr}\!
	 \left[ 
	 \PRXjjpr \ 
	\rhojjRX
	\right] \right\}    \\
&  =\text{Tr}\!\left[  
\ExpALL
\left\{  
 \PRXjjpr 
 \right\} \ 
\ExpALL
\left\{  
\rhojjRX
\right\}  \right]  \\
&  =\text{Tr}\!\left[  
	\ExpALL
	\left\{  
	 \PRXjjpr 
	 \right\} \ 
	\bar{\tau}^{\otimes n}
	\right]  \\
&  =
	\ExpALL
	\text{Tr}\!\left[ 
	\PIjjRXavg
	\PIjjRXljpr \
	\PIjjRXavg \
	\bar{\tau}^{\otimes n}
	\right]  \\
&  =
	\ExpALL
	\text{Tr}\!\left[  
	\PIjjRXljpr  \
	\PIjjRXavg 
	\bar{\tau}^{\otimes n}
	\PIjjRXavg
	\right]  \\
&   \leq
	2^{-n\left[  H\left(  B\right)  -\delta\right]  }
	\ExpALL
	\text{Tr}\!\left[
	\PIjjRXljpr \
	\PIjjRXavg
	\right]  \\
&  \leq
	2^{-n\left[  H\left(  B\right)  -\delta\right]  }
	2^{n\left[  H\left(B|X_{1}\right)  +\delta\right]  }
  =
	2^{-n\left[  I\left(  X_{1};B\right)  -2\delta\right]  }.
\end{align*}}

Combining the upper bounds on $(\textrm{B1})$ and $(\textrm{B2})$  
gives our final upper bound: 
\begin{align}
\ExpALL \!  \left\{ (\textrm{B}) \right\}   
& = 
\ExpALL 
\sum_{\ell_{j}^{\prime}\neq \ell_{j},m_{j}^{\prime}}
(\textrm{B1}) \times (\textrm{B2}) \nonumber \\
& \leq
	\sum_{\ell_{j}^{\prime}\neq \ell_{j},\ m_{j}^{\prime}}
	2^{-n\left[  I\left(  UX;B|X_{1}\right) -2\delta\right]  } 
	\times
	2^{-n\left[  I\left(X_{1};B\right)  -2\delta\right]  }	
	\nonumber  \\
& \leq 
	|\mcal{L}||\mcal{M}| \  2^{-n\left[  I\left(  X_{1};B\right)  +I\left(  UX;B|X_{1}\right)  -4\delta\right]  }.
	\label{eqn:bound-on-R}
\end{align}
By choosing the size of 
message sets to satisfy equations 
\eqref{eqn:bound-on-Rm} and 
\eqref{eqn:bound-on-R},
the expectation of the average probability of error at the destination
becomes arbitrarily small for $n$ sufficiently large.

%

%
%


\subsection{Decoding at the relay}

In this section we give the details of the POVM construction
and the error analysis for the decoding at the relay.

\noindent
\textbf{POVM\ Construction}. 
	%
	%
	During block $j$, the relay wants to decode the 
	message $\ell_j$ encoded in $u^n(\ell_j,\ell_{j-1})$,
	given the knowledge of the message $\ell_{j-1}$ from the previous block.
	%
	Consider the state obtained by tracing over the systems $X$
	and $B$ in (\ref{eq:code-state}):
	\[
	\theta^{UX_1B_1} =
	\sum_{u,x_{1}}p\!\left(  u|x_{1}\right)  p\!\left(  x_{1}\right)  \left\vert
	u\right\rangle \!\!\left\langle u\right\vert ^{U}\otimes\left\vert x_{1}%
	\right\rangle \!\!\left\langle x_{1}\right\vert ^{X_{1}}\otimes\sigma_{u,x_{1}%
	}^{B_{1}},
	\]
	where $
	\sigma_{u,x_{1}}^{B_{1}}\equiv\sum_{x}p\!\left(  x|x_{1},u\right)  \text{Tr}_{B}\!\left[  \rho_{x,x_{1}}^{B_{1}B}\right]$.
	Further tracing over the system $U$ leads to the state%
	\[
	\theta^{X_1B_1} =
	\sum_{x_{1}}p\!\left(  x_{1}\right)  \left\vert x_{1}\right\rangle \!\!\left\langle
	x_{1}\right\vert ^{X_{1}}\otimes\bar{\sigma}_{x_{1}}^{B_{1}},
	\]
	where $
	\bar{\sigma}_{x_{1}}
	\equiv\sum_{u}p\!\left(  u|x_{1}\right)  \sigma_{u,x_{1}%
	}^{B_{1}}$.
	Corresponding to the above conditional states are conditionally typical
	projectors of the following form%
	\begin{align*}
	\PIRElj 
	& \equiv
	  \Pi_{\sigma_{u^{n}\left(  \ell_{j},\ell_{j-1}\right)  ,x_{1}^{n}\left(
	\ell_{j-1}\right)  }}^{B_{1\left(  j\right)  }^{n}}, \qquad 
	\PIREavg 
	\equiv 
	 \Pi_{\bar{\sigma}_{x_{1}^{n}\left(  \ell_{j-1}\right)  }}^{B_{1\left(  j\right)
	}^{n}}.
	\end{align*}
	The relay constructs a square-root measurement $\{ \Gamma_{\ell_j|\ell_{j-1}} \}$
	using the following positive operators:
	\be
	\PRE \equiv
	\PIREavg \PIRElj \PIREavg.
	\ee
	
\noindent
\textbf{Error analysis}. 
	In this section we show that during block $j$ the relay will be able to decode the message $\ell_{j}$ 
	from the state $\rhoFULLatRE$,
	provided 
	the rate $R_\ell < I(U;B_1|X_1)=H(B_1|X_1)-H(B_1|UX_1)-\delta$.
	The bound follows from the following properties of typical projectors:
	\begin{align}
		\Tr[ \PIRElj ] &\leq 2^{n[H(B_1|UX_1)+\delta ]},  \label{RE-typ-1}\\
		 \PIREavg  \bar{\sigma}  \ \PIREavg & \leq 2^{-n[H(B_1|X_1)-\delta] } \PIREavg. \label{RE-typ-2}.
	\end{align}

	Recall that the average probability of error at the relay is given by:
	\begin{align*}
		\bar{p}_{e}^{R}
		\equiv 
			\frac{1}{|\mathcal{L}|}\sum_{\ell_j}					
			\text{Tr}\!
			\left\{  
				\left(  I-\Gamma^{B_{1(j)}^{n}} _{\ell_j|\ell_{j-1}} \right)
				\rhojRE
			\right\}.
	\end{align*}
	
	We consider the probability of error for a single message $\ell_j$
	and begin by applying the Hayashi-Nagaoka operator inequality
	(Lemma~\ref{lem:HN-inequality}) to split the error
	into two terms:
	\begin{align*}
	 \bar{p}_e^{R}
	 & \equiv
	\text{Tr}\!\left[ 
		\left(I   - \!\GAMRE \right)  
		\rhojRE
	  \right] \\
	& \leq  2
	\underbrace{
	\text{Tr}\left[  \left(  I 
	-\PRE
	\right)   \ 
	\rhojRE 
	\right]  
	}_{(\textrm{I})}
	\  + \ 4 \!
	\underbrace{
	\sum_{\ell_{j}^{\prime}   \neq \ell_{j}  }
	\!\!
	\text{Tr}
	\left[  
	\PREpr
	\
	\rhojRE
	\right]
	}_{(\textrm{II})}
	\!.
	\end{align*}
	
	We will bound the expectation 
	of the average probability of error by bounding the 
	individual terms.
	We bound the first term as follows:
	\begin{align*}
	(\text{I})  
	&=	
		\text{Tr}\left[  
			\left(  I  -\PRE\right)   \ 
			\rhojRE 
		\right] \\
	&=
		\text{Tr}\!\left[  
			\left(  I- \PIREavg \PIRElj \PIREavg \right) 
			\rhojRE  
		\right]  \\
	&=
		1 - 
		\text{Tr}\!\left[  
			\PIREavg \PIRElj \PIREavg \ 
			\rhojRE  
		\right]  \\
	&\leq
		1 -
		\text{Tr}\!\left[  
		  \PIRElj \ 
			\rhojRE  
		\right] 
	  +  \left\| \PIREavg \rhojRE \PIREavg   - \rhojRE \right\|_1,
	\end{align*}
	where the inequality follows from Lemma~\ref{eqn:tr-trick}.

	By taking the expectation over the code randomness
	we obtain the bound
	{\allowdisplaybreaks 
	\begin{align*}
	 \!\! \ExpALL   \!\!(\text{I})  
	&=	
		1 - 
		\ExpUXone
		\text{Tr}\!\left[  
		  \PIRElj \! 
			\ExpXgUXone \!\!\!\left\{ \rhojRE   \right\}
		\right]  \\
	& \qquad 
	  +   \!\!\ExpALL \left\| \PIREavg \rhojRE \PIREavg   - \rhojRE \right\|_1 \\
	& =
		1 - 
		\ExpUXone
		\text{Tr}\!\left[  
		  \PIRElj \! 
		  	\sigma_{\ell_j,\ell_{j-1}}
		\right]  \\
	& \qquad  
	  +  \!\! \ExpALL \left\| \PIREavg \rhojRE \PIREavg   - \rhojRE \right\|_1 \\
	& \leq
		1 - 
		\ExpUXone
		\text{Tr}\!\left[  
		  \PIRElj \! 
		  	\sigma_{\ell_j,\ell_{j-1}}
		\right]  + 2 \sqrt{\epsilon}  \\
	&\leq 
		1 - (1-\epsilon) + 2\sqrt{\epsilon} = \epsilon + 2\sqrt{\epsilon}.
	\end{align*}
	}
	The first inequality follows from Lemma~\ref{lem:gentle-operator}
	and the property
	\be
		\ExpUXone
		 \text{Tr}\!\left[  \PIREavg \ \bar{\sigma} \right] \geq 1- \epsilon.
	\ee
	The second inequality follows from:
	\be
		\ExpUXone \text{Tr}\!\left[ 	
			\PIRElj  \sigma_{\ell_j,\ell_{j-1}} 
		 \right] 
		 \geq 1- \epsilon.
	\ee
	
	\medskip
	To bound the second term we proceed as follows:
	\begin{align*}
	 \ExpALL  \!  \left\{ (\textrm{II}) \right\}  
	 & = 
		\ExpALL \sum_{\ell_{j}^{\prime}   \neq \ell_{j}  }
		\!
		\text{Tr}
		\left[  
		\PREpr
		\
		\rhojRE
		\right] \\
	 & = 
		\ExpXone \sum_{\ell_{j}^{\prime}   \neq \ell_{j}  }
		\!\!
		\text{Tr}
		\left[  
		\ExpUXgXone \!\!\left\{
		\PREpr
		\right\}
		\!\!
		\ExpUXgXone \{
		\rhojRE
		\}
		\right] \\
	& = 
		\ExpXone
		\sum_{\ell_{j}^{\prime}   \neq \ell_{j}  }
		\!\!
		\text{Tr}
		\left[  
		\ExpUXgXone \!\!\left\{
		\PREpr
		\right\}
		\ 
		\bar{\sigma}_{|\ell_{j-1}} 
		\right].
	\end{align*}		
	The expectation can be broken
	up because $\ell_j^\prime \neq \ell_j$
	and thus the $U^n$ codewords are
	independent.
	We have also used 
	\be
		\ExpUXgXone \! \left\{
		\rhojRE
		\right\}
		=
		\bar{\sigma}_{|\ell_{j-1}}.
	\ee

	We continue by expanding the operator
	$\PREpr$ as follows:
	{\allowdisplaybreaks 
	\begin{align*}		
	 \hspace{0.17\textwidth}
	 & = \!\!
		\ExpALL
		\sum_{\ell_{j}^{\prime}   \neq \ell_{j}  }
		\!
		\text{Tr}\!
		\left[  
		\PIREavg \!\!\PIREljpr \!\!\PIREavg
		\bar{\sigma}_{|\ell_{j-1}} 
		\right] \\
	 & = \!\!
		\ExpALL
		\sum_{\ell_{j}^{\prime}   \neq \ell_{j}  }
		\!
		\text{Tr}\!
		\left[  
		\qquad  \PIREljpr 
		\underbrace{
		\PIREavg
		\bar{\sigma}_{|\ell_{j-1}}
		\PIREavg
		}_{\textrm{\dingfive}}
		\right] \\
	 & \leq \!\!
		\ExpALL
		\sum_{\ell_{j}^{\prime}   \neq \ell_{j}  }
		\!
		\text{Tr}
		\left[  
		 \PIREljpr  \
 	 	2^{-n[H(B_1|X_1)-\delta] }
		\PIREavg
		\right] \\
	 & \leq 
	 	2^{-n[H(B_1|X_1)-\delta] }
		\ExpALL
		\sum_{\ell_{j}^{\prime}   \neq \ell_{j}  }
		\!
		\text{Tr}\!\left[  
		 \PIREljpr   
		\right] \\
	 & \leq 
	 	2^{-n[H(B_1|X_1)-\delta] }
		\ExpALL
		\sum_{\ell_{j}^{\prime}   \neq \ell_{j}  }
		\!\!
		2^{n[H(B_1|UX_1)+\delta ]} \\
	 & \leq 
	 	| \mcal{L}| \
	 	2^{-n[I(U;B_1|X_1)-2\delta] }.
	\end{align*}}%
	The first inequality follows from using \eqref{RE-typ-2} on the expression \dingfive\!\!\!.
	The second inequality follows from the fact that $\PIREavg$ is a positive semidefinite
	operator less than or equal to the identity. More precisely we have
	\begin{align*}
		\text{Tr}\!\left[  \PIREljpr \PIREavg \right]
		& = 
			\text{Tr}\!\left[  \PIREljpr \PIREavg \PIREljpr  \right] \\
		& \leq 
			\text{Tr}\!\left[  \PIREljpr  I \ \PIREljpr \right] \\
		& = 
			\text{Tr}\!\left[  \PIREljpr \right].
	\end{align*}
	The penultimate inequality follows from
	 \eqref{RE-typ-1}.

	Thus if we choose
	$R_\ell \leq I(U;B_1|X_1)-3\delta$,
	we can make the expectation 
	of the average probability of error 
	at the relay vanish in the limit of
	many uses of the channel.

\medskip
\textbf{Proof conclusion}.
Note that the \emph{gentle operator lemma for ensembles} is used
several times in the proof.
First, it is used to guarantee that the effect of
acting with one of the projectors from the ``measurement sandwich'' 
does not disturb the state too much.
Furthermore, because each of the output blocks is operated
on twice: we depend on the gentle operator lemma 
to guarantee that the disturbance to the state during the first 
decoding stage is asymptotically negligible if the correct 
messages are decoded.

%


\section{Discussion}
	\label{sec:QRC-discussion}

	In this chapter, we established the achievability of the rates given by
	the partial decode-and-forward strategy,
	thus extending the study of classical-quantum channels
	to multi-hop scenarios.
	
	The new techniques from this chapter are the use of
	the \emph{coherent codebooks} and the  {\sc and}-measurement,
	which collectively decodes messages from two blocks of the output of the channel.

	We obtain the decoding-and-forward inner bound as a corollary of Theorem~\ref{thm:PDF-for-QRC}.
	\begin{corollary}[Decode-and-forward strategy for quantum relay channel]
	The rates $R$ satisfying 
	\be
		R \leq \max_{p(x,x_1)} \min\{ \ I(X,X_1;B)_{\theta}, \ I(X;B_1|X_1)_{\theta} \}
	\ee
	where the mutual information quantities are taken with respect to the 
	state
	\be
		\theta^{XX_1B_1B} 
		= 	\sum_{x,x_1} \underbrace{ p_{X|X_1}(x|x_1) p_{X_1}(x_1) }_{p_{X,X_1} }
			\ketbra{x}{x}^{X} \otimes \ketbra{x_1}{x_1}^{X_1}
			\otimes 
			\rho_{x,x_1}^{B_1B}.
	\ee
	are achievable for quantum relay channels by setting $X=U$ in Theorem~\ref{thm:PDF-for-QRC}.
	\end{corollary}

	Note also that 
	setting the $x_1$ to a fixed input in Theorem~\ref{thm:PDF-for-QRC} 
	would give us a quantum direct coding inner bound similar to the
	one from equation \eqref{eq:direct-coding-bd}.

	%
	An interesting open question is
	to determine a compress-and-forward strategy for the quantum setting.
	This could possibly involve combining results from quantum source coding and quantum 
	channel coding \cite{datta2011quantum,jointSRCandQCH}.
	%

	Another avenue for research would be to consider \emph{quantum communication} 
	and \emph{entanglement distillation} scenarios on a quantum relay network.
	%
	Further research in this area would have applications 
	for the design of quantum repeaters \cite{collins2005quantum, dutil2011multiparty}.

	%



\chapter{Bosonic interference channels}

									\label{chapter:bosonic}

	Optical communication links form the backbone of the 
	information superhighway which is the Internet.
	A single optical fiber can carry hundreds of gigabits 
	of data per second over long distances thanks to 
	the excellent light-transmission properties of glass materials.
	Free-space optical communication is also possible
	at rates of hundreds of megabits per second \cite{tolker2002orbit}.

	 
	An optical communication system consists of
	a modulated source of photons, the optical channel
	(or more generally the \emph{bosonic} channel, since photons are bosons),
	and an optical detector. 
	Figure~\ref{fig:real-world-MAC} on page \pageref{fig:real-world-MAC} 
	illustrates 
	an example of such a communication system.
%

	As information theorists, we are interested in 
	determining the ultimate limits on the rates for 
	communication over  such channels.
	%
	For each possible combination of the optical encoding
	and optical decoding strategies, we obtain a different
	communication model for which we can calculate the capacity.
	More generally, we are interested in the \emph{ultimate}
	capacity of the bosonic channel as permitted by the laws of physics.
	For this purpose we must optimize over all possible
	encoding and decoding strategies, both practical
	and theoretical.
	%

	In this chapter we present a quantum treatment of a
	free-space optical interference channel.
	We consider the performance of laser-light encoding (coherent light) in
	conjunction with three detection strategies: 
	(1) homodyne, (2) heterodyne, 	and (3) joint detection. 
	In Section~\ref{sec:bosonic-preliminaries}, we will introduce some basic notions 
	of quantum optics which are required for the remainder of the chapter.
	In Section~\ref{sec:bosonic-channels} we will discuss previous results
	on bosonic quantum channels and describe the known capacity formulas
	for point-to-point free-space bosonic channels for the three detection strategies.
	In Section~\ref{sec:bosonicIC} we define the bosonic interference channel
	model and calculate the capacity region for the special cases
	of  ``strong'' and ``very strong'' interference for each detection strategy.
 	We also establish the Han-Kobayashi achievable rate regions for 
	homodyne, heterodyne and joint detection.


%

\section{Preliminaries}
						\label{sec:bosonic-preliminaries}

\subsection{Gaussian channels}

	We begin by introducing some notation.
	Define the real-valued Gaussian probability density function
	with mean $\mu$ and variance $\sigma^2$ as follows:
	\be
		\mcal{N}_{\mathbb{R}}(x; \mu, \sigma^2)
			\equiv 		
		\frac{1}{\sqrt{2\pi\sigma^2}}\, e^{\frac{-(x-\mu)^2}{2\sigma^2}}
		\ \ \in \ \mcal{P}(\mathbb{R}).
	\ee
	Define also the circularly symmetric complex-valued Gaussian distribution
	\be
		\!\!\!\mcal{N}_{\mathbb{C}}(z; \mu, \sigma^2)
			\!\equiv \! 		
		\frac{1}{2 \pi \sigma^2 }\, e^{ \frac{ - \left| z-\mu \right|^2 }{ 2\sigma^2 }  }
		\!\!\equiv \!
		\frac{1}{\sqrt{2\pi\sigma^2}} e^{\!\frac{-(x-\operatorname{Re}\left\{  \mu \right\})^2}{2\sigma^2}}\!\!
		\frac{1}{\sqrt{2\pi\sigma^2}} e^{\!\frac{-(y-\operatorname{Im}\left\{  \mu \right\})^2}{2\sigma^2}}
		 \in \mcal{P}(\mathbb{C}),	
	\ee
	where we identify $z=x+iy$ and assume that the 
	variance parameter is real-valued $\sigma^2 \in \mathbb{R}$.
	Note also that in the complex-valued case,
	the quantity $\sigma^2$ represents the \emph{variance per real dimension};
	a variable $Z \sim \mcal{N}_{\mathbb{C}}(\mu, \sigma^2)$ will
	have variance $\operatorname{Var}\{Z\} \equiv \mathbb{E}_Z \left[ |Z-\mu|^2 \right] = 2\sigma^2$.

	The additive white Gaussian noise (AWGN) channel is a communication model
	where the input and output are continuous random variables
	and the noise is Gaussian.
	Let $X$ be the random variable associated with the input of the channel.
	Then the output variable $Y$ will be:
	\be
		Y = X + Z,
	\ee
	where $Z \sim \mcal{N}_{\mathbb{R}}(0, N )$ is a Gaussian random 
	variable with zero-mean and variance~$N$.
	As in the discrete memoryless case, we can use a 
	codebook $\{ x^n(m) \}$, $m \in [1:2^{nR}]$,
	with codewords generated randomly and independently 
	%
	according to a probability density function $\prod^n p_X(x)$.
	%
	%
	Furthermore we impose an \emph{average power constraint}
	on the codebook:
	\be
		\ExpX \left\{ \frac{1}{n}\sum_{i=1}^n X_i^2 \right\}  \ \leq \ P.
	\ee

	The channel capacity is calculated using the \emph{differential entropy}, $h\!:\! \mcal{P}(\mathbb{R}) \to \mathbb{R}$,
	which plays the role of the Shannon entropy for continuous random variables.
	We know from Shannon's channel capacity theorem (Theorem~\ref{thm:shannon-ch-cap})
	that a rate $R$ is achievable provided it is less than the mutual information
	of the joint probability distribution induced by the input distribution
	and the channel: $(X,Y) \sim p_X p_{Y|X}$.
	For any choice of input distribution $p_X$, the following rate is achievable:
	\begin{align}
	R \leq  I(X;Y) &= h(Y) - h(Y|X)  \nonumber \\
		  	&= h(Y)-h(X+Z|X) \nonumber \\
	 	  	&= h(Y)-h(Z|X) \nonumber \\
		  	&= h(Y) - h(Z). \label{eqn:gaussian-capacity}
	\end{align}
	The last equality follows because the noise $Z$ is assumed to be independent
	of the input $X$.
	It can be shown that a Gaussian distribution with variance $P$
	is the optimal choice of input distribution \cite{CT91}.
	%
	Furthermore, when we choose $X \sim  \mcal{N}_{\mathbb{R}}(0,P)$
	it is possible to compute the above expression exactly
	and obtain the capacity:
	\be
		C = \frac{1}{2}\log_2\left( 
			1 + \frac{P}{N}
			\right) \qquad \textrm{   [bits/use].}
		\label{eq:gauss-chan}
	\ee
	We will refer to the ratio $P/N$ as the \emph{signal to noise ratio}.
	We sometimes abbreviate this expression as: $\gamma(\textrm{SNR}) \equiv 
	\frac{1}{2}\log_2\left( 1 + \textrm{SNR} \right)$.
	The above formula is one of the great successes of classical
	information theory.
	
	The Gaussian multiple access channel is defined as:
	\be
		Y \ = \ \sqrt{\alpha} X_1 + \sqrt{\beta} X_2  + Z,
	\ee
	where $\alpha,\beta \in \mathbb{R}$ are the \emph{gain coefficients}
	and $Z \sim \mcal{N}_{\mathbb{R}}(0, N )$
	is an additive Gaussian noise term with average power $N$.
	When input power constraints $\ExpXone \left\{ \frac{1}{n}\sum_{i=1}^n X_{1i}^2 \right\}  \leq  P_1$
	and $\ExpXtwo \left\{ \frac{1}{n}\sum_{i=1}^n X_{2i}^2 \right\}  \leq  P_2$
	are imposed, the capacity region is given by:
	\be
		        		 C_{\textrm{MAC}} \!
		\equiv  \!
		\left\{ \!
		(R_1,R_2) \in \mathbb{R}_+^2 
		\left| 
		\begin{array}{rcl}					
		            R_1             &\!\!\leq& \!\!\!   I(X_1;Y|X_2) \!=\! \frac{1}{2}\!\log_2\!\left(  1 + \frac{\alpha P_1}{N} \right) \\
		            R_2             &\!\!\leq& \!\!\!   I(X_2;Y|X_1) \!=\! \frac{1}{2}\!\log_2\!\left(  1 + \frac{\beta P_2}{N} \right) \\
		            \!\! R_1+R_2   &\!\!\leq& \!\!\!   I(X_1X_2;Y) \ \!=\! \frac{1}{2}\!\log_2\!\left(  1 + \frac{\alpha P_1+\beta P_2}{N} \right) \!\!\!\!\!
		           \end{array}
		          \right.
		 \right\}\!\!.
		 \nonumber
	 \ee
	 Each of the constraints on the capacity region has an intuitive interpretation 
	 in terms of signal to noise ratios.
	 In this context, we also have the expression $I(X_1;Y)=\frac{1}{2}\log_2\left(  1 + \frac{\alpha P_1}{N+\beta P_2} \right)$,
	 in which the unknown codewords of the second transmitter are treated as contributing to the noise.

\subsection{Introduction to quantum optics}
											\label{sec:q-optics}

	Photons are excitations of the electromagnetic field.
	We say that photons are \emph{bosons} because they obey Bose-Einstein statistics:
	they are indistinguishable particles that are symmetric under exchange\footnote{
	The wave function describing two photons $p_1$ and $p_2$ is even under exchange
	of the two particles: $\psi(p_1,p_2)=\psi(p_2,p_1)$.}.
	Multiple bosons with the same energy can occupy the same quantum state.
	This is in contrast with \emph{fermions} which obey Pauli's exclusion principle.
	%
	%
	%
	Bosonic channels are channels in which the inputs and the outputs are bosons. 

	In this section, we will introduce some background material
	on quantum optics which is needed for the rest of the presentation
	in this chapter.
		Recall that the states of quantum systems are described by density operators $\sigma,\rho \in \mcal{D}(\mcal{H})$,
		where $\mcal{H}$ is a Hilbert space.
		Unitary quantum operations act by conjugation,
		so that by applying $U$ to $\sigma$ we obtain $\rho = U \sigma U^\dag$ as output.
		The expectation value of some operator $\hat{A}$ when the system
		is in the state $\rho$ is denoted $\langle \hat{A} \rangle = \Tr[ \hat{A} \rho ]$.
		
		Let $\rho_0 = \ketbra{0}{0}$ be the \emph{vacuum state} of one mode of 
		the electromagnetic field.
		We define $\hat{a}^\dag$ to be the \emph{creation operator} for that mode.
		Applying $\hat{a}^{\dag}$ to the vacuum state we obtain the first excited state:
		\be
			\ketbra{1}{1} = \hat{a}^{\dag} \ketbra{0}{0} \hat{a},
		\ee
		and this process can be iterated to create further excitations in the field.
		The Hermitian conjugate of the creation operator is the \emph{annihilation}
		operator which takes away excitations from the field.
		More generally, we have
		\begin{align}
			a \ket{n}  &=\sqrt{n} \; \ket{n-1}, \\
			a^\dag \ket{n}  &=\sqrt{n+1}\; \ket{n+1}. \\
		\end{align}%
		The state space $\ketbra{0}{0}, \ketbra{1}{1}, \ketbra{2}{2}, \ketbra{3}{3}, \ldots$ is known 
		as \emph{Fock space} and it is infinite dimensional. 
		The creation and annihilation operators obey the commutation relation $[\hat{a},\hat{a}^{\dag}]  =1$.
		
		The real part and the imaginary part of the operator $\hat{a}$ are defined
		as the $x$ quadrature and the $p$ quadrature:
		\be
			\hat{X} = \frac{ \hat{a} + \hat{a}^\dag}{\sqrt{2}},
			\qquad
			\hat{P} = \frac{ \hat{a} - \hat{a}^\dag}{i\sqrt{2}},
		\ee
		and we have $[\hat{X},\hat{P}]=i$.
		
		If we want to measure how many excitations are in the field,
		we use the \emph{number operator} $\hat{N}=\hat{a}^\dag\hat{a}$.
		If the field is in excitation level $n$, the expected number of 
		excitations will be:
		\be
			\langle \hat{N} \rangle  = \Tr\left[ \hat{a}^\dag\hat{a}  \ketbra{n}{n} \right] = n.
		\ee		

		%
		The Hamiltonian that describes one non-interacting mode of the electromagnetic field is given by:
		\be
			\hat{H} = \hbar \omega \left( \hat{a}^\dag \hat{a} + \frac{1}{2} \right). 
		\ee 
		The Hamiltonian is important because it gives the time evolution operator 
		$U(t)\equiv e^{i\hat{H}t}$ and the energy of the system:
		$E_\rho \equiv \langle \hat{H} \rangle = \Tr[ \hat{H}  \rho ]$.
		Observe that the system has energy even when it is in the vacuum state:
		\be
			E_0 = \Tr[ \hat{H} \ketbra{0}{0} ] = \bra{0} \hat{H} \ket{0} 
			= \hbar \omega  \bra{0} \left( \hat{a}^\dag \hat{a} + \frac{1}{2} \right)   \ket{0} = \frac{\hbar \omega}{2}.
		\ee
		This is known as the zero-point energy or vacuum energy.

		%
		%

	\subsection{Coherent states}
	
		A composite system exhibits coherence if all its components somehow
		coincide with each other.
		This could be either coincidence in time, space coherence,  phase coherence
		or quantum coherence.
		An example of the latter is the process of \emph{stimulated emission} of photons 
		which occurs inside a laser.
		All new photons are created exactly ``in phase'' with the other photons inside
		the laser.
		Over time the number of photons in the laser will grow,
		but they will all have the same frequency, phase and polarization.
		
		The coherent state $\ket{\alpha}$ describes an oscillation of the electromagnetic field.
		In general $\alpha \in \mathbb{C}$ and we have $\alpha = |\alpha|e^{i\phi}$,
		where $|\alpha|$ is the amplitude of the oscillation and $\phi$ is the initial phase.
		In the Fock basis, the coherent state $\ket{\alpha}$ is written as:
		\begin{align}
			\ket{\alpha}
			& = e^{-{|\alpha|^2\over2}}\sum_{n=0}^{\infty}{\alpha^n\over\sqrt{n!}}|n\rangle  \\
			& = e^{-{|\alpha|^2\over2}}\left[
				\ket{0}  \ + \  |\alpha|e^{i\phi}\ket{1} \ + \  \frac{|\alpha|^2}{\sqrt{2}}e^{2i\phi}\ket{2}
					\ + \ \frac{|\alpha|^3}{\sqrt{6}}e^{3i\phi}\ket{3} \ + \ \cdots \
				\right].
		\end{align}
		The output of a laser is coherent light: the excitations at all energy levels will
		have the same phase. Coherent states remain coherent over time:
		$\ket{\alpha(t)} \equiv U(t)\ket{\alpha} = e^{i\omega t/2} \ket{ |\alpha| e^{i(\phi - \omega t)} }$.
		
		A coherent state can also be defined in terms of 
		the unitary \emph{displacement operator} which acts as:
		\be
			\ket{\alpha}	\ = \ D({\alpha})\ket{0} = \exp{\alpha \hat{a}^\dag - \alpha^* \hat{a} }\ket{0}.
		\ee
		Note that in some respect $D({\alpha})$ is similar to the creation 
		operator $\hat{a}^\dag$, since it creates excited states from the vacuum state.

		
		%
		%
		%








\section{Bosonic channels}
												\label{sec:bosonic-channels}

	Point-to-point optical communication using laser-light modulation in
	conjunction with direct-detection and coherent-detection receivers has been
	studied in detail using the semiclassical theory of photodetection~\cite{GK95}. 
	This approach treats light as a classical electromagnetic field, and the
	fundamental noise encountered in photodetection is the shot noise associated
	with the discreteness of the electron charge.

	These semiclassical treatments for systems that exploit
	classical-light modulation and conventional receivers (direct, homodyne, or
	heterodyne) have had some success,
	but we should recall that electromagnetic waves are quantized, and
	the correct assessment of systems that use non-classical light sources
	and/or general optical measurements requires a full quantum-mechanical
	framework~\cite{S09}. There are several recent theoretical studies on the
	point-to-point~\cite{GGLMSY04, guha2011structured}, broadcast~\cite{guha2007classical} and
	multiple-access~\cite{Y05} bosonic channels.
	These studies have shown that quantum communication rates (Holevo rates)
	surpass what can be obtained with conventional receivers.
	For the general quantum channel, attaining Holevo information rates may
	require collective measurements (a joint detection) across all the output systems of
	the channel.

	Before stating our results on the bosonic interference channel,
	we will briefly review some results on point-to-point bosonic channels
	in the next subsection.
	
	\subsection{Channel model}
	
		The free-space optical communication channel is a physically realistic
		model for the propagation of photons from transmitter to receiver.
		We assume that a transmitter aperture of size $A_t$ is placed at
		a distance $L$ from a receiver aperture of size $A_r$, and that
		we are using $\lambda$-wavelength laser light for the transmission.

		\begin{figure}[hb]
		\begin{center}
		\includegraphics[width=0.45\textwidth]{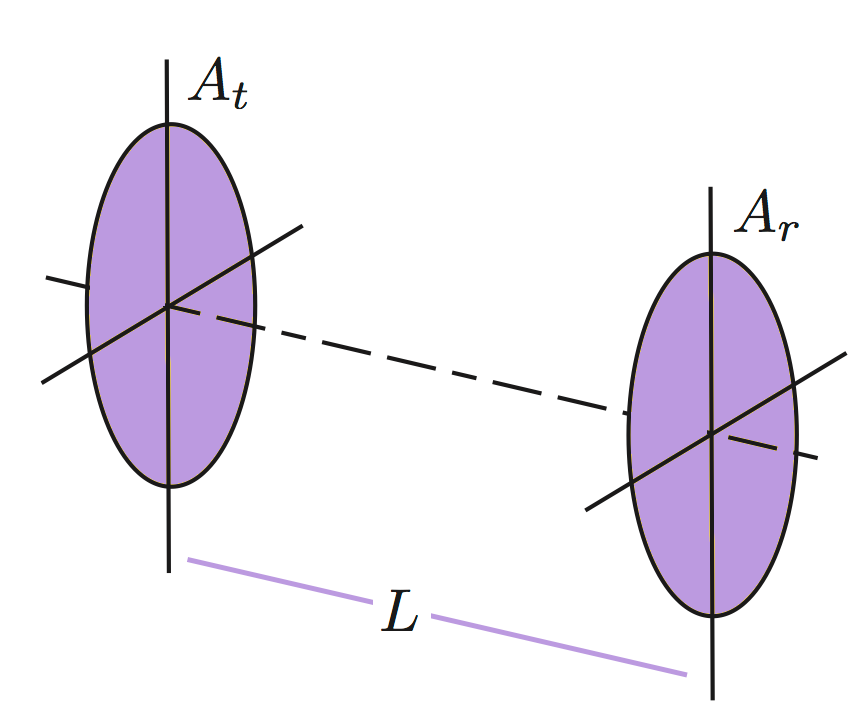}
		\caption{The free-space optical communication channel.
		Two apertures of area $A_t$ and $A_r$ are placed $L$ distance apart.
		The channel decomposes into different modes of propagation.
		We model the channel as a transformation from an annihilation operator 
		on the transmit side to an annihilation operator at the receiver side.
				}		
		\label{fig:apertures}	
		\end{center}
		\end{figure}

		To analyze the communication capacity of the bosonic channel,
		we can decompose the problem into finding the capacity for each 
		of the spatial modes of propagation, which will in general have different
		transmissivity coefficients $\eta$. 
		In the far-field propagation regime, which is when we have $A_tA_r/(\lambda L)^2 \ll 1$,
		only two orthogonal spatial modes (one for each polarization degree of freedom)
		will have significant power transmissivity. 
		We will analyze the channel for a single mode (one choice of polarization).

		The channel input is an electromagnetic field mode
		with annihilation operator $\hat{a}$, and the channel output
		is another mode with annihilation operator $\hat{b}$.
		The channel map is described by:
		\be
			\hat{b}\ = \ \sqrt{\eta}\ \hat{a} \ +\  \sqrt{ 1 -\eta}\ \hat{\nu},
		\ee
		in which $\hat{\nu}$ is associated with the noise of the
		environment and the parameter $\eta$, $0 \leq \eta \leq 1$, 
		models the channel transmissivity.
		
		We say that a channel is \emph{pure-loss} if the environmental noise  $\hat{\nu}$ 
		is in the vacuum state $\ketbra{0}{0}$.
		A channel has \emph{thermal noise} if the mode $\hat{\nu}$ 
		is in the thermal state:
		\be
			\rho_t = \int d^2\alpha \frac{ \exp{ - |\alpha|^2/N_B} }{ \pi N_B } \ \ketbra{\alpha}{\alpha},
			\label{eqn:themral-state-alpha}
		\ee
		which is Gaussian mixture of coherent states with average photon number $N_B > 0$.
		One can also write the thermal state in the number basis as follows: 
		\be
			\rho_t =  \frac{ 1 }{ N_B +1  } 
			\sum_{n=0}^{\infty} 
			\left( \frac{ N_B }{ N_B +1  } \right)^n
			\ketbra{n}{n}.
			\label{eqn:themral-state-n-basis}
		\ee
		
		\vspace{-2mm}		
		
	\subsection{Encoding}
		\vspace{-2mm}
	
		We will use coherent state encoding of the information at the transmitter.
		The codebook consists of tensor products of vacuum states displaced
		randomly and independently by an amount drawn from a distribution $p_\alpha$:
                \begin{align*}
                	\alpha^n \sim \prod^n p_\alpha \ \
			\to \ \ 
			\left\vert \alpha_{1}\alpha_{2}\cdots \alpha_{n}\right\rangle  
			\equiv D(\alpha_1)\ket{0} \otimes D(\alpha_2)\ket{0} \otimes \cdots
			\otimes D(\alpha_n)\ket{0}.
	         \end{align*}
		This encoding strategy is chosen because it is simple to implement in practice,
		and also because 		
		it is known that it suffices to achieve the ultimate capacity of the bosonic channel \cite{GGLMSY04}.
		
		When homodyne detection will be used at the receiver,
		we will encode the information using only the $x$ quadrature.
		The displacements are chosen according to:
		\be
			\alpha \sim \mathcal{N}_{\mathbb{R}}\left( 0, N_S \right).
			\label{eq:real-in-prob-dist}
		\ee
		%
		%
		The distribution is chosen so that it satisfies the constraint
		on the average number of input photons $\langle |\alpha|^2 \rangle \leq N_S$,
		which is the quantum analogue of the input power constraint
		for the AWGN channel.

		For heterodyne and joint detection, we will use
		both quadratures and choose the displacements according to
		a circularly-symmetric complex-valued Gaussian distribution:
		\be
			\alpha \sim \mathcal{N}_{\mathbb{C}}\left( 0, N_S/2 \right).
			\label{eq:cplx-in-prob-dist}
		\ee

	\subsection{Homodyne detection}
	
		Homodyne detection consists of combining on a beamsplitter 
		the incoming light and a local oscillator signal and measuring the resulting difference of the intensities.
		By tuning the relative phase between the incoming signal
		and the local oscillator it is possible to measure the incoming
		photons in any quadrature.

		When coherent state encoding is used with displacement values
		chosen as in \eqref{eq:real-in-prob-dist} and homodyne detection is used,
		the resulting channel is Gaussian:
		\begin{align*}
			Y 	&=   \sqrt{\eta} \alpha  + Z_{\textrm{hom}} ,
		\end{align*}
		where $Z_{\textrm{hom}} \sim \mathcal{N}_{\mathbb{R}}\left(0, \left(  2(1-\eta)N_{B}+1\right) / 4 \right)$.
		The \textquotedblleft$+1$\textquotedblright\ term in the noise variance
		arises physically from the zero-point fluctuations of the vacuum. 
		
		We can now use the general formula for the capacity of the AWGN channel from 
		\eqref{eq:gauss-chan} to obtain the capacity with homodyne detection:
		\be
			C_{\mathrm{hom}}=\frac{1}{2}\log\left(  1+\frac{4\eta{N_S} }{ 2(1-\eta)N_{B}+1  } \right) \ \ \textrm{ bits/use}.
		\ee

	\subsection{Heterodyne detection}

		The heterodyne detection strategy attempts to measure the incoming
		light in both quadratures.
		The sender inputs a coherent state  $\left\vert
		\alpha\right\rangle$ with $\alpha \in\mathbb{C}$.
		Heterodyne detection of the channel output 
		results in a classical complex Gaussian channel, 
		where the receiver output is a complex random variable $Y$
		described by:
		\begin{align}
			Y &=  \sqrt{\eta} \alpha + Z_{\textrm{het}},
		\end{align}
		where $Z_{\textrm{het}} \sim \mathcal{N}_{\mathbb{C}}\left( 0, \left(  (1 - \eta)N_{B}+1\right)\!/2 \right)$.
		The capacity formula for this choice of detection strategy is given by:
		\be
			C_{\mathrm{het}}=\log\left(  1+ \frac{ \eta{N_S} }{  (1-\eta)N_{B}+1 }\right) \ \ \textrm{bits/use}.
		\ee
		The factor of 1/2 in the noise variances is due to the
		attempt to measure both quadratures of the field simultaneously~\cite{S09}.

	\subsection{Joint detection}

		The capacity of the single-mode lossy bosonic channel with thermal
		background noise is thought to be equal to the channel's Holevo information:%
		\begin{equation}
			\chi \ \equiv \ g\!\left( \eta N_S +\left(  1-\eta\right)  N_{B}\right)  -g\!\left(  \left(
			1-\eta\right)  N_{B}\right) \ \ \   {  \textrm{     bits/use}},\label{eq:thermal-ach-rate}%
		\end{equation}
		where $N_S$ and $N_{B}$ are the mean photon numbers per mode for the input signal
		and the thermal noise,
		and $g\!\left( N\right)  \equiv\left(  N+1\right)  \log\left(N+1\right)  -N\log\left(  N\right)$ 
		is the entropy of a thermal state with mean photon number $N$.
		The latter formula is easily obtained from \eqref{eqn:themral-state-n-basis}:
		\begin{align*}
			h(\rho_t) 	& = - \Tr\!\left[ \rho_t \log  \rho_t \right] 	\\
					& = -  \sum_{n=0}^\infty
								\frac{ 1 }{ N +1  } 
								\left( \frac{ N }{ N +1  } \right)^n
								\log \left( 
									\frac{ 1 }{ N +1  } 
									\left( \frac{ N }{ N +1  } \right)^n
									\right)
								 \\
					& =  \sum_{n=0}^\infty
								\frac{ 1 }{ N +1  } 
								\left( \frac{ N }{ N +1  } \right)^n
								\bigg[ 
									- n \log N
									+ (n + 1 )\log (N+1)
								\bigg] \\
					& =  	 (N+ 1)\log (N+1) - N \log N = g(N).
		\end{align*}
		This capacity formula from equation \eqref{eq:thermal-ach-rate} assumes a long-standing conjecture 
		regarding the minimum-output entropy of the thermal noise channel~\cite{GGLMSY04,GHLM10}. 

		It is known that 	joint-detection (collective) measurements over long codeword
		blocks are necessary to achieve the rates in equation \eqref{eq:thermal-ach-rate} 
		for both the pure-loss and the thermal-noise lossy bosonic channel \cite{guha2011structured,guha2012explicit}.
		Note, however, that 
		quantum states of light
		are not necessary to achieve the rate $\chi$;
		coherent-state encoding is sufficient.

	\begin{figure}[htb]
		\begin{center}
		\includegraphics[width=0.5\textwidth]{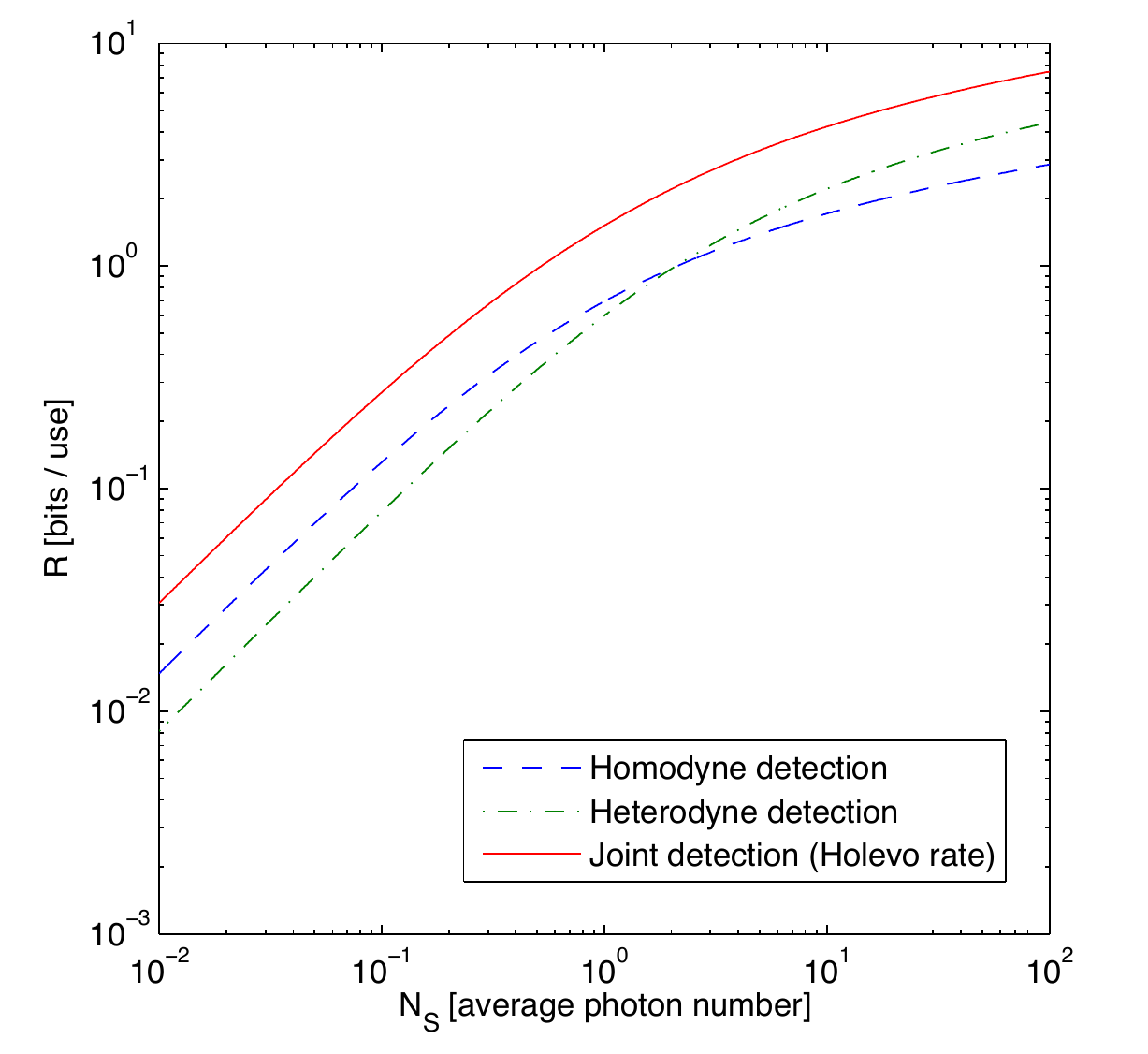}
		\caption{The achievable rates for the different decoding strategies:
		homodyne, heterodyne and joint detection in the low photon number regime
		$0.01 \leq \langle |\alpha|^2 \rangle = N_S \leq 100$.
		The channel has $\eta =0.9$ and $N_B = 1$.
		The joint detection strategy outperforms the classical strategies in which
		the outputs of the channel are measured individually, cf. Figure~\ref{fig:PtoP_example}.
		}			
		\label{fig:bosonic-capacities-plot}
		\end{center}
	\end{figure}

The rates achievable by the three different detection strategies
are illustrate in Figure~\ref{fig:bosonic-capacities-plot},
and on this we conclude our review of point-to-point bosonic communication.
In the next section, we consider the bosonic interference
channel with thermal-noise, particularly in the context of 
free-space terrestrial optical communications. 

%

\section{Free-space optical interference channels}
														\label{sec:bosonicIC}
															
	Consider now a scenario similar to the one described in Figure~\ref{fig:apertures}, 
	but now assume that there are two senders and two receivers.
	%
	Sender~1 modulates her information on the first spatial mode 
	of the transmitter-pupil, and Receiver~1 separates and demodulates
	information from the corresponding receiver-pupil spatial mode. 
	With perfect spatial-mode control at the transmitter and perfect mode separation at the
	receiver, the orthogonal spatial modes can be thought of as independent
	parallel channels with no crosstalk. However, imperfect (slightly
	non-orthogonal) mode generation or imperfect mode separation can result in
	crosstalk (interference) between the different channels.

	We will model the bosonic interference channel as a 
	passive linear mixing of the input modes 
	along with a 
	thermal environment adding zero-mean, isotropic Gaussian noise.
	%
	The channel model is given by:
	\begin{align}
	\hat{b}_{1} &  =\sqrt{\eta_{11}}\hat{a}_{1}+\sqrt{\eta_{21}}\hat{a}_{2}%
	+\sqrt{\bar{\eta}_{1}}\hat{\nu}_{1},\\
	\hat{b}_{2} &  =\sqrt{\eta_{12}}\hat{a}_{1}-\sqrt{\eta_{22}}\hat{a}_{2}%
	+\sqrt{\bar{\eta}_{2}}\hat{\nu}_{2},
	\end{align}
	where $\eta_{11}, \eta_{12}, \eta_{21}, \eta_{22}, \bar{\eta}_{1}, \bar{\eta}_{2} \in \mathbb{R}_+$,
	$\sqrt{\eta_{11} \eta_{12}} = \sqrt{\eta_{21} \eta_{22}}$,
	$
	\bar{\eta}_{1}\equiv1-\eta_{11}-\eta_{21}$, and $\bar{\eta}_{2}%
	\equiv1-\eta_{12}-\eta_{22}$.
	The following conditions ensure that the network is passive:%
	\[
	\eta_{11}+\eta_{12}\leq1,\ \ \ \eta_{11}+\eta_{21}\leq1,\ \ \ \eta_{22}%
	+\eta_{21}\leq1,\ \ \ \eta_{22}+\eta_{12}\leq1.
	\]
	We constrain the mean photon number of the
	transmitters\ $\hat{a}_{1}$ and $\hat{a}_{2}$ to be $N_{S_{1}}$ and $N_{S_{2}%
	}$ photons per mode, respectively. The environment modes ${\hat{\nu}}_{1}$ and
	${\hat{\nu}}_{2}$ are in statistically independent zero-mean thermal states
	with respective mean photon numbers $N_{B_{1}}$ and $N_{B_{2}}$ per mode \cite{S09}.

	\subsection{Detection strategies}
	
		For a coherent state encoding and coherent\footnote{We refer to both homodyne and heterodyne strategies as \emph{coherent}
		strategies.} detection at both
		receivers, the above model is a special case of the Gaussian
		interference channel, and we can study its capacity regions in various
		settings by applying the known classical results from \cite{carleial1975case, sato1981capacity} 
		and \cite{HK81}. 
		
		If the senders prepare their inputs in coherent states $\left\vert \alpha
		_{1}\right\rangle $ and $\left\vert \alpha_{2}\right\rangle $, with
		$\alpha_{1},\alpha_{2}\in~{\mathbb{R}}$, and both receivers perform
		$x$-quadrature homodyne detection on their respective modes, the result is a
		classical Gaussian interference channel~\cite{S09}, where Receivers~1 and 2
		obtain respective conditional Gaussian random variables $Y_{1}$ and $Y_{2}$
		distributed as
		\begin{align*}
		Y_{1} &  \sim\mathcal{N}_{\mathbb{R}}\left(  \sqrt{\eta_{11}}\alpha_{1}+\sqrt{\eta_{21}%
		}\alpha_{2},\ \left(  2\bar{\eta}_{1}N_{B_{1}}+1\right) / 4 \right)
		,\\
		Y_{2} &  \sim\mathcal{N}_{\mathbb{R}}\left(  \sqrt{\eta_{12}}\alpha_{2}+\sqrt{\eta_{22}%
		}\alpha_{1},\ \left(  2\bar{\eta}_{2}N_{B_{2}}+1\right) / 4  \right)  ,
		\end{align*}
		where the \textquotedblleft$+1$\textquotedblright\ term in the noise variances
		arises physically from the zero-point fluctuations of the vacuum. Suppose that
		the senders again encode their signals as coherent states $\left\vert
		\alpha_{1}\right\rangle $ and $\left\vert \alpha_{2}\right\rangle $, but this
		time with $\alpha_{1},\alpha_{2}\in\mathbb{C}$, and that the receivers both
		perform heterodyne detection. This results in a classical complex Gaussian
		interference channel~\cite{S09}, where Receivers~1 and 2 detect respective
		conditional complex Gaussian random variables $Z_{1}$ and $Z_{2}$, whose real
		parts are distributed as
		\be
		\operatorname{Re}\left\{  Z_{m}\right\}  \sim\mathcal{N}_{\mathbb{R}}\left(  \mu
		_{m},\ \left(  \bar{\eta}_{m}N_{B_{m}}+1\right) \!/  2  \right)  ,
		\ee
		where $m\in\left\{  1,2\right\}  $, $\mu_{1}\equiv\sqrt{\eta_{11}%
		}\operatorname{Re}\left\{  \alpha_{1}\right\}  +\sqrt{\eta_{21}}%
		\operatorname{Re}\left\{  \alpha_{2}\right\}  $, $\mu_{2}\equiv\sqrt{\eta
		_{12}}\operatorname{Re}\left\{  \alpha_{1}\right\}  +\sqrt{\eta_{22}%
		}\operatorname{Re}\left\{  \alpha_{2}\right\}  $, and the imaginary parts of
		$Z_{1}$ and $Z_{2}$ are distributed with the same variance as their real
		parts, and their respective means are $\sqrt{\eta_{11}}\operatorname{Im}%
		\left\{  \alpha_{1}\right\}  +\sqrt{\eta_{21}}\operatorname{Im}\left\{
		\alpha_{2}\right\}  $ and $\sqrt{\eta_{12}}\operatorname{Im}\left\{
		\alpha_{1}\right\}  +\sqrt{\eta_{22}}\operatorname{Im}\left\{  \alpha
		_{2}\right\}  $. The factor of 1/2 in the noise variances is due to the
		attempt to measure both quadratures of the field simultaneously~\cite{S09}.
		


\section{Very strong interference case}
	
	\label{sec:very-strong-int}


	Recall the setting of the interference channel which we discussed
	in Section~\ref{sec:rate-succ-decoding-IC}, where the crosstalk between
	the communication links is so strong that the receivers can fully decode
	the interfering signal and ``subtract'' it from the received signal 
	to completely cancel its effects.
	The conditions in (\ref{eq:VSI-1}) and (\ref{eq:VSI-2}) translate to the following
	ones for the case of coherent-state encoding and coherent detection:%
	\begin{align*}
	\frac{\eta_{21}}{\eta_{22}}  &  \geq\frac{4^{i}\eta_{11}N_{S_{1}}+2^{i}%
	\bar{\eta}_{1}N_{B_{1}}+1}{2^{i}\bar{\eta}_{2}N_{B_{2}}+1},\\
	\frac{\eta_{12}}{\eta_{11}}  &  \geq\frac{4^{i}\eta_{22}N_{S_{2}}+2^{i}%
	\bar{\eta}_{2}N_{B_{2}}+1}{2^{i}\bar{\eta}_{1}N_{B_{1}}+1},
	\end{align*}
	and the capacity region becomes%
	\begin{align}
	R_{1}  &  \leq\frac{1}{2^{i}}\log\left(  1+\frac{4^{i}\eta_{11}N_{S_{1}}}%
	{2^{i}\bar{\eta}_{1}N_{B_{1}}+1}\right)  ,\label{eq:hom-VSI-cap-1}\\
	R_{2}  &  \leq\frac{1}{2^{i}}\log\left(  1+\frac{4^{i}\eta_{22}N_{S_{2}}}%
	{2^{i}\bar{\eta}_{2}N_{B_{2}}+1}\right)  , \label{eq:hom-VSI-cap-2}%
	\end{align}
	where $i=1$ for homodyne detection and $i=0$ for heterodyne detection.

	\begin{figure}[htb]
	\begin{center}
	\includegraphics[
	width=0.8\textwidth
	]{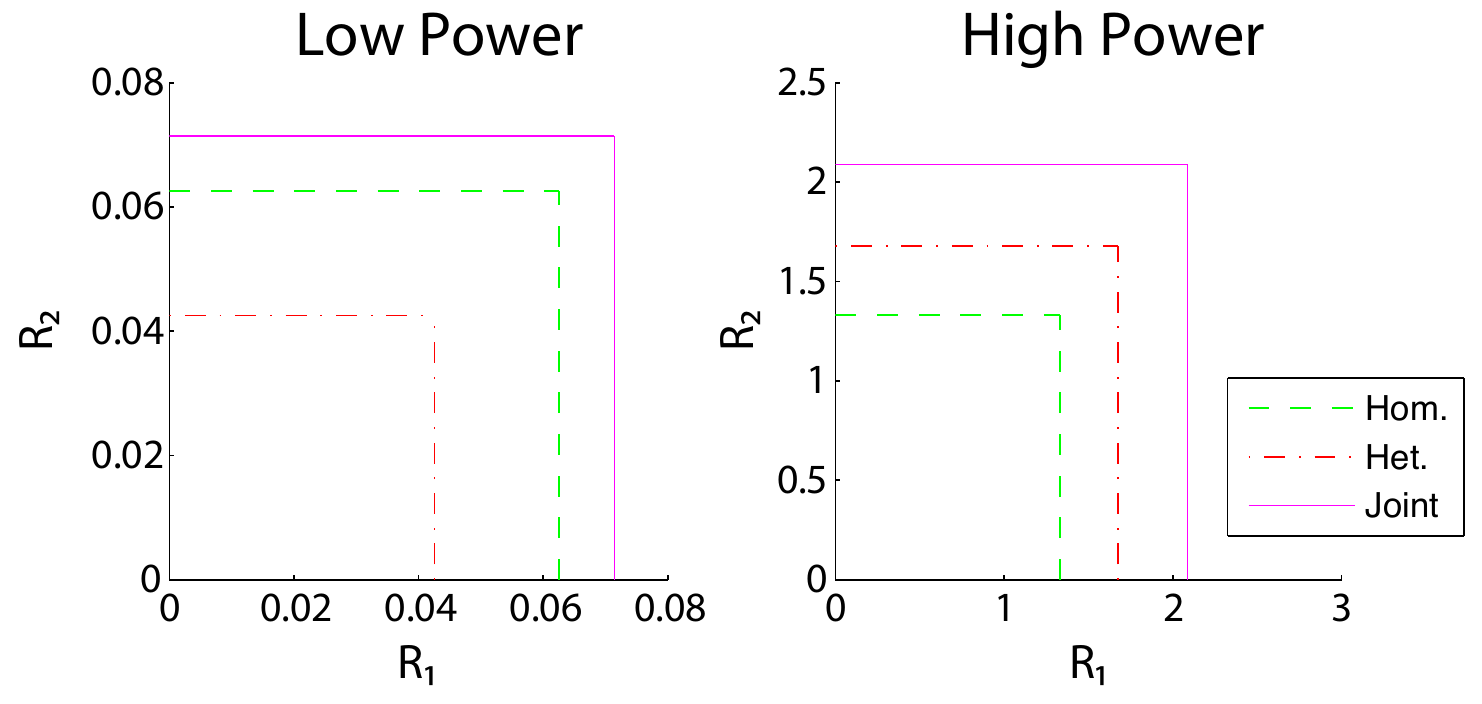}
	\end{center}
	\caption{Capacity regions for coherent-state encodings and coherent detection,
	and achievable rate regions for coherent-state encodings and joint detection
	receivers---both with $\eta_{11}=\eta_{22}= 1 / 16$ and $\eta_{12}%
	=\eta_{21}=1/2$ (\textquotedblleft very strong\textquotedblright%
	\ interference for coherent detection). 
	The LHS\ displays these regions in a low-power
	regime with $N_{S_{1}}=N_{S_{2}}=1$ and $N_{B_{1}}=N_{B_{2}}=1$, and the
	RHS\ displays these regions in a high-power regime where $N_{S_{1}}=N_{S_{2}%
	}=100$. Homodyne detection outperforms heterodyne detection in the low-power
	regime because it has a reduced detection noise, while heterodyne detection
	outperforms homodyne detection in the high-power regime because its has an
	increased bandwidth.}%
	\label{fig:carleial}%
	\end{figure}

	We can also consider the case when the senders employ coherent-state encodings
	and the receivers employ a joint detection strategy on all of their respective
	channel outputs. The conditions in (\ref{eq:VSI-1}) and (\ref{eq:VSI-2}) readily
	translate to this quantum setting where we now consider $B_{1}$ and $B_{2}$ to
	be quantum systems, and the information quantities in (\ref{eq:VSI-1}) and (\ref{eq:VSI-2}) become Holevo informations. 
	The conditions in (\ref{eq:VSI-1}) and (\ref{eq:VSI-2})
	when restricted to coherent-state encodings translate to:
	\ifthenelse{\boolean{BOOKFORM}}
	{
	\begin{align*}
		& g\!\left(  \eta_{22}N_{S_{2}}+\bar{\eta}_{2}N_{B_{2}}\right)  -g\!\left(
		\bar{\eta}_{2}N_{B_{2}}\right)  \\
		&  \qquad\qquad \leq g\!\left(  \eta_{21}N_{S_{2}}+\eta_{11}N_{S_{1}}+\bar{\eta}_{1}N_{B_{1}%
		}\right)  -g\!\left(  \eta_{11}N_{S_{1}}+\bar{\eta}_{1}N_{B_{1}}\right)
		,
		\nonumber \\
		& g\!\left(  \eta_{11}N_{S_{1}}+\bar{\eta}_{1}N_{B_{1}}\right)  -g\!\left(
		\bar{\eta}_{1}N_{B_{1}}\right)   \\
		& \qquad\qquad \leq g\!\left(  \eta_{12}N_{S_{1}}+\eta_{22}N_{S_{2}}+\bar{\eta}_{2}N_{B_{2}%
		}\right)  -g\!\left(  \eta_{22}N_{S_{2}}+\bar{\eta}_{2}N_{B_{2}}\right).
		\nonumber
	\end{align*}	
	}
	{
	\begin{align}
		g\!\left(  \eta_{22}N_{S_{2}}+\bar{\eta}_{2}N_{B_{2}}\right)  -g\!\left(
		\bar{\eta}_{2}N_{B_{2}}\right)  
		&  \leq g\!\left(  \eta_{21}N_{S_{2}}+\eta_{11}N_{S_{1}}+\bar{\eta}_{1}N_{B_{1}%
		}\right)  -g\!\left(  \eta_{11}N_{S_{1}}+\bar{\eta}_{1}N_{B_{1}}\right)
		,
		\nonumber \\
		g\!\left(  \eta_{11}N_{S_{1}}+\bar{\eta}_{1}N_{B_{1}}\right)  -g\!\left(
		\bar{\eta}_{1}N_{B_{1}}\right)  
		&  \leq g\!\left(  \eta_{12}N_{S_{1}}+\eta_{22}N_{S_{2}}+\bar{\eta}_{2}N_{B_{2}%
		}\right)  -g\!\left(  \eta_{22}N_{S_{2}}+\bar{\eta}_{2}N_{B_{2}}\right).
		\nonumber
	\end{align}
	}
	where $g\!\left(  N\right)  \equiv\left(  N+1\right)  \log\left(
	N+1\right)  -N\log\left(  N\right)$ is the entropy of a thermal state with mean photon number $N$.
	
	An achievable rate region is then%
	\begin{align*}
	R_{1} &  \leq g\!\left(  \eta_{11}N_{S_{1}}+\bar{\eta}_{1}N_{B_{1}}\right)
	-g\!\left(  \bar{\eta}_{1}N_{B_{1}}\right)  ,\\
	R_{2} &  \leq g\!\left(  \eta_{22}N_{S_{2}}+\bar{\eta}_{2}N_{B_{2}}\right)
	-g\!\left(  \bar{\eta}_{2}N_{B_{2}}\right)  . 
	\end{align*}
	These rates are achievable using a coherent-state encoding, but are not necessarily optimal
	(though they would be optimal if the minimum-output entropy conjecture from
	Refs.~\cite{GGLMSY04,GHLM10}\ were true). Nevertheless, these rates always
	beat the rates from homodyne and heterodyne detection.
	Figure~\ref{fig:carleial}\ shows examples of the achievable rate regions 
	for a bosonic interference channel with very strong interference.
	Both the low-power and high-power regimes are considered.
	Observe that the relative superiority of homodyne and heterodyne detection
	depend on power constraint and that the joint detection strategy always outperforms them.

\section{Strong interference case}
												\label{sec:bos-strong-int}
	
	Sato~\cite{sato1981capacity} 
	determined the capacity of the classical Gaussian interference
	channel under \textquotedblleft strong\textquotedblright\ interference. 
	Theorem~\ref{thm:strong-int} from Chapter~\ref{chapter:IC} gives us
	the capacity region for quantum interference channels with strong interference.
	We will now apply these results in the context of the bosonic interference channel.

	The conditions for a channel to exhibit \textquotedblleft strong\textquotedblright\ interference 
	are given in equations (\ref{eq:SI-1}) and (\ref{eq:SI-2}), 
	and they translate to the following
	ones for coherent-state encoding and coherent detection:%
	\[
	\frac{\eta_{21}}{\eta_{22}}\geq\frac{2^{i}\bar{\eta}_{1}N_{B_{1}}+1}{2^{i}%
	\bar{\eta}_{2}N_{B_{2}}+1},\ \ \ \ \ \ \ \ \frac{\eta_{12}}{\eta_{11}}%
	\geq\frac{2^{i}\bar{\eta}_{2}N_{B_{2}}+1}{2^{i}\bar{\eta}_{1}N_{B_{1}}+1},
	\]
	and the capacity region becomes:
	\begin{align}
	R_{1}  &  \leq\frac{1}{2^{i}}\log\left(  1+\frac{4^{i}\eta_{11}N_{S_{1}}}%
	{2^{i}\bar{\eta}_{1}N_{B_{1}}+1}\right)  ,\label{eq:hom-SI-cap-1}\\
	R_{2}  &  \leq\frac{1}{2^{i}}\log\left(  1+\frac{4^{i}\eta_{22}N_{S_{2}}}%
	{2^{i}\bar{\eta}_{2}N_{B_{2}}+1}\right)  , \label{eq:hom-SI-cap-2} \\
	R_{1}+R_{2} & \leq\frac{1}{2^{i}}\min\left\{
	\begin{array}
	[c]{c}%
	\log\left(  1+4^{i}\frac{\eta_{11}N_{S_{1}}+\eta_{21}N_{S_{2}}}{2^{i}\bar{\eta
	}_{1}N_{B_{1}}+1}\right)  ,\\
	\ \ \ \log\left(  1+4^{i}\frac{\eta_{22}N_{S_{2}}+\eta_{12}N_{S_{1}}}{2^{i}%
	\bar{\eta}_{2}N_{B_{2}}+1}\right)
	\end{array}
	\right\},
	\end{align}
	where again $i=1$ for homodyne detection and $i=0$ for heterodyne detection.

                \begin{figure}[htb]
                \begin{center}
                \includegraphics[
                width=0.43\textwidth                ]{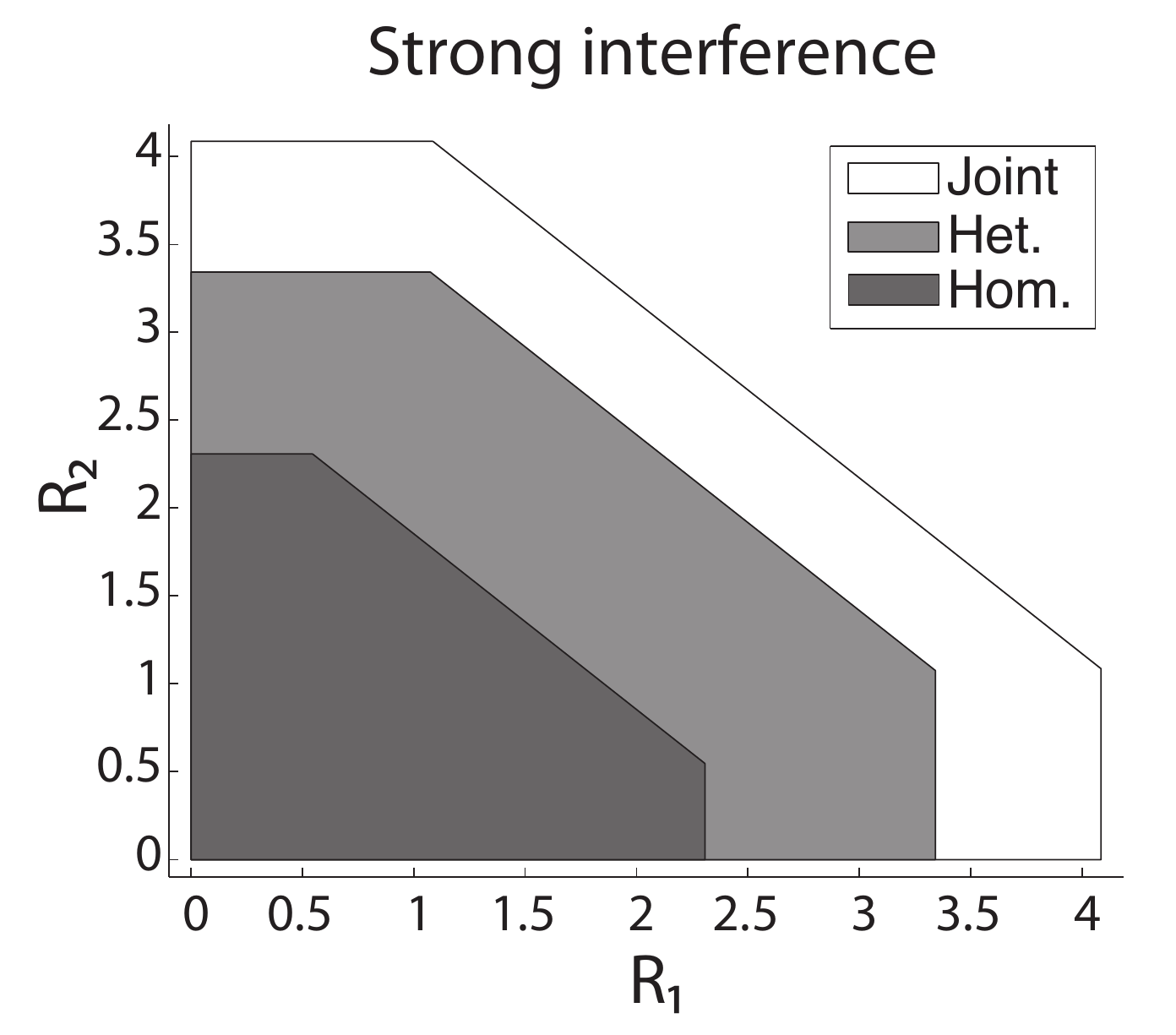}
                \end{center}
                \caption{
                The  figure depicts the \textquotedblleft strong\textquotedblright\ interference capacity 
                regions  in the high-power regime for homodyne and
                heterodyne detection, and joint detection.
		The channel in the figure is in the high-power regime:
                $N_{B_{1}}=N_{B_{2}}=1$, $\eta_{11}=\eta_{22}=0.3$, $\eta
                _{21}=\eta_{12}=0.6$, and $N_{S_{1}}=N_{S_{2}}=100$.
                Heterodyne detection outperforms homodyne detection in this case.
		}%
                \label{fig:strong-int}
                \end{figure}

We can also compute the achievable rate region using the joint detection strategy.
Figure~\ref{fig:strong-int}\ displays the different capacity and achievable
rate regions when a free-space interference channel exhibits \textquotedblleft
strong\textquotedblright\ interference.

\section{Han-Kobayashi rate regions}
	
	The Han-Kobayashi rate region is the largest known achievable rate region for the
	classical interference channel~\cite{HK81}. 
	The region was described in Theorem~\ref{thm:quantum-HK-region},
	and in Section~\ref{sec:QCMGvia2QMACproof} we established the achievability
	of the Chong-Motani-Garg, which is equivalent to the Han-Kobayashi rate region.

	The Han-Kobayashi coding strategy readily translates into a strategy for
	coherent-state encoding and coherent detection. Sender~$m$
	shares the total photon number $N_{S_{m}}$ between her personal message and
	her common message. Let $\lambda_{m}$ be the fraction of signal power that
	Sender~$m$ devotes to her personal message, and let $\bar{\lambda}_{m} \equiv (1-\lambda_m)$ denote
	the remaining fraction of the signal power that Sender$~m$ devotes to her common
	message. 
	

\def\Pone{   \eta_{11}N_{S_{1}}   }								
\def\Ponep{   \lambda_1\eta_{11}N_{S_{1}}   }						
\def\Ponec{   \bar{\lambda}_1\eta_{11}N_{S_{1}}   }					
\def\NWone{   \bar{\lambda}_2\eta_{21}N_{S_{2}}   }					
\def\NUone{   \lambda_2\eta_{21}N_{S_{2}}   }						
\def\Ione{   \eta_{21}N_{S_{2}}   }								
\def\None{  \tfrac{1}{4} \left( 2 \bar{\eta}_1N_{B_{1}}+1 \right)  }	

\def\Ptwo{   \eta_{22}N_{S_{2}}   }								
\def\Ptwop{   \lambda_2\eta_{22}N_{S_{2}}   }						
\def\Ptwoc{   \bar{\lambda}_2\eta_{22}N_{S_{2}}   }					
\def\Utwo{   \bar{\lambda}_1\eta_{12}N_{S_{1}}   }					
\def\NUtwo{   \lambda_1\eta_{12}N_{S_{1}}   }						
\def\Itwo{   \eta_{12}N_{S_{1}}   }								
\def\Ntwo{  \frac{2}{4} \left( 1(2-\eta_{22}-\eta_{12})N_{B_{2}}+2 \right)  }	



\def\HKR{\mathcal{R}_{\text{HK}}}
\def\RSato{\mathcal{R}_{\text{Sato}}}
\def\bbR{\mathbb{R}}

	When Receiver~1 uses homodyne detection to decode the messages,
	we can identify the following components that are part of his received signal:
	\begin{align}
		\Ponep 	&=  	\textrm{ power of own personal message},  \label{eqn:qua-first}\\
		\Ponec 	&=  	\textrm{ power of own common message}, \\
		\Pone 	&=  	\textrm{ total own signal power}, \\
		\Ione 	&=  	\textrm{ total interference power}, \\			
		\NWone 	&= 	\textrm{ useful part of interference (other's common)}, \\
		\NUone 	&= 	\textrm{ non-useful interference (other's personal)}, \\
		N_1 		&=	\None  = 	\textrm{ noise power} \label{eqn:qua-last}, 
	\end{align}
	Similar expressions exist for  Receiver 2.

\def\Ctwo#1#2{\gamma\!\left( \frac{#1}{#2} \right)}

\def\None{  N_{1} }	
\def\Ntwo{  N_{2} }	

	Consider now the inequalities (HK1)-(HK9) which define the Han-Kobayashi rate region
	(see page~\pageref{thm:quantum-HK-region}).
	When we evaluate each of the mutual informations  
	for the signal and noise quantities \eqref{eqn:qua-first} - \eqref{eqn:qua-last},
	we obtain the Han-Kobayashi achievable rate region for the bosonic interference channel:

	{ \allowdisplaybreaks 
	\ifthenelse{\boolean{BOOKFORM}}{\tiny}{\small}
        \begin{align}
            R_1 		&\leq		\Ctwo{\Pone}{ \NUone + \None } 
            					\tag{BHK1} \\
            R_1 		&\leq		\Ctwo{ \Ponep }{ \NUone + \None } + \Ctwo{ \Utwo }{ \NUtwo + \Ntwo } 
						\tag{BHK2} \\
            R_2 		&\leq		 \Ctwo{ \Ptwo }{ \NUtwo + \Ntwo }
					 	\tag{BHK3} \\
            R_2 		&\leq		\Ctwo{ \Ptwop }{ \NUtwo + \Ntwo} + \Ctwo{ \NWone }{ \NUone + \None } 
						\tag{BHK4} \\
            R_1 + R_2	&\leq		\Ctwo{ \Pone + \NWone }{ \NUone + \None} 
            									 + \Ctwo{ \Ptwop }{ \NUtwo + \Ntwo}
						\tag{BHK5}\\
            R_1 + R_2	&\leq		\Ctwo{ \Ptwo + \Utwo }{ \NUtwo + \Ntwo } 
            									+ \Ctwo{ \Ponep }{ \NUone + \None}
						\tag{BHK6} \\
            R_1 + R_2	&\leq		\Ctwo{ \Ponep + \NWone }{ \NUone + \None} +
            					\Ctwo{ \Ptwop + \Utwo }{ \NUtwo + \Ntwo}
						\tag{BHK7} \\
            2R_1 + R_2	&\leq		\Ctwo{ \Pone + \NWone}{\NUone + \None}
            									+ \Ctwo{ \Ponep }{ \NUone + \None }
						+ \Ctwo{ \Ptwop + \Utwo }{ \NUtwo + \Ntwo }           
						\tag{BHK8} \\
            R_1 + 2R_2	&\leq		\Ctwo{ \Ptwo + \Utwo}{ \NUtwo + \Ntwo} + 
            									\Ctwo{ \Ptwop }{ \NUtwo + \Ntwo}
						 + \Ctwo{ \Ponep + \NWone }{ \NUone + \None}
						\tag{BHK9}  
        \end{align}	
        }
	Note the shorthand notation used $\gamma(x) = \frac{1}{2}\log_2(1+x)$.

            \begin{figure}[htb]
            \begin{center}
            \includegraphics[width=0.6\textwidth]{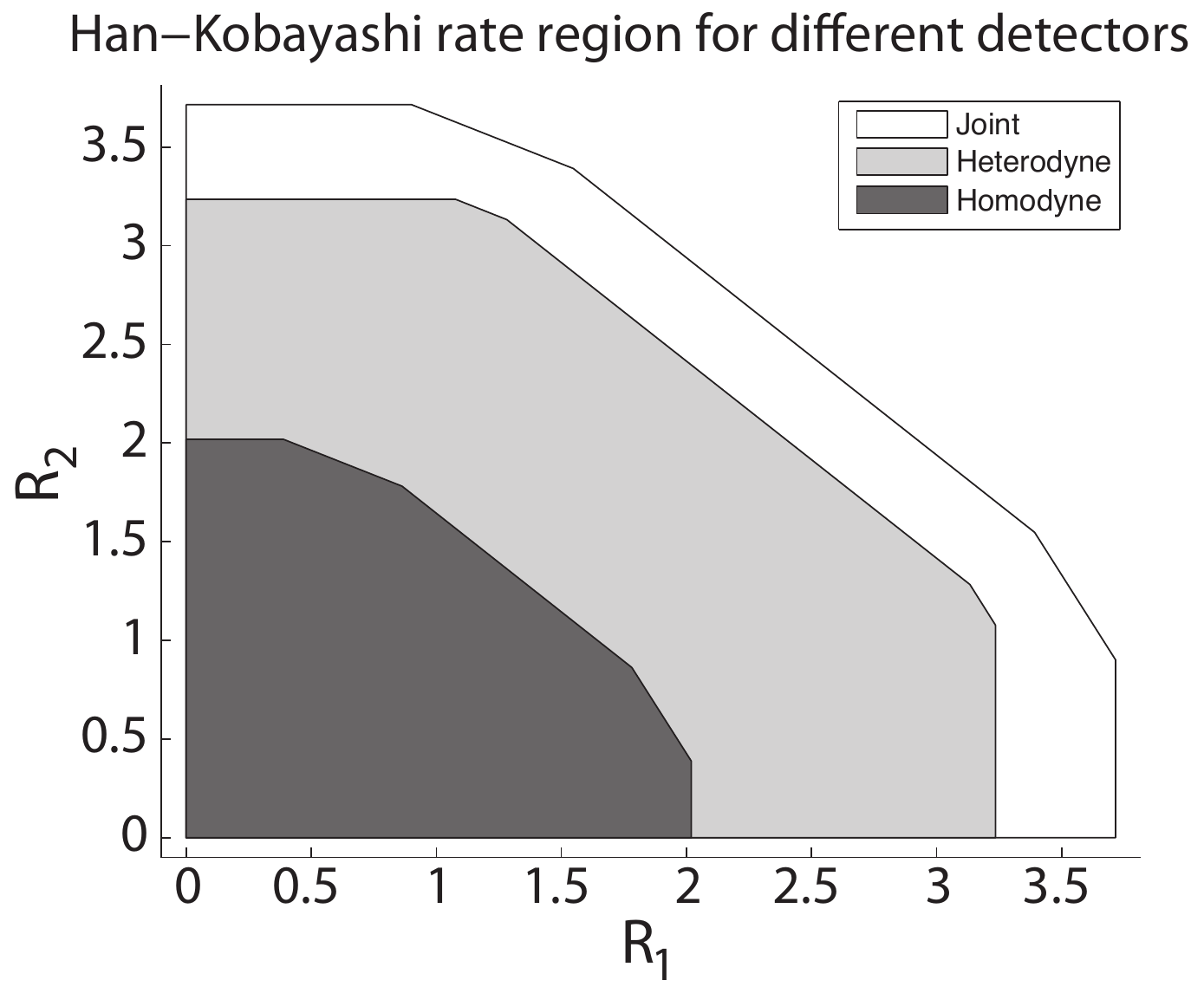}
            \end{center}
            \caption{
		  The figure depicts the achievable rate regions
	            by employing a Han-Kobayashi coding strategy for homodyne and heterodyne detection.
		  The channel parameters are $N_{S_{1}}=N_{S_{2}}=100$, 
		  $N_{B_{1}}=N_{B_{2}}=1$, $\eta_{11}=\eta_{22}=0.8$, and $\eta_{21}=\eta_{12}=0.1$. 
		All of these regions are with respect to
		a 10\%-personal, 90\%-common Han-Kobayashi power split.
            }
            \label{fig:h-k}
            \end{figure}

We can also calculate the shape of the Han-Kobayashi achievable rate region if the senders
employ coherent-state encodings and the receivers exploit heterodyne or joint detection receivers.
A statement of the inequalities for the other detection strategies has been 
omitted, because they are similar to (BHK1)-(BHK9).
Figure~\ref{fig:h-k} shows the relative sizes of the Han-Kobayashi rate regions achievable with 
coherent detection and joint detector for a particular choice of input power split: $\lambda_m = 0.1$,
$\bar{\lambda}_m = 0.9$.

\section{Discussion}

The semiclassical models for free-space optical communication are not
sufficient to understand the ultimate limits on reliable communication rates,
for both point-to-point and multiuser bosonic channels. We presented a
quantum-mechanical model for the free-space optical interference channel and determined
achievable rate regions using three different decoding strategies for the receivers.
We also determined the Han-Kobayashi inner bound for homodyne, heterodyne 
and joint detection.

Several open problems remain for this line of inquiry. 
We do not know if a coherent-state encoding is in
fact optimal for the free-space interference channel---it might be that
squeezed state transmitters could achieve higher communication rates as in
\cite{Y05}. One could also evaluate the ergodic and outage capacity
regions based on the statistics of $\eta_{ij}$, which could be derived from
the spatial coherence functions of the stochastic mode patterns under
atmospheric turbulence.



\chapter{Conclusion}
										\label{chapter:conclusion}
	
	The time has come to conclude our inquiry into the problems
	of quantum network information theory.
	We will use this last chapter to summarize our results
	and highlight the specific contribution of this thesis.
	%
	%
	We will also 
	discuss open problems and avenues for future research.

	\section{Summary}

		The present work demonstrates clearly that many of the problems of
		classical network information theory can be extended to the study
		of classical-quantum channels.
		Originally, we set out to investigate the network information 
		theory problems discussed in \cite{el1980multiple}.
		It is fair to say that we have been successful on that front,
		since we managed to develop coding strategies for
		multiple access channels (Chapter~\ref{chapter:MAC}), 	
		interference channels (Chapter~\ref{chapter:IC}),
		broadcast channels (Chapter~\ref{chapter:BC})
		and 		
		relay channels (Chapter~\ref{chapter:RC}),
		in the classical-quantum setting.

		Our proof techniques are a mix of classical and quantum ideas.
		On the classical side we have the standard tools of information theory
		like averaging, conditional averaging 
		and the use of the properties of typical sets.
		%
		On the quantum side we saw how to build a \emph{projector sandwich},
		which contains many layers of conditionally typical projectors,
		how to incorporate \emph{state smoothing}, which cuts out non-typical eigenvalues of a state,
		and the winning combination of the 
		square root measurement and the Hayashi-Nagaoka operator inequality.
		
		Above all, it is the quantum conditionally typical projectors 
		that played the biggest role in all our results.
		Conditionally typical projectors are truly amazing constructs,
		since they not only give us a basis in terms of which to analyze the quantum outputs,
		but also tell us exactly in which subspace we are likely to find the output states on average.

	\section{New results}
	
		Some of the results presented in this thesis have previously appeared
		in publications and some are original to this thesis.
		We will use this section to highlight the new results.
	
		The first contribution is the establishment of the classical/quantum packing lemmas
		using conditionally typical sets/projectors.
		While these packing lemmas are not new in themselves, 
		the proofs presented  
		highlight the correspondences between the indicator functions
		for the classical conditionally typical sets and,
		their quantum counterparts, the conditionally typical projectors.
		%
		The quantum packing lemma is an effort to abstract away 
		the details of the quantum decoding strategy 
		into a reusable component as is done in \cite{el2010lecture}.
		
		It is the author's hope that the classical and quantum packing lemmas presented 
		in this work, along with their proofs, can serve as a bridge for classical 
		information theorists to cross over to the quantum side.
		Alternately, we can say that there is only one side
		and interpret the move from classical Shannon theory to quantum
		Shannon theory as a type of system upgrade.
		Indeed, the change from indicator functions for the conditionally typical sets 
		to conditionally typical projectors can be seen in terms of the OSI layered model for network architectures:
		quantum coding techniques are a change in \emph{physical layer} (Layer 1) 
		protocols while the random coding approach of the 
		\emph{data link layer} (Layer~2) stays the same.
		Note that this analogy only works for the \emph{classical} communication problem,
		and that \emph{quantum} communication and \emph{entanglement-assisted} communication 
		are completely new problems in quantum Shannon theory, which have no direct classical analogues.
		
		The main original contribution of this thesis is the  
		achievability proof for the quantum Chong-Motani-Garg rate region,
		which requires only two-sender simultaneous decoding.
		By the equivalence 
		$\mathcal{R}^o_{\mathrm{HK}}(\mcal{N}) \equiv \mathcal{R}_{\mathrm{CMG}}(\mcal{N})$,
		we have established the achievability of the quantum Han-Kobayashi rate region.
		We can therefore close the book on the original 
		research question which prompted our investigation more than two years ago.

		An interesting open problem is to prove Conjecture~\ref{conj:sim-dec}
		on the simultaneous decoding for the three-sender quantum multiple access channels.
		This result would be a powerful building block for multiuser quantum Shannon theory.



%
%

%
%
%

%
%
%
%
%


%


\appendix

\fancyhead[RO]{\emph{Appendix \thechapter}}
\pagestyle{fancy}

%


\chapter{Classical channel coding}
\label{apdx:classical-ch-coding}

This appendix contains the proof of the classical packing lemma
(Section~\ref{apdx:classical-coding-theorem}) and a brief review on 
some of the properties of typical sets.

\section{Classical typicality}
													\label{apdx:classical-typicality}
		
	In Section~\ref{sec:typ-review}, we presented a number of properties of typical sequences and
	typical sets that were used in the proof of the classical coding theorem.
	The reader is invited to consult \cite{CT91} and \cite{wilde2011book} for the proofs.
	
	In this section, we review the properties of conditionally typical 
	sets in a more general setting where an additional random variable
	$U^n$ is present. 
	This is the setting of the classical packing lemma, which will be 
	stated and proved in Section~\ref{apdx:classical-coding-theorem}.

	Consider the probability distribution $p_U(u)p_{X|U}(x|u) \in \mcal{P}(\mcal{U},\mcal{X})$ 
	and the channel $\mcal{N} = (\mcal{U}\times\mcal{X}, p_{Y|XU}(y|x,u), \  \mcal{Y})$.
	Let $(U^n,X^n)$ be distributed according to the product 
	distribution $\prod^n_{i=1} p_{U}(u_i)p_{X|U}(x_i|u_i)$.
	Let $Y^n$ denote the random variable that corresponds to the output
	of the channel when the inputs are $(U^n,X^n)$.

		\begin{figure}[hbt]
	\begin{center}
	{ \scriptsize
	\begin{tikzpicture}[node distance=2cm,>=stealth',bend angle=45,auto]
	  \tikzstyle{cnode}=[isosceles triangle,isosceles triangle apex angle=60,thick,draw=blue!75,fill=blue!20,minimum height=0.25cm] 
	  \tikzstyle{processing}=[rectangle,thick,draw=black,fill=white,minimum height=1.1cm] 
	  \tikzstyle{measurment}=[rectangle,thick,draw=blue!75,fill=blue!20,minimum height=0.3cm] 
	  \tikzstyle{qnode}=[circle,thick,draw=blue!75,fill=blue!20,minimum size=6mm] 
	  \tikzstyle{every label}=[black, font=\footnotesize]

	  \begin{scope}
		\node [cnode] (Un) [label=left:$\prod^n p_U$,label=above:$\  \in \mcal{U}^n  $]                {\footnotesize $U^n$};
		\node [cnode] (tildeXn) [label=below:$\ \ \in \mcal{X}^n$, below right of=Un,xshift=1cm]                {\footnotesize $X^n$}
				edge  [pre]				node[]  	{$\ \ \ \ \ \ \ \prod^n p_{X|U}$}	(Un);
		\node [cnode] 
		(tildeYn) [ label=below:$\ \ \in \mcal{Y}^n$, right of=tildeXn, xshift=2cm]                {\footnotesize $Y^n$}
				edge  [pre]				node[swap]  	{$\mcal{N}^{n} \equiv \prod^n p_{Y|XU}$}	(Un);	
		\node [] (abovetildeYn) [right of=Un, xshift=3cm]                {};
		\draw[->]	(tildeXn)	-- (tildeYn);		
	  \end{scope}
	  \begin{pgfonlayer}{background}
	    \filldraw [line width=4mm,join=round,black!10]
	      ([xshift=-4mm,yshift=7mm]tildeXn.north -| tildeXn.east) rectangle (tildeYn.south -| tildeYn.west);
	    \filldraw [line width=4mm,join=round,black!10]
	      ([xshift=-4mm]Un.north -| Un.east) rectangle (abovetildeYn.south -| tildeYn.west);
	  \end{pgfonlayer}
	\end{tikzpicture}

	}

	\caption{An illustration of the conditional dependence between the random variables $(U^n,X^n,Y^n)$. 	}
	\label{fig:random-variables}
	\end{center}
	\end{figure}
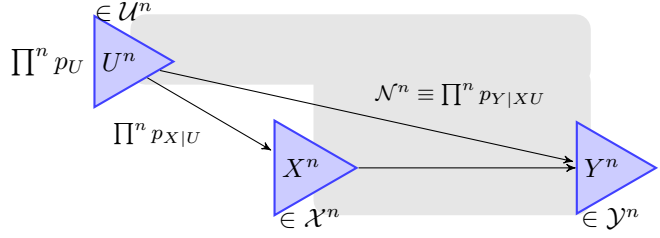

	\subsubsection{Conditionally typical sets}

		The input random variables $(U^n,X^n) \sim \prod_{i=1}^n p_U(u_i)p_{X|U}(x_i|u_i)$ 
		and the channel $\mcal{N}$ induce the following joint distribution:
		\be
			(U^n, X^n, Y^n)  \ \sim \ \prod_{i=1}^n p_U(u_i)p_{X|U}(x_i|u_i)p_{Y|XU}(y_i|x_i,u_i).
		\ee
		This corresponds to the assumption that the channel is
		memoryless, that is, the noise in the $n$ uses of the channel is independent $p_{Y^n|X^nU^n} = \prod^n p_{Y|XU}$.
		
		
		For any $\delta >0$,	define two sets of entropy conditionally typical sequences:
		\begin{align}
			\!\!\cT^{(n)}_{\delta}\!(Y|x^n,u^n)
			& \!\equiv \!
			\left\{ \!
				y^n \!\in \!\mcal{Y}^n  \colon \! \left|  -\frac{\log p_{Y^n|X^nU^n}(y^n|x^n,u^n)}{n}  - H(Y|X,U) \right| \! 
				\leq 
				\delta 
			\right\}\!\!, \\
		\cT^{(n)}_{\delta}\!(Y|u^n)
			 & \equiv 
			\left\{ 
				y^n \in \mcal{Y}^n  \colon \! \left|  -\frac{\log p_{Y^n|X^n}(y^n|u^n)}{n}  - H(Y|U) \right| \! 
				\leq 
				\delta 
			\right\},%
		\end{align}
		where $H(Y|U)=- \sum_x p_U(u)p_{Y|U}(y|u) \log p_{Y|U}(y|u)$
		is the conditional entropy of the distribution $p_{Y|U}(y|u) = \sum_x p_{X|U}(x|u) p_{Y|XU}(y|x,u)$.

		By the definition of these typical sets, we have that the following bounds
		on the probability of the sequences within these sets:
		\begin{align}
		2^{-n[ H(Y|X,U)+\delta ]}  
		\leq & \ \ 
			p_{Y^n|X^n,U^n}(y^n|x^n,u^n) \ \ 
			\leq   2^{-n[H(Y|X,U)-\delta]}  
			\ifthenelse{\boolean{BOOKFORM}}
			{
			 \nonumber \\[-1mm]
			 & \qquad \qquad \qquad \qquad    \forall y^n \in \cT^{(n)}_{\delta}\!(Y|x^n,u^n),  
			\nonumber \\[3mm]
			}
			{
			 \quad   \forall y^n \in \cT^{(n)}_{\delta}\!(Y|x^n,u^n),  
			\nonumber \\[3mm]
			}
		2^{-n[ H(Y|U)+\delta ]}   
		\leq & \ 
			p_{Y^n|U^n}(y^n|u^n) \ \ 
			\!\!\! \leq   2^{-n[H(Y|U)-\delta]}  \quad   \forall y^n \in \cT^{(n)}_{\delta}\!(Y|u^n),  \label{weightBoundsYgXU}
		\end{align}
		for any sequences $u^n$ and $x^n$.

		The channel outputs are likely to be conditionally typical sequences.
		More precisely, we have that for any $\epsilon,\delta>0$, and sufficiently large $n$
		the expectations under $U^n$ and $X^n|U^n$  obey the bounds:
		\begin{align}
		\ExpUn 
		\NExpXgU   
			 \!\!\!
			& \displaystyle
			\sum_{ y^n \in \cT^{(n)}_{\delta}\!(Y|X^n,U^n) } \!\!\!\! \! p_{Y^n|X^nU^n}\!\left( y^n|X^n,U^n\right)  
			 \ \geq \ 1-\epsilon,   \label{cc4APDX} \\
		\ExpUn 
			 \!\!\!
			 &\ \ \displaystyle
			\sum_{ y^n \in \cT^{(n)}_{\delta}\!(Y|U^n) } \!\! p_{Y^n|U^n}\!\left( y^n|U^n\right)  
			\ \geq \   1-\epsilon.   \label{cc4APDX2}
		\end{align}

		Furthermore, we have the following bounds on the size 
		of these conditionally typical sets:	
		\begin{align}
			\left|  \cT^{(n)}_{\delta}\!(Y|X^n,U^n)  \right| 
			\ \  & \leq 2^{n[ H(Y|X,U)+\delta]}, \label{cc6APDX} \\
			\left|  \cT^{(n)}_{\delta}\!(Y|U^n)  \right| 
			\ \  & \leq 2^{n[ H(Y|U)+\delta]}.  \nonumber 
		\end{align}

	\subsubsection{Conditionally typical sets}

		Equations \eqref{weightBoundsYgXU} and \eqref{cc6APDX} will play a key role
		in the proof of the classical packing lemma in the next section.
		We restate these equations here in the language of indicator functions
		for the single and double conditionally typical sets:
		\begin{align}
				\!\!\!\!\! p_{Y^n|U^n}(y^n|u^n) 
				\indicator_{  \left\{    y^n  \in   \mathcal{T}^{(n)}_\epsilon(Y|U^n)   \right\}}
				\ \ 
				& \leq   \ \ 
				2^{-n[H(Y|U)-\delta]} 
				\indicator_{  \left\{    y^n  \in   \mathcal{T}^{(n)}_\epsilon(Y|U^n)   \right\}},  \tag{$\textrm{\ref{weightBoundsYgXU}}^\prime$}
		\end{align}%
		and %
		\begin{align}
				\sum_{y^n \in \mcal{Y}^n }
				\indicator_{  \left\{    y^n  \in   \mathcal{T}^{(n)}_\epsilon(Y|x^n,u^n)   \right\}}
				\ \  & \leq  \ \ 2^{n[ H(Y|X,U)+\delta]}.
				\tag{$\textrm{\ref{cc6APDX}}^\prime$}
		\end{align}


\section{Classical packing lemma}
											\label{apdx:classical-coding-theorem}

The packing lemma is a powerful tool for proving capacity theorems \cite{el2010lecture}.
We give a proof of a packing lemma which, instead of the usual jointly typical sequences
argument, uses the properties of conditionally typical sets.
This non-standard form of the packing lemmas is preferred because 
it highlights the similarities with its quantum analogue,
the quantum conditional packing lemma stated in Appendix~\ref{apdx:q-packing}.

\begin{lemma}[Classical conditional packing lemma]
												\label{lem:classical-packing}
	
	Let $p_U(u)p_{X|U}(x|u) \in \mcal{P}(\mcal{U},\mcal{X})$ be an arbitrary code distribution,
	and let $\mcal{N} = (\mcal{U}\times\mcal{X}, p_{Y|XU}(y|x,u), \  \mcal{Y})$ be a channel.
	%
	Let $(U^n,X^n,\tilde{X}^n)$ be distributed according to 
	$\prod^n_{i=1} p_{U}(u_i)p_{X|U}(x_i|u_i)p_{X|U}(\tilde{x}_i|u_i)$.
	%
	Let $\tilde{Y}^n$ denote the random variable that corresponds to the output
	of the channel when the inputs are $(U^n,\tilde{X}^n)$.
	Define $\mcal{E}_2$ to be the event that the output $\tilde{Y}^n$ 
	will be part of the conditionally typical set 
	$\mathcal{T}^{(n)}_\epsilon(Y|X^n,U^n)$,
	given that it is part of the output-typical set $\tilde{Y}^n  \in \mathcal{T}^{(n)}_\epsilon(Y|U^n)$.
	We have that
	\begin{align}
		\mathop{\mathbb{E}}_{U^n, \atop X^n, \tilde{X}^n} 
		& \NPrYtildegXtilde
		\{
			\mcal{E}_2
		\}
		=  \nonumber \\
	& = \ExpUn
	\NExpX  
	 \NExpXtilde
	\NPrYtildegXtilde
	\!\! \left\{ \!
	 \left\{ 
	  \tilde{Y}^n
	  \in 
	  \mathcal{T}^{(n)}_\epsilon(Y|X^n,U^n)
	  \right\}
	\cap
	 \left\{ 
	  \tilde{Y}^n
	  \in 
	  \mathcal{T}^{(n)}_\epsilon(Y|U^n)
	  \right\} \!
	  \right\}	 \nonumber \\
	  &\leq 
	  2^{-n[ I(X;Y|U) - \delta(\epsilon) ]}.
	\end{align}
%
%

	Consider the random codebook $\{X^n(m) \}$, $m \in [1:2^{nR}]$
	generated randomly and independently according to $\prod_{i=1}^n p_{X|U}(x_i|u_i)$.
	There exists $\delta(\epsilon) \to 0$ as $\epsilon \to 0$ such that
	the probability that the conditionally typical decoding will 
	misinterpreting the channel output for $X^n(m)$
	incorrectly as produced by $X^n(m^\prime)$ for some $m^\prime \neq m$, that is,
	\be
		Y^n(m) \equiv \mcal{N}^n(U^n,X^n(m)),
		\quad
		Y^n(m)
		  \in 
		  \mathcal{T}^{(n)}_\epsilon(Y|X^n(m^\prime),U^n)
		  \   \mathrm{ and    } \ 
  		Y^n(m)
		  \in 
		  \mathcal{T}^{(n)}_\epsilon(Y|U^n),
		  \nonumber
	\ee
	vanishes as $n \to \infty$, if $R < I(X;Y|U) - \delta(\epsilon)$,
	where the mutual information is calculated on the
	induced joint probability distribution 
	$(U,X,Y) \sim p_{UXY}(u,x,y)=p_{Y|XU}(y|x,u)p_{X|U}(x|u)p_{U}(u)$.


\end{lemma}

	\begin{figure}[htb]
	\begin{center}
	{ \scriptsize
	\begin{tikzpicture}[node distance=1.8cm,>=stealth',bend angle=45,auto]
	  \tikzstyle{cnode}=[isosceles triangle,isosceles triangle apex angle=60,thick,draw=blue!75,fill=blue!20,minimum height=0.3cm] 
	  \tikzstyle{processing}=[rectangle,thick,draw=black,fill=white,minimum height=1.1cm] 
	  \tikzstyle{measurment}=[rectangle,thick,draw=blue!75,fill=blue!20,minimum height=0.3cm] 
	  \tikzstyle{qnode}=[circle,thick,draw=blue!75,fill=blue!20,minimum size=6mm] 
	  \tikzstyle{every label}=[black, font=\footnotesize]

	  \begin{scope}
		\node [cnode] (Un) [label=left:$\prod^n p_U$]                {$U^n$};
		\node [cnode] (Xn) [above right of=Un,xshift=1cm]                {$X^n$}
				edge  [pre]				node[swap]  	{$\ \ \ \ \prod^n p_{X|U}$}	(Un);
		\node [cnode] (tildeXn) [below right of=Un,xshift=1cm]                {$\tilde{X}^n$}
				edge  [pre]				node[]  	{$\ \ \ \ \ \prod^n p_{X|U}$}	(Un);
		\node [cnode] 
		(tildeYn) [ label=below:$\ \ \in \mcal{Y}^n$, right of=tildeXn, xshift=0.6cm]                {$\tilde{Y}^n$}
				edge  [pre]				node[swap]  	{$\prod^n p_{Y|XU}$}	(Un);	
		\node [] (abovetildeYn) [right of=Un, xshift=3cm]                {};
		\draw[->]	(tildeXn)	-- (tildeYn);		
		\node [processing] 
		(TestOutTyp) [ 
						right of=Un, xshift=5.4cm]                {$	\indicator_{ \left\{  \mathcal{T}^{(n)}_\epsilon(Y|U^n)	  \right\} }$};
		\node [processing] 
		(DecXn) [ 
						above right of=Un, xshift=5.6cm]                {$	\indicator_{ \left\{  \mathcal{T}^{(n)}_\epsilon(Y|X^n,U^n)	  \right\} }$};
		\node [and gate US,draw,logic gate inputs=nn, right of=Un, xshift=7.2cm] 
		(ANDgate)                {AND};
		\draw[->,thick]	(tildeYn)	-- (DecXn.west);		
		\draw[->,thick]	(tildeYn)	-- (TestOutTyp.west);	
		\draw[->,thick]	(DecXn.east)	-- (ANDgate.west);	
		\draw[->,thick]	(TestOutTyp.east)	-- (ANDgate.west);	
		\node [] 
		(pERR) [ 
						right of=ANDgate, xshift=-0.5cm]                {$\!\!\!\indicator_{ \{ \mcal{E}_2 \}}$}
						edge  [pre]				node[swap]  	{}	(ANDgate);	
	  \end{scope}
	  \begin{pgfonlayer}{background}
	    \filldraw [line width=4mm,join=round,black!10]
	      ([xshift=-4mm,yshift=7mm]tildeXn.north -| tildeXn.east) rectangle (tildeYn.south -| tildeYn.west);
	    \filldraw [line width=4mm,join=round,black!10]
	      ([xshift=-4mm]Un.north -| Un.east) rectangle (abovetildeYn.south -| tildeYn.west);
	  \end{pgfonlayer}
	\end{tikzpicture}

	}

	\caption{The classical packing lemma. Two random codewords $X^n$ and $\tilde{X}^n$
	are drawn randomly and independently conditional on a third 
	random variable $U^n$. 
	Assume that the random variable $U^n$ is also available at the receiver.
	What is the chance that the output of the channel which corresponds
	to $\tilde{X}$ and $U^n$ will falsely be recognized to be in the set
	of outputs which are likely to come from inputs: $X^n$ and $U^n$?
	%
	The receiver performs two tests on the output sequence~$\tilde{Y}^n$:
	(1) test membership in $\mathcal{T}^{(n)}_\epsilon(Y|U^n)$ and
	(2) test membership in $\mathcal{T}^{(n)}_\epsilon(Y|X^n,U^n)$.
	If both these are successful, the outcome will be a \emph{misidentification} error~$\mcal{E}_2$.
	}
	\label{fig:classical-pack-illustration}
	\end{center}
	\end{figure}

\bigskip

The description of the error event in the conditional packing lemma
contains four sources of randomness.
First we have $U^n \sim \prod^n p_U$, then there are two independent draws 
from $\prod^n p_{X|U}$ to produce 
$X^n$ and $\tilde{X}^n$. 
Finally, the channel-randomness produces $\tilde{Y}^n = \mcal{N}^n(U^n,\tilde{X}^n)$.
%
The fact that $\tilde{X}^n$ and $X^n$ are conditionally independent
given $U^n$ implies that $\tilde{Y}^n$ and $X^n$ are also conditionally independent
given $U^n$. 
The situation is illustrated in Figure~\ref{fig:classical-pack-illustration}.

%

%


\begin{proof} 
We give an argument based on the properties of the output-typical sequences
and a cardinality bound on the conditionally typical sets.
Assume that the output sequence $\tilde{Y}^n=\mcal{N}^n(U^n,\tilde{X}^n)$ is   
output-typical ($ \in \mathcal{T}^{(n)}_\epsilon(Y|U^n)$), 
and happens to also fall in the conditionally typical set 
for some other codeword $\mathcal{T}^{(n)}_\epsilon(Y|X^nU^n)$.
This is described by the following event: 
\be
\mcal{E}_2 = 
 \left\{ 
	  \tilde{Y}^n
	  \in 
	  \mathcal{T}^{(n)}_\epsilon(Y|X^nU^n)
  \right\}
\cap
 \left\{ 
	  \tilde{Y}^n
	  \in 
	  \mathcal{T}^{(n)}_\epsilon(Y|U^n)
  \right\}.
\ee


Now consider the expectation of the probability of 
the event $\mcal{E}_2$ under the code randomness:
{\allowdisplaybreaks \small
\begin{align*}
& \!\!\ExpUn 
\NExpX   
\NExpXtilde
\text{Pr}\left\{ \mcal{E}_2 \right\} = \\
& = 
	\ExpUn
	\NExpXgU  
	\NExpXtildegU
	\NPrYtildegXtildeU
	\!\! \left\{ 
	 \left\{ 
		  \tilde{Y}^n
		  \in 
		  \mathcal{T}^{(n)}_\epsilon(Y|X^nU^n)
	  \right\}
	\cap
	 \left\{ 
		  \tilde{Y}^n
		  \in 
		  \mathcal{T}^{(n)}_\epsilon(Y|U^n)
	  \right\}
	  \right\}
	   \\
& = 
	\ExpUn
	\NExpXgU  
	\NExpXtildegU
	\NExpYtildegXtildeU 
	\!\! 
	\indicator_{
		 \left\{ 
			  \tilde{Y}^n
			  \in 
			  \mathcal{T}^{(n)}_\epsilon(Y|X^nU^n)
		  \right\}
	  }
	\! \cdot \!\!
	\indicator_{
	 \left\{ 
		  \tilde{Y}^n
		  \in 
		  \mathcal{T}^{(n)}_\epsilon(Y|U^n)
	  \right\}
	  }
	   \\
& = 
	\ExpUn \!
	\NExpXgU  
	\sum_{ \tilde{x}^n    }      \!
	\sum_{ \tilde{y}^n    }     
	p_{X^n}( \tilde{x}^n|U^n) 
	p_{Y^n|X^nU^n}( \tilde{y}^n | \tilde{x}^n,U^n) \	
	\!\!\! 
	\indicator_{
		 \left\{ 
			  \tilde{y}^n
			  \in 
			  \mathcal{T}^{(n)}_\epsilon(Y|X^nU^n)
		  \right\}
	  }\!\!
	\indicator_{
	 \left\{ 
		  \tilde{y}^n
		  \in 
		  \mathcal{T}^{(n)}_\epsilon(Y|U^n)
	  \right\}
	  }
	   \\
&\overset{\mdingone}{=}
	\ExpUn
	\NExpXgU  
	\sum_{ \tilde{y}^n    }     
	p_{Y^n|U^n}( \tilde{y}^n|U^n ) \	
	\!\! 
	\indicator_{
		 \left\{ 
			  \tilde{y}^n
			  \in 
			  \mathcal{T}^{(n)}_\epsilon(Y|X^nU^n)
		  \right\}
	  }
	\cdot
	\indicator_{
	 \left\{ 
		  \tilde{y}^n
		  \in 
		  \mathcal{T}^{(n)}_\epsilon(Y|U^n)
	  \right\}
	  }
	   \\
&\overset{\mdingtwo}{\leq}
	\ExpUn
	\NExpXgU  
	\sum_{ \tilde{y}^n    }     	
	 2^{-n[H(Y|U)- \delta^{\prime}(\epsilon)] } 
	 \
	\indicator_{
		 \left\{ 
		  \tilde{y}^n
		  \in 
		  \mathcal{T}^{(n)}_\epsilon(Y|X^nU^n)
		  \right\}
	  }
	\cdot
	\indicator_{
	 \left\{ 
		  \tilde{y}^n
		  \in 
		  \mathcal{T}^{(n)}_\epsilon(Y|U^n)
	  \right\}
	  }
	   \\
&\overset{\mdingthree}{\leq}
	 2^{-n[H(Y|U)- \delta^{\prime}(\epsilon)] } \ 
	\ExpUn
	\NExpXgU  
	\sum_{ \tilde{y}^n    }     
	\indicator_{
		 \left\{ 
			  \tilde{y}^n
			  \in 
			  \mathcal{T}^{(n)}_\epsilon(Y|X^nU^n)
		  \right\}
	  }
	   \\
& = 
	 2^{-n[H(Y|U)- \delta^{\prime}(\epsilon)] } \ 
	\sum_{ u^n, x^n    }     
	p_{U^nX^n}( u^n, x^n) 
	\sum_{ \tilde{y}^n    }     
	\indicator_{
		 \left\{ 
			  \tilde{y}^n
			  \in 
			  \mathcal{T}^{(n)}_\epsilon(Y|x^nu^n)
		  \right\}
	  }
	   \\
& = 
	 2^{-n[H(Y|U)- \delta^{\prime}(\epsilon)] } \ 
	\sum_{ u^n, x^n    }     
	p_{U^nX^n}( u^n, x^n) 
	\left|
		  \mathcal{T}^{(n)}_\epsilon(Y|x^nu^n)
	\right|
	   \\
&\overset{\mdingfour}{\leq}
	2^{-n[H(Y|U)- \delta^{\prime}(\epsilon)] } \ 
	2^{n[H(Y|XU) +\delta^{\prime\prime}(\epsilon) ]}     \\
& =
	2^{-n[ I(X;Y|U) - \delta(\epsilon) ]}.
\end{align*}}%
The equality~\dingone follows from the definition of the conditional output  distribution:
\be
	p_{Y^n|U^n}( \tilde{y}^n|u^n ) 
	=
	\sum_{\tilde{x}^n}     
	p_{X^n|U^n}(\tilde{x}^n|u^n) 
	p_{Y^n|X^nU^n}( \tilde{y}^n | \tilde{x}^n,u^n).
\ee
Inequality~\dingtwo follows from the fact that sequence $\tilde{Y}^n$ is conditionally output-typical,
which means that $p(y^n|u^n) \leq 2^{-n[H(Y|U) - \delta]}$.
Inequality \dingthree is the consequence of dropping an indicator,
since in this way we could only be enlarging the set.
Inequality \dingfour follows from \eqref{weightBoundsYgXU}.


	The second statement in the packing lemma follows from the independence of 
	the codewords and the union bound.
 	Let the random codebook $\{X^n(m) \}$, $m \in [1:2^{nR}]$ be 
	generated randomly and independently according to $\prod_{i=1}^n p_{X|U}(x_i|U_i)$.
	Define $\{ \mcal{E}_2(m^\prime|m) \}$ to be the event that
	the channel output when message $m$ is sent,
	$Y^n(m)=\mcal{N}^n(U^n,X^n(m))$ happens  to 
	fall in the conditionally typical set for some other
	codeword $\mathcal{T}^{(n)}_\epsilon(Y|X^n(m^\prime),U^n)$
	and is also output-typical ($ \in \mathcal{T}^{(n)}_\epsilon(Y|U^n)$).
	\be
	\!\!\mcal{E}_2(m^\prime|m)
	\equiv
	\!\! \left\{ \!
	 \left\{ 
		  Y^n(m)
		  \in 
		  \mathcal{T}^{(n)}_\epsilon(Y|X^n(m^\prime),U^n)
	  \right\}
	\! \cap \!
	 \left\{ 
		  Y^n(m)
		  \in 
		  \mathcal{T}^{(n)}_\epsilon(Y|U^n)
	  \right\} \!
	  \right\}\!\!.
	\ee

	If we define $\mathbf{(E2)}$ to be the total probability of misidentifications of this kind,
	we get:
	{\allowdisplaybreaks
 \begin{align*}
 \textrm{Pr}\!\left\{ \mathbf{(E2)}  \right\}
	& =
		 \textrm{Pr}\!\left\{ \bigcup_{m^\prime \in\mathcal{M},m^\prime \neq m }  \mcal{E}_2(m^\prime|m) \right\} \\
	& \overset{\mdingfive}{\leq}
		 \sum_{m^\prime \in\mathcal{M},m^\prime \neq m } \textrm{Pr}\!\left\{   \mcal{E}_2(m^\prime|m) \right\} \\
	& \overset{\mdingsix}{=}
		 \sum_{m^\prime \in\mathcal{M},m^\prime \neq m } \textrm{Pr}\!\left\{   \mcal{E}_2 \right\} \\
	 & \leq
		  \sum_{m^\prime \in\mathcal{M},m^\prime \neq m } 
		  2^{-n[ I(X;Y|U) - \delta(\epsilon) ]} \\
	 &\leq
		 |\mathcal{M}|2^{-n[ I(X;Y|U) - \delta(\epsilon) ]} \\
	 & =
		 2^{-n[ I(X;Y|U) - R - \delta(\epsilon) ]}.
\end{align*}
}
Inequality \dingfive uses the union bound.
Inequality \dingsix is true because the all the codewords of the codebook are picked independently.

Thus if we choose $R < I(X;Y) - \delta(\epsilon)$, the probability
of error will tend to zero as $n\to \infty$.
\end{proof}

\bigskip

The reader is now invited to review the Notation page (\pageref{notation-page}) in the beginning
of the thesis.
This table can be used as a bridge from classical information theory to 
the quantum information theory.
In Appendix~\ref{apdxB}, we will discuss the properties of conditionally typical projectors
and prove a quantum packing lemma which follows \emph{exactly} the same reasoning  
as in the classical packing lemma.



\chapter{Quantum channel coding}
											\label{apdxB}

The first part of this appendix defines the quantum typical subspaces
and conditionally typical projectors associated with a quantum multiple 
access channel problem.
The second part of the appendix is the statement of 
the \emph{quantum packing lemma} which is a direct 
analogue of the classical packing lemma presented in Appendix~\ref{apdx:classical-coding-theorem}.
%


\section{Quantum typicality}
    \label{apdx:quantum-typicality}

	The concepts of entropy, and entropy-typical sets generalize to the quantum 
	setting by virtue of the spectral theorem. 
	Let $\cH^B$ be a $d_B$ dimensional Hilbert space and let 
	$\rho^B  \in \mcal{D}(\cH^B)$ be the density matrix associated with a quantum state.
	The spectral decomposition of $\rho^B$ is denoted
	$\rho^B=U\Lambda U^\dag$ where $\Lambda$ is a diagonal
	matrix of positive real eigenvalues that sum to one.
	We identify the eigenvalues of $\rho^B$ with the probability 
	distribution $p_Y(y)=\Lambda_{yy}$ and write the 
	spectral decomposition as:
	\be
		\rho^B = \sum_{y=1}^{d_B} p_Y(y) \ket{e_{\rho;y}}\bra{e_{\rho;y}}^B
		\label{eq:spectral-decomp-rho}
	\ee
	where $\ket{e_{\rho;y}}$ is the eigenvector of $\rho^B$ corresponding to eigenvalue $p_Y(y)$.
	The von Neumann entropy of the density matrix $\rho^B$ is
	\be
		H(B)_\rho=-\Tr\{\rho^B\log\rho^B\}=H(p_Y).
	\ee

	Define the set of  $\delta$-typical eigenvalues according to the eigenvalue distribution $p_Y$
	\be
		\cT^{n}_{p_Y,\delta} 
		\! \equiv \!
		\left\{ \!
			y^n \in \mcal{Y}^n  \colon \! \left|  -\frac{\log p_{Y^n}(y^n)}{n}  - H(Y) \right| \! \leq \delta 
		\right\}.%
	\ee
	For a given string $y^n = y_1y_2\ldots y_i\ldots y_n$ we define the 
	corresponding eigenvector as 
	\be
		\ket{e_{\rho;y^n}} = \ket{e_{\rho;y_1}} \otimes \ket{e_{\rho;y_2}} \otimes \cdots \otimes \ket{e_{\rho;y_n}},
	\ee
	where for each symbol where $y_i=b \in \{1,2,\ldots,d_B\}$ we select the b$^\textrm{th}$ eigenvector 
	$\ket{e_{\rho;b}}$.
	
	The typical subspace associated with the density matrix $\rho^B$
	is defined as
	\be
	{A}^n_{\rho, \delta}
		= \textrm{span} \!  \{  \ket{e_{\rho;y^n}} \colon y^n \in \cT_{p_Y,\delta}^n  \}.
	\ee
	The typical projector is defined as 
	\be
		\Pi^{n}_{\rho^B, \delta}
		= \sum_{y^n\in\cT_{p,\delta}^n}	    \ket{e_{\rho;y^n}}\! \bra{e_{\rho;y^n}}.
	\ee
	Note that the typical projector is linked twofold to the spectral decomposition
	of \eqref{eq:spectral-decomp-rho}: the sequences $y^n$ are selected 
	according to $p_Y$ and the set of typical vectors are build from tensor
	products of orthogonal eigenvectors $\ket{e_{\rho;y}}$.

	Properties analogous to (\ref{cc1BG}) -- (\ref{cc3BG}) hold.
	For any $\epsilon,\delta>0$, and all sufficiently large $n$ we have
	\vspace{-4mm}
	\begin{eqnarray}
		\label{eqn:apdx-TypP-prop-one}
		&\!\!\!\!\!\!\!\Tr\{\rho^{\otimes n} \Pi^{n}_{\rho, \delta}\}  &\geq  1-\epsilon  \\
		\label{eqn:apdx-TypP-prop-two}
		2^{-n[ H(B)_\rho+\delta]}\Pi^n_{\rho, \delta} 
		\leq  
		&\!\!\Pi^n_{\rho, \delta} \rho^{\otimes n} \Pi^{n}_{\rho, \delta}
		\!\!&
		\leq 
		2^{-n[ H(B)_\rho-\delta ]}\Pi^{n}_{\rho, \delta}, \\
	\label{eqn:apdx-TypP-prop-three}
		[1 - \epsilon] 2^{n[ H(B)_\rho-\delta]}
		\leq  
		&\Tr\{\Pi^{n}_{\rho, \delta}\} &
		\leq 
		2^{n[ H(B)_\rho+\delta]}.
	\end{eqnarray}%
	\noindent
	The interpretation of \eqref{eqn:apdx-TypP-prop-two} is that the eigenvalues 
	of the  state $\rho^{\otimes n}$ are bounded between
	$2^{-n[ H(B)_\rho-\delta ]}$ and $2^{-n[ H(B)_\rho+\delta ]}$
	on the typical subspace ${A}^n_{\rho, \delta}$.
	
	\bigskip 
	
	\noindent
	{\bf Signal states\ \ }
	Consider now a set of quantum states $\{\rho_{x_a}\}$,
	$x_a \in \mcal{X}$.
	We perform the spectral decomposition of each $\rho_{x_a}$
	to obtain 
	\be
		\rho^B_{x_a}= \sum_{y=1}^{d_B} p_{Y|X}(y|x_a) \ket{e_{\rho_{x_a};y}}\bra{e_{\rho_{x_a};y}}^B,
		\label{eq:cond-spectral-decomp-rho}
	\ee
	where $p_{Y|X}(y|x_a)$  is the $y^\textrm{th}$ eigenvalue of  $\rho^B_{x_a}$ and 
	$\ket{e_{\rho_{x_a};y}}$ is the corresponding eigenvector.
	
	We can think of $\{\rho_{x_a}\}$ as a classical-quantum (\emph{c-q}) channel
	where the input is some $x_a \in \mcal{X}$ and the output is the 
	corresponding quantum state $\rho_{x_a}$.
	If the channel is memoryless, then for  each input sequence 
	$x^n=x_1x_2\cdots x_n$ we have the corresponding tensor product output state: 
	\be
		\rho_{x^n}^{B^n} = \rho^{B_1}_{x_1} \otimes \rho^{B_2}_{x_2} \otimes \cdots \otimes \rho^{B_n}_{x_n} = \bigotimes_{i=1}^n \rho^{B_i}_{x_i}.
	\ee
	To avoid confusion with the indices, we use $i\in[n]$ to denote the index of a symbol $x$
	in the sequence $x^n$ and $a \in [1,\ldots,|\mcal{X}|]$ 
	to denote the different symbols in the alphabet~$\mcal{X}$.


	\medskip
	\noindent
	{\bf Conditionally typical projector \ \ }
	Consider the ensemble $\left\{  p_{X}\!\left(  x_a\right)  ,\rho_{x_a}\right\}$.
	The choice of distributions induces the following classical-quantum state:
	\be
		\rho^{XB} =
	       	\sum_{x_a} 
			p_{X}\!\left(x_{a}\right)
			\ketbra{x_a}{x_a}^{X} \!\!
			\otimes \!
			\rho^{B}_{x_a}.
	\ee

	We can now define the conditional entropy of this state as
	\be
		H (B|X)_\rho
		\equiv
		 \sum_{x_a \in \mcal{X}} 
			p_{X}(x_a) 
 			H (\rho_{x_a}),
	\ee	
	or equivalently, expressed in terms of the eigenvalues of the signal states,
	the conditional entropy becomes
	\be
		H (B|X)_\rho
		\equiv
		H(Y|X)
		\equiv
		\sum_{x_a}p_{X}(x_a)H(Y|x_a),
	\ee	
	where $H(Y|x_a) = - \sum_y p_{Y|X}(y|x_a) \log p_{Y|X}(y|x_a)$
	is the entropy of the eigenvalue distribution shown in \eqref{eq:cond-spectral-decomp-rho}.

	We define the $x^n$-conditionally typical projector
	as follows:
	\be
		\Pi^{n}_{\rho^B_{x^n}, \delta}
		=  \sum_{y^n \in \cT^{n}_{\rho^{B^n}_{x^n},\delta} }
			\ket{e_{\rho_{x^n};y^n}}\!\bra{e_{\rho_{x^n};y^n}},
		\label{eqn:cond-typ-projector-def-APDX}
	\ee
	where the set of conditionally typical eigenvalues $\cT^{n}_{\rho^{B^n}_{x^n},\delta}$
	consists of 
	all sequences $y^n$ which satisfy: 
	\be
		\cT^{n}_{\rho^{B^n}_{x^n},\delta} 
		\! \equiv \!
		\left\{ \!
			y^n 
			\colon \! \left|  -\frac{\log p_{Y^n|X^n}(y^n|x^n)}{n}  - H(Y|X) \right| \! \leq \delta 
		\right\},%
	\ee	
	with $p_{Y^n|X^n}(y^n|x^n) = \prod_{i=1}^n p_{Y|X}(y_i|x_i)$.
	
	%
	The states $\ket{e_{\rho_{x^n};y^n}}$ are built from tensor products of eigenvectors
	for the individual signal states:
	\[
		\ket{e_{\rho_{x^n};y^n}} = \ket{e_{\rho_{x_1};y_1}} \otimes \ket{e_{\rho_{x_2};y_2}} \otimes \cdots \otimes \ket{e_{\rho_{x_n};y_n}},
	\]
	where the string $y^n = y_1y_2\ldots y_i\ldots y_n$ varies over different choices of bases for $\mcal{H}^B$.
	For each symbol  $y_i=b \in \{1,2,\ldots,d_B\}$ we select $\ket{e_{\rho_{x_a};b}}$:
	the b$^\textrm{th}$ eigenvector from the eigenbasis of $\rho_{x_a}$ corresponding to the letter $x_i = x_a \in \mathcal{X}$.
	%

	Analogous to the three properties \eqref{eqn:apdx-TypP-prop-one}, \eqref{eqn:apdx-TypP-prop-two}
	and \eqref{eqn:apdx-TypP-prop-three}, the conditionally typical projector obeys:
	\begin{eqnarray}
		&\!\!\!\!\!\!\!
			\ExpX \Tr\!\left[  \rho^B_{X^n} 	\ 	\Pi^{n}_{\rho^B_{X^n}, \delta} 	\right]	\!\!\! &\geq  1-\epsilon  
		\label{eqn:Cond-TypP-prop-one}
		\\
		\!\!\!\!\!\!2^{-n[ H(B|X)_\rho+\delta]} \Pi^{n}_{\rho^B_{x^n}, \delta}
		\leq \!\!\!\!\!\! 
		&\!\!   \Pi^{n}_{\rho^B_{x^n}, \delta} 	\ \rho^B_{x^n} \	\Pi^{n}_{\rho^B_{x^n}, \delta}
		\!\!\!\!\!\!\!\! & 
		\leq 
		2^{-n[ H(B|X)_\rho-\delta ]} \Pi^{n}_{\rho^B_{x^n}, \delta}, 
		\label{eqn:Cond-TypP-prop-two} \\
		{}[1 - \epsilon] 2^{n[ H(B|X)_\rho-\delta]}
		\leq \!\!\!\!\!\!
		&\ExpX \Tr\!\left[  \Pi^{n}_{\rho^B_{X^n}, \delta} \right] \!\!\!\!\!&
		\leq 
		2^{n[ H(B|X)_\rho+\delta]}.
		\label{eqn:Cond-TypP-prop-three}
	\end{eqnarray}%
	
%


	\bigskip
	\noindent
	{\bf MAC code \ \ }
	Consider now a quantum multiple access channel 
	$(\mathcal{X}_1 \times \mathcal{X}_2,  \rho_{x_1,x_2}^{B}, \mathcal{H}^{B} )$
	and two input distributions $p_{X_1}$ and $p_{X_2}$.
	Define the random codebooks $\{X_1^n(m_1)\}_{m_1 \in \cM_1}$ and $\{X_2^n(m_2)\}_{m_2 \in \cM_2}$
	generated from the product distributions $p_{X_1^n}$ and $p_{X_2^n}$ respectively. 
	The choice of distributions induces the following classical-quantum state $\rho^{X_{1}X_{2}B}$
	\be
	       	\sum_{x_a,x_b} 
			p_{X_{1}}\!\left(x_a\right)p_{X_{2}}\!\left( x_b\right)
			\ketbra{x_a}{x_a\!}^{X_1} \!\!
			\otimes \!
			\ketbra{x_b}{x_b}^{X_2} \!
			\otimes \!
			\rho^{B}_{x_ax_b}.
	\ee	        	       
	and the averaged output states:
	\begin{align}
		\bar{\rho}_{x_a} &  \equiv
			\sum_{x_b}p_{X_2}\!\left(  x_b\right)  \rho_{x_a,x_b},\label{eq:rho_xap} \\
		\bar{\rho}_{x_b} &  \equiv
			\sum_{x_a}p_{X_1}\!\left(  x_a\right)  \rho_{x_a,x_b},\label{eq:rho_yap} \\
		\bar{\rho} &  \equiv
			\sum_{x_a,x_b}p_{X_1}\!\left(  x_a\right)  p_{X_2}\!\left(  x_b\right) \rho_{x_a,x_b}.\label{eq:rhoap}
	\end{align}
	
	The conditional quantum entropy $H(B|X_1X_2)_\rho$ is:
	\be
		H (B|X_1X_2)_\rho
		= \hspace*{-5mm}
		 \sum_{x_a \in \mcal{X}_1, x_b \in \mcal{X}_2} 
		 \hspace*{-5mm}
			p_{X_1}(x_a) 
			p_{X_2}(x_b) 
			H (\rho_{x_a,x_b}),
	\ee
	and using the average states we define:
	\begin{align}
		H (B|X_1)_\rho
		&= 
		 \sum_{x_a \in \mcal{X}_1} 
			p_{X_1}(x_a) 
			H (\bar{\rho}_{x_a}),
	\\
		H (B|X_2)_\rho
		&=
		 \sum_{x_b \in \mcal{X}_2} 
			p_{X_2}(x_b) 
			H (\bar{\rho}_{x_b}),
	\\
		H (B)_\rho
		&=
		H(\bar{\rho}).
	\end{align}

	Similarly to equation \eqref{eqn:cond-typ-projector-def-APDX} and for each message pair $(m_1,m_2)$
	we define the conditionally typical projector for the encoded state $\rho^B_{x_1^n(m_1)x_2^n(m_2)}$ to be $\Pi^{n}_{\rho^B_{x_1^n(m_1)x_2^n(m_2)}, \delta}$.
 	From this point on, we will not indicate the messages $m_1$, $m_2$ explicitly, 
	because the codewords are constructed identically for each message.
	
	Analogous to \eqref{eqn:bound-on-size}, the following upper bound applies:
	\be
		\mathbb{E}_{X^n_1X^n_2} \  \Tr\{  \Pi^{n}_{\rho^B_{X_1^nX_2^n}, \delta}  \} \leq 2^{n[H(B|X_1X_2)_\rho + \delta] },
	\ee
	and we can also bound from below the eigenvalues of the state 
	$\rho^B_{x_1^nx_2^n}$ as follows:
	\be
		2^{-n[H(B|X_1X_2)_\rho + \delta] } \Pi^{n}_{\rho^B_{x_1^nx_2^n}, \delta}  
		\leq 
		\Pi^{n}_{\rho^B_{x_1^nx_2^n}, \delta}  
		\rho^B_{x_1^nx_2^n}
		\Pi^{n}_{\rho^B_{x_1^nx_2^n}, \delta}
		\leq
		2^{-n[H(B|X_1X_2)_\rho - \delta] } \Pi^{n}_{\rho^B_{x_1^nx_2^n}, \delta}.
		\label{propertyBBB}
	\ee
	

	We define conditionally typical projectors for each of the
	averaged states:
	\begin{align}
		\bar{\rho}_{x_1} &  \to
			\Pi^{n}_{\bar{\rho}^B_{x_1^n}, \delta},  \\
		\bar{\rho}_{x_2} &  \to
			\Pi^{n}_{\bar{\rho}^B_{x_2^n}, \delta}, \\
		\bar{\rho} &  \to
			\Pi^{n}_{\bar{\rho}^B, \delta}.
	\end{align}
	These projectors obey the standard eigenvalue upper bounds when
	acting on the states with respect to which they are defined:
	\begin{align}
		2^{-n[H(B|X_1)_\rho + \delta] }
		\Pi^{n}_{\bar{\rho}^B_{x_1^n}, \delta}
		\leq		
			\Pi^{n}_{\bar{\rho}^B_{x_1^n}, \delta}
			\bar{\rho}_{x_1^n}
			\Pi^{n}_{\bar{\rho}^B_{x_1^n}, \delta}
		& \leq
			2^{-n[H(B|X_1)_\rho - \delta] }
			\Pi^{n}_{\bar{\rho}^B_{x_1^n}, \delta},  \\
		2^{-n[H(B|X_2)_\rho + \delta] } 
		\Pi^{n}_{\bar{\rho}^B_{x_2^n}, \delta}
		\leq		
			\Pi^{n}_{\bar{\rho}^B_{x_2^n}, \delta}	
			\bar{\rho}_{x_2^n}
			\Pi^{n}_{\bar{\rho}^B_{x_2^n}, \delta}
		& \leq
			2^{-n[H(B|X_2)_\rho - \delta] } 
			\Pi^{n}_{\bar{\rho}^B_{x_2^n}, \delta},  \\			
		2^{-n[H(B)_\rho + \delta] }
		\Pi^{n}_{\bar{\rho}^B, \delta}
		\leq
			\Pi^{n}_{\bar{\rho}^B, \delta} \ 
			\bar{\rho}^B \
			\Pi^{n}_{\bar{\rho}^B, \delta}
		& \leq 
			2^{-n[H(B)_\rho - \delta] }
			\Pi^{n}_{\bar{\rho}^B, \delta}.
	\end{align}	
	
	We have the following bounds on the rank of the conditionally typical projectors:
		\begin{align}
		%
		%
		%
		\Tr\{
			\Pi^{n}_{\bar{\rho}^B_{X_1^n}, \delta} \
		\}
		& \leq 2^{n[H(B|X_1)_\rho + \delta] },						\label{eqn:QMACtypSupAvgOneRANK}\\
		\Tr\{
			\Pi^{n}_{\bar{\rho}^B_{X_2^n}, \delta} \			
		\}
		& \leq 2^{n[H(B|X_2)_\rho + \delta] }, 						\label{eqn:QMACtypSupAvgTwoRANK}	\\
		%
		%
		%
		\Tr\{
			\Pi^{n}_{\bar{\rho}^B, \delta} \
		\}
		& \leq 2^{n[H(B)_\rho + \delta] }.						\label{eqn:QMACtypSupAvgBothRANK}
	\end{align}

	The encoded state $\rho^B_{X_1^nX_2^n}$ is well supported by
	all the typical projectors on average:
	\begin{align}
		\mathbb{E}_{X^n_1X^n_2} \left[
		\Tr\{
			\Pi^{n}_{\rho^B_{X_1^nX_2^n}, \delta} \
			\rho^B_{X_1^nX_2^n}
		\}
		\right]		
		& \geq 1 - \epsilon,						\label{eqn:QMACtypSup} \\
		\mathbb{E}_{X^n_1X^n_2} \left[
		\Tr\{
			\Pi^{n}_{\bar{\rho}^B_{X_1^n}, \delta} \
			\rho^B_{X_1^nX_2^n}
		\}
		\right]
		& \geq 1 - \epsilon,						\label{eqn:QMACtypSupAvgOne}\\
		\mathbb{E}_{X^n_1X^n_2} \left[
		\Tr\{
			\Pi^{n}_{\bar{\rho}^B_{X_2^n}, \delta} \
			\rho^B_{X_1^nX_2^n}			
		\}
		\right]
		& \geq 1 - \epsilon, 						\label{eqn:QMACtypSupAvgTwo}	\\
		\mathbb{E}_{X^n_1X^n_2} \left[
		\Tr\{
			\Pi^{n}_{\bar{\rho}^B, \delta} \
			\rho^B_{X_1^nX_2^n}
		\}
		\right]		
		& \geq 1 - \epsilon.						\label{eqn:QMACtypSupAvgBoth}
	\end{align}
		

\clearpage

	

\section{Quantum packing lemma}
								\label{apdx:q-packing}

	\begin{lemma}
								\label{lem:q-packing}
	
	Let $p_U(u)p_{X|U}(x|u) \in \mcal{P}(\mcal{U},\mcal{X})$ be an arbitrary code distribution,
	and let $\mcal{N} = (\mcal{U}\times\mcal{X}, \rho_{u,x}, \mcal{H}^B)$ be a classical-quantum channel.
	Let $(U^n,X^n,\tilde{X}^n)$ be distributed according to 
	$\prod^n_{i=1} p_{U}(u_i)p_{X|U}(x_i|u_i)p_{X|U}(\tilde{x}_i|u_i)$.		
	%
	Consider the channel $\mcal{N}^\prime$ defined by the following map:
	\be
		\mcal{N}^\prime : (u^n, x^n) 	
		\to	
		\big(u^n, 
		\underbrace{ \rho^{B_1}_{u_1,x_1}\otimes\rho^{B_2}_{u_2,x_2}\otimes \cdots 
		\otimes \rho^{B_n}_{u_n,x_n} }_{
		\rho_{u^n,x^n}^{B^n} } \big),
	\ee
	where $u^n$ 
	is available as side information to the receiver and the sender.
	Define the state $\bar{\rho}_{u^n}=\mathbb{E}_{X^n|u^n} \mcal{N}^\prime(u^n,X^n)$
	and the conditionally typical projectors $\Pi^{B^n}_{ \bar{\rho}_{u^n} }$ for the state $\bar{\rho}^{B^n}_{u^n}$
	and $\Pi^{B^n}_{ \rho_{u^n,x^n} }$ for the state $\rho^{B^n}_{u^n,x^n}$.
	
	We want to measure the expectation of the overlap between $\rho^{B^n}_{U^n, \tilde{X}^n}$
	and the operator $\Pi^{B^n}_{ \bar{\rho}_{U^n} } 	\Pi^{B^n}_{ \rho_{U^n,X^n} }  \Pi^{B^n}_{ \bar{\rho}_{U^n} }$
	associated with some $(U^n,X^n)$.
	We define this quantity to be:
	\be
		\mcal{E}_2	= \Tr\!\left[  \Pi^{B^n}_{ \bar{\rho}_{u^n} } 	\Pi^{B^n}_{ \rho_{u^n,x^n} }  \Pi^{B^n}_{ \bar{\rho}_{u^n} } \ \rho^{B^n}_{U^n, \tilde{X}^n} \right].
	\ee
	Then $\mcal{E}_2$ can be bounded as follows:
	\begin{align}
	\ExpUn
	\NExpXgU  \ 
	 \NExpXtildegU \
	\mcal{E}_2
	&\ \leq \ \ 		  
		2^{-n[ I(X;B|U) - \delta(\epsilon) ]}.
	\end{align}


 	Let the random codebook $\{X^n(m) \}$, $m \in [1:2^{nR}]$ be 
	generated randomly and independently according to $\prod_{i=1}^n p_{X|U}(x_i|U_i)$.
	%
	%
	Then there exists $\delta(\epsilon) \to 0$ as $\epsilon \to 0$ such that the
	 expectation of the total overlap between conditionally typical output spaces 
	 can be bounded from above as follows:
	\begin{align}
	\mathrm{(E2)}
	& \equiv \!\!\!\!\!
		\sum_{m^\prime \in\mathcal{M},m^\prime \neq m } \!\!\!\!
		\ExpUn
		\mathop{\mathbb{E}}_{X^n(m)|U^n} 
		\mathop{\mathbb{E}}_{X^n(m^\prime)|U^n} 
		\!\!\!\Tr\!\left[
			\Pi^{B^n}_{ \bar{\rho}_{U^n} }
			\ 
			\Pi^{B^n}_{ \rho_{U^n,X^n(m^\prime)} } 
			\
			\Pi^{B^n}_{ \bar{\rho}_{U^n} }
			\ 
			\rho^{B^n}_{U^n, X^n(m)}
		\right] \nonumber \\
	 &\leq
		 |\mathcal{M}|2^{-n[ I(X;Y|U) - \delta(\epsilon) ]}. \label{qpack-lem-bd}
	\end{align}
	%
	Thus if we choose $R < I(X;B|U) - \delta(\epsilon)$, the quantity $\mathrm{(E2)}$ will  
	tend to zero as $n\to \infty$.

	\end{lemma}

	\begin{figure}[hbt]
	\begin{center}
	{ \scriptsize
	\begin{tikzpicture}[node distance=1.7cm,>=stealth',bend angle=45,auto]
	  \tikzstyle{cnode}=[isosceles triangle,isosceles triangle apex angle=60,thick,draw=blue!75,fill=blue!20,minimum height=0.3cm] 
	  \tikzstyle{processing}=[rectangle,thick,draw=black,fill=white,minimum height=1.1cm] 
	  \tikzstyle{measurment}=[rectangle,thick,draw=blue!75,fill=blue!20,minimum height=2cm] 
	  \tikzstyle{qnode}=[circle,thick,draw=blue!75,fill=blue!20,minimum size=6mm] 
	  \tikzstyle{every label}=[black, font=\footnotesize]

	  \begin{scope}
		\node [cnode] (Un) [label=left:$\prod^n p_U$]                {$U^n$};
		\node [cnode] (Xn) [above right of=Un,xshift=1cm]                {$X^n$}
				edge  [pre]				node[swap]  	{$\ \ \ \ \prod^n p_{X|U}$}	(Un);
		\node [cnode] (tildeXn) [below right of=Un,xshift=1cm]                {$\tilde{X}^n$}
				edge  [pre]				node[]  	{$\ \ \ \ \ \prod^n p_{X|U}$}	(Un);
		\node [qnode] 
		(tildeYn) [ label=below:$\ \ \in \mcal{H}^{B^n}$, right of=tildeXn, xshift=1cm]                {$\rho^{B^n}_{U^n,\tilde{X}^n}$}
				edge  [pre]				node[swap]  	{$\prod^n p_{Y|XU}$}	(Un);	
		\node [] (abovetildeYn) [right of=Un, xshift=3cm]                {};
		\draw[->]	(tildeXn)	-- (tildeYn);		
		\node [measurment] 
		(Sandwich) [ 
						right of=Un, xshift=5.8cm, yshift=0.5cm]                {$\Pi_{ \bar{\rho}_{U^n}}\ \Pi_{\rho_{U^n\!\!,X^n}} \ \Pi_{ \bar{\rho}_{U^n}}$};
		\draw[->,thick]	(tildeYn)	-- (Sandwich.west);	
		\node [] 
		(pERR) [ 
						right of=Sandwich, xshift=0.3cm]                {$\mcal{E}_2 $}
						edge  [pre]				node[swap]  	{}	(Sandwich);	
	  \end{scope}
	  \begin{pgfonlayer}{background}
	    \filldraw [line width=4mm,join=round,black!10]
	      ([xshift=-4mm,yshift=7mm]tildeXn.north -| tildeXn.east) rectangle ([xshift=4mm]tildeYn.south -| tildeYn.west);
	    \filldraw [line width=4mm,join=round,black!10]
	      ([xshift=-4mm]Un.north -| Un.east) rectangle ([xshift=4mm,yshift=-4mm]abovetildeYn.south -| tildeYn.west);
	  \end{pgfonlayer}
	\end{tikzpicture}

	}

	\caption{The quantum packing lemma. Two random codewords $X^n$ and $\tilde{X}^n$
	are drawn randomly and independently conditional on a third 
	random variable $U^n$. 
	Assume that the random variable $U^n$ is also available at the receiver.
	What is the chance that the output of the channel which corresponds
	to $\tilde{X}$ and $U^n$ will falsely be recognized to be in the set
	of outputs which are likely to come from inputs $X^n$ and $U^n$?
	}
	\label{fig:quantum-pack-illustration}
	\end{center}
	\end{figure}

	To bound the expectation of the second term, define
	$\tilde{X}(m)$ and $X^n(m^\prime)$ to be the two 
	random codewords assigned to messages $m$ and $m^\prime$
	respectively.

	{\allowdisplaybreaks
	\begin{align*}
	\ExpUn
	\NExpXgU  \ 
	 \NExpXtildegU \
	\mcal{E}_2
	 & = 
		\ExpUn
		\NExpXgU  
		\NExpXtildegU
			\text{Tr}
			\left[  
			\Pi^{B^n}_{ \bar{\rho}_{U^n} }
			\Pi^{B^n}_{ \rho_{U^n,X^n} } 
			\Pi^{B^n}_{ \bar{\rho}_{U^n} }
			\
			\prhoMnew
			\right] \\
	 & = 
	\ExpUn
	\NExpXgU  
		\text{Tr}
		\left[  
		\Pi^{B^n}_{ \bar{\rho}_{U^n} }
		\Pi^{B^n}_{ \rho_{U^n,X^n} } 
		\Pi^{B^n}_{ \bar{\rho}_{U^n} }
		\NExpXtildegU
		\{
		\prhoMnew
		\}
		\right] \\
	 & \overset{\mdingone}{=} 
	\ExpUn
	\NExpXgU  
		\text{Tr}
		\left[  
		\Pi^{B^n}_{ \bar{\rho}_{U^n} }
		\Pi^{B^n}_{ \rho_{U^n,X^n} } 
		\Pi^{B^n}_{ \bar{\rho}_{U^n} }
		\bar{\rho}_{U^n}
		\right] \\
	 & = 
	\ExpUn
	\NExpXgU  
		\text{Tr}
		\left[  
		\Pi^{B^n}_{ \rho_{U^n,X^n} }  \
		\Pi^{B^n}_{ \bar{\rho}_{U^n} } 
		\bar{\rho}_{U^n}
		\Pi^{B^n}_{ \bar{\rho}_{U^n} }
		\right] \\
	& \overset{\mdingtwo}{\leq} 	 
		2^{-n[H(B|U)-\delta] }
		\ExpUn
		\NExpXgU  
			\text{Tr}
			\left[  
			\Pi^{B^n}_{ \rho_{U^n,X^n} }  \
			\Pi^{B^n}_{ \bar{\rho}_{U^n} }
			\right] \\
	& \overset{\mdingthree}{\leq} 	 
	 	2^{-n[H(B|U)-\delta] }
		\ExpUn
		\NExpXgU  
		\text{Tr}
		\left[  
			\Pi^{B^n}_{ \rho_{U^n,X^n} }  \
		\right] \\
	& \overset{\mdingfour}{\leq} 	 
	 	2^{-n[H(B|U)-\delta] }
		2^{n[H(B|U,X)+\delta ]} \\
	& =
		2^{-n[ I(X;Y|U) - \delta(\epsilon) ]}.
	\end{align*}
	}

	%
	Equation \dingone is true by the definition $
		\NExpXtildegU \{
		\prhoMnew 		\}
		= 		\bar{\rho}_{U^n}$.
	The inequality \dingtwo uses the eigenvalue bound as in \eqref{eqn:Cond-TypP-prop-two}.
	The inequality \dingthree follows from 
	\begin{align*}
		\Tr\!\left[
			\Pi^{B^n}_{ \rho_{U^n,X^n} }  \
			\Pi^{B^n}_{ \bar{\rho}_{U^n} }
		\right]
		& = 
			\Tr\!\left[
			\Pi^{B^n}_{ \rho_{U^n,X^n} }  \
			\Pi^{B^n}_{ \bar{\rho}_{U^n} } \
			\Pi^{B^n}_{ \rho_{U^n,X^n} }  
			\right] \\
		& \leq 
			\Tr\!\left[
				\Pi^{B^n}_{ \rho_{U^n,X^n} }  
				\ I \
				\Pi^{B^n}_{ \rho_{U^n,X^n} }  
			\right] \\
		& = 
			\Tr\!\left[
				\Pi^{B^n}_{ \rho_{U^n,X^n} }  
			\right].
	\end{align*}
	The inequality \dingfour follows from bound
	on the expected rank of the conditionally typical projector 
	like in \eqref{eqn:Cond-TypP-prop-three}.

\subsection*{Applications}


\subsubsection{Holevo-Schumacher-Westmoreland (HSW) Theorem}
Given a channel $(\mcal{X}, \rho_x, \mcal{H})$, if we set:
\begin{itemize}
	\item 	$U = \emptyset$
	\item 	$p_U(u)p_{X|U}(x|u) = p_X(x)$
	\item 	$\rho_{u^n,x^n} =  \rho_{x^n}$
	\item 	$\Pi^{B^n}_{ \bar{\rho}_{u^n} } 	\Pi^{B^n}_{ \rho_{u^n,x^n} }  
			   \Pi^{B^n}_{ \bar{\rho}_{u^n} } =    \Pi_{\bar{\rho}} \Pi_{\rho_{x^n}} \Pi_{\bar{\rho}}$,
\end{itemize}
then the quantum packing lemma tells us how many conditionally typical
subspaces we can \emph{pack} inside the output-typical subspace before
they 
start to overlap too much.

\subsubsection{Successive decoding for the quantum multiple access channel}
Given a quantum multiple access channel $(\mcal{X}_1\times \mcal{X}_2, \rho_{x_1,x_2}, \mcal{H})$, 
we set:
\begin{itemize}
	\item 	$U = X_1$
	\item 	$p_U(u)p_{X|U}(x|u) = p_{X_1}(x_1)p_{X_2}(x_2)$
	\item 	$\rho_{u^n,x^n} =  \rho_{x_1^n,x_2^n}$
	\item 	$\Pi^{B^n}_{ \bar{\rho}_{u^n} } 	\Pi^{B^n}_{ \rho_{u^n,x^n} }  
			   \Pi^{B^n}_{ \bar{\rho}_{u^n} } =    \Pi_{\bar{\rho}_{x_1^n}} \Pi_{\rho_{x_1^n,x_2^n}} \Pi_{\bar{\rho}_{x_1^n}}$,
\end{itemize}
to obtain the bound on the rate $R_2$ when using the successive decoding $m_1 \to m_2|m_1$.
%

\subsubsection{Superposition coding}

Consider the situation in which superposition encoding is used
to encode two messages $\ell$ and $m$ in a codebook
suitable for the channel $(\mcal{X}, \rho_x, \mcal{H})$:
\[
 \{ W^n(\ell) \} \sim p_{W^n}(w^n), 
 \qquad 
 \{ X^n(\ell,m) \} \sim \prod_{i=1}^n p_{X|W}\!\left( x_i|w_i(\ell) \right).
\]
Consider the following substitutions: 
\begin{itemize}
	\item 	$U = W$
	\item 	$p_U(u)p_{X|U}(x|u) = p_{W}(w)p_{X|W}(x|w)$
	\item 	$\rho_{u^n,x^n} =  \rho_{x^n}$
	\item 	$\Pi^{B^n}_{ \bar{\rho}_{u^n} } 	\Pi^{B^n}_{ \rho_{u^n,x^n} }  
			   \Pi^{B^n}_{ \bar{\rho}_{u^n} } =    \Pi_{\bar{\rho}_{w^n}} \Pi_{\rho_{x^n}} \Pi_{\bar{\rho}_{w^n}}$.
\end{itemize}
%
The packing lemma gives us a bound on the error associated with decoding a wrong 
message $m$ (the satellite message) given that we correctly decoded $\ell$ (the cloud center).



\chapter{Miscellaneous proofs}
								\label{appendix:misc-proofs}

This appendix  contains a series of proofs which were omitted 
from the text in Section~\ref{sec:QCMGvia2MAC} in order to make it more readable.

\section{Geometry of Chong-Motani-Garg rate region}			\label{apdx:CMG-MAC-geom}

We will now prove the inequalities from Lemma~\ref{lemma:geometry-of-MAConeR}
on the geometry of  $\MAConeR\cNpCMG$, the multiple access channel for Receiver~1
in the Chong-Motani-Garg coding strategy.
This inequality structure is important for the geometrical observations 
of the \Sasoglou \ argument.

\begin{proof}[Proof of Lemma~\ref{lemma:geometry-of-MAConeR}]
	If we expand the shorthand notation of equations \eqref{polymatr1} through 
	\eqref{polymatr3} we obtain the following inequalities.
	{\small
	\begin{align}
		I(X_1; B_1|W_1W_2Q)  	&\leq I(X_1; B_1|W_2Q)  \ \ \  \ \ \leq I(X_1W_2; B_1|Q), 	
			\label{exppolymatr1}\\
		I(X_1; B_1|W_1W_2Q) 	&\leq I(X_1W_2; B_1|W_1Q) \ \leq I(X_1W_2; B_1|Q), 
			\label{exppolymatr2}	 \\
		\!\!\!I(X_1; B_1|W_1W_2Q) + I(X_1W_2; B_1|Q) &\leq I(X_1; B_1|W_2Q) + I(X_1W_2; B_1|W_1Q). 
			\label{exppolymatr3}
	\end{align}
	}

	Observe that $W_2$ is independent from $W_1$ and $X_1$ thus
	\be
		H(X_1W_2) = H(X_1) + H(W_2), \quad 
		H(W_1W_2) = H(W_1) + H(W_2).
	\ee 
	Also, since $X_1$ is obtained from $W_1$, we have $H(X_1) = H(X_1W_1)$
	and we can add or subtract the random variable $W_1$ next to $X_1$ as needed
	without changing the entropy.
	
	The get the first part of the inequality \eqref{polymatr1}, we observe 
	\begin{align*}
		I(X_1; B_1|W_1W_2) 	&= I(X_1; B_1W_2|W_1)  \\ 
							&= H(X_1W_1) + H(B_1W_2W_1) - H(X_1B_1W_2W_1) - H(W_1) \\
							&\ \ \ \ \ 			   -H(W_1W_2)					      +H(W_1W_2) \\
							&= H(X_1) +  [H(B_1W_2W_1)-H(W_1W_2)]  - H(X_1B_1W_2W_1)  \\
								& \qquad \qquad \qquad - H(W_1) +H(W_1) +H(W_2) \\
							&\leq H(X_1)+  [H(B_1W_2)-H(W_2)]  - H(X_1B_1W_2W_1) + H(W_2)   \\		
							&= [H(X_1) + H(W_2)] + H(B_1W_2)  - H(X_1B_1W_2W_1)  -H(W_2) \\																&= I(X_1; B_1|W_2),
	\end{align*}
	where inequality follows from $H(B_1|W_1W_2) \leq H(B_1|W_2)$ (conditioning cannot increase entropy).
	
	The second part of inequality \eqref{polymatr1}, follows from a similar observation using $H(B_1|W_2) \leq H(B_1)$.
	\begin{align*}
		I(X_1; B_1|W_2)	 	&= H(X_1W_2) + H(B_1W_2) - H(X_1B_1W_2) - H(W_2) \\
							&= H(X_1W_2) +  [H(B_1W_2)-H(W_2)]  - H(X_1B_1W_2) \\
							&\leq H(X_1W_2)+  [H(B_1)]  - H(X_1B_1W_2)   \\		
							&= I(X_1W_2; B_1).
	\end{align*}
	
	For the first part of \eqref{polymatr2} we repeat the above argument but with extra conditioning
	on the $W_1$ system.
	\begin{align*}
		I(X_1; & B_1 |W_1W_2)  = \\
							&=  H(X_1W_1W_2) + H(B_1W_1W_2) - H(X_1B_1W_1W_2) - H(W_1W_2) \\
							&= H(X_1W_1W_2) +  [H(B_1W_1W_2)-H(W_1W_2)]  - H(X_1B_1W_1W_2) \\
							&\leq H(X_1W_2)+  [H(B_1|W_1)]  - H(X_1B_1W_1W_2)   \\		
							&= 	 H(X_1W_2)+  H(B_1W_1)  - H(X_1B_1W_1W_2)  - H(W_1)  \\		
							&= 	 H(X_1W_1W_2)+  H(B_1W_1)  - H(X_1B_1W_1W_2)  - H(W_1)  \\		
							&= I(X_1W_2; B_1|W_1).
	\end{align*}			
	
	For the second part of \eqref{polymatr2} we have
	\begin{align*}
		I(X_1W_2; B_1|W_1) 	&= 		H(X_1W_1W_2)+  H(B_1W_1)  - H(X_1B_1W_1W_2)  - H(W_1)  \\
							&= 		H(X_1W_2)+  [H(B_1W_1)- H(W_1)]  - H(X_1B_1W_2)    \\
							&\leq		H(X_1W_2)+  H(B_1)  - H(X_1B_1W_2)    \\
							&= 		I(X_1W_2; B_1).
	\end{align*}
	
	Finally for inequality \eqref{polymatr3} we need to use the strong subadditivity relation
	\be
		H(B_1W_1W_2) + H(B_1) 		\leq		H(B_1W_1) + H(B_1W_2).
	\ee
	The steps are 
	{\allowdisplaybreaks \small
	\begin{align*}
		& \hspace{-1cm} I(X_1; B_1|W_1W_2) + I(X_1W_2; B_1) = \\
							& =		H(X_1W_1W_2) + H(B_1W_1W_2) - H(X_1B_1W_1W_2) - H(W_1W_2) \\
							&	\qquad + 	H(X_1W_2)+  H(B_1)  - H(X_1B_1W_2) \\
							&=		[H(B_1W_1W_2)+  H(B_1)] +   H(X_1W_1W_2) - H(X_1B_1W_1W_2) - H(W_1) -H(W_2) \\
							&	\qquad + 	H(X_1W_2)  - H(X_1B_1W_2) \\
							&\leq		[H(B_1W_1)+  H(B_1W_2)] +   H(X_1W_1W_2) - H(X_1B_1W_1W_2) - H(W_1) -H(W_2) \\
							&	\qquad + 	H(X_1W_2)  - H(X_1B_1W_2) \\ 
							%
							%
							&=		 H(X_1W_1W_2)+  H(B_1W_1)  - H(X_1B_1W_1W_2)  - H(W_1) \\
							&	\qquad  + H(X_1W_2) + H(B_1W_2) - H(X_1B_1W_2) - H(W_2) \\
							%
							&=		I(X_1W_2; B_1|W_1) + I(X_1; B_1|W_2).
	\end{align*}
	}
			
	This completes the proof of Lemma \ref{lemma:geometry-of-MAConeR}.
\end{proof}

	\section{Detailed explanation concerning moving points} 		
												\label{sec:sasoglou-moving-bd-points}

		In Section~\ref{subsec:sasoglu-argument} we used 
		 Lemma~\ref{lem:sasoglou-moving-bd-points}
		to show that we can move any point on the (b) or (d) planes
		to an equivalent point on the (a) or (c) planes.
		We now give the proof.

		\begin{proof}
			We have to show how to move any point in 
			$b_i \cup d_i \setminus a_i \cup c_i$ to an equivalent 
			point in $a_i \cup c_i$. 
			Because the rates $R_{1c}$ and $R_{2c}$ appear in the coordinates of both $P_1$ and $P_2$,
			we cannot move each point independently.
			Indeed \Sasoglou \ points out that the points $P_1$ and $P_2$ are \emph{coupled} by the
			common rates.
			
			A priori, we have to consider all possible starting combinations the points
			However, using the following observations we can restrict the number of possibilities 
			significantly. 
			\begin{enumerate}
			
				\item		If $P_1 \in b_1 \setminus a_1$, then $P_2 \in a_2 \cup b_2$. \\
						\noindent The fact that $P_1 \in b_1 \setminus a_1$ implies that equation (b1) is tight
						\be
							R_{1p}+R_{1c} = I(b_1), 
						\ee
						and (a1) is loose
						\be
							R_{1p}   <  I(a_1). 
						\ee
						Then there exists $\delta > 0$ such that the point  
						$P'_1 = (R_{1p} + \delta, R_{1c} -\delta, R_{2c}) \in  \MAConeR(p)$. 
						Suppose for a contradiction that $P_2$ was originally in $(c_2 \cup d_2) \setminus (a_1 \cup b_1)$.
						The decrease in $R_{1c}$ associated with the move from $P_1$ to $P'_1$,
						will have allowed us to increase the one of the rates for Receiver~2 which 
						is a contradiction since we assumed the $R_2 = R_{2c}+R_{2p}$ was optimal. 
						More specifically, if $P_2 \in  c_2$, or $P_2 \in d_2$, 
						then we would be allowed to increase $R_{2p}$ by $\delta$,
						to obtain 
						$P'_2 = (R_{2p}+\delta, R_{2c}, R_{1c}-\delta)$, 
						resulting in the operating point $(R_1, R_2 + \delta)$ which contradicts the 
						assumption that the initial rate pair $(R_1,R_2)$ was on the boundary of  $\CMGR$.
						Thus, if $P_1 \in b_1 \setminus a_1$, then $P_2$ must be in $a_2 \cup b_2$.

				\item		If $P_1 \in d_1 \setminus (a_1 \cup b_1 \cup c_1)$ then $P_2 \in a_2$. \\
						Again consider moving the rates to obtain 
						$P'_1 = (R_{1p} + \delta, R_{1c} -\delta, R_{2c}) \in d_1 \setminus (a_1 \cup b_1 \cup c_1)$,
						then if then if $P_2$ was originally in $c_2$ or $d_2$, then 
						the decrease in $R_{1c}$ would allow us to move the point $P_2$ to a new 
						rate triple $P'_2 = (R_{2p}+\delta, R_{2c}, R_{1c}-\delta)$, 
						resulting in the operating point $(R_1, R_2 + \delta)$, which again leads to a contradiction.
						Therefore $P_2$ can only be in $a_2$ or $b_2$. 
						But if $P_2$ were in $b_2$, then by observation 1 (with a change of roles between $P_1$ and $P_2$) 
						we would have $P_1 \in  (a_1 \cup b_1)$ which contradicts our assumption that 
						$P_1 \in d_1 \setminus (a_1 \cup b_1 \cup c_1)$.
						Thus we see that if $P_1 \in d_1 \setminus (a_1 \cup b_1 \cup c_1)$,
						then 	$P_2 \in  a_2$.
			\end{enumerate}
			
			
			By the above reasoning we have restricted the possible combinations where the points
			$(P_1, P_2)$ could lie initially.
			To prove Theorem \ref{lem:sasoglou-moving-bd-points}, we have to show that we can 
			deal with the following combinations:
			$b_1 \times a_2$, $a_1 \times b_2$, $b_1 \times b_2$, $d_1 \times a_1$ and $a_1 \times d_2$.

			\bigskip

			 We now show that we can move any point $P_1 \in b_1 \cup d_1$ 
			 (on one of the bad planes) to an equivalent point  lying in $a_1 \cup c_1$,
			%
			
			\begin{itemize}
		
			\item 	{\bf Case $(P_1, P_2) \in b_1  \times a_2$:}  \\
					In this case, equations  (b1)  and (a2) are tight which means that the rate pairs
					are of the form
					\bea
						P_1 &= & (R_{1p},R_{1c},R_{2c}), \text{such that } R_{1p}+R_{1c}=I(b_1), \nonumber \\
						P_2 &=& (R_{2p},R_{2c},R_{1c})= (I(a_2),R_{2c},R_{1c}). \nonumber
					\eea

					If we apply a $R_{1c} \to R_{1p}$  rate moving operation to $P_1$ we can obtain a new
					point $P'_1$ with 
					\[
						P'_1  =  (R'_{1p},R'_{1c},R_{2c}) = ( I(a_1), I(b_1)-I(a_1), R_{2c})  \in a_1 \cap b_1.
					\]
					As a result of the moving the point $P_2$ will be moved to
					\be
						P'_2 = (R_{2p},R_{2c},R'_{1c})= (I(a_2),R_{2c}, I(b_1)-I(a_1) ), \nonumber
					\ee
					which continues to lie in the $a_2$ plane.
					%
					Observe that during this rate moving operation the sum rates remain unchanged 
					$(R_{1p}+ R_{1c},  R_{2p} + R_{2c} ) = (R_1 , R_2)=(R_{1p}^{\prime} + R_{1c}^{\prime} , 
					  R_{2p}^{\prime} + R_{2c}^{\prime} )$.
					 
					The case when $(P_1, P_2) \in a_1  \times b_2$ is analogous.

			\item		{\bf Case $(P_1, P_2) \in b_1  \times b_2$:}  \\
					Our starting points are  
					\bea
						P_1 &=& (R_{1p},R_{1c},R_{2c}), \text{such that } R_{1p}+R_{1c}=I(b_1), \nonumber \\
						P_2 &=& (R_{2p},R_{2c},R_{1c}),  \text{such that } R_{2p}+R_{2c}=I(b_2). \nonumber
					\eea
					
					We will first do a $R_{1c} \to R_{1p}$ rate moving operation until we get to the
					plane $a_1$. The points we obtain are
					\bea
						P^{\prime}_1 &=& (R^{\prime}_{1p},R^{\prime}_{1c},R_{2c}) 
									     = ( I(a_1), I(b_1)-I(a_1), R_{2c})  \in a_1 \cap b_1, \nonumber \\
						P^{\prime}_2 &=& (R_{2p},R_{2c},R^{\prime}_{1c})
										= (R_{2p},R_{2c},I(b_1)-I(a_1)) \in b_2. \nonumber
					\eea
					
					We then perform second rate moving operation $R_{2c} \to R_{2p}$ in order
					to move to the plane $a_2$. 
					\bea
						P^{\prime\prime}_1 &=& (R^{\prime}_{1p},R^{\prime}_{1c},R^{\prime\prime}_{2c}) 
									     = ( I(a_1), I(b_1)-I(a_1),I(b_2)-I(a_2))  \in a_1 \cap b_1, \nonumber \\
						P^{\prime\prime}_2 &=& (R^{\prime\prime}_{2p},R^{\prime\prime}_{2c},R^{\prime}_{1c})
										= (I(a_2),I(b_2)-I(a_2),I(b_1)-I(a_1)) \in a_1 \cap b_2. \nonumber
					\eea				
					
					
					Thus we have managed to move the points $(P_1,P_2) \in b_1 \times b_2$
					to equivalent points $(P_1^{\prime\prime}, P_2^{\prime\prime}) \in a_1 \times a_2$
					while leaving the sum rate 
					$(R_1, R_2)$ unchanged.

			\item		{\bf Case $(P_1, P_2) \in d_1  \times a_2$:}  \\
					If $P_1 \in d_1$, it means that the triple sum inequality (d1) is tight. 
					The starting rates are  
					\bea
						P_1 &=& (R_{1p},R_{1c},R_{2c}), \text{ such that } R_{1p}+R_{1c}+R_{2c}=I(d_1), \nonumber \\
						P_2 &=& (I(a_2),R_{2c},R_{1c}) \in a_2. \nonumber
					\eea
					
					To move $P_1$ away from the interior of the $d_1$ plane we will once again
					use a rate moving operation  $R_{1c} \to R_{1p}$.
					This operation will increase the rate $R_{1p}$ at the expense of the rate $R_{1c}$.
					We cannot increase the rate $R_{1p}$ indefinitely -- sooner or later one of the 
					two other rate constraints on $R_{1p}$ will saturate.
					
					The other constraints on $R_{1p}$ come from equations (a1) and (c1), 
					so by rate moving we will eventually reach either the $a_1$ or the $c_1$ planes.
					
					If the first case the resulting points will be
					\bea
						P^{\prime}_1 &=& (R^{\prime}_{1p},R^{\prime}_{1c},R_{2c}) 
									     = ( I(a_1),R^{\prime}_{1c}, R_{2c})  \in a_1 \cap d_1, \nonumber \\
						P^{\prime}_2 &=& (I(a_2),R_{2c},R^{\prime}_{1c}) \in a_2, \nonumber
					\eea
					where $R^{\prime}_{1c} = I(d_1)-I(a_1)-R_{2c}$ because by rate moving we stayed
					in the $d_1$ plane.
					
					In the latter case where moving the rates of $P_1 \in d_1$ puts us on the
					$c_1$ plane the resulting points will be
					\bea
						P^{\prime}_1 &=& (R^{\prime}_{1p},R^{\prime}_{1c},R_{2c}) \in c_1 \cap d_1, 
									     \text{ s.t. } R^{\prime}_{1p}+R_{2c} = I(c_1) \nonumber \\
						P^{\prime}_2 &=& (I(a_2),R_{2c},R^{\prime}_{1c}) \in a_2. \nonumber
					\eea
					
					Once again, the sum rate $(R_1, R_2)$ remains unchanged 
					by the rate moving, but the moved points $(P_1^{\prime}, P_2^{\prime})$ 
					are now either in $a_1 \times a_2$ or $c_1 \times a_2$ as claimed.
					
					The case when $(P_1, P_2) \in a_1  \times d_2$ is analogous.

			\end{itemize}
							
			Therefore, given an arbitrary point $(R_1, R_2) \in  \delCMGR\cNpCMG$, 
			there always exists a choice of common/private rates  
			such that $(P_1, P_2) \in a_1 \cup c_1 \times a_2 \cup c_2$ 
			with $(R_{1p}+R_{1c} ,R_{2p}	+R_{2c} )	=(R_1 ,R_2 )$.	
			
		\end{proof}


	\section{Redundant inequality}				\label{apdx:redundant-ineq}
						
		In Section~\ref{subsec:two-simult-decoding-for-ac}, we claimed that
		the inequality \eqref{eqn:ctightb2} is less tight than the 
		sum rate constraint obtained by adding equations 
		\eqref{eqn:ctighta2} and \eqref{eqn:ctightd2}.

		To that this is true, consider the following argument starting from the positivity 
		of the mutual information $I(W_1;W_2|B_1) \geq 0$:

		\begin{align}
			H(W_1W_2B_1)  + H(B_1)  & \leq H(W_1B_1) + H(W_2B_1).
		\end{align}
		We now add $H(X_1W_1W_2)$ and
		subtract $-H(X_1W_1W_2B_1)$ on both sides of the equation:%
		{\small
			\begin{align}
			\begin{array}[c]{c}%
				H(W_1W_2B_1)  + H(B_1) + H(X_1W_1W_2)  \\
				 - H(X_1W_1W_2B_1) 
			\end{array}
			& \leq 
			\begin{array}[c]{c}%
				H(W_1B_1) + H(W_2B_1) +H(X_1W_1W_2) \\
				 \ \ \ \  - H(X_1W_1W_2B_1). 
			\end{array} \nonumber
		\end{align}}%
		We now use the fact that $W_2$ is independent from $W_1$,
				so $H(W_1)-H(W_1W_2) = - H(W_2)$ to obtain:
		{ \small
		\begin{align}				
			\begin{array}[c]{c}%
				H(W_1W_2B_1)  + H(B_1) + H(X_1W_1W_2) \\
				 - H(X_1W_1W_2B_1) + H(W_1)-H(W_1W_2)
			\end{array}
			& \leq 
			\begin{array}[c]{c}%
				H(W_1B_1) + H(W_2B_1) +H(X_1W_1W_2) \\
				 \ \ \ \  - H(X_1W_1W_2B_1) - H(W_2).
			\end{array}			\nonumber
		\end{align}		
		}
		We move the term $H(W_1B_1)$ to the other side
		and rearrange the terms the final expression:
		{ \small
		\begin{align}		
			&
			\begin{array}[c]{c}%
				\!\!\! H(X_1W_1W_2) + H(W_1W_2B_1)  - \! H(X_1W_1W_2B_1)  - \!H(W_1W_2) \\
				  + H(W_1) + H(B_1)- H(W_1B_1) 
			\end{array}  \nonumber 
			\\
			& \qquad \qquad \qquad \qquad \qquad \qquad\qquad \leq 
			\begin{array}[c]{c}%
				H(X_1W_1W_2) + H(W_2B_1)  \\
					\quad - H(X_1W_1W_2B_1) - H(W_2)
			\end{array} \nonumber
		\end{align}
		\begin{align}			
			I(a_1) = I(X_1;B_1|W_1W_2) + I(W_1;B_1) \ \ 
			&\leq \ \ 		
			I(X_1;B_1|W_2) = I(b_1),	\nonumber			
		\end{align}
		}
		which shows that we can drop the constraint
		from equation \eqref{eqn:ctightb2}.

\pagestyle{plain}

\singlespacing
\bibliographystyle{alpha}
{\small
\bibliography{zz.references}
}

\end{document}